\documentclass[a4paper,12pt,reqno]{amsart}

\usepackage{amssymb}
\usepackage{mathtools}
\usepackage{graphicx,subfigure} 
\usepackage[rflt]{floatflt}
\usepackage{wrapfig}
\usepackage{array}
\usepackage{dsfont}

\usepackage[dvips,colorlinks=true,linkcolor=blue,urlcolor=red]{hyperref}
\usepackage[english]{babel}
\newcommand{\laplace}{\triangle}

\def \bbR{\mathbb{R}}
\def \bbP{\mathbb{P}}

\def \bbN{\mathbb{N}}

\def \bbP{\mathbb{P}}

\def \frH{\mathfrak{H}}

\def \frS{\mathfrak{S}}

\def \calL{\mathcal{L}}

\def \calE{\mathcal{E}}
\def \calP{\mathcal{P}}
\def \calN{\mathcal{N}}

\def \bfx{\mathbf{x}}

\def \wtphi{\widetilde{\phi}}

\def \wtx{\widetilde{x}}
\def \wtk{\widetilde{k}}

\def \whk{\hat{k}}

\def \barj{\bar{\jmath}}

\def \whx{\hat{x}}

\def \wtH{\widetilde{H}}

\def \wtC{\widetilde{C}}

\def \wtL{\widetilde{L}}

\def \rmd{\mathrm{d}}

\def \eps{\varepsilon}

\renewcommand{\leq}{\leqslant}

\renewcommand{\geq}{\geqslant}

\DeclareMathOperator{\meas}{meas}

\DeclareMathOperator{\card}{card}

\DeclareMathOperator{\supp}{supp}

\DeclareMathOperator{\diam}{diam}

\DeclareMathOperator{\dist}{dist}
\DeclareMathOperator{\Sym}{Sym}

\DeclareMathOperator{\sign}{sign}

\DeclareMathOperator{\Tr}{Tr}
\DeclareMathOperator{\Id}{Id}
\DeclareMathOperator{\Ex}{Ex}

\newcommand{\car}{\mathbf{1}}
\newcommand{\R}{\mathbb{R}}

\newcommand{\N}{\mathbb{N}}

\newcommand{\Z}{\mathbb{Z}}

\newcommand{\n}{\mathbb{N}}
\newcommand{\C}{\mathbb{C}}

\newcommand{\Q}{\mathbb{Q}}

\newcommand{\vers}[1]{\xrightarrow[\!#1]{}}

\newcommand{\D}{\displaystyle}

\newcommand{\Coi}{{\mathcal C}_0^{\infty}}
\newcommand{\esp}{\mathbb{E}}
\newcommand{\pro}{\mathbb{P}}

\def\n{\nabla}

\def\wtn{\widetilde n}
\def\wtx{\widetilde x}

\newcommand\Perp{\protect\mathpalette{\protect\PerP}{\perp}}
\def\PerP#1#2{\mathrel{\rlap{$#1#2$}\mkern2mu{#1#2}}}

\theoremstyle{plain}

\newtheorem{lemma}{Lemma}[section]

\newtheorem{proposition}[lemma]{Proposition}
\newtheorem{corollary}[lemma]{Corollary}

\newtheorem{Le}[lemma]{Lemma}
\newtheorem{Th}[lemma]{Theorem}
\newtheorem{Pro}[lemma]{Proposition}
\newtheorem{Cor}[lemma]{Corollary}

\theoremstyle{definition}
\newtheorem{definition}[lemma]{Definition}

\newtheorem{notation}[lemma]{Notation}

\theoremstyle{remark}
\newtheorem{Rem}[lemma]{Remark}
\newtheorem{remark}[lemma]{Remark}

\numberwithin{equation}{section}

\def \densEn{\calE}
\title[Interacting electrons in a random medium]{Interacting electrons
  in a random medium: a simple one-dimensional model}

\date{\today}

\author{Fr{\'e}d{\'e}ric Klopp} 
\address[Fr{\'e}d{\'e}ric Klopp]{ \vskip.1cm Sorbonne Universit{\'e}s, UPMC
  Univ. Paris 06, UMR 7586, IMJ-PRG, F-75005, Paris, France \vskip.1cm
  Univ. Paris Diderot, Sorbonne Paris Cit{\'e}, UMR 7586, IMJ-PRG, F-75205
  Paris, France \vskip.1cm CNRS, UMR 7586, IMJ-PRG, F-75005, Paris,
  France}
\email{\href{mailto:frederic.klopp@imj-prg.fr}{frederic.klopp@imj-prg.fr}}
\author{Nikolaj~A.~Veniaminov}
  \address[Nikolaj~A.~Veniaminov]{CEREMADE, UMR CNRS 7534
  Universit{\'e} Paris IX Dauphine, Place du Mar{\'e}chal De Lattre De
Tassigny F-75775 Paris cedex 16 France}
  \email{\href{mailto:veniaminov@ceremade.dauphine.fr}{veniaminov@ceremade.dauphine.fr}}



\thanks{This work is partially supported by the grant
  ANR-08-BLAN-0261-01. The authors also acknowledge the support of the
  IMS (NU Singapore) where part of this work was done. F.K. thanks
  T. Duquesne for his explaining the Palm formula.}
\begin{document}

\begin{abstract}
  The present paper is devoted to the study of a simple model of
  interacting electrons in a random background. In a large interval
  $\Lambda$, we consider $n$ one dimensional particles whose evolution
  is driven by the Luttinger-Sy model, i.e., the interval $\Lambda$ is
  split into pieces delimited by the points of a Poisson process of
  intensity $\mu$ and, in each piece, the Hamiltonian is the Dirichlet
  Laplacian. The particles interact through a repulsive pair potential
  decaying polynomially fast at infinity. We assume that the particles
  have a positive density, i.e., $n/|\Lambda|\to\rho>0$ as
  $|\Lambda|\to+\infty$. In the low density or large disorder regime,
  i.e., $\rho/\mu$ small, we obtain a two term asymptotic for the
  thermodynamic limit of the ground state energy per particle of the
  interacting system; the first order correction term to the non
  interacting ground state energy per particle is controlled by pairs
  of particles living in the same piece. The ground state is described
  in terms of its one and two-particles reduced density
  matrix. Comparing the interacting and the non interacting ground
  states, one sees that the effect of the repulsive interactions is to
  move a certain number of particles living together with another
  particle in a single piece to a new piece that was free of particles
  in the non interacting ground state.
  \vskip.5cm\noindent \textsc{R{\'e}sum{\'e}.}  
  Dans ce travail, nous consid{\'e}rons un mod{\`e}le simple de {\'e}lectrons en
  interaction dans un environnement al{\'e}atoire. Dans un grand
  intervalle $\Lambda$, nous consid{\'e}rons $n$ particules
  uni-dimensionnelles dont l'{\'e}volution est r{\'e}gie par le mod{\`e}le de
  Luttinger-Sy : l'intervalle $\Lambda$ est subdivis{\'e} en pi{\`e}ces
  d{\'e}limit{\'e}es par les points d'un processus de Poisson d'intensit{\'e}
  $\mu$ et, dans chaque pi{\`e}ce, le hamiltonien est le laplacien de
  Dirichlet. Les particules interagissent par paires au travers d'un
  potentiel r{\'e}pulsif d{\'e}croissant polynomialement {\`a} l'infini. On
  suppose que la densit{\'e} de particules est positive c'est-{\`a}-dire que
  $n/|\Lambda|\to\rho>0$ quand $|\Lambda|\to+\infty$. Lorsque la
  densit{\'e} est petite ou lorsque le d{\'e}sordre est grand, c'est-{\`a}-dire
  lorsque $\rho/\mu$ est petit, nous obtenons une asymptotique {\`a} deux
  termes de la limite thermodynamique de l'{\'e}nergie fondamentale par
  particule du syst{\`e}me ; le premier terme de correction {\`a} l'{\'e}nergie
  fondamentale par particule du syst{\`e}me sans interaction est contr{\^o}l{\'e}
  par les paires de particules vivant dans la m{\^e}me pi{\`e}ce. L'{\'e}tat
  fondamental est d{\'e}crit au moyen de sa matrice de densit{\'e} r{\'e}duite {\`a}
  une et {\`a} deux particules. En comparant l'{\'e}tat fondamental avec
  interaction {\`a} l'{\'e}tat fondamental sans interaction, on voit que
  l'effet des interactions est de s{\'e}parer un certain nombre de
  particules qui vivent en paire avec une autre particule dans la m{\^e}me
  pi{\`e}ce vers des pi{\`e}ces inoccup{\'e}es dans l'{\'e}tat fondamental sans
  interaction.

\end{abstract}

\maketitle

%\listoffixmes

% ============
% INTRODUCTION
% ============
\section{Introduction: the model and the main results}
\label{sec:introduction}
On $\R$, consider a Poisson point process $d\mu(\omega)$ of intensity
$\mu$. Let $(x_k(\omega))_{k\in\Z}$ denote its support (i.e., $\D
d\mu(\omega)=\sum_{k\in\Z}\delta_{x_k(\omega)}$), the points being
ordered increasingly.\\
On $L^2(\R)$, define the Luttinger-Sy or pieces model (see
e.g.~\cite{PhysRevA.7.701,MR1042095}), that is, the random operator
\begin{equation*}
  H_\omega=\bigoplus_{k\in\Z}-\Delta_{|[x_k,x_{k+1}]}^D
\end{equation*}
where, for an interval $I$, $-\Delta_{|I}^D$ denotes the Dirichlet
Laplacian on $I$. \\
Pick $L>0$ and let $\Lambda=\Lambda_L=[0, L]$. Restrict $H_\omega$
to $\Lambda$ with Dirichlet boundary conditions: on
$\frH:=L^2(\Lambda)$, define
\begin{equation}
  \label{eq:2}
  H_\omega(L)=H_\omega(\Lambda)=
  \bigoplus_{k_--1\leq k\leq k_+}-\Delta_{|\Delta_k(\omega)}^D
\end{equation}
where we have defined
$\Delta_k(\omega):=[x_k(\omega),x_{k+1}(\omega)]$ to be the {\it
  $k$-th piece} and we have set
\begin{equation*}
  \begin{aligned}
 k_-=\min\{k; x_k > 0\},\quad x_{k_--1} = 0,\\ 
 k_+=\max\{k; x_k < L\},\quad x_{k_++1} = L.   
  \end{aligned}
\end{equation*}
From now on, we let $m(\omega)$ be the number of pieces and renumber
them from $1$ to $m(\omega)$ (i.e., $k_-=2$ and $k_+=m(\omega)$). For
$L$ large, with probability $1-O(L^{-\infty})$, one has
$m(\omega)=\mu L+O(L^{2/3})$. \\
The pieces model admits an integrated density of states that can be
computed explicitly (see section~\ref{sec:analys-one-part}
or~\cite{MR1042095,Veniaminov_PhDthesis}), namely,
\begin{equation}
  \label{eq:6}
  \begin{split}
    N_\mu(E)&:=\lim_{L\to+\infty}\frac{\{\text{eigenvalues of }
      H_\omega(L)\text{ in }(-\infty,E]\}}L\\&=
    \frac{\mu \cdot \exp(-\mu \ell_E)}{1 - \exp(-\mu \ell_E)}\car_{E\geq0}
    \quad\text{ where }\quad\ell_E:=\frac{\pi}{\sqrt{E}}.
  \end{split}
\end{equation}
\subsection{Interacting electrons}
\label{sec:interacting-electrons}
Consider first $n$ free electrons restricted to the box $\Lambda$ in
the background Hamiltonian $H_\omega(\Lambda)$, that is, on the space
\begin{equation}
  \label{eq:19}
  \frH^n(\Lambda)=\frH^n(\Lambda_L)=\bigwedge_{j=1}^nL^2(\Lambda)=L_{-}^2(\Lambda^n),
\end{equation}
consider the operator
\begin{equation}
  \label{eq:Hn0Definition}
  H_\omega^0(\Lambda,n) = \sum_{i = 1}^n 
  \underbrace{\car_\frH \otimes \hdots \otimes \car_\frH}_{\mbox{$i -
      1$ times}}  
  \otimes H_\omega(\Lambda) \otimes 
  \underbrace{\car_\frH \otimes \hdots \otimes \car_\frH}_{\mbox{$n -
      i$ times}}. 
\end{equation}
This operator is self-adjoint and lower semi-bounded. Let
$E^0_\omega(\Lambda,n)$ be its ground state energy and
$\Psi^0_\omega(\Lambda,n)$ be its ground state.\\
To $H_\omega^0(\Lambda,n)$, we now add a repulsive pair finite range
interaction potential. Therefore, pick $U:\R\to\R$ satisfying
\begin{description}
\item[(HU)] $U$ is a repulsive (i.e., non negative), even pair
  interaction potential decaying sufficiently fast at infinity. More
  precisely, we assume
  \begin{equation}
    \label{eq:13}
    x^3 \int_x^{+\infty} U(t) \rmd{t}\vers{x\to+\infty}0.    
  \end{equation}
  To control the possible local singularities of the interactions, we
  require that $\D U \in L^p(\bbR)$ for some $p \in (1, +\infty]$.
\end{description}
On $\frH^{n}(\Lambda)$, we define
\begin{equation}
  \label{eq:HLambdanIntro}
  H^U_\omega(\Lambda, n) = H_\omega^0(\Lambda,n)+ W_n
\end{equation}
where 
\begin{equation}
  \label{eq:9}
   W_n(x^1,\cdots,x^n):=\sum_{i < j} U(x^i - x^j)
\end{equation}
on the domain
\begin{equation}
  \label{eq:29}
  \mathcal{D}^n(\Lambda):=\Coi\left(\left(\bigcup_{k=1}^{m(\omega)}
      ]x_k,x_{k+1}[\right)^{n}\right) \cap\frH^n(\Lambda).
\end{equation}
As $U$ is non negative, $H_\omega^U(\Lambda,n)$ is non negative. From
now on, we let $H^U_\omega(\Lambda,n)$ be the Friedrichs extension of
this operator. As $W_n$ is a sum of pair interactions, the fact that
$U\in L^p(\R)$ for some $p>1$ (see assumption~\textbf{(HU)})
guarantees that $W_n$ is $H_\omega^0(\Lambda,n)$-form bounded with
relative form bound $0$ (see, e.g.,~\cite[section
1.2]{MR883643}). Thus, the form domain of the operator
$H^U_\omega(\Lambda,n)$ is
\begin{equation}
  \label{eq:266}
  \frH_\infty^n(\Lambda):=\left(H_0^1\left(\bigcup_{k=1}^{m(\omega)}
      ]x_k,x_{k+1}[\right)\right)^{\otimes n} \cap\frH^n(\Lambda).
\end{equation}
Moreover, $H^U_\omega(\Lambda,n)$ admits $\mathcal{D}^n(\Lambda)$ as a
form core (see, e.g.,~\cite[section 1.3]{MR883643}) and it has a
compact resolvent, thus, only discrete spectrum.\\
We define $E^U_\omega(\Lambda,n)$ to be its ground state energy, that
is,
\begin{equation}
  \label{eq:15}
  E^U_\omega(\Lambda,n):=\inf_{\substack{\Psi\in\mathcal{D}^n(\Lambda)
      \\\|\Psi\|=1}}\langle H^U_\omega(\Lambda,n)\Psi,\Psi\rangle
\end{equation}
and $\Psi^U_\omega(\Lambda, n)$ to be a ground state, i.e., to be an
eigenfunction associated to the eigenvalue $E^U_\omega(\Lambda,n)$.\\
By construction, there is no unique continuation principle for the
pieces model (as the union of disjoint non empty intervals is not
connected); so, one should not expect uniqueness for the ground
state. Nevertheless due to the properties of the Poisson process, for
the non interacting system, one easily sees that the ground state
$\Psi^0_\omega(\Lambda, n)$ is unique $\omega$ almost surely (see
section~\ref{sec:free-electrons}). For the interacting system, it is
not as clear. Nonetheless, one proves
\begin{Th}[Almost sure non-degeneracy of the ground state]
  \label{th:asNonDegOfGroundState}
  Suppose that $U$ is real analytic. Then, $\omega$-almost surely, for
  any $L$ and $n$, the ground state of $H^U_\omega(L,n)$ is
  non-degenerate.
\end{Th}
\noindent For a general $U$, while we don't know whether the ground
state is degenerate or not, our analysis will show where the
degeneracy may come from: we shall actually write $\frH^n(\Lambda)$ as
an orthogonal sum of subspaces invariant by $H^U_\omega(L,n)$ such
that on each such subspace, the ground state of $H^U_\omega(L,n)$ is
unique. This will enable us to show that all the ground states of
$H^U_\omega(L,n)$ on $\frH^n(\Lambda)$ are very similar to each other,
i.e., they differ only by a small number of particles. 
\vskip.2cm\noindent
The goal of the present paper is to understand the thermodynamic
limits of $E^U_\omega(\Lambda,n)$ and $\Psi^U_\omega(\Lambda, n)$. As
usual, we define the \emph{thermodynamic limit} to be the limit $L \to
\infty$ and $n / L \to \rho$ where $\rho$ is a positive constant. The
constant $\rho$ is the \emph{density of particles}.\vskip.2cm\noindent
We will describe the thermodynamic limits of $E^U_\omega(\Lambda,n)$,
or rather $n^{-1}E^U_\omega(\Lambda,n)$, and $\Psi^U_\omega(\Lambda,
n)$ when $\rho$ is positive and small (but independent of $L$ and
$n$). We will be specially interested in the influence of the
interaction $U$, i.e., we will compare the thermodynamic limits for
the non-interacting and the interacting systems.
\subsection{The ground state energy per particle}
\label{sec:ground-state-energy}
Our first result describes the thermodynamic limit of
$n^{-1}E^U_\omega(\Lambda,n)$ when we assume the density of particles
$n/L$ to be $\rho$. For the sake of comparison, we also included the
corresponding result on the ground state energy of the free particles,
i.e., on $n^{-1}E^0_\omega(\Lambda,n)$. \\
We prove
\begin{Th}
  \label{thr:1}
  Under the assumptions made above, the following limits exist
  $\omega$-almost surely and in $L^1_\omega$
  \begin{equation}
    \label{eq:4}
    \densEn^0(\rho,\mu):=\lim_{\substack{L\to+\infty\\n/L \to \rho}}
    \frac{E^0_\omega(\Lambda,n)}{n}\quad\text{and}\quad      
    \densEn^U(\rho,\mu):=\lim_{\substack{L\to+\infty\\n/L \to \rho}}
    \frac{E^U_\omega(\Lambda,n)}{n}
  \end{equation}
  and they are independent of $\omega$.
\end{Th}
\noindent In~\cite{MR3022666} (see also~\cite{Veniaminov_PhDthesis}),
the almost sure existence of the thermodynamic limit of the ground
state energy per particle is established for quite general systems of
interacting electrons in a random medium if one assumes that the
interaction has compact support. For decaying interactions (as in
\textbf{(HU)}), only the $L^2_\omega$ convergence is proved. The
improvement needed on the results of~\cite{MR3022666} to obtain the
almost sure convergence is the purpose of Theorem~\ref{thr:8}.\\
In~\cite{MR2905791}, the authors study the existence of the above
limits in the grand canonical ensemble for Coulomb
interactions.\vskip.2cm\noindent
The energy $\densEn^0(\rho,\mu)$ can be computed explicitly for our
model (see section~\ref{sec:ground-state-energy-1}). We shall obtain a
two term asymptotic formula for $\densEn^U(\rho,\mu)$ in the case when
the disorder is not too large and the Fermi length $\ell_{\rho,\mu}$
is sufficiently large. \\
Define
\begin{itemize}
\item the~\emph{effective density} is defined as the ratio of the
  density of particles to the density of impurities, i.e.,
  $\D\rho_\mu=\frac\rho\mu$,
\item the~\emph{Fermi energy} $E_{\rho,\mu}$ is the unique solution to
  $N_\mu(E_{\rho,\mu})=\rho$,
\item the~\emph{Fermi length} $\ell_{\rho,\mu}:=\ell_{E_{\rho,\mu}}$ where
  $\ell_E$ is defined in~\eqref{eq:6}; the explicit formula for
  $N_\mu$ yields
  \begin{equation}
    \label{eq:67}\ell_{\rho,\mu}=\frac1\mu
    \left|\log\frac{\rho_\mu}{1+\rho_\mu}\right|=\frac1\mu
    \left|\log\frac{\rho}{\mu+\rho}\right|.      
  \end{equation}
\end{itemize}
For the free ground state energy per particle, a direct computation
using~\eqref{eq:6} yields
\begin{equation}
  \label{eq:5}
    \begin{split}
      \densEn^0(\rho,\mu)&=
      \frac1\rho\int_{-\infty}^{E_{\rho,\mu}}E\,dN_\mu(E)
      =E_{\rho,\mu} \left(1+O\left(\sqrt{E_{\rho,\mu}}\right)\right)
      % =\pi^2\ell^{-2}_{\rho,\mu}
      % \left(1+O\left((\mu\ell_{\rho,\mu})^{-1}\right)\right).
    \end{split}
\end{equation}
We prove
\begin{Th}
  \label{th:EnergyAsymptoticExpansion}
  Under the assumptions made above, for $\mu>0$ fixed, one computes
  \begin{equation}
    \label{eq:EnergyAsymptoticExpansion}
    \densEn^U(\rho,\mu)= \densEn^0(\rho,\mu) + \pi^2
    \,\gamma^\mu_*\,\mu^{-1}\,\rho_\mu\,
    \ell_{\rho,\mu}^{-3}\left(1+o(1)\right)\quad\text{where}\quad 
    o(1)\vers{\rho_\mu\to0} 0.
  \end{equation}
  The positive constant $\gamma^\mu_*$ depends solely on $U$ and
  $\mu$; it is defined in~\eqref{eq:12} below.
\end{Th}
\noindent At fixed disorder, in the small density regime, the Fermi
length is large and the Fermi energy is small. Moreover, the shift of
ground state energy (per particle) due to the interaction is
exponentially small compared to the free ground state energy: indeed
it is of order $\rho|\log\rho|^{-3}$ while the ground state energy is
of order $|\log\rho|^{-2}$. \\
For fixed $\mu$, a coarse version
of~\eqref{eq:EnergyAsymptoticExpansion} was established, in the PhD
thesis of the second author~\cite{Veniaminov_PhDthesis}, namely, for
$\rho$ sufficiently small, one has
\begin{equation*}
  \frac1{C_\mu} \rho|\log\rho|^{-3}\leq
  \densEn^U(\rho,\mu)-\densEn^0(\rho,\mu)\leq
  C_\mu\,\rho|\log\rho|^{-3}.
\end{equation*}
Moreover, from~\cite[Propositions~3.6 and 3.7]{MR3022666}), we know
that the function $\rho\mapsto\densEn^U(\rho,\mu)$ is a non decreasing
continuous function and that the function
$\rho^{-1}\mapsto\densEn^U(\rho,\mu)$ is convex.\vskip.2cm\noindent
Let us now define the constant $\gamma^\mu_*$. Therefore, we prove
\begin{Pro}
  \label{prop:TwoElectronProblem}
  Consider two electrons in $[0, \ell]$ interacting via an even non
  negative pair potential $U\in L^p(\R^+)$ for some $p>1$ and such that
  \begin{equation*}
    \int_\R x^2U(x)dx<+\infty.
  \end{equation*}
  That is, on $\frH^2([0, \ell]) = L^2([0, \ell]) \wedge L^2([0,
  \ell])$, consider the Hamiltonian
  \begin{equation}
    \label{eq:introTwoParticleHamiltonian}
    \left(-\Delta_{x_1|[0,\ell]}^D\right)\otimes\car_\frH +
    \car_\frH\otimes \left(-\Delta_{x_2|[0,\ell]}^D\right)+U(x_1-x_2),
  \end{equation}
  i.e., the Friedrichs extension of the same differential expression
  defined on the domain $\mathcal{C}^2([0,\ell])$ (see~\eqref{eq:29}).\\
  For large $\ell$, $E^U([0, \ell],2)$, the ground state energy of
  this Hamiltonian, admits the following expansion
  \begin{equation}
    \label{eq:32}
    E^U([0,\ell],2)=\frac{5 \pi^2}{\ell^2} + \frac{\gamma}{\ell^3} +
    o\left(\frac1{\ell^3}\right)
  \end{equation}
  where $\gamma=\gamma(U)>0$ when $U$ does not vanish a.e.
\end{Pro}
\noindent Let us first notice that the expansion~\eqref{eq:32}
immediately implies that $U\mapsto\gamma(U)$ is a non decreasing
concave function of the (non negative) interaction potential $U$ such
that $\gamma(0)=0$; for $\alpha$ small positive, one computes
\begin{equation*}
  \frac{\gamma(\alpha U)}\alpha=10\pi^2\int_\R x^2U(x)dx\,(1+O(\alpha)).
\end{equation*}
Concavity and monotony follow immediately from the
definition of $E^U([0,\ell],2)$ and the form of~\eqref{eq:32}. \\
In terms of $\gamma$, we then define
\begin{equation}
  \label{eq:12}
  \gamma^\mu_*:=1-\exp\left({-\dfrac{\mu\,\gamma}{8\pi^2}}\right).
\end{equation}
\subsection{The ground state: its one- and two-particle density
  matrices}
\label{sec:ground-state}
We shall now describe our results on the ground state. We start 
with a description of the spectral data of the one particle
Luttinger-Sy model. Then, we describe the non interacting ground
state.
\subsubsection{The spectrum of the one particle Luttinger-Sy model}
\label{sec:spectr-one-part}
Let $(E_{j,\omega}^\Lambda)_{j\geq1}$ and
$(\varphi_{j,\omega}^\Lambda)_{j\geq1}$ respectively denote the
eigenvalues (ordered increasingly) and the associated eigenfunctions
of $H_\omega(\Lambda)$ (see~\eqref{eq:2}). Clearly, the eigenvalues
and the eigenfunctions are explicitly computable from the points
$(x_k)_{1\leq k\leq m(\omega)+1}$. In particular, one sees that the
eigenvalues are simple $\omega$ almost surely.\\
As $n/L$ is close to $\rho$ and $L$ is large, the $n$ first
eigenvalues are essentially all the eigenvalues below the Fermi energy
$E_{\rho,\mu}$. These eigenvalues are the eigenvalues of
$-\Delta_{|\Delta_k(\omega)}^D$ below $E_{\rho,\mu}$ for all the
pieces $(\Delta_k(\omega))_{k_--1\leq k\leq k_+}$ of length at least
$\ell_{\rho,\mu}$ (see~\eqref{eq:6} and~\eqref{eq:5}). $\omega$-almost
surely, the number of pieces $(\Delta_k(\omega))_{1\leq k\leq
  m(\omega)}$ longer than $\ell_{\rho,\mu}$ is asymptotic to $n$ (see
section~\ref{sec:statistics-pieces}), the number of those longer than
$2\ell_{\rho,\mu}$ to $\rho_\mu\,n$, the number of those longer than
$3\ell_{\rho,\mu}$ to $\rho^2_\mu\,n$, etc. We refer to
section~\ref{sec:analys-one-part} for more details.
\subsubsection{The non interacting ground state}
\label{sec:non-iter-ground}
The ground state of the non interacting Hamiltonian
$H^0_\omega(\Lambda,n)$ is given by the (normalized) Slater
determinant
\begin{equation}
  \label{eq:7}
  \Psi^0_\omega(\Lambda,n)
  =\bigwedge_{j=1}^n\varphi_{j,\omega}^\Lambda
  =\frac1{\sqrt{n!}}
  \text{Det}\left(\left(\varphi_{j,\omega}^\Lambda(x_k)\right)\right)_{1\leq j,k\leq n}.
\end{equation}
Here and in the sequel, the exterior product is normalized so that the
$L^2$-norm of the product be equal to the product of the $L^2$-norms
of the factors (see~\eqref{eq:267} in
section~\ref{sec:proj-totally-antisym}).\\
It will be convenient to describe the interacting ground state using
its one-particle and two-particles reduced density matrices. Let us
define these now (see section~\ref{sec:proof-theorem6} for more
details). Let $\Psi\in\frH^n(\Lambda)$ be a normalized $n$-particle
wave function. The corresponding \emph{one-particle density matrix} is
an operator on $\frH^1(\Lambda)=L^2(\Lambda)$ with the kernel
\begin{equation}
  \label{eq:OneParticleDensityMatrixDef}
  \gamma_\Psi(x, y)=\gamma^{(1)}_\Psi(x, y) = n \int_{\Lambda^{n - 1}}
  \Psi(x,\tilde x)\Psi^\ast(y, \tilde x) d\tilde x
\end{equation}
where $\tilde x=(x^2,\dots,x^n)$ and $d\tilde x=dx^2\cdots dx^n$.\\
The \emph{two-particles density matrix} of $\Psi$ is an operator
acting on $\D\frH^2(\Lambda)=\bigwedge_{j=1}^2L^2(\Lambda)$ and its
kernel is given by
\begin{equation}
  \label{eq:TwoParticleDensityMatrixDef}
  \gamma_\Psi^{(2)}(x^1, x^2, y^1, y^2) 
  = \frac{n (n - 1)}{2} \int_{\Lambda^{n - 2}} \Psi(x^1, x^2, \tilde x)
  \Psi^\ast(y^1, y^2, \tilde x)d\tilde x 
\end{equation}
where $\tilde x=(x^3,\dots,x^n)$ and $d\tilde x=dx^3\cdots dx^n$.\\
Both $\gamma_\Psi$ and $\gamma_\Psi^{(2)}$ are positive trace class
operators satisfying
\begin{equation}
  \label{eq:tracecOfReducedMatrices}
  \Tr{\gamma_\Psi} = n, \quad\text{and}\quad
  \Tr{\gamma_\Psi^{(2)}} = \frac{n (n - 1)}{2}.
\end{equation}
So, for the non interacting ground state, using the description of the
eigenvalues and eigenvectors of $H_\omega(\Lambda)$ given in
section~\ref{sec:spectr-one-part}, as a consequence of
Proposition~\ref{prop:DensityMatrixStructure}, we obtain that
\begin{equation}
  \label{eq:11}
    \gamma_{\Psi^0_\omega(\Lambda,n)}=\sum_{j = 1}^n
    \gamma_{\varphi_{j,\omega}^\Lambda}
    =\sum_{\ell_{\rho,\mu}\leq
      |\Delta_k(\omega)|<3\ell_{\rho,\mu}}
    \gamma_{\varphi^1_{\Delta_k(\omega)}}+ \sum_{2\ell_{\rho,\mu}\leq
      |\Delta_k(\omega)|<3\ell_{\rho,\mu}}
    \gamma_{\varphi^2_{\Delta_k(\omega)}}+R^{(1)}
\end{equation}
where
\begin{itemize}
\item $|\Delta_k(\omega)|$ denotes the length of the piece
  $\Delta_k(\omega)$;
\item $\varphi^j_{\Delta_k(\omega)}$ denotes the $j$-th normalized
  eigenvector of $-\Delta_{|\Delta_k(\omega)}^D$;
\item the operator $R^{(1)}$ is trace class and $\D
  \|R^{(1)}\|_{\text{tr}}\leq 2\,n\,\rho^2_\mu$.
\end{itemize}
Here, $\|\cdot\|_{\text{tr}}$ denotes the trace norm in the ambient
space, i.e., in $L^2(\Lambda)$ for the one particle density matrix,
and in $L^2(\Lambda)\wedge L^2(\Lambda)$ for the two particles density
matrix.\vskip.2cm\noindent
For the two-particles density matrix, again as a consequence of
Proposition~\ref{prop:DensityMatrixStructure}, we obtain
\begin{equation}
  \label{eq:14}
  \gamma_{\Psi^0_\omega(\Lambda,n)}^{(2)} = 
  \frac{1}{2} (\Id - \Ex)\left[\gamma_{\Psi^0_\omega(\Lambda,n)} \otimes
    \gamma_{\Psi^0_\omega(\Lambda,n)}\right]+R^{(2)}
\end{equation}
where 
\begin{itemize}
\item $\Id$ is the identity operator, $\Ex$ is the exchange operator
  on a two-particles space:
  \begin{equation*}
    \Ex \left[f \otimes g\right] = g \otimes f, \quad f, g \in \frH \text{,}
  \end{equation*}
\item the operator $R^{(2)}$ is trace class and
  $\|R^{(2)}\|_{\text{tr}}\leq C_{\rho,\mu}n$.
\end{itemize}
\begin{figure}[h]
  \begin{center}
    \includegraphics[width=.60\textwidth]{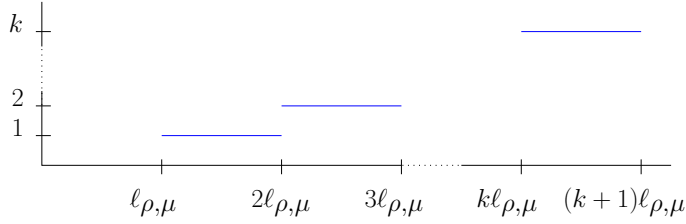}
    \caption{The distribution of particles in the non interacting
      ground state.}
    \label{fig:1}
  \end{center}
\end{figure}
\noindent One can represent graphically the ground state of the non
interacting system by representing the distribution of its particles
within the pieces: in abscissa, one puts the length of the pieces, in
ordinate, the number of particles the ground state puts in a piece of
that length. Figure~\ref{fig:1} shows the picture thus obtained.
\subsubsection{The interacting ground state}
\label{sec:iter-ground}
To describe the ground state of the interacting system, we shall
describe its one-particle and two-particles reduced density
matrices. Therefore, it will be useful to introduce the following
approximate one-particle reduced density matrices.\\
For a piece $\Delta_k(\omega)$, let $\zeta^j_{\Delta_k(\omega)}$ be
the $j$-th normalized eigenvector of
$-\laplace^D_{|\Delta_k(\omega)\times\Delta_k(\omega)} + U$ acting on
$\D L^2(\Delta_k(\omega))\wedge L^2(\Delta_k(\omega))$. We note that,
for $U = 0$, the two-particles ground state can be rewritten as
$\zeta^{1,U=0}_{\Delta_k(\omega)}= \varphi^1_{\Delta_k(\omega)} \wedge
\varphi^2_{\Delta_k(\omega)}$.\\
Define the following one-particle density matrix
\begin{equation}
  \label{10}
  \gamma_{\Psi^{\text{opt}}_{\Lambda,n}}
  =\sum_{\ell_{\rho,\mu}-\rho_\mu \gamma^\mu_*\leq |\Delta_k(\omega)|\leq
    2\ell_{\rho,\mu}-\log(1-\gamma^\mu_*)}
  \gamma_{\varphi^1_{\Delta_k(\omega)}}+
  \sum_{2\ell_{\rho,\mu}-\log(1-\gamma^\mu_*)\leq |\Delta_k(\omega)|}
  \gamma_{\zeta^1_{\Delta_k(\omega)}}.
\end{equation}
Because of the possible long range of the interaction $U$ (see the
remarks following Theorem~\ref{thr:2} below), to describe our results
precisely, it will be useful to introduce trace norms reduced to
certain pieces. For $\ell\geq0$, we define the projection onto the
pieces shorter than $\ell$
\begin{equation}
  \label{eq:241}
  \car^1_{<\ell}=\sum_{|\Delta_k(\omega)|<\ell}\car_{\Delta_k(\omega)}.
 \end{equation}
We shall use the following function to control remainder terms: define
\begin{equation}
  \label{eq:98}
  Z(x) = \sup_{x\leq v}\left(v^{3}\int_v^{+\infty} U(t)dt\right).
\end{equation}
Under assumption (HU), the function $Z$ is continuous and monotonously
decreasing on $[0,+\infty)$ and tends to $0$ at infinity.\\
We prove
\begin{Th}
  \label{thr:2}
  Fix $\mu>0$. Assume (HU) holds. Then, there exist $\rho_0>0$ such
  that, for $\rho\in(0,\rho_0)$, $\omega$-a.s., one has
  \begin{gather*}
    \limsup_{\substack{L\to+\infty\\n/L \to
        \rho}}\frac1{n}\left\|\left(\gamma_{\Psi^U_\omega(\Lambda,
          n)}-\gamma_{\Psi^{\text{opt}}_{\Lambda,n}}\right)
      \car^1_{<\ell_{\rho,\mu}+C}\right\|_{\text{tr}} \leq
    \frac1{\rho_0}\max\left(\frac{\rho_\mu}{\ell_{\rho,\mu}},
      \sqrt{\rho_\mu\,Z(\,\ell_{\rho,\mu})}\right),\\
    \limsup_{\substack{L\to+\infty\\n/L\to\rho}}
    \frac1{n}\left\|\left(\gamma_{\Psi^U_\omega(\Lambda,
          n)}-\gamma_{\Psi^{\text{opt}}_{\Lambda,n}}\right)
      \left(\car-\car^1_{<\ell_{\rho,\mu}+C}\right)\right\|_{\text{tr}}
    \leq\frac1{\rho_0}\max\left(\frac{\rho_\mu}{\ell_{\rho,\mu}},
      \rho_\mu\,\sqrt{Z(\,\ell_{\rho,\mu})}\right).
  \end{gather*}
  Here, $\|\cdot\|_{\text{tr}}$ denotes the trace norm in
  $L^2(\Lambda)$.
\end{Th}
\noindent This result calls for some comments. Let us first note that,
if $Z$, that is, $U$, decays sufficiently fast at infinity, typically
exponentially fast with a large rate, then the two estimates in
Theorem~\ref{thr:2} can be united into
\begin{equation*}
  \limsup_{\substack{L\to+\infty\\n/L \to
      \rho_\mu}}\frac1{n}\left\|\gamma_{\Psi^U_\omega(\Lambda, n)}-
    \gamma_{\Psi^{\text{opt}}_{\Lambda,n}}\right\|_{\text{tr}} \leq
  C\frac{\rho_\mu}{\ell_{\rho,\mu}}.
\end{equation*}
In this case, Theorem~\ref{thr:2} can be summarized graphically. In
Figure~\ref{fig:2}, using the same representation as in
Figure~\ref{fig:1}, we compare the non interacting and the interacting
ground state. The non interacting ground state distribution of
particles is represented in blue, the interacting one in green. We
assume that $U$ has compact support and restrict
ourselves to pieces shor\-ter than $3\ell_{\rho,\mu}$.\\
Indeed, in this case, comparing~\eqref{eq:11} and~\eqref{10}, we see
\begin{equation}
  \label{eq:gammaPsi0gammaPsiOptDiffIntro}
  \begin{split}
    \gamma_{\Psi^0_\omega(\Lambda,n)}-
    \gamma_{\Psi^{\text{opt}}_{\Lambda,n}} &=\sum_{
      2\ell_{\rho,\mu}-\log(1-\gamma^\mu_*)\leq |\Delta_k(\omega)|}
    \left(\gamma_{\varphi^1_{\Delta_k(\omega)}} +
      \gamma_{\varphi^2_{\Delta_k(\omega)}} -
      \gamma_{\zeta^1_{\Delta_k(\omega)}}\right) \\&\hskip1.5cm-\sum_{
          \ell_{\rho,\mu}-\rho_\mu\gamma^\mu_*\leq
          |\Delta_k(\omega)|\leq\ell_{\rho,\mu}}
        \gamma_{\varphi^1_{\Delta_k(\omega)}}\\
        &\hskip3cm+ \sum_{ 2\ell_{\rho,\mu}\leq
          |\Delta_k(\omega)|\leq
          2\ell_{\rho,\mu}-\log(1-\gamma^\mu_*)}
        \gamma_{\varphi^2_{\Delta_k(\omega)}} +\tilde R^{(1)}
  \end{split}
\end{equation}
where $\tilde R^{(1)}$ satisfies the same properties as $R^{(1)}$
in~\eqref{eq:11}.\\
\begin{floatingfigure}{.60\textwidth}
  \begin{center}
    \includegraphics[width=.60\textwidth]{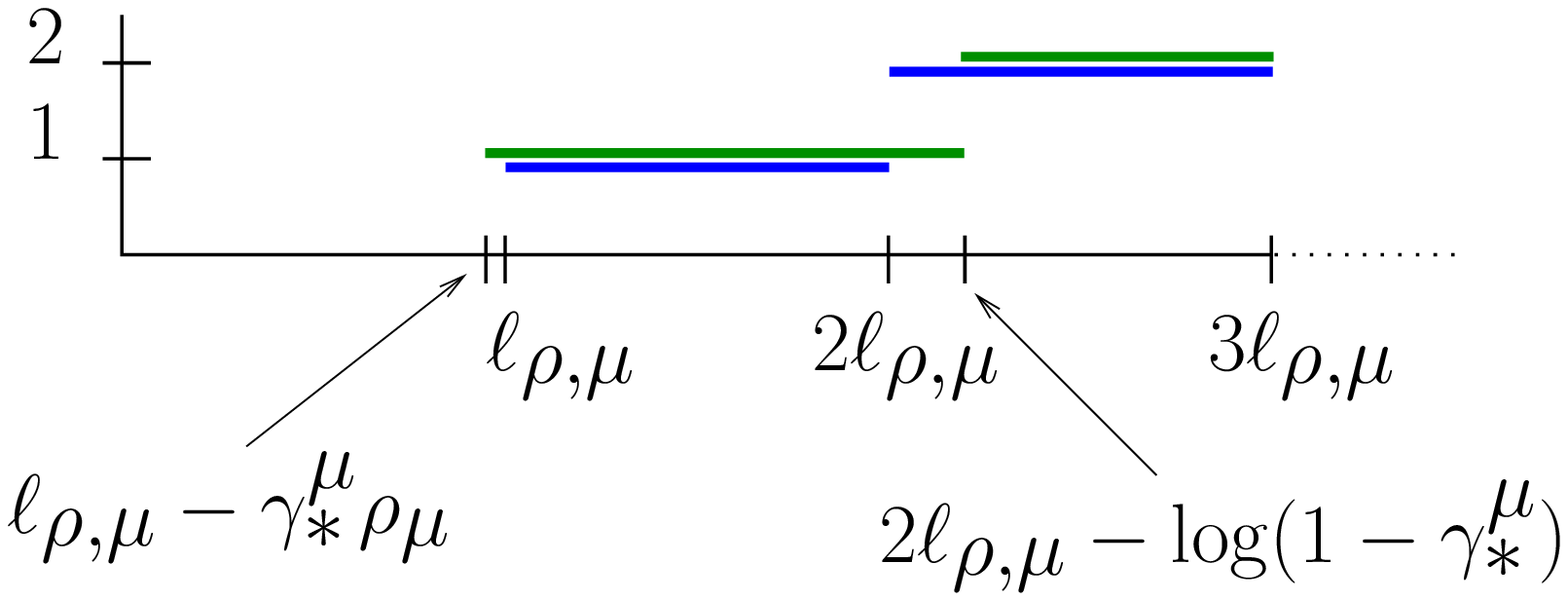}
    \caption{The distribution of particles in the interacting ground
      state.}
    \label{fig:2}
  \end{center}
\end{floatingfigure}
\noindent Thus, to obtain $\gamma_{\Psi^{\text{opt}}_{\Lambda,n}}$
from $\gamma_{\Psi^0_\omega(\Lambda,n)}$, we have displaced (roughly)
$\gamma^\mu_*\rho_\mu n$ particles living in pieces of length within
$[2\ell_{\rho,\mu},2\ell_{\rho,\mu}-\log(1-\gamma^\mu_*)]$ (i.e.,
pieces containing exactly two states below energy $E_{\rho,\mu}$ and
the energy of the top state stays above\\$E_{\rho,\mu}\left(1+
  \frac{\log(1-\gamma^\mu_*)}{\ell_{\rho,\mu}}\right)$ up to smaller
order terms in $\ell_{\rho,\mu}^{-1}$) to pieces having lengths within
$[\ell_{\rho,\mu}-\rho \gamma^\mu_*,\ell_{\rho,\mu}]$ (i.e., having
ground state energy within the interval
$\left[E_{\rho,\mu},E_{\rho,\mu}\left(1+
    \frac{2\rho\gamma^\mu_*}{\ell_{\rho,\mu}}\right)\right]$ up to
smaller order terms in $\ell_{\rho,\mu}^{-1}$). In the remaining of
(roughly) $(1 - \gamma^\mu_*) \rho n$ pieces containing exactly two
states below energy $E_{\rho,\mu}$ (that is, pieces of length within
$[2\ell_{\rho,\mu}-\log(1-\gamma^\mu_*), 3\ell_{E_{\rho, \mu}}]$ or
alternatively those with the top state below $E_{\rho,\mu}\left(1+
  \frac{\log(1-\gamma^\mu_*)}{\ell_{\rho,\mu}}\right)$ (up to smaller
order terms in $\ell_{\rho,\mu}^{-1}$), we have substituted the free
two-particles ground state (given by the anti-symmetric tensor product
of the first two Dirichlet levels in this piece) by the ground state
of the interacting system~\eqref{eq:introTwoParticleHamiltonian}.
In particular, we compute (remark that the first sum in
\eqref{eq:gammaPsi0gammaPsiOptDiffIntro} contributes only to the error
term according to Corollary~\ref{cor:twoParticleProblemComparison})
\begin{equation*}
  \lim_{\substack{L\to+\infty\\n/L \to
      \rho}}\frac1{n}\left\|\gamma_{\Psi^0_\omega(\Lambda,n)}-
    \gamma_{\Psi^{\text{opt}}_{\Lambda,n}}\right\|_{\text{tr}}
  =2\gamma^\mu_*\rho_\mu+O\left(\frac{\rho_\mu}{\ell_{\rho,\mu}}\right),
\end{equation*}
and, recalling~\eqref{eq:14}, we then compute
\begin{equation}
  \label{eq:8}
  \lim_{\substack{L\to+\infty\\n/L \to \rho}}
  \frac1{n^2}\left\|\gamma^{(2)}_{\Psi^0_\omega(\Lambda,n)}-\frac{1}{2} (\Id
    - \Ex) \left[\gamma_{\Psi^{\text{opt}}_{\Lambda,n}} \otimes
      \gamma_{\Psi^{\text{opt}}_{\Lambda,n}}\right] \right\|_{\text{tr}}
  =2\gamma^\mu_*\rho_\mu+O\left(\frac{\rho_\mu}{\ell_{\rho,\mu}}\right).    
\end{equation}
So the main effect of the interaction is to shift a macroscopic (though
small when $\rho_\mu$ is small) fraction of the particles to different
pieces.\vskip.2cm\noindent
Let us now discuss what happens when the interaction does not decay so
fast, typically, if it decays only polynomially. In this case,
Theorem~\ref{thr:2} tells us that one has to distinguish between short
and long pieces. In the long pieces, the description of the ground
state is still quite good as the error estimate is still of order
$o(\rho_\mu)$. Of course, this result only tells us something for the
pieces of length at most $3\ell_{\rho,\mu}$: the larger ones are very
few, thus, can only carry so few particles (see Lemma~\ref{le:18})
that these can be integrated into the remainder term. For short
intervals, the situation is quite different. Here, the remainder term
becomes much larger, only of order
$O\left(\sqrt{\rho_\mu}\ell_{\rho,\mu}^{-k/2}\right)$ if $Z(x)\asymp
x^{-k}$ at infinity. This loss is explained in the following way. The
short pieces carry the majority of the particles. When $U$ is of
longer range, particles in rather distant pieces start to interact in
a way that is not negligible with respect to the second term of the
expansion~\eqref{eq:EnergyAsymptoticExpansion} (which gives an average
surplus of energy per particle for the interacting ground state
compared to the free one); thus, it may become energetically
profitable to relocate some of these particles to new pieces so as to
minimize the interaction energy. When the range of the interaction
increases, the ground state will relocate more and more
particles. Nevertheless, the shift in energy will still be smaller
than the correction term obtained by relocating some of the particles
living in pairs in not too long intervals; this is going to be the
case as long as $U$ satisfies the decay assumption \textbf{(HU)}. When
$U$ decays slower than that, the main correction to the interacting
ground state energy per particle can be expected to be given by the
relocation of many particles living alone in their piece to new pieces
so as to diminish the interaction energy.\vskip.2cm\noindent
We also obtain an analogue of Theorem~\ref{thr:2} for the 2-particles
density matrix of the ground state $\Psi^U$. We prove
\begin{Th}
  \label{thr:7}
  Fix $\mu>0$. Assume (HU) holds. Then, there exist $\rho_0>0$ such
  that, for $\\rho\in(0,\rho_0)$, $\omega$-a.s., one has
  \begin{gather*}
    \begin{split}
      \limsup_{\substack{L\to+\infty\\n/L\to\rho}}&
      \frac1{n^2}\left\|\left(\gamma^{(2)}_{\Psi^U_\omega(\Lambda,
            n)}-\frac{1}{2} (\Id - \Ex)
          \left[\gamma_{\Psi^{\text{opt}}_{\Lambda,n}} \otimes
            \gamma_{\Psi^{\text{opt}}_{\Lambda,n}}\right]\right)
        \car^2_{<\ell_{\rho,\mu}+C} \right\|_{\text{tr}}
      \\&\hskip5cm\leq
      \frac1{\rho_0}\max\left(\frac{\rho_\mu}{\ell_{\rho,\mu}},
        \sqrt{\rho_\mu\,Z(\,\ell_{\rho,\mu})}\right)
    \end{split}
    \\\intertext{and}
    \begin{split}
      \limsup_{\substack{L\to+\infty\\n/L\to\rho}}&
      \frac1{n^2}\left\|\left(\gamma^{(2)}_{\Psi^U_\omega(\Lambda,
            n)}-\frac{1}{2} (\Id - \Ex)
          \left[\gamma_{\Psi^{\text{opt}}_{\Lambda,n}} \otimes
            \gamma_{\Psi^{\text{opt}}_{\Lambda,n}}\right]\right)
        \left(\car-\car^2_{<\ell_{\rho,\mu}+C}\right)
      \right\|_{\text{tr}} \\&\hskip5cm\leq
      \frac1{\rho_0}\max\left(\frac{\rho_\mu}{\ell_{\rho,\mu}},
        \rho_\mu\,\sqrt{Z(\,\ell_{\rho,\mu})}\right)
    \end{split}
  \end{gather*}
  where, for $\ell\geq0$, we recall that $\|\cdot\|_{\text{tr}}$
  denotes the trace norm in $L^2(\Lambda)\wedge L^2(\Lambda)$,
  recall~\eqref{eq:241} and define
  \begin{equation}
    \label{eq:240}
    \car^2_{<\ell}=\car^1_{<\ell}\otimes\car^1_{<\ell}.
  \end{equation}
\end{Th}
\subsection{Discussion and perspectives}
\label{sec:global-context}
While a very large body of mathematical works has been devoted to one
particle random Schr{\"o}dinger operators (see
e.g.~\cite{MR2509110,MR94h:47068}), there are only few works dealing
with many interacting particles in a random medium (for the case of
finitely many particles, see, for example,
\cite{AizenmanWarzel_LocBoundsMultipartilce} or
\cite{ChulaevskySuhov_MultiparticleAndersonLoc}).\\
The general Hamiltonian describing $n$ electrons in a random background
potential $V_\omega$ interacting via a pair potential $U$ can be
described as follows. In a $d$-dimensional domain $\Lambda$, consider
the operator
\begin{equation*}
  H_\omega(\Lambda, n) = -\laplace_{n d}\bigr|_{\Lambda^n} + \sum_{i =
    1}^n V_\omega(x^i) +  \sum_{i < j} U(x^i - x^j) \text{,}
\end{equation*}
where, for $j \in \{1, \hdots, n\}$, $x^j$ denotes the coordinates of
the $j$-th particle. The operator $H_\omega(\Lambda, n)$ acts on a
space of totally anti-symmetric functions $\D\bigwedge_{i=1}^n
L^2(\Lambda)$ which reflects the electronic nature of particles.\\
The general problem is to understand the behavior of
$H_\omega(\Lambda, n)$ in the thermodynamic limit $\Lambda \to \infty$
while $n / |\Lambda| \to \rho>0$; $\rho$ is the particle density. One
of the questions of interest is that of the behavior of the ground
state energy, say, $E_\omega(\Lambda, n)$ and of the ground state
$\Psi_\omega(\Lambda, n)$.\\
While the thermodynamic limit is known to exist for various
quantities and in various settings (see~\cite{MR3022666} for the
micro-canonical ensemble that we study in the present paper
and~\cite{MR2905791} for the grand canonical ensemble), we don't know
of examples, except for the model studied in the present paper, where
the limiting quantities have been studied. In particular, it is of
interest to study the dependence of these limiting quantities in the
different physical parameters like the density of particles, the
strength of the disorder or the interaction potential.\\
As we shall argue now, for these questions to be tractable, one needs a
good description of the spectral data of the underlying one particle
random model.
\subsubsection{Why the pieces model?}
\label{sec:motiv-piec-model}
In order to tackle the question of the behavior of $n$-electron ground
state, let us first consider the system without interactions. This is
not equivalent to a one-particle system as Fermi-Dirac statistics play
a crucial role.\\
Let us assume our one particle model is ergodic and admits an
integrated density of states (see~\eqref{eq:6} and
e.g.~\cite{MR2404176,MR94h:47068}). As described above for the
pieces model, the ground state of the $n$ non interacting
electrons is given by~\eqref{eq:7} and its energy per particle is given
by
\begin{equation}
  \label{eq:68}
  \frac{E^0_\omega(\Lambda,n)}{n}=\frac1n\sum_{j=1}^nE_{j,\omega}^\Lambda
  =|\Lambda|\int_{-\infty}^{E_{n,\omega}^\Lambda}E\,
  d\left[\frac{\#\{\text{eigenvalues of }H_\omega(\Lambda)\text{ below
      }E\}}{|\Lambda|}\right]
\end{equation}
where $E_{n,\omega}^\Lambda$ is the $n$-th eigenvalue of the one
particle random Hamiltonian $H_\omega(\Lambda)$, i.e., the smallest
energy $E$ such that
\begin{equation}
  \label{eq:70}
  \frac{\#\{\text{eigenvalues of }H_\omega(\Lambda)\text{ below
    }E\}}{|\Lambda|}=\frac n{|\Lambda|}.
\end{equation}
Here, we have kept the notations of the beginning of
section~\ref{sec:ground-state}.\\
The existence of the density of states, say $N(E)$,
(see~\eqref{eq:6}), then, ensures the convergence of $E(\Lambda,n)$ to
a solution to the equation $N(E)=\rho$, say $E_\rho$. Thus, to control
the non interacting ground state, one needs to control all (or at
least most of) the energies of the random operator $H_\omega(\Lambda)$
up to some macroscopic energy $E_\rho$. In particular, one needs to
control simultaneously a number of energies of $H_\omega(\Lambda)$
that is of size the volume of $\Lambda$. \\
To our knowledge, up to now, there are no available mathematical
results that give the simultaneous control over that many eigenvalues
for general random systems. The results dealing with the spectral
statistics of (one particle) random models deal with much smaller
intervals: in~\cite{MR97d:82046}, eigenvalues are controlled in
intervals of size $K/|\Lambda|$ for arbitrary large $K$ if $\Lambda$
is sufficiently large; in~\cite{Ge-Kl:10,MR3070753}, the interval is
of size $|\Lambda|^{1-\beta}$ for some not too large positive $\beta$.\\
The second problem is that all these results only give a very rough
picture of the eigenfunctions, a picture so rough that it actually is
of no use to control the effect of the interaction on such states: the
only information is that the eigenstates live in regions of linear
size at most $\log|\Lambda|$ and decay exponentially outside such
regions (see, e.g.,~\cite{Ge-Kl:10} and references therein).\\
The pieces model that we deal with in the present paper exhibits
the typical behavior of a random system in the localized regime: for
$H_\omega(\Lambda)$, 
\begin{itemize}
\item the eigenfunctions are localized (on a scale $\log|\Lambda|$)
\item the localization centers and the eigenvalues satisfy Poisson
  statistics.
\end{itemize}
The advantage of the pieces model is that the eigenfunctions and
eigenvalues are known explicitly and easily controlled. This is a
consequence of the fact that a crucial quantum phenomenon is missing
in the pieces model, namely, tunneling. Of course, once the particles
do interact with each other, tunneling is again
re-enabled.\vskip.2cm\noindent
All of this could lead one to think that the pieces model is very
particular. Actually, at low energies, general one-dimensional random
models exhibit the same characteristics as the pieces model up to
some exponentially small errors which are essentially due to tunneling
(see~\cite{Kl:13a}).\\
It seems reasonable to guess that the behavior will be comparable for
general random operators in higher dimensions and, thus, that the
results of the present paper on interacting electrons in a random
potentials should find their analogues for these models.
\subsubsection{Outline of the paper}
\label{sec:outline-main-results}
In section~\ref{sec:probl-prel}, after rescaling the parameters of the
problem so as to send $\mu$ to $1$ and $\rho$ to $\rho/\mu$, we first
discuss the validity of our results in a more general asymptotic
regime in $\mu$ and $\rho$. We, then, gather some basic but crucial
statistical properties of the distribution of the pieces. We first
describe the free electrons. For the pieces model, a statistical
analysis of the distribution of pieces gives exact expressions for the
one-particle integrated density of states and the Fermi energy in
Proposition~\ref{prop:IDSandFermiEnergy}. We also study the non
interacting model and introduce notations for later use.
\vskip.2cm\noindent
In section~\ref{sec:main-results-proofs-1}, we first introduce the
occupation numbers (i.e., the number of particles a given state puts
in each piece); the existence of the occupation numbers is tantamount
to the existence of a particular orthogonal sum decomposition of the
Hamiltonian $H_\omega^U(\Lambda,n)$. We prove that the ground state of
$H_\omega^U(\Lambda,n)$ restricted to a fixed occupation space is non
degenerate and, from this result, derive
Theorem~\ref{th:asNonDegOfGroundState}, the almost sure non degeneracy
of the ground state for real analytic interaction.\\
Next, still in section~\ref{sec:main-results-proofs-1}, we prove the
asymptotic formula for the interacting ground state energy per
particle. The proof relies essentially on the minimizing properties of the
ground state. This minimizing property yields a good description for
the occupation numbers associated to a ground state. To get this
description, we first study the ground state of the Hamiltonian
$H_\omega^{U^p}(\Lambda,n)$ where the interactions have been cut-off
at infinity (i.e., $U^p$ is compactly supported).  We construct an
approximate ground state $\Psi^{\text{opt}}$ which can essentially be
thought of as the ground state for the Hamiltonian
$H_\omega^{U^p}(\Lambda,n)$ restricted to the pieces shorter that
$3\ell_{\rho,\mu}$. Then, letting
$W^r(\Lambda,n):=H_\omega^{U}(\Lambda,n)-H_\omega^{U^p}(\Lambda,n)$ be
the long range behavior of the interactions, one has
\begin{equation*}
  E^{U^p}_\omega(\Lambda,n)\leq E^{U}_\omega(\Lambda,n)
  \leq\langle H_\omega^{U^p}(\Lambda,n)
  \Psi^{\text{opt}},\Psi^{\text{opt}}\rangle+\langle W^r(\Lambda,n)
  \Psi^{\text{opt}},\Psi^{\text{opt}}\rangle
\end{equation*}
The minimizing property of $\Psi^{\text{opt}}$ yields
\begin{equation*}
  E^{U^p}_\omega(\Lambda,n)\geq \langle H_\omega^{U^p}(\Lambda,n)
  \Psi^{\text{opt}},\Psi^{\text{opt}}\rangle +n\,
  o(\rho_\mu\,\mu^{-1}\,\ell^{-3}_{\rho,\mu}) 
\end{equation*}
(see Theorem~\ref{th:PsiUmPsiOptEnergyEstimate}).\\
On the other hand, the decay assumption \textbf{(HU)} on $U$ and the
explicit construction of $\Psi^{\text{opt}}$ yield
\begin{equation*}
  \langle W^r(\Lambda,n)
  \Psi^{\text{opt}},\Psi^{\text{opt}}\rangle=n\,
  o(\rho_\mu\,\mu^{-1}\,\ell^{-3}_{\rho,\mu})
\end{equation*}
(see Proposition~\ref{pro:2}).\\
This yields the proof of Theorem~\ref{th:EnergyAsymptoticExpansion}.\\
In the course of these proofs, we also prove a certain number of
estimates on the distance between the occupation numbers of the
interacting ground state(s) to the state $\Psi^{\text{opt}}$.
\vskip.2cm\noindent
Section~\ref{sec:proof-theorem6} is devoted to the proofs of
Theorems~\ref{thr:2} and~\ref{thr:7}. Therefore, we transform the
bounds of the distance between occupation numbers into bounds on the
trace class norms of the difference between the one (and the two)
particle densities of the interacting ground state(s) and the state
$\Psi^{\text{opt}}$.\\
In Theorems~\ref{thr:4} (resp. Theorem~\ref{thr:5}), we derive general
formulas for the one particle (resp. two particles) density of a state
expressed in a certain well chosen basis of $\frH^n(\Lambda)$. One of
the main steps on the path going from occupation number bounds to the
trace class norm bounds is to prove that, in most pieces, once the
particle number is known, the state must be in the ground state for
the given particle number. This is the purpose of Lemma~\ref{le:21};
it relies on the minimizing properties of the ground state; actually,
it is proved for a larger set of states, states satisfying a certain
energy bound.\\
We then use Theorems~\ref{thr:4} (resp. Theorem~\ref{thr:5}) to derive
Theorems~\ref{thr:2} (resp. Theorem~\ref{thr:7}).\vskip.2cm\noindent
Section~\ref{sec:almost-sure-conv} is devoted to the proof of the
almost sure convergence of the ground state energy per particle. The
proof is essentially identical to that found in~\cite{MR3022666}
except for the sub-additive estimate crucial to the proof. This
estimate is provided by Theorem~\ref{thr:8}.\vskip.2cm\noindent
In section~\ref{sec:two-part-probl}, we prove
Proposition~\ref{prop:TwoElectronProblem} as well as a number of
estimates on the ground states and ground state energies for a finite
number of electrons living in a fixed number of pieces and interacting.
\vskip.2cm\noindent
In three appendices, we gather a number of results used in the main body
of paper. In appendix~\ref{sec:auxil-results-calc}, we prove the
results on the statistics of the pieces stated in
section~\ref{sec:probl-prel}. Appendix~\ref{sec:simple-lemma-trace} is
devoted to a simple technical lemma used intensively in the derivation
of Theorems~\ref{thr:2} and~\ref{thr:7} in
section~\ref{sec:proof-theorem6}. Appendix~\ref{sec:proj-totally-antisym}
is devoted to anti-symmetric tensor products.

%%% Local Variables: 
%%% mode: latex
%%% TeX-master: "PiecesModelGroundState"
%%% ispell-local-dictionary: "american"
%%% End: 

\tableofcontents

% =========================
% PROBLEM AND PRELIMINARIES
% =========================
\section{Preliminary results}
\label{sec:probl-prel}
In this section, we state a number results on the Luttinger-Sy model
defined in section~\ref{sec:introduction} on which our analysis is
based. We first recall some results on the thermodynamic limit
specialized to the pieces model. Then, we describe the statistics of
the eigenvalues and eigenfunctions of the pieces model defined
in~\eqref{eq:2}; in the case of the pieces model, it suffices
therefore to describe the statistics of the pieces (see
section~\ref{sec:analys-one-part}).\\
In section~\ref{sec:free-electrons}, we describe the non interacting
system of $n$ electrons.
\subsection{Rescaling the operator}
\label{sec:rescaling}
Consider the scaling $\wtx = \mu x$, that is, define
\begin{equation}
  \label{eq:scalingX}
  \begin{aligned}
    S_\mu:\ \bigwedge_{j = 1}^nL^2([0,L])&\to \bigwedge_{j =
      1}^nL^2([0,\wtL]) \\ u&\mapsto S_\mu  u\quad \text{ where } (S_\mu 
    u)(x)=\mu^{n/2} u(\mu x)\quad\text{ and }\quad \wtL = \mu L.
  \end{aligned}
\end{equation}
One then computes
\begin{equation*}
  S^*_{\mu} H_\omega(L, n) S_\mu=\mu^2 \wtH_{\omega}(\wtL, n)
\end{equation*}
where $\wtH_{\omega}(\wtL, n)$ is the interacting pieces model on the
interval $[0,\wtL]$ defined by a Poisson process of intensity $1$ and
with pair interaction potential
\begin{equation}
  \label{eq:278}
  U^\mu(\cdot)=\mu^{-2} U(\mu^{-1}
\cdot).
\end{equation}
For $\wtH_{\omega}(\wtL, n)$, the thermodynamic limit becomes
\begin{equation*}
  \frac{n}{\wtL} = \frac{n}{\mu L} \to \frac{\rho}{\mu} = \rho_\mu.
\end{equation*}
We shall prove
Theorems~\ref{th:EnergyAsymptoticExpansion},~\ref{thr:2}
and~\ref{thr:7} under the additional assumption $\mu=1$.  Let us now
explain how Theorems~\ref{th:EnergyAsymptoticExpansion},~\ref{thr:2}
and~\ref{thr:7} get modified when one goes from $\mu=1$ to arbitrary
$\mu$.  \vskip.2cm\noindent
If one denotes by $\gamma^\mu$ the constant defined by
Proposition~\ref{prop:TwoElectronProblem} applied to the interaction
potential $U^\mu$ instead of $U$, a direct computation yields
$\gamma^\mu=\mu\gamma$.\\
In the same way, a direct computation yields that $Z^\mu$, the
analogue of $Z$ in assumption~\textbf{(HU)} for $U^\mu$, is given by
$\D Z^\mu(\cdot)=\mu^2Z(\mu^{-1}\cdot)$. Thus, for the function
$f_{Z^\mu}$ (see~\eqref{eq:98},~\eqref{eq:274} and~\eqref{eq:1})
defined for $U^\mu$, see~\eqref{eq:278}, one obtains
$f_{Z^\mu}(\cdot)=\mu^2f_{Z}(\mu^{-1}\cdot)$.  This suffices to obtain
Theorems~\ref{thr:2} and~\ref{thr:7} for $\mu$ arbitrary fixed from
the case $\mu=1$.\vskip.2cm\noindent
From now on, as we fix $\mu=1$, we shall write drop the sub- or
superscript $\mu$ and write, e.g., $\ell_\rho$ for $\ell_{\rho,\mu}$,
$E_\rho$ for $E_{\rho,\mu}$, etc. Similarly, the dependence on the
random parameter $\omega$ will be frequently dropped so as to simplify
notations.
\subsubsection{Other asymptotic regimes}
\label{sec:other-asympt-regim}
In the introduction, for the sake of simplicity we chose to state our
results at fixed $\mu$ and sufficiently small $\rho$ (depending on
$\mu$). Actually, the results that we obtained stay correct under less
restrictive conditions on $\mu$ and $\rho$. The conditions that are
required are the following. Fix $\mu_0>0$; then,
Theorems~\ref{th:EnergyAsymptoticExpansion},~\ref{thr:2}
and~\ref{thr:7} stay correct as long as $\mu\in(0,\mu_0)$, $\rho_\mu$
be sufficiently small and $\ell_{\rho,\mu}$ sufficiently large
depending only on $\mu_0$. Let us now explain this.\\
Therefore, we analyze the remainder terms of~\eqref{eq:129} (thus,
of~\eqref{eq:222}). The second term in the last equality
in~\eqref{eq:129} multiplied by $\mu^2$ (to rescale energy properly,
see above) becomes
\begin{equation*}
  \pi^2\mu^2 \gamma^\mu_*
  \frac{\rho_\mu}{|\log{\rho_\mu}|^3}=\pi^2
  \,\gamma^\mu_*\,\mu^{-1}\,\rho_\mu\,\ell_{\rho,\mu}^{-3}+
  o\left(\rho_\mu\,\ell_{\rho,\mu}^{-3}\right)
\end{equation*}
by~\eqref{eq:67}. Note that, by~\eqref{eq:12},
$\gamma^\mu_*\,\mu^{-1}$ stays bound from above and below as
$\mu\to0^+$.\\ 
The remainder term in the last equality in~\eqref{eq:129} multiplied
by $\mu^2$ (to rescale energy properly, see above) becomes
\begin{equation*}
  \mu^2\frac{\rho_\mu}{|\log{\rho_\mu}|^3}
  O\left(f_{Z^\mu}[|\log{\rho_\mu}|)]\right) =
  \frac{\rho_\mu\,\mu^4}{\ell^3_{\rho,\mu}}
  O\left(f_Z\left[\ell_{\rho,\mu}(1+o(1))\right] \right)=o
  \left(\frac{\rho_\mu\,\mu^{-1}}{\ell^3_{\rho,\mu}}\right)
\end{equation*}
when $\rho_\mu\to0$ and $\ell_{\rho,\mu}\to+\infty$ while $\mu$ stays
bounded.\\
This then yields Theorem~\ref{th:EnergyAsymptoticExpansion} for
$(\mu,\rho)$ arbitrary in the regime described above from the case
$\mu=1$ and $\rho$ small.\vskip.2cm\noindent
To obtain Theorems~\ref{thr:2} and~\ref{thr:7} for $\mu$ arbitrary, we
just use $Z^\mu(\cdot)=\mu^2Z(\mu^{-1}\cdot)$ and the fact that $Z$ is
decaying; indeed, this implies that
\begin{equation*}
  Z^\mu(2|\log\rho_\mu|)=\mu^2Z(2\ell_{\rho,\mu}(1+o(1)))\leq 
  \mu^2Z(\ell_{\rho,\mu})
\end{equation*}
when $\rho_\mu\to0$ and $\ell_{\rho,\mu}\to+\infty$ while $\mu$ stays
bounded.\\
This suffices to obtain Theorems~\ref{thr:2} and~\ref{thr:7} for
$(\mu,\rho)$ arbitrary in the regime described above from the case
$\mu=1$ and $\rho$ small.\vskip.2cm\noindent
From now on, we fix $\mu=1$ and assume $\rho$ be small. Thus, we shall
drop the sub- or superscript $\mu$ and write, e.g., $\ell_\rho$ for
$\ell_{\rho,\mu}$, $E_\rho$ for $E_{\rho,\mu}$, etc. Similarly, the
dependence on the random parameter $\omega$ will be frequently dropped
so as to simplify notations.
\subsection{The analysis of the one-particle pieces model}
\label{sec:analys-one-part}
Most of the proofs of the results stated in the present section can be
found in Appendix~\ref{sec:auxil-results-calc}.\\
Recall that we partition $[0, L]$ using a Poisson process of intensity
$1$ and write
\begin{equation}
  \label{eq:IntervalDivision}
  [0, L] = \bigcup_{j = 1}^{m(\omega)} \Delta_j(\omega).
\end{equation}
Note that, by a standard large deviation principle, for
$\beta\in(0,1/2)$, with probability at least $1-e^{-L^\beta}$, one has
$m=L+O(L^\beta)$.\\
Moreover, with probability one,
\begin{itemize}
\item $\D\min_{1 \leq j \leq m(\omega)}|\Delta_j(\omega)|>0$,
\item if $j\not=j'$ then
  $\D\frac{|\Delta_j(\omega)|^2}{|\Delta_{j'}(\omega)|^2}\not\in\Q$.
\end{itemize}
Thus, distinct pieces generate distinct Dirichlet Laplacian energy
levels. In particular, with probability one, all the eigenfunctions of
the one-particle Hamiltonian $H_\omega(L)=H_\omega(L,1)$ are supported
on a single piece $\Delta_j(\omega)$ and the corresponding eigenvalues
are simple.\\
Hence, we will enumerate the eigenvalues and the eigenfunctions of
$H_\omega(L)$ using a two-component index $(\Delta_j, k)$ where
\begin{itemize}
\item $\Delta_j$ is the piece of the partition
  \eqref{eq:IntervalDivision} on which the eigenfunction is supported,
\item $k$ is the index of the eigenvalue within the ordered list of
  eigenvalues of this piece,
\end{itemize}
i.e.,
\begin{equation*}
  \psi_{(\Delta_j, k)}(x) = \sqrt{\frac{2}{|\Delta_j|}} 
  \sin\left(\dfrac{\pi k (x - \inf{\Delta_j})}{|\Delta_j|}\right)
  \car_{\Delta_j}(x) 
\end{equation*}
and the corresponding energy
\begin{equation}
  \label{eq:17}
  E_{(\Delta_j, k)} = \left(\frac{\pi k}{|\Delta_j|}\right)^2 \text{.}
\end{equation}
Let $\calP=\calP(\omega)$ denote the set of all available indices
enumerating single-particle states, i.e., $\calP =\{\Delta_j\}_{j =
  1}^{m(\omega)} \times \bbN$.\\
In parallel to this two-component enumeration system, we will use a
direct indexing procedure: $\{(E_j,\psi_j)\}_{j \in \bbN}$ are the
eigenvalues and associated eigenfunctions of the one particle
Hamiltonian $H_\omega(L)$ counted with multiplicity ordered
with increasing energy.
\subsection{The statistics of the pieces}
\label{sec:statistics-pieces}
We first study the statistical distribution of the pieces generated by
the Poisson process. We will primarily be interested in the joint
distributions of their lengths. These statistics immediately provide
the statistics of the eigenvalues and eigenfunctions of the pieces
model. These results are presumably well known; as we don't know a
convenient reference, we provide their proofs in
Appendix~\ref{sec:auxil-results-calc} for the sake of
completeness.\\
In the sequel, the probability of the events will typically be
$1-O(L^{-\infty})$: we recall that $A_k=O(k^{-\infty})$ if $\forall
N\geq0$, $\D\lim_{k\to+\infty}k^{N}A_k=0$. Actually, the proofs show
that the probabilities lie at an exponentially small distance from
$1$, i.e., $O(L^{-\infty})=e^{-L^\beta}$ for some $\beta>0$.\\
We prove
\begin{Pro}
  \label{pro:3}
  With probability $1-O(L^{-\infty})$, the largest piece has length
  bounded by $\log L\cdot\log\log L$, i.e.,
  \begin{equation*}
    \max_{1\leq k\leq m(\omega)}|\Delta_k(\omega)|\leq \log L\cdot\log\log
    L.
  \end{equation*}
\end{Pro}
\noindent On the distribution of the length of the pieces, one proves
\begin{proposition}
  \label{prop:IntervStatistics}
  Fix $\beta\in(2/3,1)$. Then, for $L$ large, for any
  $(a_L,b_L)\in[0,\log L\cdot\log\log L]^2$, with probability
  $1-O(L^{-\infty})$, the number of pieces of length contained in
  $[a_L,a_L+b_L]$ is equal to
  \begin{equation*}
    e^{-a_L} (1-e^{-b_L})\cdot L+R_L\cdot L^\beta\quad \text{ where } 
    \quad|R_L|\leq \kappa
  \end{equation*}
  and the positive constant $\kappa$ is independent of $a_L,b_L$.
\end{proposition}
\noindent The proof of Proposition~\ref{prop:IntervStatistics} is
given in Appendix~\ref{sec:auxil-results-calc}.\\
We will also use the joint distributions of pairs and triplets of
pieces that are close to each other. We prove
\begin{proposition}
  \label{prop:IntervStatistics2}
  Fix $\beta\in(2/3,1)$. Then, for any $a,b$ positive and $b,d,g,f$
  all non negative, with probability $1-O(L^{-\infty})$, the number of
  pairs of pieces such that
  \begin{itemize}
  \item the length of the left most piece is contained in $[a,a+b]$,
  \item the length of the right most piece is contained in $[c,c+d]$,
  \item the distance between the two pieces belongs to $[g,g+f]$
  \end{itemize}
  is equal to
  \begin{equation}
    \label{eq:97}
    f\,e^{-a-c} (1-e^{-b})(1-e^{-d})\cdot L+R_L\cdot L^\beta
    \quad\text{where}\quad |R_L|\leq
    \kappa
  \end{equation}
  and the positive constant $\kappa$ may depend on $(a,b,c,d,f,g)$.
\end{proposition}
\noindent For pairs of pieces, we shall also use
\begin{proposition}
  \label{prop:NeighborssStatistics}
  For $\ell,\ell',d>0$, with probability $1-O(L^{-\infty})$, one
  has
    \begin{equation*}
      \#\left\{
        \begin{aligned}
          \text{pairs of pieces at most at a dis-}\\
          \text{tance $d$ from each other such that}\\
          \text{the left most piece longer than }\ell,\\
          \text{the right most piece longer than }\ell'.
        \end{aligned}
      \right\}\leq (2+d)e^{-\ell-\ell'} L.
    \end{equation*}
\end{proposition}
\noindent Finally, for triplets of pieces, we shall use
\begin{proposition}
  \label{prop:NeighborssStatistics2}
  For $\ell,\ell',\ell'',d>0$, with probability
  $1-O(L^{-\infty})$, one has
    \begin{equation*}
      \#\left\{(\Delta,\ \Delta',\ \Delta'')\text{ s.t.}
        \left|\begin{aligned}
          &\Delta'\text{ between }\Delta\text{  and }\Delta''\\
          &\text{dist}(\Delta,\Delta')\leq d,\ 
          \text{dist}(\Delta',\Delta'')\leq d\\
          &|\Delta|\leq \ell,\ |\Delta'|\leq \ell',\ |\Delta''|\leq
          \ell''.
        \end{aligned}
      \right.\right\}\leq (2+d^2) e^{-\ell-\ell'-\ell''} L.
    \end{equation*}
\end{proposition}
\noindent As a straightforward consequence of
Proposition~\ref{prop:IntervStatistics}, exploiting the
formula~\eqref{eq:17} for the Dirichlet eigenvalues of the Laplacian
on an interval, one obtains the explicit formula~\eqref{eq:6} for the
one-particle integrated density of states for the pieces model defined
in~\eqref{eq:6} (here, $\mu=1$) That is, one proves
\begin{proposition}[The one particle IDS]
  \label{prop:IDSandFermiEnergy}
  The one-particle integrated density of states for the pieces model
  is given by
  \begin{equation}\label{eq:NLifschitzTail}
    N(E) = \frac{\exp(-\ell_E)}{1 - \exp(-\ell_E)}\car_{E>0}
  \end{equation}
  where $\ell_E$ is defined in~\eqref{eq:6}.
\end{proposition}
\noindent Formula~\eqref{eq:NLifschitzTail} was already obtained
in~\cite{PhysRevA.7.701}; in Appendix~\ref{sec:facts-poiss-proc}, we
give a short proof for the readers convenience.\\
Recalling the scaling defined in section~\ref{sec:rescaling}
immediately yields~\eqref{eq:6} for general $\mu$.
\subsection{Free electrons}
\label{sec:free-electrons}
Understanding the system without interactions will be key to answering
the main questions raised in the present work. For free electrons,
i.e., when the interactions are absent, $U \equiv 0$, the energy per
particle $\densEn^0(\rho)$ can be expressed in terms of one-particle
density of states measure.
\subsubsection{The ground state energy per particle}
\label{sec:ground-state-energy-1}
Recall that (see Theorem~\ref{th:EnergyAsymptoticExpansion}), for a
density of particles $\rho$, the \emph{Fermi energy} $E_\rho$ is a
solution of the equation $N(E_\rho) = \rho$. In the present case, as
$N$ is continuous and strictly increasing from $0$ to $+\infty$, the
solution to this equation is unique for any $\rho>0$. The length of
the interval whose Dirichlet Laplacian has the Fermi energy $E_\rho$
as ground state energy is the Fermi length $\ell_\rho$ given by
\begin{equation}
  \label{eq:86}
 \ell_\rho:=\pi/\sqrt{E_\rho} 
\end{equation}
As a direct corollary to~\eqref{eq:6} (recall that $\mu=1$) or
equivalently Proposition~\ref{prop:IDSandFermiEnergy}, we see that the
Fermi energy is given by
\begin{equation}
  \label{eq:FermiEnergyExpression}
  E_\rho = \pi^2 \left(\log(\rho^{-1} + 1)\right)^{-2}
  \sim \pi^2 |\log{\rho}|^{-2}\quad\text{when}\quad \rho \to 0
\end{equation}
and the Fermi length by:
\begin{equation}
  \label{eq:ellRho}
  \ell_\rho = \log\left(\rho^{-1} + 1\right) \sim 
  \left|\log{\rho}\right|\quad\text{when}\quad \rho \to 0 \text{.}
\end{equation}
We recall
\begin{proposition}[{\cite[Theorem~5.13 and Lemma~5.14]{MR3022666}}]
  \label{pro:2}
  Let $E_{n,\omega}^\Lambda$ denote the $n$-th energy level of
  $H_\omega(L)$ (counting multiplicity). Then, $\omega$-a.s., one has
  \begin{equation}
    \label{eq:convEnFerm}
    E_{n,\omega}^\Lambda\vers{\substack{L\to\infty\\n/L\to\rho}}E_\rho 
    \quad\text{and}\quad
    \densEn^0(\rho) = \frac{1}{\rho} \int_{-\infty}^{E_\rho}\, E\,
    {\rmd}N(E) \text{.} 
  \end{equation}
\end{proposition}
\noindent Proposition~\ref{pro:2} follows easily from
Lemma~\ref{lem:FermiEnergyConvergenceRate},~\eqref{eq:68},~\eqref{eq:70}
and~\eqref{eq:densEtatsPoissonModel1}.\\
We see that
\begin{itemize}
\item the highest energy level occupied by a system of non interacting
  electrons tends to the Fermi energy in the thermodynamic limit;
\item the $n$-electron ground state energy per particle is the energy
  averaged with respect to the density of states measure of the
  one-particle system conditioned on energies less than the Fermi
  energy.
\end{itemize}
Combining formulas~\eqref{eq:FermiEnergyExpression}
and~\eqref{eq:convEnFerm}, one can expand $\densEn^0(\rho)$ into
inverse powers of $\log{\rho}$ up to an arbitrary order. Taking the
scaling defined in section~\ref{sec:rescaling} into
account,~\eqref{eq:convEnFerm} immediately implies~\eqref{eq:5}.
\subsubsection{The eigenfunctions}
\label{sec:eigenfucntions}
Let us now describe the eigenfunctions of $H_\omega^0(L, n)$.  Let us
recall that $(E_p)_{p\in\mathcal{P}}$ are the eigenvalues of the
one-particle operator $H_\omega(L)$ and $(\psi_p)_{p\in\mathcal{P}}$
are the corresponding normalized eigenfunctions; here $p$ in $\calP$
is a ( piece - energy level ) index.  The $n$-electron eigenstates
without interactions are given by the following procedure. Pick a set
$\alpha:=\{\alpha_1,\dots,\alpha_n\} \subset \calP$ of $n$ indices,
$\card{\alpha} = |\alpha| = n$.  The normalized eigenstate associated
to $\alpha$ is given by the Slater determinant
\begin{equation}
  \label{eq:OmegaAlphaConstruction}
  \Psi_\alpha(x^1, x^2, \cdots, x^n) :=
  \psi_{\alpha_1}\wedge\cdots\wedge\psi_{\alpha_n}:= \frac{1}{\sqrt{n!}} 
  \det\left(\psi_p(x^j)\right)_{\substack{p \in \alpha\\1\leq j \leq n}}.
\end{equation}
One easily checks that $\D
(H_\omega^0(\Lambda,n)-E_\alpha)\Psi_\alpha=0$ for the energy
$E_\alpha$ defined by
\begin{equation}
  \label{eq:zetaAlphaConstruction}
E_\alpha = \sum_{p \in \alpha} E_p.
\end{equation}
The subset $\alpha$ indicates which one-particle energy levels are
occupied in the multi-particle state $\Psi_\alpha$. For instance, in
the ground state of $n$ electrons, one chooses the states with lowest
possible energy.
\begin{notation}
  \label{not:Particle}
  For a Slater determinant $\Psi_\alpha$
  (see~\ref{eq:OmegaAlphaConstruction})) and $p\in\alpha$, we will
  refer to the one-particle functions $\psi_p$ as \emph{particles}
  that constitute the $n$-electron state indexed by $\alpha$.
  Moreover, with a slight abuse of terminology, we will refer to an
  multi-index $\alpha$ as a ($n$-electron) state and to $p$ in $\alpha$
  as a particle.
\end{notation}
%

%%% Local Variables: 
%%% mode: latex
%%% TeX-master: "PiecesModelGroundState"
%%% ispell-local-dictionary: "american"
%%% End: 

% =======================
% MAIN RESULTS AND PROOFS
% =======================
\section{The asymptotics for the ground state energy per particle}
\label{sec:main-results-proofs-1}
In this section, we prove Theorem~\ref{th:EnergyAsymptoticExpansion}
on the asymptotic expansion of the ground state energy per particle in
terms of small particle density. We assume that the pair interaction
potential $U$ satisfies condition \textbf{(HU)}.
\subsection{Decomposition by occupation numbers}
\label{sec:decomp-occup}
We give a definition of the number of particles occupying a given
piece. Therefore, we shall use the special structure of the
Hamiltonian $H^0_\omega(\Lambda,n)$, that is, that of
$H_\omega(L)$ (see~\eqref{eq:Hn0Definition} and~\eqref{eq:2}).\\
Fix $\omega$. Recall that $(\Delta_j(\omega))_{1\leq j\leq m})$ are
the pieces defined in~\eqref{eq:IntervalDivision} ($m=m(\omega)$). The
one particle space is then decomposed into
\begin{equation}
  \label{eq:18}
  L^2(\Lambda)=L^2([0,L])=
  \overset{\perp}{\bigoplus_{1\leq j\leq m}}L^2(\Delta_j(\omega)).
\end{equation}
Thus, for the $n$-particle space $\frH^n$ (see~\eqref{eq:19}), we
obtain the decomposition
\begin{equation}
  \label{eq:frHfrHQdecomposition}
  \frH^n=\frH^n(\Lambda)=\bigwedge_{j = 1}^n L^2(\Lambda)= 
  \bigoplus_{\substack{Q=(Q_1,\cdots,Q_m)\in
      \bbN^m\\ Q_1+\cdots Q_m=n}} \frH_Q
\end{equation}
where we have defined
\begin{definition}
  \label{def:Occupation}
  For $Q=(Q_1,\cdots,Q_m)\in\bbN^m$ s.t. $Q_1+\cdots Q_m=n$, the space
  of states of fixed occupation $Q$ denoted by $\frH_Q$ is given by
  \begin{equation}
    \label{eq:28}
    \frH_Q=\bigwedge_{j=1}^m\left(\bigwedge_{k=1}^{Q_j}
      L^2(\Delta_j(\omega))\right).
  \end{equation}
  Here, as usual, we set $\bigwedge_{k=1}^0 L^2(\Delta_j(\omega))=\C$.
\end{definition}
\noindent An occupation $Q$ is a multi-index of length $m$ and of
``modulus'' $n$. Note that, as
$\Delta_j(\omega)\cap\Delta_{j'}(\omega)=\emptyset$ for $j\not=j'$, we
can identify
\begin{equation*}
  \frH_Q=\bigotimes_{j=1}^m\left(\bigwedge_{k=1}^{Q_j}
    L^2(\Delta_j(\omega))\right).
\end{equation*}
\begin{Rem}
  \label{rem:1}
  The spaces of fixed occupation could also be defined starting from
  the eigenstates of $H^0_\omega(L,n)$ as
  in~\cite{Veniaminov_PhDthesis}. Indeed, each of the eigenstates of
  $H^0_\omega(L,n)$, the non interacting Hamiltonian, belongs to a
  state of fixed occupation. More precisely, if $\Psi_\alpha \in
  \frH^n$ is the eigenstate of $H^0_\omega(L,n)$ given
  by~\eqref{eq:OmegaAlphaConstruction} where $\alpha \subset\calP$,
  $\card{\alpha} = n$, then, defining the occupation $\D Q(\alpha) =
  (Q_1(\alpha),\cdots,Q_m(\alpha))$ where, for $1\leq j\leq m$,
  $Q_j(\alpha):=\#\left\{p\in\alpha|\ \supp{\psi_p}=\Delta_j\right\}$,
  we see that $\Psi_\alpha\in\frH_Q$.
\end{Rem}
\noindent The following lemma is crucial in our analysis as it gives
global information on the structure of the ground state of the
Hamiltonian $H^U_\omega(L,n)=H^0_\omega(L,n)+W_n$. We prove
\begin{Le}
  \label{lem:occupationDecomposition}
  Let $\omega$ be fixed and let $\alpha$ and $\beta$ be two
  $n$-electron indices corresponding each to an eigenstate of
  $H^0_\omega(L,n)$.\\
  If their occupations are different, then the corresponding
  $n$-particle states do not interact:
  \begin{equation*}
    Q(\alpha) \ne Q(\beta) \Rightarrow
    \langle\Psi_\alpha,W_n\Psi_\beta\rangle = 0 .
  \end{equation*}
\end{Le}
\begin{proof}
  If $\alpha$ and $\beta$ have different occupation numbers, the
  supports of $\Psi_\alpha$ and $\Psi_\beta$ in $\Lambda^n$ intersect
  at a set of measure zero: indeed, these supports are obtained by
  symmetrizing different collections of products of pieces (with
  repetitions for the pieces that are occupied more than once):
  \begin{equation*}
    Q(\alpha) \ne Q(\beta) \quad \Rightarrow \quad
    \meas\left(\supp{\Psi_\alpha} \cap \supp{\Psi_\beta}\right)=0.
  \end{equation*}
  The latter means that $\Psi_\alpha \cdot\Psi_\beta\equiv0$ as a
  function in $L^2\left(\Lambda^n\right)$.  Then, clearly, by
  definition, for the matrix elements, one obtains
  \begin{equation*}
    \langle\Psi_\alpha,W_n\Psi_\beta\rangle = \int_{\Lambda^n} W_n(\bfx)
    \Psi_\alpha(\bfx) 
    \Psi_\beta^\ast(\bfx) \rmd{\bfx} = 0.
  \end{equation*}
  Lemma~\ref{lem:occupationDecomposition} is proved.
\end{proof}
\noindent As an immediate corollary to
Lemma~\ref{lem:occupationDecomposition}, we obtain
\begin{Cor}[Decomposition by occupation]
  \label{cor:OccupationSubspaces}
  Fix $\omega$. For any $Q\in\N^m$ (here and in the sequel,
  $\N=\{0,1,\cdots\}$), $m=m(\omega)$, the subspace $\frH_Q$ are
  invariant under the action of the $n$-particle Hamiltonian
  $H^U_\omega(L,n)=H^0_\omega(L,n)+W_n$, i.e.,
  \begin{equation}
    \label{eq:frHQinvar}
    (H^U_\omega(L,n)+i)^{-1}\frH_Q \subset \frH_Q.
  \end{equation}
  Thus, the total Hamiltonian $H^U_\omega(L,n)$ is decomposed
  according to~\eqref{eq:frHfrHQdecomposition} in direct sum of its
  parts $H_Q$ on subspaces of fixed occupation, i.e.,
  \begin{equation}
    \label{eq:HQdecomposition}
    H^U_\omega(L,n) = 
    \bigoplus_{\substack{Q \in \bbN^m\\ Q_1+\cdots+Q_m = n}} H_{Q},
  \end{equation}
  where $H_Q = \left.H^U_\omega(L,n)\right|_{\frH_Q}$.
\end{Cor}
\begin{remark}
  All terms of this decomposition as well as the number of pieces $m$
  depend on the randomness $\omega$, i.e., the configuration of
  pieces.
\end{remark}
\begin{proof}[Proof of
  Corollary~\textup{\ref{cor:OccupationSubspaces}}]
  Fix $\omega$. The space
  \begin{equation*}
    \mathcal{D}^n_\omega:=\Coi\left(\left(\bigcup_{1\leq j\leq
          m}\overset{\circ}{\Delta}_j(\omega)\right)^n\right)\bigcap\frH^n
  \end{equation*}
  is a core for $H^U_\omega(L,n)$. Here, $\overset{\circ}{\Delta}_j(\omega)$
  denotes the interior of $\Delta_j(\omega)$.\\
  It, thus, suffices to check that, for $\D
  H^U_\omega(L,n)\left(\frH_Q\cap\mathcal{D}^n_\omega\right)\subset\frH_Q$;
  this follows immediately from
  Lemma~\ref{lem:occupationDecomposition}. This ensures the existence
  of the decomposition~\eqref{eq:HQdecomposition} and completes the
  proof of Corollary~\ref{cor:OccupationSubspaces}.
\end{proof}
\noindent Corollary~\ref{cor:OccupationSubspaces} states that the
interaction operator $W_n$ is partially diagonalized in the basis of
eigenfunctions of $H^0_\omega(L,n)$, i.e., its matrix representation
has a block structure corresponding to the subspaces of constant
occupation.
\subsection{Almost sure non-degeneracy of the interacting ground
  state}
\label{sec:almost-sure-non}
We first restrict ourselves to spaces with fixed occupation to prove
\begin{Le}
  \label{le:2}
  Fix an occupation $Q$. The ground state of $
  \left.\left(H^U_\omega(L,n)\right)\right|_{\frH_Q}$ is
  non-degenerate.
\end{Le}
\begin{proof}
  To simplify notations, let us write $H=H^U_\omega(L,n)$ and
  $H^0=H^0_\omega(L,n)$. Let $(\Delta_{j_p})_{1\leq p\leq n}$ be the
  pieces such that $Q_{j_p}\geq1$; in the list $(\Delta_{j_p})_{1\leq
    p\leq n}$, each piece $\Delta_{j_p}$ is repeated $Q_{j_p}$
  times. We enumerate the pieces so that their left endpoints are non
  decreasing (i.e., from the leftmost piece to the rightmost
  piece). So, $p\mapsto j_p$ is non decreasing. Then, the operator
  $H^0_Q$ is the Dirichlet Laplacian on a space of anti-symmetric
  functions defined on the symmetrized domain
  \begin{equation}
    \label{eq:SymDomain}
    \Delta_Q = \Sym\left(\bigtimes_{p=1}^n \Delta_{j_p}\right) 
    := \bigcup_{\sigma\in\frS_n}\bigtimes_{p=1}^n\Delta_{\sigma(j_p)}.
  \end{equation}
  Anti-symmetric functions on the domain (\ref{eq:SymDomain}) that
  vanish on the boundary $\partial(\Delta_Q)$ are in one-to-one
  correspondence with functions defined on the domain
  \begin{equation}
    \label{eq:DesymDomain}
    \delta_Q=\left\{(x^1,\hdots, x^n)\text{
        s.t. }x^p\in\Delta_{j_p}\text{ and }
      x^p\leq x^q \text{ for }p<q\right\}
  \end{equation}
  that vanish on $\partial(\delta_Q)$, the boundary of
  $\delta_Q$. Actually,
  \begin{equation*}
    \Delta_Q=\bigcup_{\sigma\in\frS_n}\sigma(\delta_Q)\quad\text{and,
      for }(\sigma,\sigma')\in\frS^2_n,
    \quad\sigma(\delta_Q)\cap\sigma'(\delta_Q)=\emptyset\text{ if
    }\sigma\not=\sigma'. 
  \end{equation*}
  Here, for $\sigma\in\frS_n$, we have set $\sigma:\
  (x^1,\cdots,x^n)\mapsto(x^{\sigma(1)},\cdots,x^{\sigma(n)})$.\\
  Thus, finding the ground state of $H_Q=H^0+W$ is equivalent to
  finding the ground state of the Schr{\"o}dinger operator $-\Delta+W$
  with Dirichlet boundary conditions on the domain $\delta_Q$. As the
  domain $\delta_Q$ is connected and has a piecewise linear boundary,
  the ground state of $-\Delta+W$ is non-degenerate
  (see~\cite[Theorems 1.4.3, 1.8.2 and 3.3.5]{MR1103113}
  and~\cite[Section XIII.12]{ReedSimonIV}). This completes the proof
  of Lemma~\ref{le:2}.
\end{proof}
\subsection{The proof of Theorem~\ref{th:asNonDegOfGroundState}}
\label{sec:proof-theorem-1}
Considering the decomposition (\ref{eq:HQdecomposition}),
Lemma~\ref{le:2} implies that the only possible source of degeneracy
of the ground state is that different occupations, i.e., distributions
of particles in the pieces, provide the same ground state
energy. Let us show that, almost surely, this does not happen.\\
Let $\Pi$ be the support of $d\mu(\omega)$, the Poisson process of
intensity $1$ on $\R_+$.  Let $\#(\Pi\cap[0,L])$ be the number of
points the Poisson process puts into $(0,L)$. Suppose now that the
probability that the ground state of $H^U_\omega(L,n)$ is degenerate
is positive.\\
Thus, for some $m\geq0$, conditioned on the fact that the Poisson
process puts $m$ points into $(0,L)$ (i.e., $\#(\Pi\cap[0,L])=m$), the
probability that the ground state of $H^U_\omega(L,n)$ be degenerate
is positive. Let $(\ell_j)_{j}$ be the lengths of the pieces
$(\Delta_j(\omega))_j$, i.e., the $(\Delta_j)_j$ are connected and
$\cup_j\Delta_j(\omega)=(0,L)\setminus(\Pi\cap[0,L])$. Conditioned
$\#(\Pi\cap[0,L])=m$, the joint distribution of the vector
$(\ell_j)_{j}$ is known.
\begin{proposition}[\cite{Molchanov_OneDimDisord}]
  \label{prop:IntervDistr}
  Under the condition $\#(\Pi\cap[0,L])=m$, the vector $(\ell_1,
  \hdots, \ell_{m+1})$ has the same distribution as the random vector
  \begin{equation}\label{eq:IntervDistr}
    \left(\frac{L \cdot \eta_1}{\eta_1 + \hdots + \eta_{m + 1}}, 
      \frac{L \cdot \eta_2}{\eta_1 + \hdots + \eta_{m + 1}}, \hdots,
      \frac{L \cdot \eta_{m + 1}}{\eta_1 + \hdots + \eta_{m + 1}}\right) \text{,}
  \end{equation}
  where $(\eta_i)_{1\leq i\leq m}$ are i.i.d. exponential random
  variables of parameter $1$.
\end{proposition}
\noindent As the lengths $(\ell_j)_{j}$ are continuous functions of
the parameters $(\eta_j)_j$, we know that there exists an open set in
$(\R^+)^{m+1}$, say $O$, such that, for each $(\ell_j)_{1\leq j\leq
  m+1}\in O$, there are at least two occupations $Q_1((\ell_j)_{1\leq
  j\leq m+1})$ and $Q_2((\ell_j)_{1\leq j\leq m+1})$ that have the
same ground state energy (which is at the same time the smallest
possible among the ground state energies for all the occupations). Let
us denote these branches of energy by $(\ell_j)_{1\leq j\leq
  m+1}\mapsto E_1((\ell_j)_{1\leq j\leq m+1})$ and $(\ell_j)_{1\leq
  j\leq m+1}\mapsto E_2((\ell_j)_{1\leq j\leq m+1})$ respectively.\\
For a fixed number of pieces, there are finitely many occupations and
a change in the number of pieces occurs only when a wall, i.e., an
endpoint of a piece, crosses $0$ or $L$. Thus, there exists a subset
$O_1 \subset O$ of positive measure, such that $Q_1((\ell_j)_{1\leq
  j\leq m+1})$ and $Q_2((\ell_j)_{1\leq j\leq m+1})$ are constant on
$O_1$.\\
Now, let us fix an initial set of lengths $(\ell^0_j)_{1\leq j\leq
  m+1}$ in $O_1$ and move it continuously inside this exceptional set
$O_1$. This actually corresponds to moving continuously walls inside
the interval $(0, L)$.  As $Q_1$ and $Q_2$ are two different
occupations, there exists a piece $[a, b] \subset [0, L]$, such that
$Q_1$ and $Q_2$ put different number of particles in this piece, i.e.,
$Q_1([a,b])\ne Q_2([a,b])$.\\
Now, we move $a$ continuously towards $b$; if $a=0$, we will move $b$
towards $a$. Let $a^0$ be the value of $a$ in the configuration
$(\ell^0_j)_{1\leq j\leq m+1}$. Let $E_1(a)$ and $E_2(a)$ be the
ground state energies corresponding to the two different occupations
$Q_1$ and $Q_2$. In a small neighborhood of $a_0$, by the definition
of $O_1$, one has
\begin{equation*}
  E_1(a) = E_2(a)
\end{equation*}
As $U$ is real analytic and as the ground state of $H_Q$ is simple for
any occupation $Q$, the functions $E_1(a)$ and $E_2(a)$ are analytic
in the open interval $(c,b)$ where $c$ is the end of the piece $[c,
a]$ to the left of the piece $[a,b]$. Indeed, $E_1$ (and $E_2$) is
analytic around $a_0$. Assume that $E_1(a)$ stops being analytic
somewhere inside $(c,b)$. This would mean that the eigenvalue $E_1(a)$
of $H_{Q_1}$ becomes degenerate, thus, that the ground state of
$H_{Q_1}$ becomes degenerate. This was already ruled out.\\
This immediately implies that $E_1(a) = E_2(a)$ for all $a\in(c,b)$.\\
But this cannot be. Indeed, if $Q_1$ puts $k_1$ particles in the piece
$[a,b]$, and $Q_2$ puts $k_2$ particles in the piece $[a,b]$ with $k_1
\ne k_2$, the functions $E_1$ and $E_2$ have different asymptotics as
$a$ approaches $b$, indeed,
\begin{equation*}
  E_i(a) \sim k_i^3 / (b - a)^2\quad\text{as}\quad a\to b.
\end{equation*}
This contradicts the fact that the two functions agree on the whole
interval.  This completes the proof of Theorem
\ref{th:asNonDegOfGroundState}.\qed\vskip.2cm
\noindent Finally, we use the results from
sections~\ref{sec:decomp-occup} together with
Theorem~\ref{th:asNonDegOfGroundState} to obtain the following
\begin{Cor}
  \label{cor:SameOccupationSubspace}
  Assume $U$ is real analytic. Then, $\omega$-almost surely, for any
  $L$ and $n$, the ground state of $H^U_\omega(L, n)$ belongs to the a
  unique occupation subspace $\frH_Q$.
\end{Cor}
\begin{proof}
  Consider the orthogonal decomposition~\eqref{eq:HQdecomposition}.
  As any projection of $\Psi_\omega(L, n)$ on $\frH_Q$ is either a
  ground state or zero and as the ground state is
  $\omega$-a.s. simple, only one of the projections of the ground
  state on a space of fixed occupation is different from zero. Thus,
  $\Psi_\omega(L, n)$ belongs to one of the subspaces $\frH_Q$. This
  completes the proof of Corollary~\ref{cor:SameOccupationSubspace}.
\end{proof}
\subsection{The approximate ground state $\Psi^{\text{opt}}$}
\label{sec:appr-ground-state}
The basic idea of the construction of $\Psi^{\text{opt}}$ is to find
the optimal configuration with respect to different occupations. All
the $n$-electron states are considered as deformations of the
unperturbed ground state $\Psi^{0}$ which, we
recall~\eqref{eq:OmegaAlphaConstruction}, is given by the Slater
determinant:
\begin{equation*}
  \Psi^0 = \psi_1 \wedge \psi_2 \wedge \hdots \wedge \psi_n.
\end{equation*}
When the interactions are turned on, the particles in the state
$\Psi^0$ start to interact. For some particles, these interactions may
be quite large. In particular, it may become energetically favorable
to ``decouple'' some particles by moving them apart from each other to
unoccupied pieces; obviously, it is better to move the more excited
particles. One, thus, reduces the interaction energy but this will
necessarily result in an increase of the ``non interaction'' energy of
the state, i.e., of $\langle H^0_\omega(L,n)\Psi,\Psi\rangle$: indeed,
in the non interacting ground state, the $n$ particles occupy the $n$
lowest levels of the system. Nevertheless the decrease of the
interaction energy, i.e., $\langle W_n\Psi,\Psi\rangle$ may compensate
the increase in ``non interacting'' energy. The ``optimal''
configuration then arises through the optimization on the occupation
governed by the interplay between the loss of interaction energy and
the gain of ``non interacting'' energy: it is achieved when loss and
gain balance.\\
Let us note that a ground state $\Psi$ is obviously the ground state
of the Hamiltonian restricted to the appropriate fixed occupation
subspace, i.e., $\Psi$ is the ground state of $H_{Q(\Psi)}$
(see~\eqref{eq:HQdecomposition}). This corresponds to writing the
minimization problem in the form
\begin{equation}
  \label{eq:TwoStageMinimizing}
  \inf_{\substack{\Phi \in \frH^n\\\|\Phi\| = 1}} \langle H_\omega(L, n)
  \Phi, \Phi \rangle
  = \inf_{\substack{Q \in \bbN^m\\|Q| = n}} \inf_{\substack{\Phi \in
      \frH_Q\\\|\Phi\| = 1}} \langle H_Q \Phi, \Phi \rangle.
\end{equation}
This reduces the problem to finding the optimal occupations rather
than the optimal $n$-electron state itself.\vskip.2cm\noindent
Recalling that the constant $\gamma$ is defined in
Proposition~\ref{prop:TwoElectronProblem}, we set
\begin{equation}
  \label{eq:AstarXstarDef}
  A_* := \frac{\gamma}{8 \pi^2}, \qquad 
  x_* := 1 - e^{-\tfrac{\gamma}{8 \pi^2}} \text{.}
\end{equation}
Note that
\begin{equation*}
  A_* = -\log(1 - x_*) \text{.}
\end{equation*}
\noindent Let us now define $\Psi^{\text{opt}}$. Therefore, recall
that the pieces in the model are denoted by $(\Delta_k(\omega))_{1\leq
  k\leq m(\omega)}$ (see section~\ref{sec:introduction}) and that for
$\Delta_k(\omega)$, a piece, we define (see
sections~\ref{sec:non-iter-ground} and~\ref{sec:iter-ground})
\begin{itemize}
\item $\varphi^j_{\Delta_k(\omega)}$ to be the $j$-th normalized
  eigenvector of $-\Delta_{|\Delta_k(\omega)}^D$,
\item $\zeta^j_{\Delta_k(\omega)}$ to be the $j$-th normalized
  eigenvector of $-\laplace^D_{|\Delta_k(\omega)^2} + U$ acting on
  $\D\bigwedge_{j=1}^2L^2(\Delta_k(\omega))$.
\end{itemize}
We will define the state $\Psi^{\text{opt}}$ in two steps. We first
define $\Psi_m^{\text{opt}}$: it will contain less than $n$ particles
and will be the main part of $\Psi^{\text{opt}}$. We, then, add the
missing particles to get the $n$-particle state $\Psi^{\text{opt}}$.
\begin{definition}
  \label{def:PsiOptm}
  Consider all the pieces in $[0, L]$. For each piece, depending on
  its length, do one of the following:
  \begin{enumerate}
  \item keep the pieces of length in $[0,\ell_\rho-\rho
    x_*)\cup[3\ell_\rho,\infty)$ empty;
  \item put one particle in its ground state in each piece of length
    in $[\ell_\rho-\rho x_*,2\ell_\rho+A_*)$;
  \item\label{item:1} in pieces of length in
    $[2\ell_\rho+A_*,3\ell_\rho)$, put the ground state of a
    two-particles system with interactions (see
    Proposition~\ref{prop:TwoElectronProblem} and section~\ref{sec:two-int-electrons});
  \end{enumerate}
  We define the state $\Psi_m^{\text{opt}}=\Psi_m^{\text{opt}}(L,n)$
  to be the anti-symmetric tensor product of the thus constructed one-
  and two-particles sub-states, that is,
  \begin{equation}
    \label{eq:79}
    \Psi_m^{\text{opt}}(L,n)=\bigwedge_{|\Delta_j(\omega)|
      \in[\ell_\rho-\rho x_*,2\ell_\rho+A_*)}\varphi^1_{\Delta_j(\omega)}
    \wedge \bigwedge_{|\Delta_j(\omega)|\in[2\ell_\rho+A_*,
      3\ell_\rho)}\zeta^1_{\Delta_j(\omega)}.
  \end{equation}
\end{definition}
\noindent Note that, as the $(\zeta^1_{\Delta_j(\omega)})_j$ carry two
particles, $\Psi_m^{\text{opt}}(L,n)$ is not given by a Slater
determinant; an explicit formula for such an anti-symmetric tensor
product is given in~\eqref{eq:267} in
Appendix~\ref{sec:proj-totally-antisym}.
\begin{Rem}
  Note that, in step (\ref{item:1}) of Definition~\ref{def:PsiOptm}, we
  put two interacting particles within these pieces. Because of the
  interactions, this is different from putting separately two
  particles on the two lowest one-particle energy levels (see
  appendix~\ref{sec:two-part-probl}).
\end{Rem}
\noindent Let us now compute the total number of particles contained
in $\Psi_m^{\textup{opt}}$. We prove
\begin{Le}
  \label{le:3}
  With probability $1-O(L^{-\infty})$, for $L$ sufficiently
  large, in the thermodynamic limit, the total number of particles in
  $\Psi_m^{\textup{opt}}$ constructed in
  Definition~\textup{\ref{def:PsiOptm}} is given by
  \begin{equation*}
    \D\calN(\Psi_m^{\textup{opt}})=n \left[1-\rho^2 \left(3  - x_* - 
          \frac{x_*^2}{2}\right)+O(\rho^3)\right].
  \end{equation*}
\end{Le}
\begin{proof}
  It suffices to count the number of pieces of each type and multiply
  by the corresponding number of particles.  We recall that,
  by~\eqref{eq:AstarXstarDef}, one has $\D \exp(-\ell_\rho) =
  \frac{\rho}{1 + \rho}$ and $\D \exp(-A_*) = 1 - x_*$. Thus, for
  $\beta\in(0,1/2)$, using Proposition~\ref{prop:IntervStatistics} and
  the second equation in~\eqref{eq:AstarXstarDef}, with probability
  $1-O(L^{-\infty})$, one computes
  \begin{equation*}
    \begin{split}
      \calN(\Psi_m^{\text{opt}})&=\sharp\{l\in[\ell_\rho-\rho x_*,
      2 \ell_\rho+A_*)\}+2\cdot \sharp\{l\in[2 \ell_\rho+A_*, 3\ell_\rho)\} \\
      &= L \left[e^{-(\ell_\rho-\rho x_*)}-e^{-(2\ell_\rho+A_*)} +
        2e^{-(2 \ell_\rho+A_*)} - 2e^{-3\ell_\rho}\right]
      + O\left(L^{1/2+\beta}\right) \\
      &= \frac{L\rho}{1+\rho}\left[e^{\rho x_*}+\rho e^{-A_*}-\rho^2
        e^{-A_*}-2\rho^2 +O(\rho^3)\right]
      \\&=\frac{L\rho}{1+\rho}\left[1+\rho-\rho^2 \left(e^{-A_*}+2  -
          \frac{x_*^2}{2}\right) + O(\rho^3)\right] \\
      &=n \left[1-\rho^2 \left(3  - x_* - 
          \frac{x_*^2}{2}\right)+O(\rho^3)\right].
    \end{split}
  \end{equation*}
  This completes the proof of Lemma~\ref{le:3}.
\end{proof}
\noindent Lemma~\ref{le:3} shows that, for $\rho$ small,
$\Psi_m^{\text{opt}}$ contains less than $n$ particles. Let us now add
particles to $\Psi_m^{\text{opt}}$ to complete it into
$\Psi^{\text{opt}}$. Therefore, we prove
\begin{Le}
  \label{le:16}
  Let $(\widetilde\varphi_k)_{1\leq k\leq k_\rho(\omega)}$ be the
  particles that $\Psi^0$, the non interacting ground state, puts in
  the pieces longer than $3\ell_\rho$
  ordered by increasing energy.\\
  With probability $1-O(L^{-\infty})$, for $L$ sufficiently large, one
  has $k_\rho(\omega)\geq n\rho^2(3-18\rho)$.
\end{Le}
\begin{proof}
  By Proposition~\ref{prop:IntervStatistics}, with probability
  $1-O(L^{-\infty})$, the number of pieces of length in
  $\ell_\rho[3+\rho,4)$ is equal to
  \begin{equation*}
    n\frac{\rho^2}{(1+\rho)^3}
    \left(e^{-\rho}-\frac{\rho}{1+\rho}\right)+o(L)\geq
    n\rho^2 \left(1-6\rho\right)
  \end{equation*}
  for $L$ large.\\
  To complete the proof of Lemma~\ref{le:16}, let us now establish
  some auxiliary results.  By~\eqref{eq:convEnFerm} in
  Proposition~\ref{pro:2}, we know that $E_{n,\omega}^\Lambda$
  converges to $E_\rho$ in the thermodynamic limit. We will first
  investigate the rate of convergence in~\eqref{eq:convEnFerm}.
  \begin{Le}
    \label{lem:FermiEnergyConvergenceRate}
    Denote by $\ell_{n,L}$ the length of an interval having a ground
    state energy equal to $E_{n,\omega}^\Lambda$, i.e.,
    \begin{equation*}
      \ell_{n, L} = \frac{\pi}{\sqrt{E_{n,\omega}^\Lambda}} \text{.}
    \end{equation*}
    Let $\rho > 0$ be fixed.  For any $\delta>0$, in the thermodynamic
    limit $L\to\infty$, $n/L\to\rho$, with probability
    $1-O(L^{-\infty})$, one has
    \begin{equation*}
      \begin{aligned}
        \ell_{n, L} &= \ell_\rho + O(L^{-(1/2 - \delta)})+O
        \left(\left|\frac nL-\rho\right|\right), \\
        E_{n,\omega}^\Lambda &= E_\rho + O(L^{-(1/2 - \delta)})+O
        \left(\left|\frac nL-\rho\right|\right).
      \end{aligned}
    \end{equation*}
  \end{Le}
  \noindent In view of Lemma~\ref{lem:FermiEnergyConvergenceRate} and
  by the definition of $\Psi^0$, for $L$ sufficiently large, each
  piece of length in $\ell_\rho[3+\rho,4)$ contains at least 3
  particles of $\Psi^0$. This completes the proof of
  Lemma~\ref{le:16}.
\end{proof}
\begin{proof}[Proof of Lemma~\ref{lem:FermiEnergyConvergenceRate}]
  By~\eqref{eq:densEtatsPoissonModel1}, with probability
  $1-O(L^{-\infty})$, the normalized counting function for the
  Dirichlet eigenvalues of $H_\omega(L,1)$ (see~\eqref{eq:17})
  satisfies
  \begin{equation*}
    \frac nL=N_L^D(E_{n,\omega}^\Lambda) = \frac{\exp(-\ell_{n, L})}{1 - \exp(-\ell_{n,
        L})} + O(L^{-(1/2 - \delta)}).
  \end{equation*}
  Taking into account the fact that
  \begin{equation*}
    \rho = N(E_\rho) = \frac{\exp(-\ell_\rho)}{1 - \exp(-\ell_\rho)},
  \end{equation*}
  we deduce that
  \begin{equation*}
    \frac{\exp(-\ell_{n,L})}{1-\exp(-\ell_{n,L})}=\frac{\exp(-\ell_\rho)}
    {1-\exp(-\ell_\rho)}+O(L^{-(1/2 - \delta)})+O
    \left(\left|\frac nL-\rho\right|\right).
  \end{equation*}
  This immediately yields
  \begin{equation*}
    \exp(-\ell_{n, L}) = \exp(-\ell_\rho) +  O(L^{-(1/2 - \delta)})+O
    \left(\left|\frac nL-\rho\right|\right).
  \end{equation*}
  The proof of Lemma~\ref{lem:FermiEnergyConvergenceRate} is complete.
\end{proof}
\noindent For $\rho$ small, by Lemmas~\ref{le:3} and~\ref{le:16}, one
has $n-\calN(\Psi_m^{\text{opt}})<k_\rho(\omega)$. Thus, to construct
$\Psi^{\text{opt}}$, we just add $n-\calN(\Psi_m^{\text{opt}})$
particles of $\Psi^0$ living in pieces of length in
$\ell_\rho[3+\rho,4)$ to $\Psi_m^{\text{opt}}$.
\begin{definition}
  \label{def:PsiOptr}
  We define
  \begin{equation}
    \label{eq:80}
    \Psi^{\text{opt}}=\Psi^{\text{opt}}(L,n):=\Psi_m^{\text{opt}}(L,n) \wedge
    \bigwedge_{k=1}^{n-\calN(\Psi_m^{\text{opt}})}\widetilde\varphi_k.
  \end{equation}
\end{definition}
\begin{Rem}
  \label{rem:7}
  Let us give an alternative approach to defining $\Psi^{\text{opt}}$
  which does not result in exactly the same $\Psi^{\text{opt}}$ but
  which can serve exactly the same purpose in the subsequent arguments.\\
  We start with the non interacting ground state $\Psi^0$ and describe
  how it is modified:
  \begin{itemize}
  \item for pairs of particles living in the same piece, the
    modification depends on the length of this piece:
    \begin{itemize}
    \item for the pieces of length between $2 \ell_\rho$ and $2
      \ell_\rho + A_*$, remove the more excited particle and put it
      into an unoccupied piece of length between $\ell_\rho - \rho
      x_*$ and $\ell_\rho$;
    \item for the remaining pieces, i.e., the pieces of length between
      $2 \ell_\rho + A_*$ and $3 \ell_\rho$, the factorized
      two-particles state corresponding to $\Psi^0$ should be replaced
      by a true ground state of a two-particles system with
      interaction in this piece (see
      section~\ref{sec:two-int-electrons} for a description of such a
      two-particle state);
    \end{itemize}
  \item do not modify any of the particles in $\Psi^0$ that are either
    alone or live in groups of three or more pieces.
  \end{itemize}
  One can easily verify that, in the above procedure, up to a small
  relative error, the number of pieces to which the excited particles
  are displaced is equal to the number of pieces where we decouple the
  particles. Indeed, according to
  Proposition~\ref{prop:IntervStatistics}, with probability at least
  $1-O(L^{-\infty})$, for the former, one has
  \begin{equation}
    \label{eq:NPiecesDonors}
    \begin{split}
      \sharp\{l \in (2 \ell_\rho, 2 \ell_\rho - \log(1 - x_*))\}
      &= L \exp(-2 \ell_\rho) x_* + O(L^{1/2 + \beta}) \\
      &= n \rho x_* (1 + O(\rho)),
    \end{split}
  \end{equation}
  and, for the latter, one has
  \begin{equation}
    \label{eq:NPiecesAcceptors}
    \begin{split}
      \sharp\{l \in (\ell_\rho - \rho x_*, \ell_\rho)\}
      &= L\exp(-\ell_\rho) (\exp(\rho x_*) - 1) + O(L^{1/2 + \beta}) \\
      &= n\rho x_* (1 + O(\rho)).
    \end{split}
  \end{equation}
  Thus, both sets contain the same number of pieces (up to an error of
  order $n\rho^2$). This completes the construction of
  $\Psi^{\text{opt}}$.
\end{Rem}
\subsection{Comparing $\Psi^{\text{opt}}$ with the ground state of the
  interacting system}
\label{sec:comp-psiopt-ground}
Our goal in the sections to come is to estimate how much
$\Psi^{\textup{opt}}$ differs from a true ground state $\Psi^U =
\Psi^U_\omega(L,n)$ (and to show that it doesn't differ much). This
will be done through the comparison of their occupation numbers. We
shall see that the ground states of the interacting Hamiltonian must
live in subspaces with special occupation numbers (see
Corollary~\ref{cor:3}).\\
To compare occupation numbers, we introduce the distance dist$_1$.
\begin{definition}
  \label{def:2}
  Let $m=m(\omega)$ be the number of pieces in $[0,L]$. For
  $j\in\{1,2\}$, pick an occupation
  \begin{equation*}
    Q^j = (Q_1^j, Q_2^j, \hdots, Q_m^j) \in \bbN^m, 
    \quad |Q^j| = n.
  \end{equation*}
  Define
  \begin{equation*}
    \dist_1(Q^1, Q^2) = \sum_{i = 1}^{m} |Q_i^1- Q_i^2|.
  \end{equation*}
\end{definition}
\begin{Rem}
  \label{rem:3}
  Recall that the non interacting ground state $\Psi^0$ has a single
  occupation $Q(\Psi^0)$: all the states with energy below
  $E_{n,\omega}^\Lambda$ (where we recall that $E_{n,\omega}^\Lambda$
  denote the $n$-th (counting multiplicity) energy level of the
  one-particle Hamiltonian $H_\omega(L)$); moreover, only those states
  are occupied. In~\cite{Veniaminov_PhDthesis}, for $U$ compactly
  supported, for $\Psi^U$ an interacting ground state, it was proved
  that
  \begin{equation}
    \label{eq:diffQPsiQPsi0}
    C^{-1} n \rho \leq \dist_0(Q(\Psi^U), Q(\Psi^0)) \leq C n \rho
    \text{.}
  \end{equation}
  where $\dist_0$ is defined by $\D \dist_0(Q^1, Q^2) = \sum_{i =
    1}^{m} \car_{Q_i^1 \ne Q_i^2}$. Clearly, one has
  dist$_0\leq$dist$_1$.\\
  In the sequel, we shall prove that $\Psi^{\textup{opt}}$ is a better
  approximation of a ground state of the interacting system than is
  the non-interacting ground state $\Psi^0$ (compare~\eqref{eq:142}
  with~\eqref{eq:diffQPsiQPsi0}).
\end{Rem}
\noindent For interaction potentials that decrease at infinity
sufficiently fast (see~\textbf{(HU)}), we will prove that the main
modification to the ground state energy comes from $U$ restricted to
some (sufficiently large) compact set.\\
Fix a constant $B > 2$. We decompose the interaction potential in the
sum of the ``principal'' and ``residual'' parts that is, write $U =
U^p + U^r$ where
\begin{equation}
  \label{eq:99}
  U^p:=\car_{[-B \ell_\rho, B \ell_\rho]}U\quad\text{and}\quad U^r
  :=\car_{(-\infty, - B \ell_\rho) \cup (B \ell_\rho, +\infty)}U. 
\end{equation}
As the sum of pair interactions $W_n$ is linear in $U$, this yields
the following decomposition for the full Hamiltonian:
\begin{equation}
  \label{eq:245}
  H^U = H^0 + W_n = H^0 + W_n^{U^p} + W_n^{U^r} = H^{U^p} + W_n^r.
\end{equation}
Our analysis is done in the following steps:
\begin{enumerate}
\item first, we prove that $\Psi^{\textup{opt}}$ approximates well the
  ground state for the system with compactified interactions
  $\Psi^{U^p}$;
\item second, we show that the quadratic form of the residual
  interactions $W^r$ on $\Psi^{\textup{opt}}$ contributes only to the
  error term; this will imply~\eqref{eq:32};
\item finally, we will conclude that the same $\Psi^{\textup{opt}}$
  gives also a good approximation for the full Hamiltonian $H^U$
  ground state $\Psi^U$ in terms of the distance $\dist_1$ for the
  respective occupations.
\end{enumerate}
\begin{Rem}
  \label{rem:6}
  Let us clarify a point of terminology: we will minimize the
  quadratic form $\langle H_Q\Psi,\Psi\rangle=\langle
  H^0_Q\Psi,\Psi\rangle+ \langle W_n\Psi,\Psi\rangle$; the term
  $\langle H^0_Q\Psi,\Psi\rangle$ is referred to as the ``non
  interacting energy'' term and $\langle W_n\Psi,\Psi\rangle$ the
  ``interaction energy'' term; we use the same decomposition and
  terminology for smaller groups of particles or at the single
  particle level.
\end{Rem}
\subsection{The analysis of $H^{U^p}$}
\label{sec:analysis-hum}
We start with the analysis of $H^{U^p}$, in particular, of its ground
state energy and ground state(s). Later, we show that the addition of
$W_n^r$ will not change much in the ground state energy and ground
state(s).\\
%
% Define
% %
% \begin{equation}
%   \label{eq:276}
%   \tilde Z(x):=x^{-3}\int_x^+\inftyU(t)dt.
% \end{equation}
% %
First, we compute the energy of $\Psi^{\textup{opt}}$. We prove
\begin{Th}
  \label{thr:6}
  There exists $\rho_0>0$ such that, for $\rho\in(0,\rho_0)$, in the
  thermodynamic limit, with probability $1$, one has
  \begin{multline}
    \label{eq:89}
    \lim_{\substack{L \to \infty\\ n / L \to \rho}}
    \frac1{n}\bigl\langle H_\omega^{U^p}(L,n)\Psi^{\textup{opt}}(L,n),
    \Psi^{\textup{opt}}(L,n)\bigr\rangle \\
    = \densEn^0(\rho)+ \pi^2 \gamma_* \rho |\log{\rho}|^{-3} \left(1 +
      O\left(f_Z(|\log{\rho}|)\right)\right)
  \end{multline}
  where $\gamma_*$ is defined in~\eqref{eq:12} and $f_Z$ is a
  continuous function satisfying $f_Z(x)\to0$ as $x \to +\infty$ no
  faster than $1/x$ (for more details, see~\eqref{eq:1}).
\end{Th}
\begin{proof}
  To shorten the notations, we will frequently drop the arguments $L$,
  $n$ and the subscript $\omega$ in this proof. We will show that, up
  to error terms, the only terms that contribute to $\langle H^{U^p}
  \Psi^{\textup{opt}}, \Psi^{\textup{opt}}\rangle - \langle H^0
  \Psi^0, \Psi^0\rangle$ are those due to
  \begin{enumerate}
  \item the interactions between two particles in the same piece,
  \item the decoupling of a fraction of these particles following the
    construction of $\Psi^{\text{opt}}$.
  \end{enumerate}
  In~\eqref{eq:89}, the interactions between neighboring distinct
  pieces will be shown to contribute only to the error term where we
  have defined
  \begin{definition}
    \label{def:4}
    A pair of {\it neighboring} or {\it interacting} pieces is a pair
    of distinct pieces at distance at most $B \ell_\rho$ from one
    another, particular, particles in two such pieces can still
    interact via the potential $U^p$.
  \end{definition}
  \noindent Let us now outline the main idea of the proof of
  Theorem~\ref{thr:6}. The pieces longer than $2 \ell_\rho + A_*$
  contain two particles both in $\Psi^0$ and
  $\Psi^{\text{opt}}$. Hence, for each piece of this type, the energy
  difference is given by the second term in the
  asymptotics~\eqref{eq:32} in
  Proposition~\ref{prop:TwoElectronProblem}.  On the contrary, in
  pieces of length between $2\ell_\rho$ and $2\ell_\rho+A_*$, in
  $\Psi^0$, the two particles were decoupled in order to construct
  $\Psi^{\text{opt}}$, keeping one intact and displacing another to a
  piece of length between $\ell_\rho - \rho x_*$ and $\ell_\rho$.  In
  this case, the energy difference is given by the increase of non
  interacting energy of the second (displaced) particle. The single
  particles in $\Psi^0$ remain untouched in $\Psi^{\text{opt}}$ and
  groups of three and more particles contribute only to the error term
  (as they carry only a small number of particles).\\
  To put the above arguments into a rigorous form, we will use the
  following partition of the set of available pieces according to
  their length.  Choose $K$ large but independent of $L$. For $k \in
  \{1, \hdots, K\}$, consider the sets of pieces
  \begin{equation*}
    \begin{aligned}
      \calL_k^1 &= \left\{\text{pieces of length in }\left[\ell_\rho -
          \tfrac{k}{K} \rho, \ell_\rho - \tfrac{k - 1}{K}
          \rho\right)\right\}, \\
      \calL_k^2 &= \left\{\text{pieces of length in } \left[2
          \ell_\rho - \log\left(1 - \tfrac{k - 1}{K}\right), 2
          \ell_\rho - \log\left(1 - \tfrac{k}{K}\right)\right)\right\}
      \text{.}
    \end{aligned}
  \end{equation*}
  As $K$ is independent of $L$, with probability $1-O(L^{-\infty})$,
  the number of pieces in the classes
  $\D((\calL_k^j))_{\substack{j\in\{1,2\}\\ k\in\{1,\dots,K\}}}$ is
  given by Proposition~\ref{prop:IntervStatistics}. We will,
  henceforth, use these estimates without reference to probabilities.\\
  As in~\eqref{eq:NPiecesDonors} and~\eqref{eq:NPiecesAcceptors}, one
  shows that these two sets map one-to-one onto one another up to an
  error estimated as follows
  \begin{equation*}
    \card{\calL_k^1} = \card{\calL_k^2} + O(n \rho^2 K^{-1}) = n \rho
    K^{-1} (1 + O(\rho)).
  \end{equation*}
  Recall that $x_*$ is defined in~\eqref{eq:AstarXstarDef}. For $k\leq
  K x_*$, according to our scheme, the pairs of particles in pieces
  belonging to $\calL_k^2$ get decoupled, one of the particles being
  sent to occupy a piece belonging to $\calL_k^1$. For $k>Kx_*$, the
  pairs of particles in the pieces of $\calL_k^2$ are kept
  untouched. The latter pieces are those of size at least
  $2\ell_\rho+A_*$. It is easily seen that the number of such pieces
  is given by
  \begin{equation*}
    \sharp\{j:\ |\Delta_j(\omega)|\geq2\ell_\rho+A_*\}=n\rho
    e^{-A_*}(1+O(\rho)) =n\rho(1-x_*)+O(n \rho^2).
  \end{equation*}
  The majority of these pieces is smaller than $2 \ell_\rho + A_* +
  \log\ell_\rho$; indeed,
  \begin{equation*}
    \sharp\{j:\ |\Delta_j(\omega)|\in2\ell_\rho+A_*+[0,\log\ell_\rho]\} 
    =n\rho(1-x_*)+O(n \rho |\log{\rho}|^{-1}).
  \end{equation*}
  By Proposition~\ref{prop:TwoElectronProblem}, for a piece of length
  $\ell$ in $2\ell_\rho+A_*+[0,\log\ell_\rho]$, the interaction energy
  of the two-particles system is given by
  \begin{equation*}
    \frac{\gamma}{\ell^3} + o(\ell^{-3}) = \frac{\gamma}{8 \ell_\rho^3} +
    o(\ell_\rho^{-3}).
  \end{equation*}
  For the difference of energies, this yields
  \begin{equation}
    \label{eq:115}
    \begin{split}
      \langle H^{U^p} \Psi^{\textup{opt}}, \Psi^{\textup{opt}}\rangle
      - \langle H^0 \Psi^0, \Psi^0\rangle &= \frac{n \rho}{K}\sum_{k =
        1}^{K x_*} \left[\frac{\pi^2}{\left(\ell_\rho - \tfrac{k}{K}
            \rho\right)^2} - \frac{4 \pi^2}{\left(2 \ell_\rho -
            \log\left(1-\tfrac{k}{K}\right)\right)^2}\right]\\
      &\hskip3cm+ \frac{\gamma}{8 \ell_\rho^3} n \rho (1 - x_*)+
      o\left(n \rho |\log{\rho}|^{-3}\right) \text{.}
    \end{split}
  \end{equation}
  Taking $K$ large, we approximate the Riemann sum in the last
  expression by an integral
  \begin{equation*}
    \begin{split}
      \frac1K\sum_{k = 1}^{K x_*} &\left[\frac{\pi^2}{\left(\ell_\rho
            - \tfrac{k}{K} \rho\right)^2} - \frac{4 \pi^2}{\left(2
            \ell_\rho - \log\left(1 -
              \tfrac{k}{K}\right)\right)^2}\right] \\&= x_* \int_0^1
      \left[\frac{\pi^2}{\left(\ell_\rho - t x_* \rho\right)^2} -
        \frac{\pi^2}{\left(\ell_\rho - \tfrac{1}{2}
            \log(1 - t x_*)\right)^2}\right] \rmd{t}+O\left(\frac1K\right) \\
      &= x_* \left(-\int_0^1 \frac{\pi^2}{\ell_\rho^3} \log(1 - t x_*)
        \rmd{t} + o(\ell_\rho^{-3})\right)+O\left(\frac1K\right) \\
      &= \pi^2 \ell_\rho^{-3} (x_* - (1 - x_*) A_*) (1 +
      o(1))+O\left(\frac1K\right).
    \end{split}
  \end{equation*}
  Picking $\delta\in(0,1)$, letting $K = \rho^{-\delta}$ and
  recalling~\eqref{eq:AstarXstarDef} for $A_*$ and $x_*$, for $\delta$
  small, we get
  \begin{equation}
    \label{eq:265}
    \begin{split}
      \langle H^{U^p} \Psi^{\textup{opt}}, \Psi^{\textup{opt}}\rangle
      - \langle H^0 \Psi^0, \Psi^0\rangle &= n \rho \ell_\rho^{-3}
      \left(\pi^2 (x_* - (1 - x_*) A_*) + \frac{\gamma}{8} (1 -
        x_*)\right) \\&\hskip6.5
      cm+ o\left(n \rho\ell_\rho^{-3}\right) \\
      &= n \rho \ell_\rho^{-3} \pi^2 \left (1 - e^{-\tfrac{\gamma}{8
            \pi^2}}\right) + o\left(n \rho \ell_\rho^{-3}\right).
    \end{split}
  \end{equation}
  In order to finish the proof of~\eqref{eq:89} and, thus, of
  Theorem~\ref{thr:6}, it suffices to upper bound the interactions
  between distinct pieces.  Recall that $\Psi^{\textrm{opt}}$ is an
  anti-symmetric exterior product of one- and two-particles eigenstates
  (see~\eqref{eq:79} and~\eqref{eq:80}):
  \begin{equation}
    \label{eq:112}
    \Psi^{\textup{opt}} = \bigwedge_{i = 1}^{\whk_1} \varphi_i
    \wedge \bigwedge_{j = 1}^{k_2}\zeta_{j}
    \wedge
    \bigwedge_{i = 1}^{\wtk_1}\wtphi_i,
  \end{equation}
  where the numbers of sub-states in each group are respectively
  \begin{equation*}
    \begin{aligned}
      \whk_1 &= n \left(1 - 2 \rho (1 - x_*) + \rho^2 \left(3(1 - x_*)
          + \frac{x_*^2}{2}\right) + O(\rho^3)\right), \\
      k_2 &= n \rho (1 - x_* - \rho(3 - 2 x_*) + O(\rho^2)),\\
      \wtk_1 &= n-\calN(\Psi_m^{\text{opt}}) = n \rho^2 \left(3 - x_*
        - \frac{x_*^2}{2}\right) (1 + O(\rho)).
    \end{aligned}
  \end{equation*}
  The functions $\varphi_i$ and $\wtphi_i$ are one-particle ground
  states in certain and the functions $\zeta_j$ are two-particles
  ground states in certain pieces. Of course, $\whk_1+k_2+\wtk_1=n$.
  As in what follows we will only need to distinguish between one- and
  two-particles states, let us put the two groups of one-particle
  sub-states from~\eqref{eq:112} together, i.e. write
  \begin{equation}
    \label{eq:78}
    \Psi^{\textup{opt}} = \bigwedge_{i = 1}^{k_1} \phi_i
    \wedge \bigwedge_{j = 1}^{k_2}\zeta_{j},
  \end{equation}
  where $k_1 = \whk_1 + \wtk_1$ and $\{\phi_i\}_{i = 1}^{k_1} =
  \{\varphi_i\}_{i = 1}^{\whk_1} \cup \{\wtphi_i\}_{i = 1}^{\wtk_1}$.
  As $W^p$ is a totally symmetric sum of pair interaction potentials,
  one computes
  \begin{equation}
    \label{eq:114}
    \begin{split}
    \langle W^p \Psi^{\textup{opt}},\Psi^{\textup{opt}}\rangle&=
    \sum_{1\leq i<j\leq n}\int_{[0,L]^n}U(x_i-x_j)
  \left|\Psi^{\textup{opt}}(x)\right|^2dx\\&=
    \frac{n(n-1)}2\int_{[0,L]^n}U(x_1-x_2)
  \left|\Psi^{\textup{opt}}(x)\right|^2dx=
    \Tr\left(U^p \gamma^{(2)}_{\Psi^{\textup{opt}}}\right).      
    \end{split}
  \end{equation}
  According to Proposition~\ref{prop:DensityMatrixStructure}, for
  $\Psi^{\textup{opt}}$ having the structure~\eqref{eq:78}, its
  two-particle density matrix is given by
  \begin{equation}
    \label{eq:113}
    \begin{split}
      \gamma^{(2)}_{\Psi^{\textup{opt}}} = \sum_{j = 1}^{k_2}
      \gamma^{(2)}_{\zeta_j} &+ (\Id - \Ex) \sum_{\substack{i, j = 1,
          \hdots, k_1\\i < j}} \gamma_{\phi_i} \otimes^s
      \gamma_{\phi_j} + (\Id - \Ex) \sum_{i = 1}^{k_1} \sum_{j =
        1}^{k_2}
      \gamma_{\phi_i} \otimes^s \gamma_{\zeta_j} \\
      &+ (\Id - \Ex) \sum_{\substack{i, j = 1, \hdots, k_2\\i < j}}
      \gamma_{\zeta_i} \otimes^s \gamma_{\zeta_j}.
    \end{split}
  \end{equation}
  As $\zeta_j$ is a two-particle state and $\phi_j$ is a one-particle
  state, one has
  \begin{equation*}
    \gamma^{(2)}_{\zeta_j} = \langle \cdot, \zeta_j \rangle
    \zeta_j\quad\text{and}\quad
    \gamma_{\phi_j} = \langle \cdot, \phi_j \rangle \phi_j.
  \end{equation*}
  The decomposition~\eqref{eq:113} being plugged in the r.h.s. of
 ~\eqref{eq:114} reads as follows:
  \begin{enumerate}
  \item the first term corresponds to the interaction of two particles
    living in the same piece; this term is the leading one in the
    difference~\eqref{eq:115} and has been already taken into account
    in the first part of the proof;
  \item the second term is the interaction between two one-particle
    sub-states living in distinct pieces;
  \item the third term is due to the interaction between a
    one-particle sub-state in one piece and a two-particle sub-state
    (represented by its one-particle reduced density matrix) in
    another piece;
  \item finally, the last term describes the interaction between two
    distinct two-particle sub-states.
  \end{enumerate}
  Thus, we are interested in upper bounds on $\Tr(U^p \beta)$ where
  $\beta$ is any of the last three terms in~\eqref{eq:113}. Let
  $\gamma_1$ and $\gamma_2$ be two arbitrary one-particle density
  matrices encountered in the above expressions. Then, the kernel of
  $(\Id - \Ex) \gamma_1 \otimes^s \gamma_2$ is given by
  \begin{equation}
    \label{eq:116}
    \begin{split}
      (\Id - \Ex) (\gamma_1 \otimes^s \gamma_2)(x, y, x^\prime,
      y^\prime) = \frac{1}{2}\bigl(&\gamma_1(x, x^\prime) \gamma_2(y,
      y^\prime)
      + \gamma_2(x, x^\prime) \gamma_1(y, y^\prime) \\
      &\hskip1cm- \gamma_1(y, x^\prime) \gamma_2(x, y^\prime) -
      \gamma_2(y, x^\prime) \gamma_1(x, y^\prime)\bigr).
    \end{split}
  \end{equation}
  Taking into account the fact that in our case $\gamma_1$ and
  $\gamma_2$ live on distinct pieces $\Delta_1$ and $\Delta_2$
  respectively,~\eqref{eq:116} implies
  \begin{equation}
    \label{eq:117}
    \begin{split}
      \Tr\left(U^p (\Id - \Ex) \gamma_1 \otimes^s \gamma_2\right)&=
      \int_{\bbR^2} U^p(x - y) (\Id - \Ex) (\gamma_1 \otimes^s
      \gamma_2)(x, y, x, y) \rmd{x}
      \rmd{y} \\
      & =\int_{\Delta_1} \int_{\Delta_2} U^p(x - y) \gamma_1(x, x)
      \gamma_2(y, y) \rmd{x} \rmd{y}.
    \end{split}
  \end{equation}
  To upper bound the last expression, we use the estimates proved in
  section~\ref{sec:ferm-neighb-piec}. We now study the different sums
  in~\eqref{eq:113}.\\
  For pairs of one-particle states, we estimate the number of pairs of
  pieces at a certain distance by
  Proposition~\ref{prop:IntervStatistics2} and we bound individual
  terms by Lemma~\ref{lem:inter11_closeEstimate}. We compute that, for
  any $\eta>0$ and $\varepsilon>0$, for $L$ sufficiently large, with
  probability $1-O(L^{-\infty})$, one has
  \begin{equation*}
    \begin{split}
      \Tr\Bigl(U^p (\Id - \Ex)& \sum_{\substack{i, j = 1, \hdots,
          k_1\\i < j}} \gamma_{\phi_i} \otimes^s \gamma_{\phi_j}\Bigr)
      \leq \sum_{\substack{|\Delta_i|\geq \ell_\rho - \rho x_*\\
          |\Delta_j| \geq\ell_\rho - \rho x_*\\\dist(\Delta_i,
          \Delta_j) \leq B \ell_\rho}} \int_{\Delta_i \times \Delta_j} U(x -
      y) |\varphi^1_{\Delta_i}(x)|^2
      |\varphi^1_{\Delta_j}(y)|^2 \rmd{x} \rmd{y} \\
      &\leq \sum_{k=0}^{B\ell_\rho/\eta}
      \sum_{\substack{|\Delta_i|\geq \ell_\rho - \rho x_*\\
          |\Delta_j| \geq\ell_\rho - \rho x_*\\
          k\eta\leq\dist(\Delta_i,\Delta_j)<(k+1)\eta}}
      \int_{\Delta_i \times \Delta_j} U(x - y) |\varphi^1_{\Delta_i}(x)|^2
      |\varphi^1_{\Delta_j}(y)|^2 \rmd{x} \rmd{y} \\
      &\leq C\sum_{k=0}^{B\ell_\rho/\eta} \#\left\{
        \begin{aligned}
          |\Delta_i|&\geq\ell_\rho- \rho x_*,\\
          |\Delta_j|&\geq\ell_\rho- \rho x_*\\
          k\eta\leq\dist(&\Delta_i,\Delta_j)<(k+1)\eta
        \end{aligned}
      \right\}\ell_\rho^{-4+\eps}\,((k+1)\eta)^{-\eps} Z((k+1)\eta) \\
      & \leq C L e^{-2\ell_\rho}\,\ell_\rho^{-4 + \eps}
      \sum_{k=0}^{B\ell_\rho/\eta} ((k+1)\eta)^{-\eps}\,
      Z((k+1)\eta)\eta.
    \end{split}
  \end{equation*}
  Here, to get line three from line two, we have used
  Lemma~\ref{lem:inter11_closeEstimate}, and to get line four from
  line three, we have used Proposition~\ref{prop:IntervStatistics2} to
  bound the counting function with a probability $1-O(L^{-\infty})$.\\
  Thus, by the continuity and local integrability of $x\mapsto
  x^{-\varepsilon} Z(x)$, choosing $\eta$ small and
  $\varepsilon\in[0,1)$, we obtain that, for $L$ sufficiently large,
  with probability $1-O(L^{-\infty})$, one has
  \begin{equation}
    \label{eq:143}
    \Tr\Bigl(U^p (\Id - \Ex) \sum_{\substack{i, j = 1, \hdots,
        k_1\\i < j}}
    \gamma_{\phi_i} \otimes^s \gamma_{\phi_j}\Bigr) \\
    \leq  C\,n\,\rho\,\ell_\rho^{-4 + \eps} \int_0^{B \ell_\rho} 
    a^{-\eps} Z(a) \rmd{a}.
  \end{equation}  
  Let us now estimate the last integral. For
  $\eps\in[0,1)$ and $0\leq Y<X$, one computes
  \begin{equation*}
    \begin{split}
      \int_0^X a^{-\eps} Z(a) \rmd{a} &\leq \left(\int_0^Y +
        \int_Y^X\right) a^{-\eps} Z(a) \rmd{a} \\
      &\leq (1 - \eps)^{-1} \left[Z(0) Y^{1 - \eps} + Z(Y) X^{1 -
          \eps} - Z(Y) Y^{1 - \eps}\right] \\
      &= (1 - \eps)^{-1} X^{1 - \eps} \left[(Y / X)^{1 - \eps}(Z(0) -
        Z(Y)) + Z(Y)\right].
    \end{split}
  \end{equation*}
  Let us now optimize the last expression with respect to
  $\alpha=Y/X\in[0,1]$. Consider
  \begin{equation}
    \label{eq:274}
    f(X,\alpha) := \alpha^{1 - \eps} (Z(0) - Z(\alpha X)) +
    Z(\alpha X).
  \end{equation}
  In general, the more rapidly $Z$ goes to zero at infinity, the
  smaller the optimal $\alpha$ and, thus, the smaller is the minimal
  value.  Let us define the following functional of $Z$ (depending
  also on $X$):
  \begin{equation}
    \label{eq:1}
    f_Z(X) = \inf_{\alpha \in [0, 1]} f(X,\alpha).
  \end{equation}
  Obviously, as soon as $Z(X) = o(1)$ for $X \to +\infty$., one finds
  that $f_Z(X) = o(1)$ for $X \to +\infty$. Then, plugging this
  into the estimate~\eqref{eq:143}, we obtain
  \begin{equation}
    \label{eq:144}
    \Tr\Bigl(U^p (\Id - \Ex) \sum_{\substack{i, j = 1, \hdots, k_1\\i < j}}
    \gamma_{\phi_i} \otimes^s \gamma_{\phi_j}\Bigr) 
    \leq C_1\,n\,\rho\,\ell_\rho^{-3} \cdot f_Z(B \ell_\rho).
  \end{equation}
  In particular, the last expression is $o(n\rho\ell_\rho^{-3})$. Note
  also that, it can never be made better than $O(n\rho\ell_\rho^{-4})$
  as there is no control of the size of $Z$ near the origin.\\
  To estimate the interactions between a one-particle state and a
  one-particle density matrix of a two-particle state, we use the
  bound derived in Lemma~\ref{lem:inter12_closeEstimate}. We estimate
  the number of pairs of pieces of this type at a certain distance by
  Proposition~\ref{prop:NeighborssStatistics} (in this case, there is
  no need in for the more precise
  Proposition~\ref{prop:IntervStatistics2} as in the derivation
  of~\eqref{eq:144} above). This yields
  \begin{equation}
    \label{eq:119}
    \begin{split}
      &\Tr\Bigl(U^p (\Id - \Ex) \sum_{i = 1}^{k_1} \sum_{j = 1}^{k_2}
      \gamma_{\phi_i} \otimes^s \gamma_{\zeta_j}\Bigr) \\
      &\leq \sum_{\substack{|\Delta_i|\geq\ell_\rho - \rho x_*
          \\|\Delta_j| \in [2 \ell_\rho + A_*, 3 \ell_\rho)\\
          \dist(\Delta_i, \Delta_j) \leq B \ell_\rho}} \int_{\Delta_i
        \times \Delta_j} U(x - y) |\varphi^1_{\Delta_i}(x)|^2
      \gamma_{\zeta^1_{\Delta_j}}(y, y) \rmd{x} \rmd{y}\\&\leq
      C\,n\,\rho^2\,\ell_\rho\,\ell_\rho^{-7 / 2 +
        \varepsilon}\int_0^{B\ell_\rho} a^{-\varepsilon}Z(a)da.
    \end{split}
  \end{equation}
  Finally, for interactions between two reduced density matrices of
  two-particles sub-states, we proceed as before; using
  Lemma~\ref{lem:inter22_estimate} for each term, we compute
  \begin{equation}
    \label{eq:118}
    \begin{split}
      &\Tr\Bigl(U^p (\Id - \Ex) \sum_{\substack{i, j = 1, \hdots,
          k_2\\i < j}}
      \gamma_{\zeta_i} \otimes^s \gamma_{\zeta_j}\Bigr) \\
      &= \sum_{\substack{|\Delta_i|, |\Delta_j| \in
          [2 \ell_\rho + A_*, 3 \ell_\rho)\\
          i < j \\
          \dist(\Delta_i, \Delta_j) \leq B \ell_\rho}} \int_{\Delta_i
        \times \Delta_j} U(x - y) \gamma_{\zeta^1_{\Delta_i}}(x, x)
      \gamma_{\zeta^1_{\Delta_j}}(y, y) \rmd{x} \rmd{y} \\
      &\leq C\,n\,\rho^3\,\int_0^{B\ell_\rho}\min(1,a^{-2}Z(a))da.
    \end{split}
  \end{equation}
  Summing~\eqref{eq:144},~\eqref{eq:119},~\eqref{eq:118}, we obtain
  \begin{equation}
    \label{eq:246}
    \langle W^p \Psi^{\textup{opt}}, \Psi^{\textup{opt}} \rangle\leq
    C\,n\,\rho\,\ell_\rho^{-3} \cdot f_Z(B \ell_\rho).
  \end{equation}
  Taking~\eqref{eq:265} into account, this completes the proof of
  Theorem~\ref{thr:6}.
\end{proof}
\noindent To formulate our next result, we will first need to define
the notion of occupation restricted to a subset of the total set of
pieces.
\begin{definition}
  \label{def:partialOccupation}
  Let $\mathcal{P}_\omega = \{\Delta_k(\omega)\}_{k = 1}^{m(\omega)}$
  be the total set of pieces and let $Q \in \bbN^{m}$ be an
  occupation.  For $P\subset\mathcal{P}_\omega$ a subset of pieces,
  define the corresponding sub-occupation (or a restriction of
  occupation) as an occupation vector containing only those components
  that are singled out by $P$:
  \begin{equation*}
    Q|_P = (Q_k)_{k\colon\Delta_k\in P}.
  \end{equation*}
  When the subset $P$ is defined by a condition on the length of the
  pieces, we will use a shorthand notation involving only this
  condition, e.g., $Q|_{>\ell_\rho}$ stands for the occupation $Q$
  restricted to the pieces of length greater than the Fermi length
  $\ell_\rho$.
\end{definition}
\noindent Recall that $\Psi^{\textup{opt}}$ is constructed in
Definition~\textup{\ref{def:PsiOptr}}.
\begin{Th}
  \label{thr:3}
  For any non negative function $r:[0,\rho_0]\to\R^+$ such that
  $r(\rho) = o(1)$ when $\rho\to0^+$, there exist $C>0$ and $\rho_r>0$
  such that, for $\rho\in(0,\rho_r)$, in the thermodynamic limit, with
  probability $1-O(L^{-\infty})$, if $\Psi$ is a normalized
  $n$-particles state in $\frH_{Q(\Psi)}\cap \frH_\infty^n([0,L])$
  (see~\eqref{eq:28}) satisfying
  \begin{equation}
    \label{eq:133}
    \frac1{n}\bigl\langle H_\omega^{U^p}(L,n) \Psi, \Psi\bigr\rangle
      \leq 
      \frac1{n}\bigl\langle H_\omega^{U^p}(L,n) \Psi^{\textup{opt}},
      \Psi^{\textup{opt}}\bigr\rangle
      + \rho |\log{\rho}|^{-3} (r(\rho))^2,
  \end{equation}
  then
  \begin{equation}
    \label{eq:81}
    \begin{aligned}
      &\dist_1\left(Q|_{\geq \ell_\rho + C}(\Psi), Q|_{\geq \ell_\rho
          + C}(\Psi^{\textup{opt}})\right)
      \leq C n \rho \cdot \max(r(\rho), |\log{\rho}|^{-1}),\\
      &\dist_1\left(Q|_{< \ell_\rho + C}(\Psi), Q|_{< \ell_\rho +
          C}(\Psi^{\textup{opt}})\right) \leq C n \max(\sqrt{\rho}
      \cdot r(\rho), \rho |\log{\rho}|^{-1}).
    \end{aligned} 
  \end{equation}
\end{Th}
\begin{proof}[Proof of Theorem~\textup{\ref{thr:3}}]
  First of all, taking into account the form of the first inequality
  in~\eqref{eq:81}, while dealing with its proof we may suppose
  without loss of generality that $|\log{\rho}|^{-1}$ is
  asymptotically bounded by $r(\rho)$, i.e., for $\rho$ small,
  \begin{equation}
    \label{eq:134}
    |\log{\rho}|^{-1}\lesssim r(\rho).
  \end{equation}
  For the proof of the second inequality in~\eqref{eq:81}, we will no
  longer assume~\eqref{eq:134}.\\
  Consider now the pieces $(\Delta_k(\omega))_{1\leq k\leq m(\omega)}$
  (see section~\ref{sec:introduction}). Fix $\varepsilon>0$. We say
  that a piece $\Delta_k(\omega)$ is of $\varepsilon$-type
  \begin{enumerate}
  \item\label{ptype:1} if $|\Delta_k(\omega)|\geq
    3\ell_\rho(1-\varepsilon)$, that is, it has length at least
    $3\ell_\rho(1-\varepsilon)$;
  \item\label{ptype:2} if $|\Delta_k(\omega)|\geq
    2\ell_\rho(1-\varepsilon)$ and $\Delta_k(\omega)$ has at least one
    neighbor (in the sense of interactions $U^p$ from~\eqref{eq:99})
    of length at least $\ell_\rho(1-\varepsilon)$;
  \item\label{ptype:3} if $|\Delta_k(\omega)|\geq
    \ell_\rho(1-\varepsilon)$ and $\Delta_k(\omega)$ has at least two
    neighbors, each of length at least $\ell_\rho(1-\varepsilon)$.
  \end{enumerate}
  \noindent Note that, by~\eqref{eq:99}, as $U^p$ is of compact
  support of radius at most $B \ell_\rho$, there exists $\rho_0>0$
  such that for $\rho\in(0,\rho_0)$ and $\varepsilon\in(0,1/2)$, a
  given piece can have at most $2 B$ neighbors of length at least
  $\ell_\rho(1-\varepsilon)$.\\
  We first prove that ``exceptional'' pieces contribute only to the
  error term.
  \begin{Le}
    \label{le:17}
    Fix $\eta\in(0,1/3)$. There exists $\varepsilon\in(0,1/2)$ and
    $\rho_0>0$ such that, for $\rho\in(0,\rho_0)$, in the
    thermodynamic limit, with probability $1-O(L^{-\infty})$, if
    $\Psi\in\frH_{Q(\Psi)}\cap\frH_\infty^n([0,L])$ satisfies
    \begin{equation}
      \label{eq:210}
      \langle H_\omega^{U^p}(L,n)\Psi,\Psi\rangle\leq
      2\densEn^0(\rho)n\|\Psi|^2,
    \end{equation}
    then
    \begin{equation}
      \label{eq:82}
      \sum_{\bullet\in\{a,b,c\}}\sum_{\Delta_k(\omega)\text{ of
          $\varepsilon$-type }(\bullet)}Q_k(\Psi)\leq n\rho^{1+\eta}/2.
    \end{equation}
    and
    \begin{equation}
      \label{eq:209}
      \sum_{\Delta_k(\omega)\text{ of
          $\varepsilon$-type
        }(a)}\left[Q_k(\Psi)\right]^2\lesssim\densEn^0(\rho)n\cdot\log
      n\cdot\log\log n.
    \end{equation}
  \end{Le}
  \noindent Let us postpone the proof of this result for a while and
  continue with the proof of Theorem~\ref{thr:3}.  The following lemma
  estimates the total contribution of ``normal'' pieces (i.e., that
  are not of $\eps$-type) that carry too many particles.
  \begin{Le}
    \label{lem:normPiecesWithParticleExcess}
    Recall that $\{\Delta_k\}_{k = 1}^{m(\omega)}$ denote the
    pieces.\\
    There exists $C>0$ such that, for $L$ sufficiently large, with
    probability $1-O(L^{-\infty})$, for a normalized $n$-state $\Psi$
    in $\frH_{Q(\Psi)}\cap\frH_\infty^n([0,L])$
    satisfying~\eqref{eq:133} and $Q(\Psi)=(Q_k)_{1\leq k\leq
      m(\omega)}$, the occupation number of the state $\Psi$, one has
    \begin{gather}
      \label{eq:124}
      \sum_{\substack{|\Delta_k| \leq \ell_\rho (1-\rho^2)\\Q_k \geq
          2}} Q_k + \sum_{\substack{|\Delta_k| \in [\ell_\rho (1 -
          \rho^2), 2 \ell_\rho (1 - \rho^2))\\Q_k \geq 3}} Q_k +
      \sum_{\substack{|\Delta_k| \in [2 \ell_\rho (1 - \rho^2), 3
          \ell_\rho(1 - \rho^2))\\Q_k \geq 4}} Q_k \leq C n \rho
      \ell_\rho^{-1}
      \\\intertext{and}
      \label{eq:211} \sum_{\substack{|\Delta_k| \leq 3 \ell_\rho(1 -
          \rho^2)}} Q^2_k \leq C n \rho \ell_\rho^{-1}
    \end{gather}
    and, for $\varepsilon\in(\rho^2,1/4)$,
    \begin{equation}
      \label{eq:172}
      \sum_{\substack{|\Delta_k|\leq\ell_\rho(1-\eps)\\Q_k \geq 1}} 
      Q_k + \sum_{\substack{|\Delta_k|\leq2\ell_\rho(1-\eps)\\Q_k \geq
          2}} Q_k +\sum_{\substack{|\Delta_k|\leq3
          \ell_\rho(1-\eps)\\Q_k \geq 3}} Q_k
      \leq C n \frac{\rho}{\varepsilon-\rho^2}\ell_\rho^{-1}.
    \end{equation}
  \end{Le}
  \begin{proof}
    First, note that by Theorem~\ref{thr:6} and~\eqref{eq:133}, there
    exists a constant $\wtC$ such that
    \begin{equation}
      \label{eq:125}
      \langle H_\omega^{U^p} \Psi, \Psi \rangle \leq \langle
      H_\omega^{U^p} \Psi^{\textup{opt}}, \Psi^{\textup{opt}}
      \rangle + n \rho |\log{\rho}|^{-3} (r(\rho))^2\leq \langle
      H_\omega^0 \Psi^0, \Psi^0 \rangle + \wtC n\rho\ell_\rho^{-3}.
    \end{equation}
    Moreover, if $\D -\Delta_{\Delta_k}^{Q_k}$ denotes the Laplacian
    with Dirichlet boundary conditions on
    $\D\bigwedge^{Q_k}L^2(\Delta_k)$, one has
    \begin{equation}
      \label{eq:162}
      (H_\omega^{U^p})_{\frH_{Q(\Psi)}}\geq
      (H_\omega^0)_{\frH_{Q(\Psi)}}\geq
      \sum_{k=1}^{m(\omega)}\inf(\sigma(-\Delta_{\Delta_k}^{Q_k}))=
      \sum_{k=1}^{m(\omega)} \sum_{j = 1}^{Q_k} \frac{\pi^2
        j^2}{|\Delta_k|^2}= \sum_{k=1}^{m(\omega)} \frac{\pi^2
        P(Q_k)}{|\Delta_k|^2}
    \end{equation}
    where $\D P(X):= \frac{(2 X + 1) (X + 1) X}6$.\\
    On the other hand, by the description of $\Psi^0$, for some $C>0$,
    one has
    \begin{equation*}
      \langle
      H_\omega^0 \Psi^0, \Psi^0 \rangle\leq
      \sum_{|\Delta_k| \in [\ell_\rho (1 - \rho^2),2
        \ell_\rho (1 - \rho^2))}\frac{P(1)\,\pi^2}{|\Delta_k|^2} +
      \sum_{|\Delta_k| \in [2 \ell_\rho (1 - \rho^2), 3
        \ell_\rho(1 - \rho^2))} \frac{P(2)\,\pi^2}{|\Delta_k|^2}+C n\rho^2
    \end{equation*}
    Plugging this and~\eqref{eq:162} into~\eqref{eq:125}, we obtain
    \begin{equation}
      \label{eq:163}
      \begin{split}
        &\sum_{|\Delta_k| \leq \ell_\rho (1-\rho^2)}
        \frac{\pi^2}{|\Delta_k|^2}P(Q_k) + \sum_{|\Delta_k| \in
          [\ell_\rho (1 - \rho^2), 2 \ell_\rho (1 - \rho^2))}
        \frac{\pi^2 }{|\Delta_k|^2}(P(Q_k)-P(1))\\&\hskip3cm +
        \sum_{|\Delta_k| \in [2 \ell_\rho (1 - \rho^2), 3 \ell_\rho(1
          - \rho^2))} \frac{\pi^2 }{|\Delta_k|^2}(P(Q_k)-P(2)) \leq C
        n \rho \ell_\rho^{-3}.
      \end{split}
    \end{equation}
    By Lemma~\ref{le:17} and the explicit description of the non
    interacting ground state $\Psi^0$ (see the beginning of
    section~\ref{sec:comp-psiopt-ground}), for some $C>0$ and $\rho$
    sufficiently small, for $L$ sufficiently large, with probability
    $1-O(L^{-\infty})$, one has
    \begin{equation}
      \label{eq:164}
      \begin{split}
        &\sum_{|\Delta_k| \leq \ell_\rho (1-\rho^2)}Q_k +
        \sum_{|\Delta_k| \in [\ell_\rho (1 - \rho^2), 2 \ell_\rho (1 -
          \rho^2))} Q_k + \sum_{|\Delta_k| \in [2 \ell_\rho (1 -
          \rho^2), 3 \ell_\rho(1 - \rho^2))} Q_k\\
        &\hskip1cm\geq n(1-C\rho^2)\\
        &\hskip1cm\geq\left[ \sum_{|\Delta_k| \in [\ell_\rho (1 +
            \rho^2), 2 \ell_\rho (1 - \rho^2))}1 + \sum_{|\Delta_k|
            \in [2 \ell_\rho (1 + \rho^2), 3
            \ell_\rho(1 - \rho^2))}2\right]-2Cn\rho^2 \\
        &\hskip1cm\geq\left[ \sum_{|\Delta_k| \in [\ell_\rho (1 -
            \rho^2), 2 \ell_\rho (1 - \rho^2))}1 + \sum_{|\Delta_k|
            \in [2 \ell_\rho (1 - \rho^2), 3 \ell_\rho(1 -
            \rho^2))}2\right]-3Cn\rho^2
      \end{split}
    \end{equation}
    as
    \begin{equation*}
      \#\{k;\ |\Delta_k| \in [\ell_\rho (1 - \rho^2),
      \ell_\rho (1 + \rho^2))\cup[2
      \ell_\rho (1 - \rho^2), 2 \ell_\rho(1 +
      \rho^2))\}\leq Cn\rho^2.
    \end{equation*}
    Thus,~\eqref{eq:164} yields
    \begin{equation}
      \label{eq:165}
      \begin{split}
        &\sum_{\substack{|\Delta_k| \leq \ell_\rho
            (1-\rho^2)\\Q_k\geq1}}Q_k + \sum_{\substack{|\Delta_k| \in
            [\ell_\rho (1 - \rho^2), 2 \ell_\rho (1 -
            \rho^2))\\Q_k\geq2}} (Q_k-1) + \sum_{\substack{|\Delta_k|
            \in [2 \ell_\rho (1 - \rho^2), 3 \ell_\rho(1 -
            \rho^2))\\Q_k\geq3}} (Q_k-2) \\&\geq
        \left[\sum_{\substack{|\Delta_k| \in [\ell_\rho (1 - \rho^2),
              2 \ell_\rho (1 - \rho^2))\\Q_k=0}}1 +
          \sum_{\substack{|\Delta_k| \in [2 \ell_\rho (1 - \rho^2), 3
              \ell_\rho(1 -
              \rho^2))\\Q_k\leq1}}2\right]-3n\rho^{1+\eta}
      \end{split}
    \end{equation}
    Rewrite~\eqref{eq:163} as
    \begin{equation*}
      \begin{split}
        C n \rho \ell_\rho^{-1}&\geq \sum_{\substack{|\Delta_k| \leq
            \ell_\rho (1-\rho^2)\\Q_k\geq1}}
        \frac{\pi^2}{|\Delta_k|^2}P(Q_k) + \sum_{\substack{|\Delta_k|
            \in [\ell_\rho (1 - \rho^2), 2 \ell_\rho (1 -
            \rho^2))\\Q_k\geq2}} \frac{\pi^2
        }{|\Delta_k|^2}(P(Q_k)-P(1))\\&\hskip3cm +
        \sum_{\substack{|\Delta_k| \in [2 \ell_\rho (1 - \rho^2), 3
            \ell_\rho(1 - \rho^2))\\Q_k\geq3}} \frac{\pi^2
        }{|\Delta_k|^2}(P(Q_k)-P(2)) \\&\hskip1cm-
        \sum_{\substack{|\Delta_k| \in [\ell_\rho (1 - \rho^2), 2
            \ell_\rho (1 - \rho^2))\\Q_k=0}} \frac{P(1)\pi^2
        }{|\Delta_k|^2}-\sum_{\substack{|\Delta_k| \in [2 \ell_\rho (1
            - \rho^2), 3 \ell_\rho(1 - \rho^2))\\Q_k\leq1}}
        \frac{(P(2)-P(Q_k))\pi^2 }{|\Delta_k|^2} \\&\geq
        \sum_{\substack{|\Delta_k| \leq \ell_\rho
            (1-\rho^2)\\Q_k\geq1}} \frac{\pi^2}{|\Delta_k|^2}P(Q_k) +
        \sum_{\substack{|\Delta_k| \in [\ell_\rho (1 - \rho^2), 2
            \ell_\rho (1 - \rho^2))\\Q_k\geq2}} \frac{\pi^2
        }{|\Delta_k|^2}(P(Q_k)-P(1))\\&\hskip3cm +
        \sum_{\substack{|\Delta_k| \in [2 \ell_\rho (1 - \rho^2), 3
            \ell_\rho(1 - \rho^2))\\Q_k\geq3}} \frac{\pi^2
        }{|\Delta_k|^2}(P(Q_k)-P(2)) \\&\hskip1cm-
        P(1)\left(\sum_{\substack{|\Delta_k| \in [\ell_\rho (1 -
              \rho^2), 2 \ell_\rho (1 - \rho^2))\\Q_k=0}} \frac{\pi^2
          }{|\Delta_k|^2}\right)- P(2)\left(\sum_{\substack{|\Delta_k|
              \in [2 \ell_\rho (1 - \rho^2), 3 \ell_\rho(1 -
              \rho^2))\\Q_k\leq1}} \frac{\pi^2 }{|\Delta_k|^2}\right).
      \end{split}
    \end{equation*}
    Hence,
    \begin{equation*}
      \begin{split}
        C n \rho \ell_\rho^{-1}&\geq \sum_{\substack{|\Delta_k| \leq
            \ell_\rho (1-\rho^2)\\Q_k\geq1}}
        \frac{\pi^2}{|\Delta_k|^2}P(Q_k) + \sum_{\substack{|\Delta_k|
            \in [\ell_\rho (1 - \rho^2), 2 \ell_\rho (1 -
            \rho^2))\\Q_k\geq2}} \frac{\pi^2
        }{|\Delta_k|^2}(P(Q_k)-P(1))\\&\hskip1cm +
        \sum_{\substack{|\Delta_k| \in [2 \ell_\rho (1 - \rho^2), 3
            \ell_\rho(1 - \rho^2))\\Q_k\geq3}} \frac{\pi^2
        }{|\Delta_k|^2}(P(Q_k)-P(2)) \\&\hskip2cm-
        \frac{\pi^2}{|\ell_\rho (1 - \rho^2)|^2}\left[
          \sum_{|\Delta_k| \in [\ell_\rho (1 - \rho^2), 2 \ell_\rho (1
            - \rho^2))}1 + \sum_{|\Delta_k| \in [2 \ell_\rho (1 -
            \rho^2), 3 \ell_\rho(1 - \rho^2))}2\right]
      \end{split}
    \end{equation*}
    as $P(1)=1$ and $P(2)=5\leq 8=2^3P(1)$.\\
    Using~\eqref{eq:165}, we then obtain
    \begin{equation}
      \label{eq:167}
      \begin{split}
        C n \rho \ell_\rho^{-1}&\geq \sum_{\substack{|\Delta_k| \leq
            \ell_\rho (1-\rho^2)\\Q_k\geq1}}
        \left(\frac{\pi^2}{|\Delta_k|^2}P(Q_k)-\frac{\pi^2}{|\ell_\rho
            (1 - \rho^2)|^2}Q_k\right)\\& + \sum_{\substack{|\Delta_k|
            \in [\ell_\rho (1 - \rho^2), 2 \ell_\rho (1 -
            \rho^2))\\Q_k\geq2}}\left( \frac{\pi^2
          }{|\Delta_k|^2}(P(Q_k)-P(1))-\frac{\pi^2}{|\ell_\rho (1 -
            \rho^2)|^2}(Q_k-1)\right)\\&+ \sum_{\substack{|\Delta_k|
            \in [2 \ell_\rho (1 - \rho^2), 3 \ell_\rho(1 -
            \rho^2))\\Q_k\geq3}}\left( \frac{\pi^2
          }{|\Delta_k|^2}(P(Q_k)-P(2))-\frac{\pi^2}{|\ell_\rho (1 -
            \rho^2)|^2}(Q_k-2)\right).
      \end{split}
    \end{equation}
    Now, we note that, for $X\geq n+1$, $X$ integer, one has
    \begin{equation}
      \label{eq:275}
      P(X)-P(n)=\sum_{k=n+1}^Xk^2\geq (n+1)^2(X-n).
    \end{equation}
    This yields
    \begin{itemize}
    \item for $Q_k\geq1$ and $|\Delta_k|\leq\ell_\rho (1-\rho^2)$, one
      has
      \begin{equation}
        \label{eq:168}
        \frac{\pi^2}{|\Delta_k|^2}P(Q_k)-\frac{\pi^2}{|\ell_\rho
          (1 - \rho^2)|^2}Q_k>\frac{\pi^2Q_k(Q_k-1)(2Q_k+3)}{6{|\ell_\rho
            (1 - \rho^2)|^2}}\geq0;
      \end{equation}
      if, moreover, $|\Delta_k|\leq\ell_\rho (1-\varepsilon)$
      ($\rho^2<\varepsilon<1/2$), by~\eqref{eq:275}, one has
      \begin{equation}
        \label{eq:169}
        \frac{\pi^2}{|\Delta_k|^2}P(Q_k)-\frac{\pi^2}{|\ell_\rho
          (1 - \rho^2)|^2}Q_k\geq
        \left(\frac{\pi^2}{|\Delta_k|^2}-\frac{\pi^2}{|\ell_\rho
            (1 - \rho^2)|^2}\right)Q_k\geq \frac{(8\pi)^2
          (\varepsilon-\rho^2)}{|\ell_\rho|^2}\, Q_k;
      \end{equation}
    \item for $Q_k\geq2$ and $|\Delta_k|\leq 2\ell_\rho(1-\rho^2)$,
      one has
      \begin{equation}
        \label{eq:212}
        \begin{split}
          \frac{\pi^2}{|\Delta_k|^2}(P(Q_k)-P(1))-\frac{\pi^2}{|\ell_\rho
            (1-\rho^2)|^2}(Q_k-1)&>
          \frac{\pi^2(2Q_k+9)(Q_k-2)(Q_k-1)}{24|\ell_\rho
            (1-\rho^2)|^2}\\&\geq0;
        \end{split}
      \end{equation}
      if, moreover, $|\Delta_k| \leq 2\ell_\rho (1-\varepsilon)$
      ($\rho^2<\varepsilon<1/2$), by~\eqref{eq:275}, one has
      \begin{equation}
        \label{eq:170}
        \begin{split}
          \frac{\pi^2}{|\Delta_k|^2}(P(Q_k)-P(1))-\frac{\pi^2}{|\ell_\rho
            (1-\rho^2)|^2}(Q_k-1)&\geq
          \left(\frac{4\pi^2}{|\Delta_k|^2}-\frac{\pi^2}{|\ell_\rho
              (1-\rho^2)|^2}\right)(Q_k-1)\\&\geq \frac{(8\pi)^2
            (\varepsilon-\rho^2)}{|\ell_\rho|^2}\, (Q_k-1);
        \end{split}
      \end{equation}
    \item for $Q_k\geq3$ and $|\Delta_k|\leq3\ell_\rho(1-\rho^2)$, one
      has
      \begin{equation}
        \label{eq:213}
        \begin{split}
          \frac{\pi^2}{|\Delta_k|^2}(P(Q_k)-P(2))-\frac{\pi^2}{|\ell_\rho
            (1-\rho^2)|^2}(Q_k-2)&>
          \frac{\pi^2(2Q_k+13)(Q_k-3)(Q_k-2)}{|\ell_\rho
            (1-\rho^2)|^2}\\&\geq0;
        \end{split}
      \end{equation}
      if, moreover, $|\Delta_k| \leq 3\ell_\rho (1-\varepsilon)$
      ($\rho^2<\varepsilon<1/2$), by~\eqref{eq:275}, one has
      \begin{equation}
        \label{eq:171}
        \begin{split}
        \frac{\pi^2}{|\Delta_k|^2}(P(Q_k)-P(2))-\frac{\pi^2}{|\ell_\rho
          (1-\rho^2)|^2}(Q_k-2)&\geq
          \left(\frac{9\pi^2}{|\Delta_k|^2}-\frac{\pi^2}{|\ell_\rho
              (1-\rho^2)|^2}\right)(Q_k-9)\\&\geq \frac{(8\pi)^2
          (\varepsilon-\rho^2)}{|\ell_\rho|^2}\, (Q_k-2).          
        \end{split}
      \end{equation}
    \end{itemize}
    Plugging~\eqref{eq:168}~-~\eqref{eq:171} into~\eqref{eq:167}
    immediately yields~\eqref{eq:124} and~\eqref{eq:172}, thus,
    completes the proof of~\eqref{eq:124} and~\eqref{eq:172} in
    Lemma~\ref{lem:normPiecesWithParticleExcess}.\\
    To derive~\eqref{eq:211}, we proceed as follows. Clearly, for
    $Q_k\geq4$, the right hand sides of~\eqref{eq:168}, \eqref{eq:212}
    and~\eqref{eq:213} is larger than $\delta\cdot Q^2_k$ (for some
    $\delta\in(0,1)$). Thus,~\eqref{eq:167} implies
    \begin{equation*}
      \sum_{\substack{|\Delta_k| \leq 3 \ell_\rho(1 -
          \rho^2)\\Q_k\geq4}} Q^2_k \leq C n \rho \ell_\rho^{-1}.
    \end{equation*}
    On the other hand, by~\eqref{eq:124}, one clearly has
    \begin{equation*}
      \sum_{\substack{|\Delta_k| \leq 3 \ell_\rho(1 -
          \rho^2)\\Q_k\leq3}} Q^2_k \leq 3
      \sum_{\substack{|\Delta_k| \leq 3 \ell_\rho(1 - 
          \rho^2)\\Q_k\leq3}} Q_k \leq C n \rho \ell_\rho^{-1}.
    \end{equation*}
    Thus, the proof of~\eqref{eq:211} is complete. This completes the
    proof of Lemma~\ref{lem:normPiecesWithParticleExcess}.
  \end{proof}
  \noindent We also remark the following
  \begin{Le}
    \label{lem:minimalOccupiedInterval}
    Consider $\Psi^{U^p}_\omega$, the ground state of
    $H^{U^p}_\omega(L, n)$.\\
    There exists $C>0$ such that for $L$ sufficiently large, with
    probability at least $1 - O(L^{-\infty})$, no piece of length
    smaller than
    \begin{equation}
      \label{eq:120}
      \ell_{min} = \ell_\rho - C \rho \ell_\rho
    \end{equation}
    is occupied by particles of $\Psi^{U^p}$.
  \end{Le}
  \begin{Rem}
    \label{rem:4}
    The proof of Lemma~\ref{lem:minimalOccupiedInterval} shows that it
    suffices to take $C> 4B+4$ for $\rho$ sufficiently small; here,
    $B$ is the constant defining $U^p$ (see~\eqref{eq:99}).
  \end{Rem}
  \begin{proof}
    Suppose that the claim of the lemma is false. Then, a piece
    shorter than $\ell_{min}$ is occupied.\\
    Let us show now that, as there are too many such pieces, pieces
    longer than $\ell_{min}$ cannot be all in interaction with $n$
    particles, no matter where these $n$ particles are.\\
    First of all, according to
    Proposition~\ref{prop:IntervStatistics}, the total number of
    pieces longer than $\ell_{min}$ is
    \begin{equation*}
      \begin{split}
        \sharp\{j:\ |\Delta_j(\omega)| \geq \ell_{min}\} &= L
        e^{-\ell_{min}} + O(L^{1/2 + 0})= L \frac{\rho}{1 + \rho} (1 +
        C \rho \ell_\rho + O(\rho^2 \ell_\rho^2) \\ &= n (1 + C \rho
        \ell_\rho + O(\rho)).
      \end{split}
    \end{equation*}
    The number of pieces of length larger than $2 \ell_\rho$ is $n
    \rho (1 + O(\rho))$. If a particle lies in one of these pieces, it
    can interact with at most $2 B$ other pieces of length greater
    than $\ell_{min}$.\\
    For pieces smaller than $2 \ell_\rho$ (but as always larger than
    $\ell_{min}$), we remark that if two such pieces are at a distance
    greater than $(2 B + 2) \ell_\rho$ from one another then they
    cannot interact with the same particle, except for the cases
    already taken into account above. \\
    Moreover, according to Proposition~\ref{prop:IntervStatistics2},
    the number of pairs of such pieces at distance at most $(2 B + 2)
    \ell_\rho$ is given by
    \begin{equation*}
      \begin{split}
        &\sharp\{(\Delta_i, \Delta_j), |\Delta_i| > \ell_{min},
        |\Delta_j| > \ell_{min}, \dist(\Delta_i, \Delta_j) \leq (2 B +
        2) \ell_\rho\} \\
        &= 2 (2 B + 2) \ell_\rho L \left(e^{-\ell_{min}}\right)^2 +
        O(L^{3/4}) \\
        &= (4 B + 4) n \rho \ell_\rho (1 + O(\rho \ell_\rho)).
      \end{split}
    \end{equation*}
    Consequently, the rest of these pieces are at larger distances
    from each other. This leaves at least
    \begin{multline*}
      n (1 + C \rho \ell_\rho + O(\rho)) - (2 B + 1) n \rho (1 +
      O(\rho)) - (4 B + 4) n \rho \ell_\rho (1 + O(\rho \ell_\rho)) \\
      = n (1 + (C - 4 B - 4) \rho \ell_\rho + O(\rho))
    \end{multline*}
    pieces such that no two of them can interact with the same
    particle.  Remark that it suffices to take $C > 4 B + 4$ to ensure
    that this number is larger than $n$ for $\rho$ small. This proves
    that there exists at least one piece longer than $\ell_{min}$
    which neither occupied nor interacting with any particle in a
    ground state $\Psi^{U^p}_\omega(L,n)$.\\
    This leads to a contradiction with the fact that the ground state
    $\Psi^{U^p}_\omega(L,n)$ puts at least one particle in a piece
    smaller than $\ell_{min}$: indeed, moving this particle to the
    piece longer than $\ell_{min}$ which was singled out just above
    would result in a decrease of energy as no interaction energy
    would be added and non interacting energy would obviously decrease
    with the increase of the piece's length. This completes the proof
    of Lemma~\ref{lem:minimalOccupiedInterval}.
  \end{proof}
  \noindent Let us now resume the proof of Theorem~\ref{thr:3}. In
  what follows, $\Psi$ is a function satisfying
  condition~\eqref{eq:133}. By Theorem~\ref{thr:6}, using
  $\Psi^{\text{opt}}(L,n)$ as a trial function, we see that both
  $\Psi$ and $\Psi^{\text{opt}}(L,n)$ satisfy the assumptions of
  Lemma~\ref{le:17}. Thus, picking $\eta \in(0, 1/3)$ and
  $\varepsilon$ sufficiently small, by Lemma~\ref{le:17}, for $\rho$
  sufficiently small and $L$ sufficiently large, with probability
  $1-O(L^{-\infty})$, we have
  \begin{equation}
    \label{eq:90}
    \sum_{\bullet\in\{a,b,c\}}\sum_{\Delta_k(\omega)\text{ of
        $\varepsilon$-type
      }(\bullet)}\left(Q_k(\Psi^{\text{opt}}(L,n))+
      Q_k(\Psi)\right)\leq n\rho^{1+\eta}.
  \end{equation}
  We will now reason on the particles in $\Psi^{U^p}_\omega(L,n)$ that
  live in pieces that are not of $\varepsilon$-type (a), (b) or (c).\\
  Recall that, by definition (see Definitions~\ref{def:PsiOptm}
  and~\ref{def:PsiOptr}), $\Psi^{\textup{opt}}(L,n)$ puts
  \begin{itemize}
  \item no particle in each piece of length in $(0,\ell_\rho-x_*\rho)$;
  \item one particle in each piece of length in
    $[\ell_\rho-x_*\rho,2\ell_\rho+A_*)$;
  \item two particles (as a true two-particles state) in each piece of
    length in $[2\ell_\rho+A_*,3\ell_\rho)$;
  \end{itemize}
  Let $C$ be the constant from the claim of Theorem~\ref{thr:3} that
  we will fix later on. Define
  \begin{itemize}
  \item $n_0^+$ to be the total number of pieces of length in
    $(0,\ell_\rho-x_*\rho)$ where $\Psi$ puts exactly 1 particle;
  \item $n_1^-$ to be the total number of pieces
     of length in $[\ell_\rho-x_*\rho, \ell_\rho + C)$ where
    $\Psi$ puts no particle;
  \item $n_1^+$ to be the total number of pieces of length in
    $[\ell_\rho-x_*\rho, \ell_\rho + C)$ where
    $\Psi$ puts exactly $2$ particles;
  \item $\wtn_1^-$ to be the total number of pieces
     of length in $[\ell_\rho + C, 2 \ell_\rho + A_*)$ where
    $\Psi$ puts no particle;
  \item $\wtn_1^+$ to be the total number of pieces of length in
    $[\ell_\rho + C, 2 \ell_\rho + A_*)$ where
    $\Psi$ puts exactly $2$ particles;
  \item $n_2^-$ to be the total number of pieces of length in
    $[2\ell_\rho+A_*,3\ell_\rho(1-\varepsilon))$ where
    $\Psi$ puts exactly $1$ particle;
  \item $n_2^+$ to be the total number of pieces of length in
    $[2\ell_\rho+A_*,3\ell_\rho(1-\varepsilon))$ where
    $\Psi$ puts exactly $3$ particles.
  \end{itemize}
  The general idea of the forthcoming proof is the following. On the
  one hand, Lemma~\ref{le:17} tells that pieces with too many
  neighbors are a sort of exception in a sense that they occur
  relatively rarely and carry relatively few particles.  From the
  other hand, according to
  Lemma~\ref{lem:normPiecesWithParticleExcess}, pieces with too
  many particles are also relatively exceptional.\\
  Finally, let us complement these two observations by noting that no
  particle in a piece of length in $[2 \ell_\rho + A_*, 3 \ell_\rho (1
  - \eps))$ can also occur for a small fraction of them. Therefore, we
  first note that it is sufficient to argue for pieces that are not of
  $\eps$-type (as those of $\eps$-type are already handled by
  Lemma~\ref{le:17}). Let us now take a look at the distribution of
  particles in the state $\Psi^{\textup{opt}}$ in the pieces of length
  in $[2\ell_\rho+A_*, 3\ell_\rho(1-\eps))$ that have no particles and
  no neighbors (as they are not of $\eps$-type) in $\Psi$.  Obviously,
  moving a particle from a piece of length greater than $2 \ell_\rho +
  A_*$ to a smaller piece induces an increase of the non interacting
  energy of order $\ell_\rho^{-2}$ just because the pieces longer than
  $\ell_\rho-\rho x_*$ are already occupied by at least one particle
  (thus the non interacting energy of a second particle is a best
  $4\pi^2/(2\ell_\rho+A_*)^2$ and $\pi^2/(\ell_\rho-\rho x_*)^2$ if a
  particle is placed in a non occupied piece). Thus, the total number
  of pieces of length greater than $2\ell_\rho+A_*$ with no particles
  is bounded by $O(n\rho\ell_\rho^{-1})$.\vskip.2cm\noindent
  The last three arguments together prove essentially that the
  distances $\dist_0$ and $\dist_1$ coincide for the matter of the
  current proof up to an admissible error i.e. of size $O(n \rho
  \ell_\rho^{-1})$. Namely, by the definition of the distance
  dist$_1$, one has
  \begin{equation}
    \label{eq:91}
    \begin{split}
      \dist_1(Q|_{<\ell_\rho + C}(\Psi), Q|_{<\ell_\rho +
        C}(\Psi^{\text{opt}})) 
      &= n_0^++n_1^++n_1^- + r, \\
    \dist_1(Q|_{\geq \ell_\rho + C}(\Psi), Q|_{\geq \ell_\rho +
      C}(\Psi^{\text{opt}})) &= \wtn_1^+ + \wtn_1^-+n_2^++n_2^-+r', 
    \end{split}
  \end{equation}
  and, by the fact that the total number of particles in both states
  is the same, one gets
  \begin{equation}
    \label{eq:92}
    n_0^++n_1^+ + \wtn_1^+ +n_2^++r'' = n_1^- + \wtn_1^- +n_2^-+r'''
  \end{equation}
  where 
  \begin{equation}
    \label{eq:93}
    \max(r,r',r'', r''')\leq C n \rho \ell_\rho^{-1}.  
  \end{equation}
  Recall that $r(\rho)$ is of order at most $|\log{\rho}|^{-1}$.
  Hence, if~\eqref{eq:81} does not hold, for any constant $C_1$, if $L$
  is large enough, either one has
  \begin{equation}
    \label{eq:94}
    \wtn_1^+ + \wtn_1^- + n_2^+ + n_2^- \geq C_1 n \rho \cdot r(\rho)
  \end{equation}
  or one has
  \begin{equation}
    \label{eq:128}
    n_0^+ + n_1^+ + n_1^- \geq C_1 n \sqrt{\rho} \cdot r(\rho).
  \end{equation}
  First, we simplify~\eqref{eq:94}. Suppose that, for some $C_1$
  large, one has
  \begin{equation}
    \label{eq:121}
    n_2^+ \geq \frac{C_1}{4} n \rho \cdot r(\rho).
  \end{equation}
  The number of pieces of length in
  $\left[\tfrac{5}{2}\ell_\rho,3\ell_\rho(1-\eps)\right)$ is given by
  \begin{equation*}
    \sharp\left\{j:\ |\Delta_j(\omega)|\in\left[\frac{5}{2}\ell_\rho,
        3\ell_\rho(1-\eps)\right)\right\}=O(n\rho^{3/2}). 
  \end{equation*}
  Thus, at least $\tfrac{C_1}{5} n \rho \cdot r(\rho)$ of the pieces
  with three particles (as given by~\eqref{eq:121}) have their length
  in $\left[2 \ell_\rho + A_*, \tfrac{5}{2} \ell_\rho\right)$. Hence,
  the non interacting energy excess (compared to the non interacting
  energy in the ground state) for each of these pieces is lower
  bounded by $O(\ell_\rho^{-2})$ which, in turn, being multiplied by
  their total number, contradicts~\eqref{eq:133}. This
  simplifies~\eqref{eq:94} into
  \begin{equation}
    \label{eq:122}
    \wtn_1^+ + \wtn_1^- + n_2^- \geq C_1 n \rho \cdot r(\rho).  
  \end{equation}
  The conditions~\eqref{eq:92},~\eqref{eq:93} and either
  ~\eqref{eq:128} or~\eqref{eq:122} lead us to a number of
  possibilities that we will now study one by one. More precisely,
  there are nine possible variants as at least one among $n_1^-$,
  $\tilde n_1^-$ and $n_2^-$ should be ``large'' and the same is true
  for either $n_0^+$, $n_1^+$, $n_2^+$ and $\tilde n_1^+$. We now
  discuss these cases.
  \begin{enumerate}
  \item Consider first the case when %either
    \begin{equation}
      \label{eq:77}
      \min(\wtn_1^+, n_2^-) \geq C_2 n \rho \cdot r(\rho)
    \end{equation}
    with $C_2 < C_1 / 3$. \\
    This corresponds to taking the same configuration of particles as
    in $\Psi^{\textup{opt}}$ and move some of them from pieces of
    length in $[2\ell_\rho+A_*,3\ell_\rho(1-\eps))$ to pieces of
    length in $[\ell_\rho+C,2\ell_\rho+A_*)$ that already contain one
    particle each. As we are now dealing only with pieces that are not
    of $\eps$-type, this implies in particular that the pieces of
    length in $[2\ell_\rho+A_*,3\ell_\rho(1-\eps))$ from which we
    withdraw particles and that originally contain $2$ particles, do
    not have any neighbors.\\
    Taking the smallest available pieces for particle donors and the
    largest available for particle acceptors gives a lower bound on
    the total energy increase induced by this operation. Suppose that
    $C_2n\rho r(\rho)$ smallest pieces have their length between $2
    \ell_\rho+A_*$ and $2\ell_\rho+A_*+\delta$. Then, choosing $C_1$
    (thus, $C_2$) much larger than the constant in
    Lemma~\ref{lem:normPiecesWithParticleExcess} for the case when
    $r(\rho)\asymp|\log{\rho}|^{-1}$, we obtain
    \begin{equation*}
      L e^{-2\ell_\rho-A_*}(1-e^{-\delta})\geq\frac{C_2}{2}n\rho\cdot r(\rho),
    \end{equation*}
    which yields
    \begin{equation}
      \label{eq:135}
      \delta \geq \frac{C_2 e^{A_*}}{2} r(\rho).
    \end{equation}
    Moreover, analogous calculations show that at least
    $\tfrac{C_2}{3} n \rho r(\rho)$ of these pieces have length in
    $(2\ell_\rho+A_*+\delta/2,2\ell_\rho+A_*+\delta)$. For the
    particles in these pieces, the increase of energy is lower bounded
    by
    \begin{multline}
      \label{eq:137}
      \frac{4 \pi^2}{(2 \ell_\rho + A_* + \delta / 2)^2} +
      \frac{\gamma}{(2 \ell_\rho + A_* + \delta / 2)^3} \\
      - \frac{4 \pi^2}{(2 \ell_\rho + A_*)^2} - \frac{\gamma}{(2
        \ell_\rho + A_*)^3} + O(\ell_\rho^{-4}) \geq C_3 r(\rho)
      \ell_\rho^{-3},
    \end{multline}
    where $C_3>0$. Multiplying the number of pieces by the lower bound
    ~\eqref{eq:137} gives a total energy excess that contradicts
    ~\eqref{eq:133} if we choose $C_2$ (hence, $C_1$) sufficiently
    large.
  \item The case
    \begin{equation*}
      \min(n_1^+, n_2^-) \geq C_2 n \rho \cdot r(\rho)
    \end{equation*}
    is even simpler than the previous one. Indeed, in
    $\Psi^{\textup{opt}}$, the occupations of the pieces of length in
    $[\ell_\rho-\rho x_*,\ell_\rho+C)$ and in
    $[\ell_\rho+C,2\ell_\rho+A_*)$ are the same but the lengths
    considered in the previous case are smaller. Hence, the arguments
    developed in point (a) above enable one to conclude with the only
    difference that the increase of energy is even larger. Moreover,
    there is no need to remove the small interval of size $\delta$.
  \item Next, the situation when
    \begin{equation}
      \label{eq:126}
      \min(n_0^+, n_2^-) \geq C_2 n \rho \cdot r(\rho)
    \end{equation}
    corresponds to moving excited particles, i.e., particles occupying
    the second energy level, from pieces of length in $[2 \ell_\rho +
    A_*, 3 \ell_\rho (1 - \eps))$ to empty pieces of length smaller
    than $\ell_\rho - \rho x_*$. Recall that actually the approximate
    equilibrium between the gain in interaction energy due to
    decoupling and the increase of non-interaction energy was part of
    the definition of values of $x_*$ and $A_*$, i.e.,
    \begin{equation}
      \label{eq:123}
      \frac{4 \pi^2}{(2 \ell_\rho + A_*)^2} + \gamma \ell_\rho^{-3} =
      \frac{\pi^2}{(\ell_\rho - \rho x _*)^2} + O(\ell_\rho^{-4}).
    \end{equation}
    Obviously, the smaller the piece we choose to remove the second
    particle from, the more energy one gains. On the other hand, the
    larger the piece where one puts the particle, the smaller
    the non interacting energy increase, thus, the better.\\
    According to these two observations, we choose to move particles
    from the $C_2n \rho\cdot r(\rho)$ smallest pieces longer than
    $2\ell_\rho+A_*$. Suppose that the largest of these pieces has
    length $2\ell_\rho+A_*+B_2$. Then, by
    Proposition~\ref{prop:IntervStatistics}, $B_2$ satisfies
    \begin{equation*}
      Le^{-2\ell_\rho-A_*}(1-e^{-B_2})+O(L^{1/2+0})=C_2n\rho\cdot r(\rho).
    \end{equation*}
    Hence, $B_2=C_2 e^{A_*}r(\rho)(1+O(r(\rho)))$. Moreover, the
    number of such pieces with length in
    $[2\ell_\rho+A_*+B_2/2,2\ell_\rho+A_*+B_2)$ is
    \begin{equation}
      \label{eq:138}
      \begin{split}
        \sharp\{k;\ |\Delta_k(\omega)|-2\ell_\rho-A_*\in[B_2/2,B_2)\}
        &=Le^{-2\ell_\rho-A_*}(e^{-B_2/2}-e^{-B_2})+O(L^{1/2+0}) \\
        &\geq \frac{C_2}{3} n \rho \cdot r(\rho).
      \end{split}
    \end{equation}
    Clearly, for all these $\tfrac{C_2}{3} n \rho \ell_\rho^{-1}$
    pieces, the non interacting energy excess is proportional to
    $C_2\ell_\rho^{-3}r(\rho)$; thus, multiplied by their total
    number~\eqref{eq:138}, for large $C_2$, this energy excess does
    not fit within the margin allowed by~\eqref{eq:133}.
  \item Yet another possibility for~\eqref{eq:122} is that
    \begin{equation*}
      \min(\max(n_1^+, \wtn_1^+), \max(n_1^-, \wtn_1^-)) 
      \geq C_2 n \rho \cdot r(\rho).
    \end{equation*}
    Obviously, the variant
    \begin{equation*}
      \min(\wtn_1^+, n_1^-) \geq C_2 n \rho \cdot r(\rho).
    \end{equation*}
    is more advantageous from the energetic point of view.  The
    question here is whether it is worth moving a particle from a
    piece of length close to the lower bound of the corresponding
    group, i.e., $\ell_\rho-\rho x_*$, to another piece (but as the
    second particle because there is already another particle in that
    piece) of length close to the upper bound, i.e., $2\ell_\rho+A_*$.
    In a certain sense, this is the opposite to the case~(c) as the
    latter tells that the threshold value $A_*$ is not too small,
    while the current case will explain why $A_*$ is not too big.\\
    As above, one shows that, in order to choose the $C_2n\rho\cdot
    r(\rho)$ largest pieces of length in $[\ell_\rho-\rho x_*, 2
    \ell_\rho+A_*)$, it is sufficient to solve
    \begin{equation*}
      L e^{-2\ell_\rho-A_*}(e^{B_1}-1)+O(L^{1/2+0})=C_2n\rho\cdot r(\rho),
    \end{equation*}
    which also implies $B_1=C_2e^{A_*}r(\rho)(1+O(r(\rho)))$. Then, as
    above, the energy excess is proportional to $C_2 \ell_\rho^{-3}
    r(\rho)$ (where the constant $C_2$ can be chosen arbitrarily
    large) whereas the interaction terms are uniformly bounded by
    $O(\ell_\rho^{-4 + 0})$. Thus, the total energy gained by such an
    operation exceeds the limits imposed by~\eqref{eq:133}.
  \item The next possible option is that
    \begin{equation}
      \label{eq:139}
      \min(n_0^+, \wtn_1^-) \geq C_2 n \rho \cdot r(\rho).      
    \end{equation}
    This corresponds to moving particles in $\Psi^{\textup{opt}}$ from
    pieces of longer than $\ell_\rho+C$ to pieces shorter than
    $\ell_\rho-\rho x_*$. Remark first that the increase of non
    interacting energy is at least
    \begin{equation}
      \label{eq:140}
      \frac{\pi^2}{(\ell_\rho - \rho x_*)^2} - \frac{\pi^2}{(\ell_\rho +
        C)^2} \geq \frac{2 \pi^2 C}{\ell_\rho^3},
    \end{equation}
    which always dominates the possible interaction with a particle in
    a neighboring piece: this interaction is $O(\ell_\rho^{-4 + 0})$
    by Lemma~\ref{lem:inter11_closeEstimate}. Multiplying the left
    hand sides of~\eqref{eq:139} and~\eqref{eq:140} gives a lower
    estimate on the energy excess that contradicts~\eqref{eq:133}
    because $r(\rho)=o(1)$.
  \item Finally, the only case left is when
    \begin{equation}
      \label{eq:127}
      \min(n_0^+, n_1^-) \geq C_2 n \sqrt{\rho} \cdot r(\rho).    
    \end{equation}
    Informally speaking, this is about the question if the threshold
    $\ell_\rho - \rho x_*$ between occupation zero and occupation one
    is placed correctly.\\
    It is also remarkable that the allowed number of particle
    displacements for this case is much larger than in the other
    cases: one has to compare $o(n \sqrt{\rho})$ to $o(n \rho)$. This
    is due to the following mechanism. First, note that moving a
    particle that interacts with another particle in a neighboring
    piece may result to a decrease of the total energy. Obviously, the
    contribution of the displacement of such particles is upper
    bounded by $O(n \rho \ell_\rho^{-4 + 0})$ because there are at
    most $O(n \rho)$ neighboring particles and the size of interaction
    is $O(\ell_\rho^{-4+0})$ by Lemma~\ref{lem:inter11_closeEstimate}.
    Thus, these particles may be neglected for the precision of the
    current proof.\\
    Then, reasoning as we did many times above, we observe that, at
    least $\tfrac{C_2}{2}n\sqrt{\rho}r(\rho)$ of particles that are
    removed from pieces of length in $[\ell_\rho-\rho
    x_*,\ell_\rho+C)$ have their length greater than $\ell_\rho+C_3
    \sqrt{\rho}r(\rho)$, where the constant $C_3$ grows together with
    $C_2$.  But, for each of these particles the non interacting
    energy increase is of order $C_3 \ell_\rho^{-3} \sqrt{\rho}
    \cdot r(\rho)$. As above, multiplying the number of particles involved
    by the lower bound on the energy change, we get a contradiction
    with~\eqref{eq:133}.
  \end{enumerate}
  This completes the proof of Theorem~\ref{thr:3}.
\end{proof}
\noindent We are now left with proving Lemmas~\ref{le:17}.
\begin{proof}[The proof of Lemma~\textup{\ref{le:17}}]
  We first prove the estimate~\eqref{eq:82}. It will be a consequence
  of the fact that the number of pieces in any of the three type is
  small and of the following
  \begin{Le}
    \label{le:18}
    Pick $k$ pieces of respective lengths $l_1\leq l_2\leq\cdots\leq
    l_k$. Assume that, for $1\leq i\leq k$, the state
    $\Psi\in\frH^n_{Q(\Psi)}\cap\frH_\infty^n([0,L])$ puts exactly
    $\nu_i$ particles in the piece $i$ so that
    $\nu_1+\cdots+\nu_k=\nu$. Then, one has
    \begin{equation}
      \label{eq:83}
      \frac{\pi^2 \nu^3}{3 l_k^2k^2}\leq
      \langle H^0(L,n)\Psi,\Psi\rangle\leq\langle
      H_\omega^{U^p}(L,n)\Psi,\Psi\rangle\leq\langle
      H_\omega^U(L,n)\Psi,\Psi\rangle.
    \end{equation}
  \end{Le}
  \noindent Let us postpone the proof of this result for a while and
  complete the proof of Lemma~\ref{le:17}. We shall write out the
  proof for pieces of type (a). Those for pieces of type (b) and (c)
  is similar. \\
  Pick $\eta\in(0,1)$ and $\varepsilon>0$ such that
  $\eta+2\varepsilon<1/6$. The proof of
  Propositions~\ref{prop:IntervStatistics} and~\ref{pro:3}) show that
  there exists $\rho_\varepsilon>0$ such that, for
  $\rho\in(0,\rho_\varepsilon)$, for $L$ sufficiently large, with
  probability $1-O(L^{-\infty})$, for one has
  \begin{equation}
    \label{eq:87}
    \#\left\{k;\ |\Delta_k(\omega)|\in[3\ell_\rho(1-\varepsilon),
      4\ell_\rho)\right\}\leq 
    n\rho^{2-3\varepsilon}
  \end{equation}
  and, for $4\leq k\leq\log L\cdot\log\log L$,
  \begin{equation}
    \label{eq:88}
    \#\left\{k;\ |\Delta_k(\omega)|\in[k\ell_\rho,(k+1)\ell_\rho)\right\}\leq
    n\rho^{k-1-\varepsilon}.
  \end{equation}
  Now, if $\Psi$ places more than $n\rho^{1+\eta}$ particles in pieces
  of type $a$ then
  \begin{itemize}
  \item either it places at least $2^{-1}n\rho^{1+\eta}$ particles in
    pieces of length in $[3\ell_\rho(1-\varepsilon),4\ell_\rho)$; in
    this case, by Lemma~\ref{le:18}, as $3(\eta+2\varepsilon)<1$, we
    know that
    \begin{equation}
      \label{eq:84}
      \langle H^0(L,n)\Psi,\Psi\rangle\geq\frac{\pi^2(n\rho^{1+\eta})^3}
      {8(4\ell_\rho)^2(n\rho^{2-3\varepsilon})^2}\gtrsim n
      \ell_\rho^{-2}\rho^{-1+3(\eta+2\varepsilon)}\gg n
      \ell_\rho^{-2}
    \end{equation}
    for $\rho$ small;
  \item or, for some $4\leq k\leq\log L$, it places at least
    $n\rho^{1+\eta}2^{-k+2}$ particles in pieces of length in
    $[k\ell_\rho,(k+1)\ell_\rho)$; in this case, by Lemma~\ref{le:18},
    we know that
    \begin{equation}
      \label{eq:85}
      \langle H^0(L,n)\Psi,\Psi\rangle\gtrsim
      \frac{n\rho^{3+3\eta-2k-2\varepsilon}}{((k+1)\ell_\rho)^{2}2^{3k}}
      \gtrsim n\ell_\rho^{-2}\rho^{-1}\frac{(8\rho)^{-k}}{(k+1)^2} \geq
      n\ell_\rho^{-2}\rho^{-1}
    \end{equation}
    for $\rho$ sufficiently small.
  \end{itemize}
  Hence, for $\rho$ sufficiently small, recalling~\eqref{eq:5}
  and~\eqref{eq:86} (and that here $\mu=1$), one has $\langle
  H^0(L,n)\Psi,\Psi\rangle> 2\densEn^0(\rho)n$. \\
  This completes the proof of~\eqref{eq:82} in Lemma~\ref{le:17} for
  particles of type (a).\\
  To deal with the particles of type (b) (resp. (c)), we replace the
  upper bounds~\eqref{eq:87} and~\eqref{eq:88} obtained using
  Proposition~\ref{prop:IntervStatistics} by analogous upper bounds on
  the numbers of pieces of type (b) (resp. (c)) obtained through
  Proposition~\ref{prop:NeighborssStatistics}
  (resp. Proposition~\ref{prop:NeighborssStatistics2}).\\
  This completes the proof of~\eqref{eq:82} in Lemma~\ref{le:17}.\\
  Let us now prove~\eqref{eq:209}. By~\eqref{eq:162}, one has 
  \begin{equation*}
    \sum_{k=1}^{m(\omega)} \frac{\pi^2
      P(Q_k(\Psi))}{|\Delta_k|^2}\leq       \langle
    H_\omega^{U^p}(L,n)\Psi,\Psi\rangle\leq 2\densEn^0(\rho)n
  \end{equation*}
  where $P$ is defined in~\eqref{eq:162}.\\
  Taking Proposition~\ref{pro:3} into account immediately
  yields~\eqref{eq:209} and completes the proof of~Lemma~\ref{le:17}.
\end{proof}
\begin{proof}[The proof of Lemma~\ref{le:18}]
  The form of the
  Hamiltonians~\eqref{eq:Hn0Definition},~\eqref{eq:99} (the
  definition of $U^p$),~\eqref{eq:HLambdanIntro} and the non
  negativity of the interactions guarantee that
  \begin{equation*}
    \langle
    H_\omega^{U^p}(L,n)\Psi,\Psi\rangle\geq
    \langle H^0(L,n)\Psi,\Psi\rangle\geq
    \sum_{i=1}^k\sum_{m=1}^{\nu_i}\left(\frac{\pi \alpha^i_m}{l_i}\right)^2
  \end{equation*}
  where $(\alpha^i_m)_{1\leq m\leq \nu_i}\in(\N^*)^{\nu_i}$ and
  $\alpha^i_1<\alpha^i_2<\cdots<\alpha^i_{\nu_i}$.\\
  Thus
  \begin{equation*}
    \langle H^0(L,n)\Psi,\Psi\rangle\geq
    \sum_{i=1}^k\sum_{m=1}^{\nu_i}\left(\frac{\pi m}{l_i}\right)^2\geq
    \frac{\pi^2}{3l^2_k}\sum_{i=1}^k\nu^3_i\geq       \frac{\pi^2
      \nu^3}{3 l_k^2k^2}
  \end{equation*}
  as $\nu_1+\cdots+\nu_k=\nu$.\\
  This completes the proof of Lemma~\ref{le:18}.
\end{proof}
  \begin{Th}
    \label{th:PsiUmPsiOptEnergyEstimate}
    For $\rho$ sufficiently small, in the thermodynamic limit, with
    probability $1-O(L^{-\infty})$, for any function
    $\Psi\in\frH^n\cap\frH_\infty^n([0,L])$,
    \begin{equation}
      \label{eq:130}
      \frac{1}{n} \langle H^{U^p}_\omega(L, n) \Psi, \Psi \rangle 
      \geq \frac{1}{n} \langle H^{U^p}_\omega(L, n) \Psi^{\textup{opt}},
      \Psi^{\textup{opt}} \rangle - o(\rho |\log{\rho}|^{-3}).
    \end{equation}
  \end{Th}
  \begin{proof}
    This result can easily be traced throughout the proof of
    Theorem~\ref{thr:3} by considering each of the cases. Before doing
    so, let us give some preliminary remarks that correspond exactly
    to the three remarks found in the beginning of the proof of
    Theorem~\ref{thr:3}.\\
    First, the energy gain due to moving a single particle is always
    bounded by $O(\ell_\rho^{-2})$ just because each individual
    particle in $\Psi^{\textup{opt}}$ brings to the system at most
    this amount of energy.\\
    Next, the number of pieces of $\eps$-type is $O(n\rho^{1+\eta})$
    (see Lemma~\ref{le:17}); thus, the energy gain due to them is at
    most $O(n\rho^{1+\eta}\ell_\rho^{-2})$.\\
    The pieces with too many particles are also rare by
    Lemma~\ref{lem:normPiecesWithParticleExcess}. Moreover, the many
    particles in these pieces always bring an excess of energy and
    never an energy gain.\\
    Finally, the analysis of $n_2^+$ large (see~\eqref{eq:121}) shows
    that moving an extra particle to the majority of these pieces
    results in an energy increase of order of $O(\ell_\rho^{-2})$,
    whereas for only $O(n\rho^{3/2})$ of them
    adding a particle may be energetically favorable.\\
    We treat now the cases from (a) to (f) of the last part of the
    proof of Theorem~\ref{thr:3}. For the matter of the current proof
    we shall put $r(\rho) = 0$ (because we are interested only in
    those states that have the energy smaller that
    $\Psi^{\textup{opt}}$), thus, reducing the claim of
    Theorem~\ref{thr:3} to
    \begin{equation*}
      \dist_1(Q(\Psi), Q(\Psi^{\textup{opt}})) = O(n \rho \ell_\rho^{-1}).
    \end{equation*}
    \begin{itemize}
    \item For those displacements when the possible energy gain is due
      to removing interaction with neighbors (this includes the cases
      (d), (e) and (f)), it suffices to remark that, by
      Lemma~\ref{lem:inter11_closeEstimate}, the size of the
      interacting energy is bounded by $O(\ell_\rho^{-4+0})$.
      Combined with the fact that, in total, there are $O(n \rho)$
      pairs of neighboring particles, this yields a total energy gain
      of size $O(n\rho\ell_\rho^{-4+0})$.
    \item For those displacements when the possible energy gain is due
      to decoupling particles living in the same piece (cases (a), (b)
      and (c)), the individual interacting energy is of size
      $O(\ell_\rho^{-3})$ while their total number is $O(n \rho
      \ell_\rho^{-1})$. This yields a total energy gain of size
      $O(n\rho\ell_\rho^{-4})$.
    \item Finally, when the energy gain results from a non interacting
      energy decrease (like in the case (d)), it is at most
      $O(\ell_\rho^{-3})$ and the total number of displacements that
      result in energy decrease is $O(n \rho \ell_\rho^{-1})$. This
      again yields a total energy gain of size
      $O(n\rho\ell_\rho^{-4})$.
    \end{itemize}
    This concludes the proof of~\eqref{eq:130}.
  \end{proof}

\begin{Cor}
  \label{cor:4}
  There exists $\rho_0 > 0$ such that for $\rho \in (0, \rho_0)$, in
  the thermodynamic limit, with probability $1-O(L^{-\infty})$,
  \begin{equation}
    \label{eq:129}
      \begin{split}
        \frac{1}{n} \langle H^{U^p}_\omega(L, n) \Psi^{U^p},
        \Psi^{U^p} \rangle &= \frac{1}{n} \langle H^{U^p}_\omega(L, n)
        \Psi^{\textup{opt}}, \Psi^{\textup{opt}} \rangle +
        O(\rho |\log{\rho}|^{-4}) \\
        &= \densEn^0(\rho) + \pi^2 \gamma_*
        \frac{\rho}{|\log{\rho}|^3} +
        \frac{\rho}{|\log{\rho}|^3}O\left(f_Z(|\log{\rho}|)\right),
      \end{split}
  \end{equation}
  where the constant $\gamma_*$ is given in~\eqref{eq:12}, $Z$
  describes the behavior of $U$ at infinity and $f_Z$ is defined in
  Theorem~\textup{\ref{thr:6}}. 
\end{Cor}

\begin{proof}
  The upper bound is given by the fact that $\Psi^{U^p}$ is the ground
  state of $H_\omega^{U^p}$. The lower bound is a direct consequence
  of~\eqref{eq:130} and~\eqref{eq:89}. This proves~\eqref{eq:129}.
\end{proof}

\begin{remark}
  The ground state $\Psi^{U^p}$ satisfies the conditions of
  Theorem~\ref{thr:3}. Hence, the inequalities~\eqref{eq:81} hold for
  the distance between the occupations of $\Psi^{U^p}$ and
  $\Psi^{\textup{opt}}$.
\end{remark}
\subsection{The proof of Theorem~\ref{th:EnergyAsymptoticExpansion}}
\label{sec:proof-theor-text}
Theorem~\ref{thr:3} and Theorem~\ref{th:PsiUmPsiOptEnergyEstimate}
give a rather complete description of the ground state for the
operator with compactified interactions $H^{U^p}_\omega(L, n)$. The
description is given in terms of comparison with $\Psi^{\textup{opt}}$
(see Definitions~\ref{def:PsiOptm} and \ref{def:PsiOptr}).  In this
section, we complement it with estimates on the residual part of
interactions $W^r$ (see~\eqref{eq:99}).

\begin{proposition}
  \label{prop:WrPsiOptFormEstimate}
  There exists $\rho_0$ such that, for $\rho\in(0,\rho_0)$, in the
  thermodynamic limit, for $L$ sufficiently large, with probability
  $1-O(L^{-\infty})$, one has
  \begin{equation}
    \label{eq:131}
    \frac{1}{n} \langle W^r \Psi^{\textup{opt}}, \Psi^\textup{{opt}}
    \rangle = O(\rho |\log{\rho}|^{-3} Z(2|\log{\rho}|)). 
  \end{equation}
\end{proposition}

\begin{proof}
  We will mostly follow the lines of the second part of the proof of
  Theorem~\ref{thr:6} (see formula~\eqref{eq:112} and what
  follows). First, as in~\eqref{eq:114}, one computes
  \begin{equation*}
      \langle W^r \Psi^{\textup{opt}}, \Psi^{\textup{opt}} \rangle =
      \Tr\left(U^r \gamma^{(2)}_{\Psi^{\textup{opt}}}\right)
  \end{equation*}
  where $\gamma^{(2)}_{\Psi^{\textup{opt}}}$ is given
  by~\eqref{eq:113}.  Let us treat here only the contribution of the
  second sum ~\eqref{eq:113}. It corresponds to interactions between
  single particles in pieces of length in $[\ell_\rho-\rho
  x_*,2\ell_\rho+A_*)$. The other three sums only contribute error
  terms as the number of $2$-particles sub-states in
  $\Psi^{\textup{opt}}$ is by a factor $\rho$ smaller than that of
  single-particle sub-states.  For the second sum in~\eqref{eq:113}.,
  using Lemma~\ref{lem:inter11_farEstimate}, one obtains
  \begin{equation*}
    \begin{split}
      &\Tr\Bigl(U^r (\Id - \Ex) \sum_{\substack{i, j = 1, \hdots, k_1\\i < j}}
    \gamma_{\phi_i} \otimes^s \gamma_{\phi_j}\Bigr) \\
      &\leq \sum_{\substack{|\Delta_i|, |\Delta_j| \in [\ell_\rho - \rho
          x_*, 2 \ell_\rho +  A_*) \cup [3 \ell_\rho, +\infty)\\i <
          j\\\dist(\Delta_i, \Delta_j) > B \ell_\rho}}
      \int_{\Delta_i \times \Delta_j} U(x - y) |\varphi^1_{\Delta_i}(x)|^2
      |\varphi^1_{\Delta_j}(y)|^2 \rmd{x} \rmd{y} \\
      &\leq C_1 n \rho \int_{B \ell_\rho}^{+\infty} \ell_\rho^{-1}
      a^{-3} Z(a) \rmd{a}.
    \end{split}
  \end{equation*}
  Recall that $Z$ is defined in~\eqref{eq:98}.\\
  We compute next
  \begin{equation*}
    \int_{B \ell_\rho}^{+\infty} a^{-3} Z(a) \rmd{a} = \int_{B
      \ell_\rho}^{+\infty} \int_a^{+\infty} U(x) \rmd{x}
    \rmd{a} \leq \int_{B \ell_\rho}^{+\infty} x U(x) \rmd{x}\leq C
    \ell_\rho^{-2} Z(B \ell_\rho),
  \end{equation*}
  where the last inequality is just~\eqref{eq:111} for $\eps = 2$.
  This completes the proof of~\eqref{eq:131}.
\end{proof}
\begin{proof}[Proof of
  Theorem~\textup{\ref{th:EnergyAsymptoticExpansion}}]
  Proposition~\ref{prop:WrPsiOptFormEstimate} immediately entails the
  asymptotics of the interacting ground state energy
  $\densEn^{U}(\rho)$. Indeed, as $H^{U^p} \leq H^U$, one has
  $\densEn^{U^p}(\rho)\leq\densEn^{U}(\rho)$; thus, the announced
  lower bound is given by~\eqref{eq:129}. On the other hand, by
  Theorem~\ref{thr:6} and Proposition~\ref{prop:WrPsiOptFormEstimate},
  one has
  \begin{equation}
    \label{eq:222}
    \begin{split}
      \langle H^U \Psi^U, \Psi^U \rangle &\leq \langle H^U
      \Psi^\textup{{opt}}, \Psi^{\textup{opt}} \rangle = \langle
      H^{U^p} \Psi^{\textup{opt}}, \Psi^{\textup{opt}} \rangle +
      \langle W^r \Psi^{\textup{opt}}, \Psi^{\textup{opt}} \rangle \\
      &= \densEn^{0}(\rho) + \pi^2\gamma_* \rho |\log{\rho}|^{-3}
      \left(1 + O\left(f_Z(|\log{\rho}|)\right)\right),
    \end{split}
  \end{equation}
  which gives the announced upper bound.\\
  This, the facts that $B>2$ and that $Z$ is decreasing complete the
  proof of Theorem~\ref{th:EnergyAsymptoticExpansion}.
\end{proof}
\noindent Our analysis yields the following description for the
possible occupations of the ground state of the full Hamiltonian.
\begin{Cor}
  \label{cor:3}
  There exists $C>0$ such that, $\omega$ almost surely, in the
  thermodynamic limit, with probability $1-O(L^{-\infty)}$, for any
  $\Psi^{U}$, ground state of the full Hamiltonian of fixed occupation
  $Q(\Psi^U)$, one has
  \begin{equation}
    \label{eq:142}
    Q(\Psi^U)\in \mathcal{Q}_\rho:=
    \left\{Q\text{ occ.};
      \begin{aligned}
        &\dist_1\left(Q|_{\geq \ell_\rho + C},
          Q|_{\geq \ell_\rho + C}(\Psi^{\textup{opt}})\right) \\
        &\hspace{16mm} \leq C n \rho \max\left(\sqrt{Z(2|\log{\rho}|)},
          |\log{\rho}|^{-1}\right), \\
        &\dist_1\left(Q|_{< \ell_\rho + C},
          Q|_{< \ell_\rho + C}(\Psi^{\textup{opt}})\right) \\
        &\hspace{16mm} \leq C n \max\left(\sqrt{\rho\,Z(2|\log{\rho}|)},
          \rho |\log{\rho}|^{-1}\right).
      \end{aligned}
    \right\}
  \end{equation}
\end{Cor}
\begin{proof}
  Note that
  \begin{equation*}
      \langle H^{U^p} \Psi^{U}, \Psi^U \rangle \leq \langle H^{U}
      \Psi^{U}, \Psi^U \rangle \leq \langle H^{U^p}
      \Psi^{\textup{opt}}, \Psi^{\textup{opt}} \rangle + \langle W^r
      \Psi^{\textup{opt}}, \Psi^{\textup{opt}} \rangle.
  \end{equation*}
  Thus, according to Proposition~\ref{prop:WrPsiOptFormEstimate},
  $\Psi^U$ satisfies the condition~\eqref{eq:133} with
  \begin{equation*}
    r(\rho) = C\sqrt{Z(2|\log{\rho}|)}
  \end{equation*}
  for some $C>0$ sufficiently large.\\
  Then, Theorem~\ref{thr:3} is applicable and
  yields~\eqref{eq:142}. This completes the proof of
  Corollary~\ref{cor:3}.
\end{proof}

%%% Local Variables: 
%%% mode: latex
%%% TeX-master: "PiecesModelGroundState"
%%% ispell-local-dictionary: "american"
%%% End: 

\section{From the occupation and energy bounds to the control of the
  density matrices}
\label{sec:proof-theorem6}
In this section, we will derive Theorem~\ref{thr:2} from
Theorem~\ref{th:EnergyAsymptoticExpansion}, Corollary~\ref{cor:3} and
a computation of the reduced one particle and two particles density
matrix of a (non factorized) state. More precisely, from
Theorem~\ref{th:EnergyAsymptoticExpansion} and Corollary~\ref{cor:3},
we will infer a description of the ground state $\Psi^U$ in most of
the pieces: roughly, in most of the pieces, the only occupied state is
the ground state (up to a controllable error). We then use this
knowledge to compute the reduced one particle and two particles
density matrix of $\Psi^U$ (up to a controllable error).
\subsection{From the occupation decomposition to the reduced
  density matrices}
\label{sec:ground-state-psiu}
Fix a configuration of the Poisson points, say, $\omega$, and a state
$\Psi\in\frH^n(\Lambda)$. Recall that, in the configuration $\omega$,
the pieces are denoted by $(\Delta_j(\omega))_{1\leq j\leq
  m}=(\Delta_j)_{1\leq j\leq m}$ (where $m=m(\omega)$, see
section~\ref{sec:analys-one-part}). For $1\leq j\leq m$ and $q\geq1$,
let $(E^j_{q,n})_{1\leq n}$ be the eigenvalues (ordered increasingly)
and $(\varphi^j_{q,n})_{1\leq n}$ be the associated eigenvectors of
$q$ interacting electronic particles in the piece $\Delta_j(\omega)$
i.e. the eigenvalues and eigenvectors of the Hamiltonian
\begin{equation}
  \label{eq:161}
  H^q_{\Delta_j(\omega)}=-\sum_{l=1}^q\frac{d^2}{dx_l^2}+ \sum_{1\leq
    l<l'\leq q}U^p(x_l-x_{l'})  
\end{equation}
acting on $\D\bigwedge_{l=1}^q L^2(\Delta_j(\omega))$ with Dirichlet
boundary conditions. Recall that $U^p$ is defined in
section~\ref{sec:comp-psiopt-ground} (see~\eqref{eq:99}).\\
The occupation number decomposition (see
section~\ref{sec:decomp-occup}) implies that one can write
\begin{equation}
  \label{eq:146}
  \Psi=\sum_{Q}\Psi_Q\quad\text{and}\quad
  \Psi_Q=\sum_{\overline{n}\in\N^m}a^Q_{\overline{n}}\Phi^Q_{\overline{n}}
  =\sum_{\substack{(n_j)_{1\leq j\leq m}\\\forall j,\ n_j\geq1}}
  a^Q_{n_1,\cdots,n_{m}}(\Psi)\bigwedge_{j=1}^{m}\varphi^j_{Q_j,n_j}
\end{equation}
where
\begin{itemize}
\item the first sum is taken over the occupation number
  $Q=(Q_j)_{1\leq j\leq m}$; recall $\D\sum_{j=1}^mQ_j=n$;
\item we have defined $\D\Phi^Q_{\overline{n}}:=
  \bigwedge_{j=1}^{m}\varphi^j_{Q_j,n_j}$; we refer to~\eqref{eq:267}
  in section~\ref{sec:proj-totally-antisym} for an explicit
  description of the anti-symmetric tensor product.
\end{itemize}
\begin{Rem}
  \label{rem:9}
  In~\eqref{eq:146}, the convention in the exterior product is that,
  if $Q_j=0$, then the corresponding basis vector drops out of the
  exterior product. Thus, the product is only at most $n$
  fold. Moreover, in this case, $a^Q_{n_1,\cdots,n_{m}}=0$ if
  $n_j\geq2$.
\end{Rem}
\noindent For $\overline{n}=(n_1,\cdots,n_{m})\in\N^m$, we write
$a^Q_{\overline{n}}=a^Q_{n_1,\cdots,n_{m}}=a^Q_{n_1,\cdots,n_{m}}(\Psi)$.
These coefficients are uniquely determined by $\Psi$.
\subsubsection{The one-particle density matrix}
\label{sec:one-particle-density}
We shall first compute the $1$ particle reduced density matrix in
terms of the coefficients $(a^Q_{\overline{n}})_{Q,\overline{n}}$
coming up in the occupation number decomposition~\eqref{eq:146}. We
prove
\begin{Th}
  \label{thr:4}
  The $1$-particle density $\gamma^{(1)}_\Psi$
  (see~\eqref{eq:OneParticleDensityMatrixDef}) is written as $\D
  \gamma_\Psi^{(1)}=\gamma_\Psi^{(1),d}+\gamma_\Psi^{(1),o}$ where
  \begin{gather}
    \label{eq:160}
    \gamma_\Psi^{(1),d}=\sum_{j=1}^m\sum_{\substack{Q\text{ occ.}\\
        Q_j\geq1}} \sum_{\substack{n_j\geq1\\ n'_j\geq1}}\sum_{\tilde
      n\in\N^{m-1}}a^Q_{\tilde n_j} \overline{a^Q_{\tilde
        n'_j}}\gamma^{(1)}_{\substack{Q_j\\n_j,n'_j}} \\\label{eq:174}
    \gamma_\Psi^{(1),o}=\sum_{\substack{i,j=1\\i
        \not=j}}^m\sum_{\substack{Q,\text{ occ. }Q_j\geq1\\ Q':\
        Q'_k=Q_k\text{ if }k\not\in\{i,j\}\\Q'_i=Q_i+1\\Q'_j=Q_j-1}}
    C_1(Q,i,j)\sum_{\tilde
      n\in\N^{m-2}}\sum_{\substack{n_i,n_j\geq1\\ n'_i,n'_j\geq1}} a^Q_{\tilde
      n_{i,j}} \overline{a^{Q'}_{\tilde n'_{i,j}}}
    \gamma^{(1)}_{\substack{Q_i,Q_j\\n_i,n_j\\n'_i,n'_j}}
  \end{gather}
  and
  \begin{itemize}
  \item we have used the shorthands
    \begin{itemize}
    \item $\tilde n_j$ for the vector $(\tilde n_1\cdots,\tilde
      n_{j-1},n_j,\tilde n_{j},\cdots,\tilde n_{m-1})$ when $\tilde
      n=(\tilde n_1,\cdots,\tilde n_{m-1})$,
    \item $\tilde n_{i,j}$ for $(\tilde n_1,\cdots,\tilde
      n_{i-1},n_i,\tilde n_i,\cdots,\tilde n_{j-2},n_j,\tilde
      n_{j-1},\cdots,\tilde n_{m-2})$ when $i<j$ and $\tilde n=(\tilde
      n_1,\cdots,\tilde n_{m-2})$,
    \end{itemize}
  \item the trace class operator $\gamma^{(1)}_{\substack{Q_j\\n_j,n'_j}}:\
    L^2(\Delta_j)\to L^2(\Delta_j)$ has the kernel
    \begin{equation*}
      \gamma^{(1)}_{\substack{Q_j\\n_j,n'_j}}(x,y)=Q_j\int_{\Delta^{Q_j-1}_j}
      \varphi^j_{Q_j,n_j}(x,z)
      \overline{\varphi^j_{Q_j,n'_j}(y,z)}dz,
    \end{equation*}
  \item $\D C_1(Q,i,j)=\frac{(n-Q_j-Q_i-1)!Q_i!Q_j!}{(n-1)!}$;
  \item the rank one operator
    $\gamma^{(1)}_{\substack{Q_i,Q_j\\n_i,n_j\\n'_i,n'_j}}:\ L^2(\Delta_i)\to
    L^2(\Delta_j)$ has the kernel
    \begin{equation}
      \label{eq:16}
      \gamma^{(1)}_{\substack{Q_i,Q_j\\n_i,n_j\\n'_i,n'_j}}(x,y)=
      \int_{\Delta^{Q_j-1}_j}\varphi^j_{Q_j,n_j}(x,z)
      \overline{\varphi^j_{Q_j-1,n'_j}(z)}dz\,\int_{\Delta^{Q_i}_i}
      \varphi^i_{Q_i,n_i}(z) \overline{\varphi^i_{Q_i+1,n'_i}(y,z)}dz.
    \end{equation}
  \end{itemize}
\end{Th}
\begin{Rem}
  \label{rem:8} 
  In~\eqref{eq:16}, in accordance with remark~\ref{rem:9}, we use the
  following convention
  \begin{itemize}
  \item if $Q_j=1$ and $Q_i=0$ then $n'_j=1$ and $n_i=1$ (i.e. for
    different indices, the coefficient $ a^Q_{\tilde n_{i,j}}
    \overline{a^{Q'}_{\tilde n'_{i,j}}}$ vanishes) and
    \begin{equation}
      \label{eq:145}
      \gamma^{(1)}_{\substack{Q_i,Q_j\\1,n_j\\n'_i,1}}(x,y)= 
      \varphi^j_{1,n_j}(x)\cdot\overline{\varphi^i_{1,n'_i}(y)},
    \end{equation}
  \item if $Q_j\geq2$ and $Q_i=0$ then $n_i=1$ and
    \begin{equation}
      \label{eq:155}
      \gamma^{(1)}_{\substack{Q_i,Q_j\\1,n_j\\n'_i,n'_j}}(x,y)=
      \overline{\varphi^i_{1,n'_i}(y)}
      \int_{\Delta^{Q_j-1}_j}\overline{\varphi^j_{Q_j-1,n'_j}(z)}
      \varphi^j_{Q_j,n_j}(x,z)dz,
    \end{equation}
  \item if $Q_j=1$ and $Q_i\geq1$ then $n'_j=1$ and
    \begin{equation}
      \label{eq:156}
      \gamma^{(1)}_{\substack{Q_i,Q_j\\n_i,n_j\\n'_i,1}}(x,y)=
      \varphi^j_{1,n_j}(x)\int_{\Delta^{Q_i}_i}
      \varphi^i_{Q_i,n_i}(z)\overline{\varphi^i_{Q_i+1,n'_i}(y,z)}dz.
    \end{equation}
  \end{itemize}  
\end{Rem}
\begin{proof}
  Theorem~\ref{thr:4} follows from a direct computation that we now
  perform. First, by the bilinearity of
  formula~\eqref{eq:OneParticleDensityMatrixDef}, one has

  \begin{equation}
    \label{eq:151}
    \gamma_\Psi^{(1)}=n\sum_{\substack{Q\text{ occ.}\\
        \overline{n}\in\N^m}}\sum_{\substack{Q'\text{ occ.}\\
        \overline{n}'\in\N^m}}
    a^Q_{\overline{n}}\overline{a^{Q'}_{\overline{n}'}}
    \gamma^{(1)}_{\substack{Q,\overline{n}\\Q',\overline{n}'}}
  \end{equation}
  where the trace class operator $\gamma^{(1)}_{Q,n,Q',n'}$ acts on
  $L^2([0,L])$ and has the kernel
  \begin{equation*}
    \gamma^{(1)}_{\substack{Q,\overline{n}\\Q',\overline{n}'}}(x,y):=\int_{[0,L]^{n-1}}
    \left[\bigwedge_{j=1}^{m}\varphi^j_{Q_j,n_j}\right](x,z)
    \overline{\left[\bigwedge_{j=1}^{m}\varphi^j_{Q'_j,n'_j}\right] (y,z)}dz.
  \end{equation*}
  Recall~\eqref{eq:267}, that is, in the present case
  \begin{multline}
    \label{eq:152}
    \left[\bigwedge_{j=1}^{m}\varphi^j_{Q_j,n_j}\right](z_1,z_2,\cdots,z_n)
    \\=c(Q)\cdot\sum_{\substack{|A_j|=Q_j,\
        \forall1\leq j\leq m\\A_1\cup\cdots\cup A_m=\{1,\cdots,n\}\\
        A_j\cap A_{j'}=\emptyset\text{ if
        }j\not=j'}}\varepsilon(A_1,\cdots,A_m)
    \prod_{j=1}^m\varphi^j_{Q_j,n_j}((z_l)_{l\in A_j})
  \end{multline}
  where
  \begin{itemize}
  \item $\varepsilon(A_1,\cdots,A_m)$ is the signature of
    $\sigma(A_1,\cdots,A_m)$, the unique permutation of
    $\{1,\cdots,n\}$ such that, if $A_j=\{a_{ij};\ 1\leq i\leq Q_j\}$
    for $1\leq j\leq m$ then $\sigma(a_{ij})=Q_1+\cdots+Q_{j-1}+i$,
  \item and $c(Q)$ is such that $\|\wedge_j\varphi^j_{Q_j,n_j}\|=1$
    i.e.
    \begin{equation}
      \label{eq:154}
      c(Q)=\sqrt{\frac{\prod_{j=1}^m Q_j!}{n!}}.
    \end{equation}
  \end{itemize}
  Thus, by~\eqref{eq:OneParticleDensityMatrixDef}, one has
  \begin{equation}
    \label{eq:158}
    \frac{\gamma^{(1)}_{\substack{Q,\overline{n}\\Q',\overline{n}'}}(x,y)}{c(Q)c(Q')}
    =\sum_{\substack{|A_j|=Q_j,\
        \forall1\leq j\leq m\\A_1\cup\cdots\cup A_m=\{1,\cdots,n\}\\
        A_j\cap A_{j'}=\emptyset\text{ if }j\not=j'}}
    \sum_{\substack{|A'_j|=Q'_j,\ \forall1\leq j\leq
        m\\A'_1\cup\cdots\cup A'_m=\{1,\cdots,n\}\\ A'_j\cap
        A'_{j'}=\emptyset\text{ if }j\not=j'}}
    (-1)^{\varepsilon((A_j))+\varepsilon((A'_j))}I((A_j)_j,(A'_j)_j)
  \end{equation}
  where
  \begin{equation}
    \label{eq:196}
    \begin{split}
      I(A,A')&:=I((A_j)_j,(A'_j)_j)\\&=
      \int_{[0,L]^{n-1}}\left[\prod_{j=1}^m\varphi^j_{Q_j,n_j}((x_l)_{l\in
          A_j}) \overline{\varphi^j_{Q'_j,n'_j}((y_l)_{l\in A'_j})}\right]_{
        \substack{x_1=x\\y_1=y\\x_j=y_j\text{ if }j\geq2}}dx_2\cdots dx_n.
    \end{split}
  \end{equation}
  To evaluate this last integral, we note that, for any pair of
  partitions $(A_j)_j$ and $(A'_j)_j$ (as in the indices of the sum
  in~\eqref{eq:158}), if there exists $j\not=j'$ such that $A_j\cap
  A'_{j'}\cap\{2,\cdots,n\}\not=\emptyset$, then the integral
  $I(A,A')$ vanishes.\\
  Now, note that, if $d_1(Q,Q')>2$, then, for any pair of partitions
  $(A_j)_j$ and $(A'_j)_j$, there exists $j\not=j'$ such that $A_j\cap
  A'_{j'}\cap\{2,\cdots,n\}\not=\emptyset$; thus, the integral
  $I(A,A')$ above always vanishes and, summing this, one
  has
  \begin{equation*}
   \gamma^{(1)}_{\substack{Q,\overline{n}\\Q',\overline{n}'}}=0\quad\text{ if
   }\quad d_1(Q,Q')>2.
  \end{equation*}
  So we are left with the case $Q=Q'$ or $d_1(Q,Q')=2$. \\
  Assume first $Q=Q'$. Consider the sums in~\eqref{eq:158}. If $1\in
  A_{j_0}$ and $1\not\in A'_{j_0}$, then, as $\forall j$,
  $|A'_j|=|A_j|$, there exists $\alpha\in
  A'_{j_0}=A'_{j_0}\cap\{2,\cdots,n\}$ and $j\not=j_0$ such that
  $\alpha\in A_j$. That is, there exists $j\not=j'$ such that $A_j\cap
  A'_{j'}\cap\{2,\cdots,n\}\not=\emptyset$, thus, the integral
  $I(A,A')$ vanishes. Thus, we rewrite
  \begin{equation}
    \label{eq:153}
    \begin{split}
      \frac{\gamma^{(1)}_{Q,\overline{n},Q,\overline{n}'}(x,y)}{c^2(Q)}&=
      \sum_{\substack{j_0=1\\Q_{j_0}\geq1}}^m\sum_{\substack{1\in
          A_{j_0}\\ |A_j|=Q_j,\ \forall1\leq j\leq m
          \\A_1\cup\cdots\cup A_m=\{1,\cdots,n\}\\
          A_j\cap A_{j'}=\emptyset\text{ if }j\not=j'}}
      \sum_{\substack{1\in A'_{j_0}\\ |A'_{j_0}|=Q_{j_0}\\
          A'_{j}=A_j\text{ if }j\not=j_0}}
      (-1)^{\varepsilon((A_j))+\varepsilon((A'_j))}I(A,A')\\
      &= \sum_{\substack{j_0=1\\Q_{j_0}\geq1}}^m\sum_{\substack{1\in
          A_{j_0}\\ |A_j|=Q_j,\ \forall1\leq j\leq
          m\\A_1\cup\cdots\cup A_m=\{1,\cdots,n\}\\
          A_j\cap A_{j'}=\emptyset\text{ if }j\not=j'}} I(A)
    \end{split}
  \end{equation}
  where, using the support and orthonormality properties of the
  functions $(\varphi^j_{q,n})_{1\leq n}$, one computes
  \begin{equation*}
    \begin{split}
      I(A)&:=\left(\int_{\Delta_{j_0}^{Q_{j_0}-1}}
        \varphi^{j_0}_{Q_{j_0},n_{j_0}}(x,z)\,
        \overline{\varphi^{j_0}_{Q_{j_0},n'_{j_0}}(y,z)} dz\right)
      \,\prod_{\substack{j=1\\j\not=j_0}}^m\int_{\Delta_{j}^{Q_{j}}}\varphi^j_{Q_j,n_j}(z)
      \overline{\varphi^j_{Q_j,n'_j}(z)} dz
      \\&=\prod_{j\not=j_0}\delta_{n_j=n'_j}\cdot\left(
      \int_{\Delta_{j_0}^{Q_{j_0}-1}}\varphi^{j_0}_{Q_{j_0},n_{j_0}}(x,z)
      \overline{\varphi^{j_0}_{Q_{j_0},n'_{j_0}}(y,z)}dz\right).
    \end{split}
  \end{equation*}
  As 
  \begin{equation*}
    \#\{(A_j)_j;\ 1\in A_{j_0},\ \forall j,\ |A_j|=Q_j\}
    =\frac{(n-1)! Q_{j_0}}{\prod_{j=1}^m Q_j!}
  \end{equation*}
  by~\eqref{eq:154} and~\eqref{eq:153}, one computes
  \begin{equation*}
    \gamma^{(1)}_{Q,\overline{n},Q,\overline{n}'}(x,y)
    =\sum_{\substack{j=1\\Q_j\geq1}}^m\frac{Q_j}n \int_{\Delta^{Q_j-1}_j}
    \varphi^j_{Q_j,n_j}(x,z)\,\overline{\varphi^j_{Q_j,n'_j}(y,z)}dz
    \prod_{j\not=j_0}\delta_{n_j=n'_j}
    =\frac1n\sum_{\substack{j=1\\Q_j\geq1}}^m\gamma^{(1)}_{\substack{Q_j\\n_j,n'_j}}(x,y).
  \end{equation*}
  We now assume that $d_1(Q,Q')=2$. Thus, there exist $1\leq i_0\not= j_0\leq
  m$ such that $Q_{j_0}\geq1$, $Q'_{i_0}=Q_{i_0}+1$,
  $Q_{j_0}=Q'_{j_0}+1$ and $Q_k=Q'_k$ for $k\not\in\{i_0,j_0\}$.\\
  Consider the sums in~\eqref{eq:158}. If $1\not\in A_{j_0}$ (or
  $1\not\in A'_{i_0}$), then as $|A'_{j_0}|=Q'_{j_0}=Q_{j_0}-1$, there
  exists $\alpha\in A_{j_0}=A_{j_0}\cap\{2,\cdots,n\}$ and $i\not=j_0$
  such that $\alpha\in A'_i$. That is, there exists $j\not=j'$ such
  that $A_j\cap A'_{j'}\cap\{2,\cdots,n\}\not=\emptyset$, thus, the
  integral $I(A,A')$ vanishes. The reasoning is the same if $1\not\in
  A'_{i_0}$. Moreover, if $1\in A_{j_0}$ and $1\in A'_{i_0}$, then, as
  in the derivation of~\eqref{eq:153}, we see that $I(A,A')=0$ except
  if $A_j=A'_j$ for all $j\not\in\{i_0,j_0\}$. Therefore, if
  $d_1(Q,Q')=2$, we rewrite
  \begin{equation}
    \label{eq:157}
    \begin{split}
      \frac{\gamma^{(1)}_{\substack{Q,\overline{n}\\Q',\overline{n}'}}(x,y)}{c^2(Q)}&=\sum_{\substack{j_0,i_0=1\\i_0\not=j_0\\Q_{j_0}\geq1}}^m
      \sum_{\substack{1\in A_{j_0}\\ |A_l|=Q_l,\
          \forall1\leq j\leq m\\A_1\cup\cdots\cup A_m=\{1,\cdots,n\}\\
          A_j\cap A_{j'}=\emptyset\text{ if }j\not=j'}}
      \sum_{\substack{A'_{i_0}=\{1\}\cup A_{i_0}\\
          A'_{j_0}=A_{j_0}\setminus\{1\}\\
          A'_{j}=A_j\text{ if }j\not\in\{i_0,j_0\}}}
      (-1)^{\varepsilon((A_j))+\varepsilon((A'_j))}I(A,A').
    \end{split}
  \end{equation}
  For such $(A_j)_j$ and $(A'_j)_j$, one has
  $(-1)^{\varepsilon((A_j))+\varepsilon((A'_j))}=1$ and we compute
  \begin{equation}
    \label{eq:159}
    \begin{split}
      I(A,A')&=\int_{\Delta^{Q_{j_0}-1}_{j_0}}
      \varphi^{j_0}_{Q_{j_0},n_{j_0}}(x,z)
      \overline{\varphi^{j_0}_{Q_{j_0}-1,n'_{j_0}}(z)}dz\\&\hskip3cm
      \int_{\Delta^{Q_{i_0}}_{i_0}}\varphi^{i_0}_{Q_{i_0},n_{i_0}}(z)
      \overline{\varphi^{i_0}_{Q_{i_0}+1,n'_{i_0}}(y,z)}dz
      \prod_{j\not\in\{i_0,j_0\}}\delta_{n_j=n'_j} 
    \end{split}
  \end{equation}
  with the convention described in Remark~\ref{rem:8}.\\
  The number of partitions coming up in~\eqref{eq:157} is given by
  \begin{equation*}
    \sum_{\substack{1\in A_{j_0}\\ |A_l|=Q_l,\
        \forall1\leq j\leq m\\A_1\cup\cdots\cup A_m=\{1,\cdots,n\}\\
        A_j\cap A_{j'}=\emptyset\text{ if }j\not=j'}}
    \sum_{\substack{A'_{i_0}=\{1\}\cup A_{i_0}\\
        A'_{j_0}=A_{j_0}\setminus\{1\}\\
        A'_{j}=A_j\text{ if }j\not\in\{i_0,j_0\}}}1
    =\frac{(n-Q_{j_0}-Q_{i_0}-1)!Q_{i_0}!Q_{j_0}!}{Q_1!\cdots Q_m!}.
  \end{equation*}
  Plugging this and~\eqref{eq:159} into~\eqref{eq:157}, we
  obtain~\eqref{eq:16}. This completes the proof of
  Theorem~\ref{thr:4}. 
\end{proof}
\subsubsection{The two-particle density matrix}
\label{sec:two-particle-density}
We shall now compute the $2$ particles reduced density matrix in terms
of the coefficients $(a^Q_{\overline{n}})_{Q,\overline{n}}$ coming up
in the occupation number decomposition~\eqref{eq:146}. We prove
\begin{Th}
  \label{thr:5}
  The $2$-particle density $\gamma^{(2)}_\Psi$
  (see~\eqref{eq:OneParticleDensityMatrixDef}) is written as
  \begin{equation}
    \label{eq:186}
    \gamma_\Psi^{(2)}=\gamma_\Psi^{(2),d,d}+\gamma_\Psi^{(2),d,o}
    +\gamma_\Psi^{(2),2}+\gamma_\Psi^{(2),4,2}+\gamma_\Psi^{(2),4,3}
    +\gamma_\Psi^{(2),4,3'} +\gamma_\Psi^{(2),4,4}  
  \end{equation}
  where
  \begin{gather}
    \label{eq:3}
    \gamma_\Psi^{(2),d,d}=\sum_{j=1}^m\sum_{\substack{Q\text{ occ.}\\
        Q_j\geq2}} \sum_{\substack{n_j\geq1\\ n'_j\geq1}}\sum_{\tilde
      n\in\N^{m-1}}a^Q_{\tilde n_j} \overline{a^Q_{\tilde
        n'_j}}\gamma^{(2),d,d}_{\substack{Q_j\\n_j,n'_j}}
    \\
    \label{eq:182}
    \gamma_\Psi^{(2),d,o}=\sum_{1\leq i<j\leq m}\sum_{\substack{Q\text{ occ.}\\
        Q_i\geq1\\Q_j\geq1}} \sum_{\tilde n\in\N^{m-2}}
    \sum_{\substack{n_j,n'_j\geq 1\\n_i,n'_i\geq 1}}a^Q_{\tilde
      n_{i,j}} \overline{a^Q_{\tilde
        n'_{i,j}}}\gamma^{(2),d,o}_{\substack{Q_i,Q_j\\n_i,n_j\\n'_i,n'_j}}
    \\
    \label{eq:183} \gamma_\Psi^{(2),2}=\sum_{\substack{i,j=1\\i
        \not=j}}^m\sum_{\substack{Q,\text{ occ. }Q_j\geq1\\ Q':\
        Q'_k=Q_k\text{ if }k\not\in\{i,j\}\\Q'_i=Q_i+1\\Q'_j=Q_j-1}}
    \sum_{\tilde n\in\N^{m-1}}C_2(Q,i,j) \sum_{\substack{n_j,n'_j\geq
        1\\n_i,n'_i\geq 1}} a^Q_{\tilde n_{i,j}} \overline{a^{Q'}_{\tilde
        n'_{i,j}}} \gamma^{(2),2}_{\substack{Q_i,Q_j\\n_i,n_j\\n'_i,n'_j}}
    \\ \label{eq:198} \gamma_\Psi^{(2),4,2}= \sum_{i\not=j}\sum_{\tilde
      n\in\N^{m-2}} \sum_{\substack{Q\text{ occ.}\\
        Q_j\geq2\\ Q':\ Q'_k=Q_k\text{ if
        }k\not\in\{i,j\}\\Q'_i=Q_i+2\\Q'_j=Q_j-2}}C_2(Q,i,j)
    \sum_{\substack{n_j,n'_j\geq 1\\n_i,n'_i\geq 1}}
    a^Q_{\tilde{n}_{i,j}}\overline{a^{Q'}_{\tilde{n}_{i,j}'}}
    \gamma^{(2),4,2}_{\substack{Q_i,Q_j\\n_i,n_j\\n'_i,n'_j}}
  \end{gather}
  \begin{gather}
    \label{eq:200} \gamma_\Psi^{(2),4,3}=\sum_{\substack{i,j,k\\\text{
          distinct}}}\sum_{\tilde n\in\N^{m-3}} \sum_{\substack{Q\text{
          occ.}\\ Q_j\geq2\\Q':\ Q'_l=Q_l\text{ if }l\not\in\{i,j,k\}\\Q'_i=Q_i+1\\
        Q'_j=Q_j-2\\Q'_k=Q_k+1}} C_3(Q,i,j,k) \sum_{\substack{n_i,n_j,n_k\geq
        1\\n'_i,n'_j,n'_k\geq 1}}
    a^Q_{\tilde{n}_{i,j,k}}\overline{a^{Q'}_{\tilde{n}_{i,j,k}'}}
    \gamma^{(2),4,3}_{\substack{Q_i,Q_j,Q_k\\n_i,n_j,n_k\\n'_i,n'_j,n'_k}}
    \\\label{eq:202} \gamma_\Psi^{(2),4,3'}=
    \sum_{\substack{i,j,k\\\text{distinct}}}\sum_{\tilde n\in\N^{m-3}}
    \sum_{\substack{Q\text{ occ.}\\ Q_i\geq1,\ Q_k\geq1\\ Q':\ Q'_l=Q_l\text{
          if }l\not\in\{i,j,k\}\\Q'_i=Q_i-1\\ Q'_j=Q_j+2\\Q'_k=Q_k-1}}
    C_3(Q,i,j,k)\sum_{\substack{n_i,n_j,n_k\geq 1\\n'_i,n'_j,n'_k\geq 1}}
    a^Q_{\tilde{n}_{i,j,k}}\overline{a^{Q'}_{\tilde{n}_{i,j,k}'}}
    \gamma^{(2),4,3'}_{\substack{Q_i,Q_j,Q_k\\n_i,n_j,n_k\\n'_i,n'_j,n'_k}},
    \\\intertext{and}\label{eq:203} \gamma_\Psi^{(2),4,4}=
    \sum_{\substack{i,j,k,l\\\text{distinct}}}\sum_{\tilde
      n\in\N^{m-4}} \sum_{\substack{Q\text{ occ.}\\ Q_i\geq1,\
        Q_j\geq1\\ Q':\ Q'_l=Q_l\text{ if
        }l\not\in\{i,j,k,l\}\\Q'_i=Q_i-1,\ Q'_j=Q_j-1\\Q'_k=Q_k+1,\
        Q'_l=Q_l+1}} C_4(Q,i,j,k,l)\sum_{\substack{n_i,n_j,n_k,n_l\geq
        1\\n'_i,n'_j,n'_k,n'_l\geq 1}}
    a^Q_{\tilde{n}_{i,j,k}}\overline{a^{Q'}_{\tilde{n}_{i,j,k}'}}
    \gamma^{(2),4,4}_{\substack{Q_i,Q_j,Q_k,Q_l\\n_i,n_j,n_k,n_l\\n'_i,n'_j,n'_k,n'l}},
  \end{gather}
  where
  \begin{itemize}
  \item we have used the shorthands defined in Theorem~\ref{thr:4} and
    defined
    \begin{itemize}
    \item $\tilde n_{i,j,k}$ for $(\tilde n_1,\cdots,\tilde
      n_{i-1},n_i,\tilde n_i,\cdots,\tilde n_{j-2},n_j,\tilde
      n_{j-1},\cdots,\tilde n_{k-3},n_j,\tilde n_{k-2},\cdots,\tilde
      n_{m-3})$\\when $i<j<k$ and $\tilde n=(\tilde n_1,\cdots,\tilde
      n_{m-3})$,
    \item $\tilde n_{i,j,k,l}$ for $(\tilde n_1,\cdots,\tilde
      n_{i-1},n_i,\tilde n_i,\cdots,\tilde n_{j-2},n_j,\tilde
      n_{j-1},\cdots,\tilde n_{k-3},n_k,\tilde n_{k-2},\cdots,$\\
      $\cdots,\tilde n_{l-4},n_l,\tilde n_{l-3},\cdots,\tilde n_{m-4})$
      when $i<j<k<l$ and $\tilde n=(\tilde n_1,\cdots,\tilde
      n_{m-4})$,
    \end{itemize}
  \item the trace class operator $\gamma^{(2),d,d}_{\substack{Q_j\\n_j,n'_j}}:\
    L^2(\Delta_j)\bigwedge L^2(\Delta_j)\to L^2(\Delta_j)\bigwedge
    L^2(\Delta_j)$ has the kernel
    \begin{equation}
      \label{eq:184}
      \gamma^{(2),d,d}_{\substack{Q_j\\n_j,n'_j}}(x,x',y,y')=\frac{Q_j(Q_j-1)}2
      \int_{\Delta^{Q_{j}-2}_{j}}\varphi^{j}_{Q_{j},n_{j}}(x,x',z)
      \overline{\varphi^{j}_{Q_{j},n'_{j}}(y,y',z)}dz
    \end{equation}
  \item the trace class operator
    $\gamma^{(2),d,o}_{\substack{Q_i,Q_j\\n_i,n_j\\n'_i,n'_j}}:\
    L^2(\Delta_i)\otimes L^2(\Delta_j)\to L^2(\Delta_i)\otimes
    L^2(\Delta_j)$ has the kernel
    \begin{equation}
      \label{eq:185}
      \begin{split}
        \gamma^{(2),d,o}_{\substack{Q_i,Q_j\\n_i,n_j\\n'_i,n'_j}}
        (x,x',y,y')&=\frac{Q_iQ_j}2\times
        \int_{\Delta^{Q_{i}-1}_{i}\times\Delta^{Q_{j}-1}_{j}}dzdz'
        \\&\hskip1cm \left|\begin{matrix}
            \varphi^{i}_{Q_{i},n_{i}}(x,z)
            &\varphi^{i}_{Q_{i},n_{i}}(x,z') \\
            \varphi^{j}_{Q_{j},n_{j}}(x',z)&
            \varphi^{j}_{Q_{j},n_{j}}(x',z')
          \end{matrix}\right|
        \cdot\overline{\left|\begin{matrix}
              \varphi^{i}_{Q_{i},n'_{i}}(y,z)&
              \varphi^{i}_{Q_{i},n'_{i}}(y',z)\\
              \varphi^{j}_{Q_{j},n'_{j}}(y,z')&
              \varphi^{j}_{Q_{j},n'_{j}}(y',z')
            \end{matrix}\right|}
      \end{split}
    \end{equation}
  \item $\D C_2(Q,i,j)=\frac{(n-Q_j-Q_i-2)!Q_i!Q_j!}{2\,(n-2)!}$;
  \item the trace-class operator $\gamma^{(2),2}_{\substack{Q_i,Q_j\\n_i,n_j\\n'_i,n'_j}}:\
    L^2(\Delta_j)\bigwedge L^2(\Delta_j)\to L^2(\Delta_i)\bigwedge
    L^2(\Delta_i)$ has the kernel
    \begin{equation}
      \label{eq:76}
      \begin{split}
        \gamma^{(2),2}_{\substack{Q_i,Q_j\\n_i,n_j\\n'_i,n'_j}}(x,y) &=
        \car_{Q_j\geq2}\int_{\Delta^{Q_j-2}_j\times\Delta^{Q_i}_i}
        \varphi^j_{Q_j,n_j}(x,x',z)\varphi^i_{Q_i,n_i}(z') \overline{\left|
              \begin{matrix}
                \varphi^j_{Q_j-1,n'_j}(y',z)&\varphi^j_{Q_j-1,n'_j}(y,z)
                \\\varphi^i_{Q_i+1,n'_i}(y',z')&\varphi^i_{Q_i+1,n'_i}(y,z') 
              \end{matrix}
            \right|} dzdz'
          \\&+\car_{Q_i\geq1}\int_{\Delta^{Q_j-1}_j\times\Delta^{Q_i-1}_i}
          \left|\begin{matrix}
              \varphi^j_{Q_j,n_j}(x',z)&\varphi^j_{Q_j,n_j}(x,z)
              \\\varphi^i_{Q_i,n_i}(x',z')&\varphi^i_{Q_i,n_i}(x,z')
              \end{matrix}
            \right|
            \overline{\varphi^j_{Q_j-1,n'_j}(z)\varphi^i_{Q_i+1,n'_i}(y,y',z')}
        dzdz',
      \end{split}
    \end{equation}
  \item the rank one operator $\gamma^{(2),4,2}_{\substack{Q_i,Q_j\\n_i,n_j\\n'_i,n'_j}}:\
    L^2(\Delta_j)\bigwedge L^2(\Delta_j)\to L^2(\Delta_i)\bigwedge
    L^2(\Delta_i)$ has the kernel
    \begin{equation}
      \label{eq:194}
     \begin{split}
        \gamma^{(2),4,2}_{\substack{Q_i,Q_j\\n_i,n_j\\n'_i,n'_j}}(x,x',y,y')
        &=\int_{\Delta^{Q_{j}-2}_{j}}
        \varphi^{j}_{Q_{j},n_{j}}(x,x',z)
        \overline{\varphi^{j}_{Q_{j}-2,n'_{j}}(z)}dz%\\&\hskip3cm
        \int_{\Delta^{Q_{i}}_{i}}\varphi^{i}_{Q_{i},n_{i}}(z)
        \overline{\varphi^{i}_{Q_{i}+2,n'_{i}}(y,y',z)}dz.
      \end{split}
    \end{equation}
  \item the rank 2 operator
    $\gamma^{(2),4,3}_{\substack{Q_i,Q_j,Q_k\\n_i,n_j,n_k\\n'_i,n'_j,n'_k}}:\
    L^2(\Delta_i)\otimes L^2(\Delta_k)\to L^2(\Delta_j)\bigwedge
    L^2(\Delta_j)$ has the kernel
    \begin{equation}
      \label{eq:199}
      \begin{split}
        \gamma^{(2),4,3}_{\substack{Q_i,Q_j,Q_k\\n_i,n_j,n_k\\n'_i,n'_j,n'_k}}&(x,x',y,y')=
        \int_{\Delta^{Q_{j}-2}_{j}} \varphi^{j}_{Q_{j},n_{j}}(x,x',z)
        \overline{\varphi^{j}_{Q_{j}-2,n'_{j}}(z)}dz\\&\times
        \left|
          \begin{matrix}
            \int_{\Delta^{Q_{i}}_{i}}\varphi^{i}_{Q_{i},n_{i}}(z)
            \overline{\varphi^{i}_{Q_{i}+1,n'_{i}}(y,z)}dz&
            \int_{\Delta^{Q_{i}}_{i}}\varphi^{i}_{Q_{i},n_{i}}(z)
            \overline{\varphi^{i}_{Q_{i}+1,n'_{i}}(y',z)}dz \\
            \int_{\Delta^{Q_{k}}_{k}}\varphi^{k}_{Q_{k},n_{k}}(z)
            \overline{\varphi^{k}_{Q_{k}+1,n'_{k}}(y,z)}dz &
            \int_{\Delta^{Q_{k}}_{k}}\varphi^{k}_{Q_{k},n_{k}}(z)
            \overline{\varphi^{k}_{Q_{k}+1,n'_{k}}(y',z)}dz
          \end{matrix}
        \right|,
      \end{split}
    \end{equation}
  \item $\D C_3(Q,i,j,k)=\frac{(n-Q_i-Q_j-Q_k-2)!Q_i!Q_j!Q_k!}{2\,(n-2)!}$;
  \item the rank 2 operator
    $\gamma^{(2),4,3'}_{\substack{Q_i,Q_j,Q_k\\n_i,n_j,n_k\\n'_i,n'_j,n'_k}}:\
    L^2(\Delta_j)\bigwedge L^2(\Delta_j)\to L^2(\Delta_i)\otimes
    L^2(\Delta_k)$ has the kernel
    \begin{equation}
      \label{eq:53}
      \begin{split}
        \gamma^{(2),4,3'}_{\substack{Q_i,Q_j,Q_k\\n_i,n_j,n_k\\n'_i,n'_j,n'_k}}(x,x',y,y')&=
        \left|
          \begin{matrix}
            \int_{\Delta^{Q_{i}-1}_{i}}\varphi^{i}_{Q_{i},n_{i}}(x,z)
            \overline{\varphi^{i}_{Q_{i}-1,n'_{i}}(z)}dz&
            \int_{\Delta^{Q_{i}-1}_{i}}\varphi^{i}_{Q_{i},n_{i}}(x',z)
            \overline{\varphi^{i}_{Q_{i}-1,n'_{i}}(z)}dz\\
            \int_{\Delta^{Q_{k}-1}_{k}}\varphi^{k}_{Q_{k},n_{k}}(x,z)
            \overline{\varphi^{k}_{Q_{k}-1,n'_{k}}(z)}dz &
            \int_{\Delta^{Q_{k}-1}_{k}}\varphi^{k}_{Q_{k},n_{k}}(x',z)
            \overline{\varphi^{k}_{Q_{k}-1,n'_{k}}(z)}dz
          \end{matrix}
        \right|\\&\times
        \int_{\Delta^{Q_{j}}_{j}} \varphi^{j}_{Q_{j},n_{j}}(z)
        \overline{\varphi^{j}_{Q_{j}+2,n'_{j}}(y,y',z)}dz,
      \end{split}
    \end{equation}
  \item the rank 4 operator
    $\gamma^{(2),4,4}_{\substack{Q_i,Q_j,Q_k,Q_l\\n_i,n_j,n_k,n_l
        \\n'_i,n'_j,n'_k,n'_l}}:\ L^2(\Delta_k)\otimes
    L^2(\Delta_l)\to L^2(\Delta_i)\otimes L^2(\Delta_j)$ has the
    kernel
    \begin{multline}
      \label{eq:166}
      \gamma^{(2),4,4}_{\substack{Q_i,Q_j,Q_k,Q_l\\n_i,n_j,n_k,n_l\\n'_i,n'_j,n'_k,n'_l}}(x,x',y,y')=
      \left|\begin{matrix}\int_{\Delta^{Q_{i}-1}_{i}}
          \varphi^{i}_{Q_{i},n_{i}}(x,z)
          \overline{\varphi^{i}_{Q_{i}-1,n'_{i}}(z)}dz&
          \int_{\Delta^{Q_{i}-1}_{i}}\varphi^{i}_{Q_{i},n_{i}}(x',z)
          \overline{\varphi^{i}_{Q_{i}-1,n'_{i}}(z)}dz\\
          \int_{\Delta^{Q_{j}-1}_{j}} \varphi^{j}_{Q_{j},n_{j}}(x,z)
          \overline{\varphi^{j}_{Q_{j}-1,n'_{j}}(z)}dz &
          \int_{\Delta^{Q_{j}-1}_{j}} \varphi^{j}_{Q_{j},n_{j}}(x',z)
          \overline{\varphi^{j}_{Q_{j}-1,n'_{j}}(z)}dz          
        \end{matrix}\right|
      \\\times \left|\begin{matrix}\int_{\Delta^{Q_{k}}_{k}}
          \varphi^{k}_{Q_{k},n_{k}}(z)
          \overline{\varphi^{k}_{Q_{k}+1,n'_{k}}(y,z)}dz&
          \int_{\Delta^{Q_{k}}_{k}} \varphi^{k}_{Q_{k},n_{k}}(z)
          \overline{\varphi^{k}_{Q_{k}+1,n'_{k}}(y',z)}dz
          \\\int_{\Delta^{Q_{l}}_{l}} \varphi^{l}_{Q_{l},n_{l}}(z)
          \overline{\varphi^{l}_{Q_{l}+1,n'_{l}}(y,z)}dz &
          \int_{\Delta^{Q_{l}}_{l}} \varphi^{l}_{Q_{l},n_{l}}(z)
          \overline{\varphi^{l}_{Q_{l}+1,n'_{l}}(y',z)}dz
        \end{matrix}\right|
    \end{multline}
  \item $\D
    C_4(Q,i,j,k,l)=\frac{(n-Q_i-Q_j-Q_k-Q_l-2)!Q_i!Q_j!Q_k!Q_l!}{2\,(n-2)!}$;
  \end{itemize}
\end{Th}
\begin{Rem}
  \label{rem:2}
  In~\eqref{eq:184}~-~\eqref{eq:166}, in accordance with
  Remark~\ref{rem:9}, in the degenerate cases, we use the conventions
  derived from those in Remark~\ref{rem:8} in a obvious way.\\
  For example, in~\eqref{eq:185}, if $Q_i=Q_j=1$, one has
  \begin{equation}
    \label{eq:220}
    \begin{split}
      \gamma^{(2),d,o}_{\substack{Q_i,Q_j\\n_i,n_j\\n'_i,n'_j}}
      (x,&x',y,y')=\frac{Q_iQ_j}2\left|\begin{matrix}
            \varphi^{i}_{Q_{i},n_{i}}(x)&
            \varphi^{i}_{Q_{i},n_{i}}(x')\\
            \varphi^{j}_{Q_{j},n_{j}}(x)&
            \varphi^{j}_{Q_{j},n_{j}}(x')
          \end{matrix}\right| \cdot\overline{
          \left|\begin{matrix} \varphi^{i}_{Q_{i},n'_{i}}(y)&
          \varphi^{i}_{Q_{i},n'_{i}}(y')\\
          \varphi^{j}_{Q_{j},n'_{j}}(y)&
          \varphi^{j}_{Q_{j},n'_{j}}(y')
          \end{matrix}\right|}.
    \end{split}
  \end{equation}
\end{Rem}
\begin{proof}[Proof of Theorem~\ref{thr:5}]
  Theorem~\ref{thr:5} follows from a direct computation that we now
  perform. First, by the bilinearity of
  formula~\eqref{eq:TwoParticleDensityMatrixDef}, one has
  \begin{equation}
    \label{eq:132}
    \gamma_\Psi^{(2)}=\frac{n(n-1)}2\sum_{\substack{Q\text{ occ.}\\
        \overline{n}\in\N^m}}\sum_{\substack{Q'\text{ occ.}\\
        \overline{n}'\in\N^m}}
    a^Q_{\overline{n}}\overline{a^{Q'}_{\overline{n}'}}
    \gamma^{(2)}_{\substack{Q,\overline{n}\\Q',\overline{n}'}}
  \end{equation}
  where the trace class operator
  $\gamma^{(2)}_{\substack{Q,\overline{n}\\Q',\overline{n}'}}$ acts on
  $L^2([0,L])\bigwedge L^2([0,L])$ and has the kernel
  \begin{multline}
    \label{eq:136}
    \gamma^{(2)}_{\substack{Q,\overline{n}\\Q',\overline{n}'}}(x,x',y,y')
    :=\int_{[0,L]^{n-2}}
    \left[\bigwedge_{j=1}^{m}\varphi^j_{Q_j,n_j}\right](x,x',z_3,\cdots,z_n)\\
    \overline{\left[\bigwedge_{j=1}^{m}\varphi^j_{Q'_j,n'_j}\right]
      (y,y',z_3,\cdots,z_n)}dz_3\cdots dz_n.
  \end{multline}
  By~\eqref{eq:152}, one has
  \begin{equation}
    \label{eq:141}
    \frac{\gamma^{(2)}_{\substack{Q,\overline{n}\\Q',\overline{n}'}}(x,x',y,y')}{c(Q)c(Q')}
    =\sum_{\substack{|A_j|=Q_j,\
        \forall1\leq j\leq m\\A_1\cup\cdots\cup A_m=\{1,\cdots,n\}\\
        A_j\cap A_{j'}=\emptyset\text{ if }j\not=j'}}
    \sum_{\substack{|A'_j|=Q'_j,\ \forall1\leq j\leq m\\A'_1\cup\cdots\cup
        A'_m=\{1,\cdots,n\}\\ A'_j\cap A'_{j'}=\emptyset\text{ if }j\not=j'}}
    (-1)^{\varepsilon((A_j))+\varepsilon((A'_j))}I(A,A')
  \end{equation}
  where
  \begin{equation}
    I(A,A'):=
    \int_{[0,L]^{n-2}}\left[\prod_{j=1}^m\varphi^j_{Q_j,n_j}((z_l)_{l\in
        A_j}) \overline{\varphi^j_{Q'_j,n'_j}((y_l)_{l\in
          A'_j})}\right]_{ \substack{x_1=x,\ x_2=x'\\y_1=y,\
        y_2=y'\\x_j=y_j\text{ if }j\geq3}}dx_3\cdots dx_n
  \end{equation}
  To evaluate this last integral, we note that, for any pair of
  partitions $(A_j)_j$ and $(A'_j)_j$ (as in the indices of the above
  sum), if there exists $j\not=j'$ such that $A_j\cap
  A'_{j'}\cap\{3,\cdots,n\}\not=\emptyset$, then the integral
  $I(A,A')$ vanishes.\\
  Now, note that, if $d_1(Q,Q')>4$:, then, for any pair of partitions
  $(A_j)_j$ and $(A'_j)_j$, there exists $j\not=j'$ such that $A_j\cap
  A'_{j'}\cap\{3,\cdots,n\}\not=\emptyset$; thus, the integral
  $I(A,A')$ above always vanishes and, summing this, one
  has
  \begin{equation*}
    \gamma^{(2)}_{\substack{Q,\overline{n}\\Q',\overline{n}'}}=0\quad\text{ if
    }\quad d_1(Q,Q')>4.
  \end{equation*}
  So we are left with the cases $Q=Q'$, $d_1(Q,Q')=2$ or $d_1(Q,Q')=4$. \\
  Assume first $Q=Q'$. Consider the sums in~\eqref{eq:158}. If
  $\{1,2\}\subset A_{i_0}\cup A_{j_0}$ and $\{1,2\}\not\subset
  A'_{i_0}\cup A'_{j_0}$ then, as $\forall j$, $|A'_j|=|A_j|$, there
  exists $\alpha\in (A'_{i_0}\cup A'_{j_0})\cap\{3,\cdots,n\}$ and
  $j\not\in\{i_0,j_0\}$ such that $\alpha\in A_j$. That is, there
  exists $j\not=j'$ such that $A_j\cap
  A'_{j'}\cap\{3,\cdots,n\}\not=\emptyset$, thus, the integral
  $I(A,A')$ vanishes. Moreover, if $\{1,2\}\subset
  A_{j_0}$ and $\{1,2\}\not\subset A'_{j_0}$ then, there exists
  $\alpha\in A'_{j_0}\cap\{3,\cdots,n\}$ and $j\not=j_0$ such that
  $\alpha\in A_j$, thus, the integral $I(A,A')$ vanishes.
  Thus, we rewrite
  \begin{equation}
    \label{eq:106}
    \begin{split}
      \frac{\gamma^{(2)}_{Q,\overline{n},Q,\overline{n}'}(x,x',y,y')}{c^2(Q)}&=
      \sum_{i_0,j_0=1}^m\sum_{\substack{\{1,2\}\subset A_{i_0}\cup A_{j_0}\\
          |A_l|=Q_l,\
          \forall1\leq j\leq m\\A_1\cup\cdots\cup A_m=\{1,\cdots,n\}\\
          A_j\cap A_{j'}=\emptyset\text{ if }j\not=j'}}
      \sum_{\substack{\{1,2\}\subset A'_{i_0}\cup A'_{j_0}\\\
          |A'_{i_0}|=Q_{i_0}\\ |A'_{j_0}|=Q_{j_0}\\
          A'_{j}=A_j\text{ if }j\not\in\{i_0,j_0\}}}
      (-1)^{\varepsilon((A_j))+\varepsilon((A'_j))}I(A,A')\\
      &=\sum_{\substack{j_0=1\\Q_{j_0}\geq2}}^m\sum_{\substack{\{1,2\}\subset A_{j_0}\\
          |A_j|=Q_j,\ \forall1\leq j\leq m\\A_1\cup\cdots\cup A_m=\{1,\cdots,n\}\\
          A_j\cap A_{j'}=\emptyset\text{ if }j\not=j'}} I(A)+
      \sum_{\substack{i_0\not=j_0\\Q_{i_0}\geq1\\Q_{j_0}\geq1}}
      \sum_{\substack{1\in A_{i_0},\ 2\in A_{j_0}\\
          |A_j|=Q_j,\ \forall1\leq j\leq m\\A_1\cup\cdots\cup A_m=\{1,\cdots,n\}\\
          A_j\cap A_{j'}=\emptyset\text{ if }j\not=j'}} J(A)
    \end{split}
  \end{equation}
  where
  \begin{gather*}
    I(A):= \prod_{j\not=j_0}\delta_{n_j=n'_j}
    \int_{\Delta^{Q_{j_0}-2}_{j_0}}\varphi^{j_0}_{Q_{j_0},n_{j_0}}(x,x',z)
    \overline{\varphi^{j_0}_{Q_{j_0},n'_{j_0}}(y,y',z)}dz\\
    \intertext{and}
    \begin{split}
      J(A)&:= \prod_{j\not\in\{i_0,j_0\}}\delta_{n_j=n'_j}\left(
        \int_{\Delta^{Q_{i_0}-1}_{i_0}}\varphi^{i_0}_{Q_{i_0},n_{i_0}}(x,z)
        \overline{\varphi^{i_0}_{Q_{i_0},n'_{i_0}}(y,z)}dz\right.\\&\hskip7cm
      \left.\cdot\int_{\Delta^{Q_{j_0}-1}_{j_0}}\varphi^{j_0}_{Q_{j_0},n_{j_0}}(x',z')
        \overline{\varphi^{j_0}_{Q_{j_0},n'_{j_0}}(y',z')}dz'
      \right.\\&\hskip3.5cm -\left.
        \int_{\Delta^{Q_{i_0}-1}_{i_0}}\varphi^{i_0}_{Q_{i_0},n_{i_0}}(x,z)
        \overline{\varphi^{i_0}_{Q_{i_0},n'_{i_0}}(y',z)}dz\right.\\&\hskip7cm
      \left.\cdot\int_{\Delta^{Q_{j_0}-1}_{j_0}}\varphi^{j_0}_{Q_{j_0},n_{j_0}}(x',z')
        \overline{\varphi^{j_0}_{Q_{j_0},n'_{j_0}}(y,z')}dz'
      \right).
    \end{split}
  \end{gather*}
  As 
  \begin{gather*}
    \#\{(A_j)_j;\ \{1,2\}\subset A_{j_0},\ \forall j,\ |A_j|=Q_j\}
    =\frac{(n-2)! Q_{j_0}(Q_{j_0}-1)}{\prod_{j=1}^m Q_j!}\\\intertext{and}
    \#\{(A_j)_j;\ 1\in A_{i_0},\ 2\in A_{j_0},\ \forall j,\
    |A_j|=Q_j\} =\frac{(n-2)! Q_{i_0} Q_{j_0}}{\prod_{j=1}^m Q_j!}
    \quad\text{if}\quad i_0\not=j_0
  \end{gather*}
  by~\eqref{eq:154} and~\eqref{eq:106}, one obtains
  \begin{equation}
    \label{eq:187}
    \sum_{\substack{Q\text{
          occ.}\\\overline{n}\in\N^m\\\overline{n}'\in\N^m}}
    a^Q_{\overline{n}}\overline{a^{Q}_{\overline{n}'}}
    \gamma^{(2)}_{Q,\overline{n},Q,\overline{n}'}=\gamma_\Psi^{(2),d,d}
    +\gamma_\Psi^{(2),d,o} 
  \end{equation}
  where $\gamma_\Psi^{(2),d,d}$ and $\gamma_\Psi^{(2),d,o}$ are
  defined in Theorem~\ref{thr:5}.\\
  Let us now assume $d_1(Q,Q')=2$. Thus, there exists $1\leq i_0\not=
  j_0\leq m$ such that $Q_{j_0}\geq1$, $Q'_{i_0}=Q_{i_0}+1$,
  $Q_{j_0}=Q'_{j_0}+1$ and $Q_k=Q'_k$ for $k\not\in\{i_0,j_0\}$.\\
  Consider now the sums in~\eqref{eq:141}. If $\{1,2\}\cap
  A_{j_0}=\emptyset$, then as $|A'_{j_0}|=Q'_{j_0}=Q_{j_0}-1$, there
  exists $\alpha\in A_{j_0}=A_{j_0}\cap\{3,\cdots,n\}$ and $i\not=j_0$
  such that $\alpha\in A'_i$. Thus, the integral $I(A,A')$
  vanishes. If $A_{j_0}=\{1\}\cup B$ (resp. $A_{j_0}=\{2\}\cup B$)
  with $B\subset\{3,\cdots,n\}$, either $A'_{j_0}=B$ (and
  $\{1,2\}\subset A'_{i_0}$) or the integral $I(A,A')$
  vanishes. Finally, if $A_{j_0}=\{1,2\}\cup B$ with
  $B\subset\{3,\cdots,n\}$, then, $A'_{j_0}=\{1\}\cup B$ or
  $A'_{j_0}=\{2\}\cup B$ or $I(A,A')=0$. The
  same holds true for $A_{j_0}$ replaced with $A'_{i_0}$.\\
  Therefore, using the definition of $\varepsilon((A_j))$, if
  $d_1(Q,Q')=2$, we rewrite
  \begin{equation}
    \label{eq:188}
    \frac{\gamma^{(2)}_{\substack{Q,\overline{n}\\Q',\overline{n}'}}(x,y)}{c^2(Q)}=
    \sum_{\substack{i_0\not=j_0\\Q_{j_0}\geq2}}
    \Sigma_1(i_0,j_0)-\Sigma_2(i_0,j_0)+
    \sum_{\substack{i_0\not=j_0\\Q_{i_0}\geq1}}
    \Sigma_3(i_0,j_0)-\Sigma_4(i_0,j_0) 
  \end{equation}
  where
  \begin{gather}
    \label{eq:189}
    \Sigma_1(i_0,j_0):=\sum_{\substack{\{1,2\}\subset A_{j_0}\\
        |A_j|=Q_j,\ \forall1\leq j\leq m\\A_1\cup\cdots\cup A_m=\{1,\cdots,n\}\\
        A_j\cap A_{j'}=\emptyset\text{ if }j\not=j'}}
    \sum_{\substack{A'_{i_0}=\{1\}\cup A_{i_0}\\
        A'_{j_0}=A_{j_0}\setminus\{1\}\\
        A'_{j}=A_j\text{ if }j\not\in\{i_0,j_0\}}}I(A,A'),\\
    \label{eq:190}
    \Sigma_2(i_0,j_0):=\sum_{\substack{\{1,2\}\subset A_{j_0}\\
        |A_j|=Q_j,\ \forall1\leq j\leq m\\A_1\cup\cdots\cup A_m=\{1,\cdots,n\}\\
        A_j\cap A_{j'}=\emptyset\text{ if }j\not=j'}}
    \sum_{\substack{A'_{i_0}=\{2\}\cup A_{i_0}\\
        A'_{j_0}=A_{j_0}\setminus\{2\}\\
        A'_{j}=A_j\text{ if }j\not\in\{i_0,j_0\}}}I(A,A'),\\
    \label{eq:191}
    \Sigma_3(i_0,j_0):=\sum_{\substack{\{1,2\}\subset A'_{i_0}\\
        |A_j|=Q_j,\ \forall1\leq j\leq m\\A_1\cup\cdots\cup A_m=\{1,\cdots,n\}\\
        A_j\cap A_{j'}=\emptyset\text{ if }j\not=j'}}
    \sum_{\substack{A_{i_0}=A'_{i_0}\setminus\{1\} \\
        A_{j_0}=A'_{j_0}\cup\{1\}\\
        A'_{j}=A_j\text{ if }j\not\in\{i_0,j_0\}}}I(A,A'),\\
    \label{eq:192}
    \Sigma_4(i_0,j_0):=\sum_{\substack{\{1,2\}\subset A'_{i_0}\\
        |A_j|=Q_j,\ \forall1\leq j\leq m\\A_1\cup\cdots\cup A_m=\{1,\cdots,n\}\\
        A_j\cap A_{j'}=\emptyset\text{ if }j\not=j'}}
    \sum_{\substack{A_{i_0}=A'_{i_0}\setminus\{2\} \\
        A_{j_0}=A'_{j_0}\cup\{2\}\\
        A'_{j}=A_j\text{ if }j\not\in\{i_0,j_0\}}}I(A,A')
  \end{gather}
  and
  \begin{itemize}
  \item for the summands in $\Sigma_1(i_0,j_0)$:
   \begin{equation*}
    \begin{split}
      I(A,A')&=\int_{\Delta^{Q_{j_0}-2}_{j_0}}
      \varphi^{j_0}_{Q_{j_0},n_{j_0}}(x,x',z)
      \overline{\varphi^{j_0}_{Q_{j_0}-1,n'_{j_0}}(y',z)}dz\\&\hskip3cm
      \int_{\Delta^{Q_{i_0}}_{i_0}}\varphi^{i_0}_{Q_{i_0},n_{i_0}}(z')
      \overline{\varphi^{i_0}_{Q_{i_0}+1,n'_{i_0}}(y,z')}dz'
      \prod_{j\not\in\{i_0,j_0\}}\delta_{n_j=n'_j} 
    \end{split}
  \end{equation*}
  \item for the summands in $\Sigma_2(i_0,j_0)$:
   \begin{equation*}
    \begin{split}
      I(A,A')&=\int_{\Delta^{Q_{j_0}-2}_{j_0}}
      \varphi^{j_0}_{Q_{j_0},n_{j_0}}(x,x',z)
      \overline{\varphi^{j_0}_{Q_{j_0}-1,n'_{j_0}}(y,z)}dz\\&\hskip3cm
      \int_{\Delta^{Q_{i_0}}_{i_0}}\varphi^{i_0}_{Q_{i_0},n_{i_0}}(z')
      \overline{\varphi^{i_0}_{Q_{i_0}+1,n'_{i_0}}(y',z')}dz'
      \prod_{j\not\in\{i_0,j_0\}}\delta_{n_j=n'_j} 
    \end{split}
  \end{equation*}
  \item for the summands in $\Sigma_3(i_0,j_0)$:
   \begin{equation*}
    \begin{split}
      I(A,A')&=\int_{\Delta^{Q_{j_0}-1}_{j_0}}
      \varphi^{j_0}_{Q_{j_0},n_{j_0}}(x',z)
      \overline{\varphi^{j_0}_{Q_{j_0}-1,n'_{j_0}}(z)}dz\\&\hskip3cm
      \int_{\Delta^{Q_{i_0}-1}_{i_0}}\varphi^{i_0}_{Q_{i_0},n_{i_0}}(x,z')
      \overline{\varphi^{i_0}_{Q_{i_0}+1,n'_{i_0}}(y,y',z')}dz'
      \prod_{j\not\in\{i_0,j_0\}}\delta_{n_j=n'_j}
    \end{split}
  \end{equation*}
  \item for the summands in $\Sigma_4(i_0,j_0)$:
   \begin{equation*}
    \begin{split}
      I(A,A')&=\int_{\Delta^{Q_{j_0}-1}_{j_0}}
      \varphi^{j_0}_{Q_{j_0},n_{j_0}}(x,z)
      \overline{\varphi^{j_0}_{Q_{j_0}-1,n'_{j_0}}(z)}dz\\&\hskip3cm
      \int_{\Delta^{Q_{i_0}-1}_{i_0}}\varphi^{i_0}_{Q_{i_0},n_{i_0}}(x',z')
      \overline{\varphi^{i_0}_{Q_{i_0}+1,n'_{i_0}}(y,y',z')}dz'
      \prod_{j\not\in\{i_0,j_0\}}\delta_{n_j=n'_j} 
    \end{split}
  \end{equation*}
  \end{itemize}
  with the convention described in Remark~\ref{rem:9}.\\
  The number of partitions coming up
  in~\eqref{eq:189},~\eqref{eq:190},~\eqref{eq:191} and~\eqref{eq:192}
  are the same: indeed, it suffices to invert the roles of $1$ and $2$
  and $i_0$ and $j_0$. We compute
  \begin{equation*}
    \sum_{\substack{\{1,2\}\subset A_{j_0}\\
        |A_j|=Q_j,\ \forall1\leq j\leq m\\A_1\cup\cdots\cup
        A_m=\{1,\cdots,n\}\\ A_j\cap A_{j'}=\emptyset
        \text{ if }j\not=j'}}\sum_{\substack{A'_{i_0}=\{1\}\cup
        A_{i_0}\\ A'_{j_0}=A_{j_0}\setminus\{1\}\\
        A'_{j}=A_j\text{ if }j\not\in\{i_0,j_0\}}}1= 
    \frac{(n-Q_{j_0}-Q_{i_0}-2)!Q_{i_0}!Q_{j_0}!}{Q_1!\cdots
      Q_m!}.
  \end{equation*}
  Hence, we get that
  \begin{equation}
    \label{eq:206}
    \frac{n(n-1)}2\sum_{\substack{Q,\ Q'\text{ occ.}\\ d_1(Q,Q')=2\\
        \overline{n},\,\overline{n}'\in\N^m}}
    a^Q_{\overline{n}}\overline{a^{Q'}_{\overline{n}'}}
    \gamma^{(2)}_{\substack{Q,\overline{n}\\Q',\overline{n}'}}=
    \sum_{i\not=j}\sum_{\tilde n\in\N^{m-2}} \sum_{\substack{Q\text{ occ.}\\
        Q_j\geq1\\ Q':\ Q'_k=Q_k\text{ if
        }k\not\in\{i,j\}\\Q'_i=Q_i+1\\Q'_j=Q_j-1}}
    \sum_{\substack{n_j,n'_j\geq 1\\n_i,n'_i\geq 1}}
    a^Q_{\tilde{n}_{i,j}}\overline{a^{Q'}_{\tilde{n}_{i,j}'}}
    \gamma^{(2),2}_{\substack{Q_i,Q_j\\n_i,n_j\\n'_i,n'_j}}
  \end{equation}
  where $\gamma^{(2),2}_{\substack{Q_i,Q_j\\n_i,n_j\\n'_i,n'_j}}$ is defined
  in~\eqref{eq:76}.\\
  Let us now assume $d_1(Q,Q')=4$. Thus,
  \begin{enumerate}
  \item either there exist $1\leq i_0\not= j_0\leq m$ such that
    $Q_{j_0}\geq2$, $Q'_{i_0}=Q_{i_0}+2$, $Q_{j_0}=Q'_{j_0}+2$ and
    $Q_k=Q'_k$ for $k\not\in\{i_0,j_0\}$.\\
    In this case, either $A_{j_0}=\{1,2\}\cup A'_{j_0}$ and
    $A'_{i_0}=\{1,2\}\cup A_{i_0}$ with $A_{i_0},A'_{j_0},\subset\{3,\cdots,n\}$ or
    $I(A,A')=0$ vanishes. Thus,
    \begin{equation}
      \label{eq:193}
      \frac{\gamma^{(2)}_{\substack{Q,\overline{n}\\Q',\overline{n}'}}(x,y)}{c^2(Q)}=
      \sum_{\substack{\{1,2\}\subset A_{j_0}\\
          |A_j|=Q_j,\ \forall1\leq j\leq m\\A_1\cup\cdots\cup
          A_m=\{1,\cdots,n\}\\A_j\cap A_{j'}=\emptyset\text{ if }j\not=j'}}
      \sum_{\substack{A'_{i_0}=\{1,2\}\cup A_{i_0}\\
          A'_{j_0}=A_{j_0}\setminus\{1,2\}\\
          A'_{j}=A_j\text{ if }j\not\in\{i_0,j_0\}}}
      (-1)^{\varepsilon((A_j))+\varepsilon((A'_j))}I(A,A'),
    \end{equation}
    and
    \begin{equation*}
      \begin{split}
        I(A,A')&=\int_{\Delta^{Q_{j_0}-2}_{j_0}}
        \varphi^{j_0}_{Q_{j_0},n_{j_0}}(x,x',z)
        \overline{\varphi^{j_0}_{Q_{j_0}-2,n'_{j_0}}(z)}dz\\&\hskip3cm
        \int_{\Delta^{Q_{i_0}}_{i_0}}\varphi^{i_0}_{Q_{i_0},n_{i_0}}(z')
        \overline{\varphi^{i_0}_{Q_{i_0}+2,n'_{i_0}}(y,y',z')}dz'
        \prod_{j\not\in\{i_0,j_0\}}\delta_{n_j=n'_j}.
      \end{split}
    \end{equation*}
    Hence, taking~\eqref{eq:194} into account, we get
    \begin{multline}
      \label{eq:195}
      \frac{n(n-1)}2\sum_{\substack{Q,\ Q'\text{ occ.}\\ \exists i\not=j,\ Q_j\geq2\\
          Q':\ Q'_k=Q_k\text{ if }k\not\in\{i,j\}\\Q'_i=Q_i+2\\Q'_j=Q_j-2}}
      \sum_{\substack{\overline{n}\in\N^m\\\overline{n}'\in\N^m}}
      a^Q_{\overline{n}}\overline{a^{Q'}_{\overline{n}'}}
      \gamma^{(2)}_{\substack{Q,\overline{n}\\Q',\overline{n}'}}\\=
      \sum_{i\not=j}\sum_{\tilde n\in\N^{m-2}} \sum_{\substack{Q\text{ occ.}\\
          Q_j\geq2\\ Q':\ Q'_k=Q_k\text{ if
          }k\not\in\{i,j\}\\Q'_i=Q_i+2\\Q'_j=Q_j-2}}C_2(Q,i,j)
      \sum_{\substack{n_j,n'_j\geq 1\\n_i,n'_i\geq 1}}
      a^Q_{\tilde{n}_{i,j}}\overline{a^{Q'}_{\tilde{n}_{i,j}'}}
      \gamma^{(2),4,2}_{\substack{Q_i,Q_j\\n_i,n_j\\n'_i,n'_j}}
    \end{multline}
    as
    \begin{equation*}
      \sum_{\substack{\{1,2\}\subset A_{j}\\
          |A_l|=Q_l,\ \forall1\leq l\leq m\\A_1\cup\cdots\cup
          A_m=\{1,\cdots,n\}\\A_l\cap A_{l'}=\emptyset\text{ if }l\not=l'}}
      \sum_{\substack{A'_{i}=\{1,2\}\cup A_{i}\\
          A'_{j}=A_{j}\setminus\{1,2\}\\
          A'_{l}=A_l\text{ if }j\not\in\{i,j\}}}1=
      \frac{(n-Q_{j}-Q_{i}-2)!Q_{i}!Q_{j}!}{Q_1!\cdots
        Q_m!}=\frac{2\,C_2(Q,i,j)}{n(n-1)\,c(Q)^2}.
    \end{equation*}
  \item or there exist $1\leq i_0,j_0,k_0\leq m$ distinct such that
    $Q_{j_0}\geq2$, $Q'_{j_0}=Q_{j_0}-2$, $Q_{i_0}=Q'_{i_0}+1$,
    $Q_{k_0}=Q'_{k_0}+1$, and $Q_k=Q'_k$ for
    $k\not\in\{i_0,j_0,k_0\}$.\\
    In this case, either $A_{j_0}=\{1,2\}\cup A'_{j_0}$ and
    (($A'_{i_0}=\{1\}\cup A_{i_0}$ and $A'_{k_0}=\{2\}\cup A_{k_0}$) or
    ($A'_{i_0}=\{2\}\cup A_{i_0}$ and $A'_{k_0}=\{1\}\cup A_{k_0}$) ) with
    $A_{j_0}, A'_{i_0}, A'_{k_0}\subset\{3,\cdots,n\}$ or
    $I(A,A')=0$ vanishes. Thus,
    \begin{equation}
      \label{eq:268}
      \frac{\gamma^{(2)}_{\substack{Q,\overline{n}\\Q',\overline{n}'}}(x,y)}{c^2(Q)}=
      \sum_{\substack{\{1,2\}\subset A_{j_0}\\
          |A_j|=Q_j,\ \forall1\leq j\leq m\\A_1\cup\cdots\cup
          A_m=\{1,\cdots,n\}\\A_j\cap A_{j'}=\emptyset\text{ if
          }j\not=j'}}\left(
        \sum_{\substack{A'_{i_0}=\{1\}\cup A_{i_0}\\
            A'_{k_0}=\{2\}\cup A_{k_0}\\
            A'_{j_0}=A_{j_0}\setminus\{1,2\}\\
            A'_{j}=A_j\text{ if }j\not\in\{i_0,j_0,k_0\}}}I(A,A')
        -        \sum_{\substack{A'_{i_0}=\{2\}\cup A_{i_0}\\ 
            A'_{k_0}=\{1\}\cup A_{k_0}\\
            A'_{j_0}=A_{j_0}\setminus\{1,2\}\\
            A'_{j}=A_j\text{ if }j\not\in\{i_0,j_0,k_0\}}}
        I(A,A')\right)
    \end{equation}
    and, if $A'_{i_0}=\{1\}\cup A_{i_0}$ and $A'_{k_0}=\{2\}\cup A_{k_0}$, one has
    \begin{gather*}
      \begin{split}
        I(A,A')&=\int_{\Delta^{Q_{j_0}-2}_{j_0}}
        \varphi^{j_0}_{Q_{j_0},n_{j_0}}(x,x',z)
        \overline{\varphi^{j_0}_{Q_{j_0}-2,n'_{j_0}}(z)}dz\\&\hskip1.5cm
        \int_{\Delta^{Q_{i_0}}_{i_0}}\varphi^{i_0}_{Q_{i_0},n_{i_0}}(z')
        \overline{\varphi^{i_0}_{Q_{i_0}+1,n'_{i_0}}(y,z')}dz'\\&\hskip3cm
        \int_{\Delta^{Q_{k_0}}_{k_0}}\varphi^{i_0}_{Q_{k_0},n_{k_0}}(z'')
        \overline{\varphi^{k_0}_{Q_{k_0}+1,n'_{k_0}}(y',z'')}dz''
        \prod_{j\not\in\{i_0,j_0,k_0\}}\delta_{n_j=n'_j}
      \end{split}
      \\ \intertext{and, if $A'_{i_0}=\{2\}\cup A_{i_0}$ and
        $A'_{k_0}=\{1\}\cup A_{k_0}$, one has}
      \begin{split}
        I(A,A')&=\int_{\Delta^{Q_{j_0}-2}_{j_0}}
        \varphi^{j_0}_{Q_{j_0},n_{j_0}}(x,x',z)
        \overline{\varphi^{j_0}_{Q_{j_0}-2,n'_{j_0}}(z)}dz\\&\hskip1.5cm
        \int_{\Delta^{Q_{i_0}}_{i_0}}\varphi^{i_0}_{Q_{i_0},n_{i_0}}(z')
        \overline{\varphi^{i_0}_{Q_{i_0}+1,n'_{i_0}}(y',z')}dz'\\&\hskip3cm
        \int_{\Delta^{Q_{k_0}}_{k_0}}\varphi^{i_0}_{Q_{k_0},n_{k_0}}(z'')
        \overline{\varphi^{k_0}_{Q_{k_0}+1,n'_{k_0}}(y,z'')}dz''
        \prod_{j\not\in\{i_0,j_0,k_0\}}\delta_{n_j=n'_j}.
      \end{split}
    \end{gather*}
    For $i_0,j_0,k_0$ distinct, one has
    \begin{equation*}
      \begin{split}
        \sum_{\substack{\{1,2\}\subset A_{j_0}\\
            |A_j|=Q_j,\ \forall1\leq j\leq m\\A_1\cup\cdots\cup
            A_m=\{1,\cdots,n\}\\A_j\cap A_{j'}=\emptyset\text{ if }j\not=j'}}
        \sum_{\substack{A'_{i_0}=\{1\}\cup A_{i_0}\\
            A'_{k_0}=\{2\}\cup A_{k_0}\\
            A'_{j_0}=A_{j_0}\setminus\{1,2\}\\
            A'_{j}=A_j\text{ if }j\not\in\{i_0,j_0,k_0\}}}1
        &=\frac{(n-Q_{j_0}-Q_{i_0}-Q_{k_0}-2)!Q_{i_0}!Q_{j_0}!Q_{k_0}!}{Q_1!\cdots
          Q_m!}
        \\&=\frac{2\,C_3(Q,i_0,j_0,k_0)}{n(n-1)\,c(Q)^2}.
      \end{split}
    \end{equation*}
    Inverting the roles of $1$ and $2$ we see that the number of
    partitions coming up in the second sum in~\eqref{eq:268} is the
    same. Thus, taking~\eqref{eq:194} into account, we get
    \begin{multline}
      \label{eq:197}
      \frac{n(n-1)}2\sum_{\substack{Q,\ Q'\text{ occ.}\\ \exists
          i,j,k\text{ distinct}\\ Q_j\geq2\\
          Q':\ Q'_l=Q_l\text{ if }l\not\in\{i,j,k\}\\
          Q'_j=Q_j-2\\Q'_i=Q_i+1,\ Q'_k=Q_k+1}}
      \sum_{\substack{\overline{n}\in\N^m\\\overline{n}'\in\N^m}}
      a^Q_{\overline{n}}\overline{a^{Q'}_{\overline{n}'}}
      \gamma^{(2)}_{\substack{Q,\overline{n}\\Q',\overline{n}'}}\\=
      \sum_{\substack{i,j,k\\\text{distinct}}}\sum_{\tilde n\in\N^{m-3}}
      \sum_{\substack{Q\text{ occ.}\\ Q_j\geq2\\ Q':\ Q'_l=Q_l\text{ if
          }l\not\in\{i,j,k\}\\ Q'_j=Q_j-2\\Q'_i=Q_i+1,\ Q'_k=Q_k+1}}
      C_3(Q,i,j,k)\sum_{\substack{n_i,n_j,n_k\geq 1\\n'_i,n'_j,n'_k\geq 1}}
      a^Q_{\tilde{n}_{i,j,k}}\overline{a^{Q'}_{\tilde{n}_{i,j,k}'}}
      \gamma^{(2),4,3}_{\substack{Q_i,Q_j,Q_k\\n_i,n_j,n_k\\n'_i,n'_j,n'_k}}.
    \end{multline}
  \item or there exist $1\leq i_0,j_0,k_0\leq m$ distinct such that
    $Q_{i_0}\geq1$, $Q_{k_0}\geq1$, $Q'_{j_0}=Q_{j_0}+2$,
    $Q_{i_0}=Q'_{i_0}-1$, $Q_{k_0}=Q'_{k_0}-1$, and $Q_k=Q'_k$ for
    $k\not\in\{i_0,j_0,k_0\}$.\\
    We see that we are back to case (b) if we invert the roles of $Q$
    and $Q'$. Thus, we get
    \begin{multline}
      \label{eq:201}
      \frac{n(n-1)}2\sum_{\substack{Q,\ Q'\text{ occ.}\\ \exists i,j,k\text{
            distinct}\\ Q_i\geq1,\ Q_k\geq1\\Q':\ Q'_l=Q_l\text{ if
          }l\not\in\{i,j,k\}\\Q'_j=Q_j+2\\Q'_i=Q_i-1,\ Q'_k=Q_k-1}}
      \sum_{\substack{\overline{n}\in\N^m\\\overline{n}'\in\N^m}}
      a^Q_{\overline{n}}\overline{a^{Q'}_{\overline{n}'}}
      \gamma^{(2)}_{\substack{Q,\overline{n}\\Q',\overline{n}'}}\\=
      \sum_{\substack{i,j,k\\\text{distinct}}}\sum_{\tilde n\in\N^{m-3}}
      \sum_{\substack{Q\text{ occ.}\\ Q_i\geq1,\ Q_k\geq1\\ Q':\
          Q'_l=Q_l\text{ if }l\not\in\{i,j,k\}\\Q'_j=Q_j+2\\Q'_i=Q_i-1,\
          Q'_k=Q_k-1}} C_3(Q,i,j,k)\sum_{\substack{n_i,n_j,n_k\geq
          1\\n'_i,n'_j,n'_k\geq 1}}
      a^Q_{\tilde{n}_{i,j,k}}\overline{a^{Q'}_{\tilde{n}_{i,j,k}'}}
      \gamma^{(2),4,3'}_{\substack{Q_i,Q_j,Q_k\\n_i,n_j,n_k\\n'_i,n'_j,n'_k}}.
    \end{multline}
  \item or there exist $1\leq i_0,j_0,k_0,l_0\leq m$ distinct such
    that $Q_{j_0}\geq1$, $Q_{l_0}\geq1$, $Q'_{i_0}=Q_{i_0}-1$,
    $Q_{j_0}=Q'_{j_0}-1$, $Q'_{k_0}=Q_{k_0}+1$, $Q_{l_0}=Q'_{l_0}+1$
    and $Q_k=Q'_k$ for $k\not\in\{i_0,j_0,k_0,l_0\}$.\\
    Then, either $I(A,A')=0$ or
    \begin{enumerate}
    \item[(i)] either $A_{i_0}=\{1\}\cup A'_{i_0}$ and $A_{j_0}=\{2\}\cup
      A'_{j_0}$ and $A'_{i_0}, A'_{j_0}\subset\{3,\cdots,n\}$,
    \item[(ii)] or $A_{i_0}=\{2\}\cup A'_{i_0}$ and $A_{j_0}=\{1\}\cup A'_{j_0}$ and
      $A'_{i_0}, A'_{j_0}\subset\{3,\cdots,n\}$, in which case 
    \end{enumerate}
    Moreover, in each of the cases (i) and (ii), either $I(A,A')=0$ or
    \begin{enumerate}
    \item[(i)] either $A'_{k_0}=\{1\}\cup A_{k_0}$ and $A'_{l_0}=\{2\}\cup
      A_{l_0}$ and $A_{k_0},A_{l_0}\subset\{3,\cdots,n\}$,
    \item[(ii)] or $A'_{k_0}=\{2\}\cup A_{k_0}$ and $A'_{l_0}=\{1\}\cup A_{l_0}$
      and $A_{k_0},A_{l_0}\subset\{3,\cdots,n\}$
    \end{enumerate}
    In the 4 cases when $I(A,A')$ does not vanish, one computes
    \begin{itemize}
    \item $I(A,A')=\alpha(x,x',y,y')$ in case (i.i),
    \item $I(A,A')=\alpha(x',x,y,y')$ in case (ii.i),
    \item $I(A,A')=\alpha(x,x',y,y')$ in case (i.ii),
    \item $I(A,A')=\alpha(x',x,y',y)$ in case (ii.ii),
    \end{itemize}
    where
    \begin{equation*}
      \begin{split}
        \alpha(x,x',y,y')&:=\int_{\Delta^{Q_{i}-1}_{i}}
        \varphi^{i}_{Q_{i},n_{i}}(x,z)
        \overline{\varphi^{i}_{Q_{i}-1,n'_{i}}(z)}dz
        \int_{\Delta^{Q_{j}-1}_{j}} \varphi^{j}_{Q_{j},n_{j}}(x',z)
        \overline{\varphi^{j}_{Q_{j}-1,n'_{j}}(z)}dz\\&\hskip1cm
        \times\int_{\Delta^{Q_{k}}_{k}} \varphi^{k}_{Q_{k},n_{k}}(z)
        \overline{\varphi^{k}_{Q_{k}+1,n'_{k}}(y,z)}dz
        \int_{\Delta^{Q_{l}}_{l}} \varphi^{l}_{Q_{l},n_{l}}(z)
        \overline{\varphi^{l}_{Q_{l}+1,n'_{l}}(y',z)}dz.
      \end{split}
    \end{equation*}
    Hence, if $d_1(Q,Q')=4$, we obtain
    \begin{equation}
      \label{eq:204}
      \begin{split}
        \frac{\gamma^{(2)}_{\substack{Q,\overline{n}\\Q',
              \overline{n}'}}(x,y)}{c^2(Q)}&=
        \sum_{\substack{1\in A_{i_0},\ 2\in A_{j_0}\\
            |A_j|=Q_j,\ \forall1\leq j\leq m\\A_1\cup\cdots\cup
            A_m=\{1,\cdots,n\}\\A_j\cap A_{j'}=\emptyset\text{ if
            }j\not=j'}}\left(\sum_{\substack{A'_{k_0}=\{1\}\cup
              A_{k_0},\ A'_{l_0}=\{2\}\cup
              A_{l_0}\\A'_{i_0}=A_{i_0}\setminus\{1\},\
              A'_{j_0}=A_{j_0}\setminus\{2\}\\ A'_{j}=A_j\text{ if
              }j\not\in\{i_0,j_0,k_0,l_0\}}}
          I(A,A')-\sum_{\substack{A'_{k_0}=\{2\}\cup A_{k_0},\
              A'_{l_0}=\{1\}\cup
              A_{l_0}\\A'_{i_0}=A_{i_0}\setminus\{1\},\
              A'_{j_0}=A_{j_0}\setminus\{2\}\\
              A'_{j}=A_j\text{ if }j\not\in\{i_0,j_0,k_0,l_0\}}}
          I(A,A')\right)
        \\&-\sum_{\substack{2\in A_{i_0},\ 1\in A_{j_0}\\
            |A_j|=Q_j,\ \forall1\leq j\leq m\\A_1\cup\cdots\cup
            A_m=\{1,\cdots,n\}\\A_j\cap A_{j'}=\emptyset\text{ if
            }j\not=j'}}\left(\sum_{\substack{A'_{k_0}=\{1\}\cup
              A_{k_0},\ A'_{l_0}=\{2\}\cup
              A_{l_0}\\A'_{i_0}=A_{i_0}\setminus\{2\},\
              A'_{j_0}=A_{j_0}\setminus\{1\}\\
              A'_{j}=A_j\text{ if }j\not\in\{i_0,j_0,k_0,l_0\}}}
          I(A,A')-\sum_{\substack{A'_{k_0}=\{2\}\cup A_{k_0},\
              A'_{l_0}=\{1\}\cup
              A_{l_0}\\A'_{i_0}=A_{i_0}\setminus\{2\},\
              A'_{j_0}=A_{j_0}\setminus\{1\}\\
              A'_{j}=A_j\text{ if }j\not\in\{i_0,j_0,k_0,l_0\}}}
          I(A,A')\right).
      \end{split}
    \end{equation}
    For $i_0,j_0,k_0,l_0$ distinct, the number of partitions coming up
    in the first sum in~\eqref{eq:204} is given by
    \begin{equation*}
      \begin{split}
        \sum_{\substack{1\in A_{i_0},\ 2\in A_{j_0}\\
            |A_j|=Q_j,\ \forall1\leq j\leq m\\A_1\cup\cdots\cup
            A_m=\{1,\cdots,n\}\\A_j\cap A_{j'}=\emptyset\text{ if
            }j\not=j'}}\sum_{\substack{A'_{k_0}=\{1\}\cup A_{k_0},\
            A'_{l_0}=\{2\}\cup
            A_{l_0}\\A'_{i_0}=A_{i_0}\setminus\{1\},\
            A'_{j_0}=A_{j_0}\setminus\{2\}\\ A'_{j}=A_j\text{ if
            }j\not\in\{i_0,j_0,k_0,l_0\}}}1 &=
        \frac{(n-Q_{j_0}-Q_{i_0}-Q_{k_0}-Q_{l_0}-2)!Q_{i_0}!Q_{j_0}!
          Q_{k_0}!Q_{l_0}!}{Q_1!\cdots Q_m!}
        \\&=\frac{2\,C_4(Q,i_0,j_0,k_0,l_0)}{n(n-1)\,c(Q)^2}.
      \end{split}
    \end{equation*}
    Inverting the roles of $i_0,j_0,k_0,l_0$, we see that the number
    of partitions involved is the same in the three remaining sums
    of~\eqref{eq:204}.\\
    Thus, taking~\eqref{eq:194} into account, we get
    \begin{multline}
      \label{eq:207}
      \frac{n(n-1)}2\sum_{\substack{Q,\ Q'\text{ occ.}\\ \exists
          i,j,k,l\text{ distinct}\\ Q_i\geq1,\ Q_j\geq1\\
          Q':\ Q'_p=Q_p\text{ if }p\not\in\{i,j,k,l\}\\Q'_i=Q_i-1,\
          Q'_j=Q_j-1\\Q'_k=Q_k+1,\ Q'_l=Q_l+1}}
      \sum_{\substack{\overline{n}\in\N^m\\\overline{n}'\in\N^m}}
      a^Q_{\overline{n}}\overline{a^{Q'}_{\overline{n}'}}
      \gamma^{(2)}_{\substack{Q,\overline{n}\\Q',\overline{n}'}}\\=
      \sum_{\substack{i,j,k,l\\\text{distinct}}}\sum_{\tilde n\in\N^{m-4}}
      \sum_{\substack{Q\text{ occ.}\\ Q_i\geq1,\ Q_j\geq1\\
          Q':\ Q'_p=Q_p\text{ if }p\not\in\{i,j,k,l\}\\Q'_i=Q_i-1,\
          Q'_j=Q_j-1\\Q'_k=Q_k+1,\ Q'_l=Q_l+1}}
      C_4(Q,i,j,k)\sum_{\substack{n_i,n_j,n_k,n_l\geq
          1\\n'_i,n'_j,n'_k,n'_l\geq 1}}
      a^Q_{\tilde{n}_{i,j,k,l}}\overline{a^{Q'}_{\tilde{n}_{i,j,k,l}'}}
      \gamma^{(2),4,4}_{\substack{Q_i,Q_j,Q_k,Q_l\\n_i,n_j,n_k,n_l\\n'_i,n'_j,n'_k,n'_l}}.
    \end{multline}
  \end{enumerate}
  Plugging this,~\eqref{eq:195} and~\eqref{eq:206} into~\eqref{eq:132}, we
  obtain~\eqref{eq:186}. This completes the proof of Theorem~\ref{thr:5}.
\end{proof}
\subsubsection{A particular case}
\label{sec:particular-case}
Let us now explain how the structure of the one-particle and
two-particles density matrices may be simplified in the particular
case when the ground state is factorized. This in particular
immediately yields the expansions~\eqref{eq:11} and~\eqref{eq:14} for
the one-particle and two-particles density matrices of the non
interacting ground state.
\begin{definition}
  \label{def:completeOrthgonalityCondition}
  Let $\alpha \in \frH^i(L)$ and $\beta \in \frH^j(L)$ be two states
  describing $i$ and $j$ electrons respectively.  We say $\alpha$ and
  $\beta$ \emph{do not interact} if for all $(x^2, \hdots, x^{i}, y^2,
  \hdots, y^{j}) \in [0, L]^{i + j - 2}$,
  \begin{equation}
    \label{eq:completeOrthgonalityCondition}
    \int_0^L \alpha(x^1, \hdots, x^{i}) \beta^\ast(y^1, \hdots,
    y^{j}) \bigr|_{x^1 = y^1} \rmd{x^1} = 0 \text{.}
  \end{equation}
  To denote this complete orthogonality, we will write $\alpha\Perp\beta$.
\end{definition}
\begin{remark}
  Because of the anti-symmetric nature of the states $\alpha$ and
  $\beta$ in the above definition, it is sufficient to impose the
  orthogonality only on the first variables. Thus, an integral of the
  type~\eqref{eq:completeOrthgonalityCondition} vanishes for any pair
  of coordinates $x^{i_1} = y^{j_1}$ for $i_1\in\{1,\hdots,i\}$, and
  $j_1\in\{1, \hdots,j\}$.
\end{remark}
\noindent We prove
\begin{proposition}
  \label{prop:DensityMatrixStructure}
  Suppose that a $n$-particle state $\Psi \in \frH^n(L)$ is decomposed
  in its non interacting parts:
  \begin{equation*}
    \Psi = \bigwedge_{j = 1}^k \zeta_j \text{,}
  \end{equation*}
  where each $\zeta_j \in \frH^{k_j}(L)$ is a $k_j$-particle state
  describing a packet of particles that do not interact with other
  packets, i.e., for $i \ne j$, $\zeta_i \Perp \zeta_j$ in the sens of
  Definition~\textup{\ref{def:completeOrthgonalityCondition}}.  Then
  \begin{equation}
    \label{eq:OneParticleDensityMatrixReduction}
    \gamma_\Psi = \sum_{j = 1}^k \gamma_{\zeta_j}
  \end{equation}
  and
  \begin{equation}
    \label{eq:TwoParticleDensityMatrixReduction}
    \gamma_\Psi^{(2)} = \sum_{j = 1}^k \left[\gamma_{\zeta_j}^{(2)} -
      \frac{1}{2} (\Id - \Ex) \gamma_{\zeta_j} \otimes
      \gamma_{\zeta_j}\right]
    + \frac{1}{2} (\Id - \Ex) \gamma_{\Psi} \otimes \gamma_{\Psi},
  \end{equation}
  where $\Id$ is the identity, $\Ex$ is the exchange operator on the
  two-particles space defined as
  \begin{equation*}
    \Ex f \otimes g = g \otimes f, \quad f,g\in\frH,
  \end{equation*}
  and with the obvious convention that $\gamma_{\zeta_j}^{(2)} = 0$ if
  $\zeta_j$ is a one-particle state.
\end{proposition}
\noindent While Proposition~\ref{prop:DensityMatrixStructure} could be
obtained as a consequence of Theorems~\ref{thr:4} and~\ref{thr:5}, we
will derive it from the following auxiliary lemma.
\begin{lemma}
  \label{lem:DensityMatrixFactorizationReduction}
  Let $\alpha \in \frH^n(L)$ and $\beta \in \frH^m(L)$ be two vectors
  describing $n$ and $m$ electrons respectively.  Suppose that $\alpha$
  and $\beta$ do not interact:
  \begin{equation*}
    \alpha \bot \beta.
  \end{equation*}
  Then,
  \begin{equation}
    \label{eq:OneParticleDensityMatrixOfTensorProduct}
    \gamma_{\alpha \wedge \beta} = \gamma_\alpha + \gamma_\beta
  \end{equation}
  and
  \begin{equation}
    \label{eq:TwoParticleDensityMatrixOfTensorProduct}
    \gamma_{\alpha \wedge \beta}^{(2)} = \gamma_\alpha^{(2)} +
    \gamma_\beta^{(2)} + (\Id - \Ex) \gamma_\alpha \otimes^s \gamma_\beta
  \end{equation}
  where $\otimes^s$ denotes the symmetrized tensor product:
  \begin{equation*}
    A \otimes^s B = \frac{1}{2} (A \otimes B + B \otimes A) \text{.}
  \end{equation*} 
\end{lemma}
\begin{proof}
  Define $\bbN_n:=\{1, \hdots, n\}$. Consider the two-particles
  density matrix. By~\eqref{eq:267}, the anti-symmetrized product of
  two eigenfunctions in respectively $n$ and $m$ variables is given by
  \begin{equation*}
    (\alpha \wedge \beta)(x^1, \hdots, x^{n + m})
    = \frac{1}{\sqrt{\binom{n + m}{n}}} \sum_{\substack{J\cup
        J'=\bbN_{n + m}\\J\cap J'=\emptyset,\ |J| = n}}
    (-1)^{\sign{J}} \alpha(x^J) \beta(x^{J'}). 
  \end{equation*}
  Thus, the corresponding two-particles density matrix can be written
  as
  \begin{equation}
    \label{eq:AlphaBetaGamma2Expansion}
    \begin{split}
      &\gamma_{\alpha \wedge \beta}^{(2)}(x^1, x^2, y^1, y^2)\\& =
      \frac{n (n - 1)}{2} \int_{[0, L]^{n + m - 2}} (\alpha \wedge
      \beta)(x^1, x^2, \overline{x})\,(\alpha \wedge \beta)^\ast(y^1,
      y^2, \overline{x}) d\overline{x}
      \\
      &=\frac{n (n - 1)}{2\binom{n + m}{n}} \sum_{\substack{I\cup
          I'=\bbN_{n + m}\\I\cap I'=\emptyset,\ |I| = n\\ J\cup
          J'=\bbN_{n + m}\\J\cap J'=\emptyset,\ |J| = n}} \int_{[0,
        L]^{n + m - 2}}(-1)^{\sign{I} + \sign{J}}\alpha(x^I)
      \beta(x^{I'}) \alpha^\ast(y^J)
      \beta^\ast(y^{J'})\Bigr|_{\substack{y^j = x^j\\ j \in \{3,
          \hdots, n + m\}}} d\overline{x}.
    \end{split}
  \end{equation}
  As $\alpha$ and $\beta$ do not interact, the integrals in the sum in
  the last part of~\eqref{eq:AlphaBetaGamma2Expansion} vanish if $I$
  differs from $J$ by more than two elements, i.e., $|I\setminus J|
  \geq2$. Moreover, if such an integral does not vanish, one
  distinguishes the following cases:
  \begin{enumerate}
  \item if $\{1,2\}\subset I$, then $I=J$; indeed, otherwise $J$ would
    contain an index in $I'$ and the integration of
    $\beta(x^{I'})\alpha^\ast(y^J)\Bigr|_{\substack{y^j = x^j\\ j \in
        \{3, \hdots, n + m\}}}$ over the corresponding variable would
    produce zero because $\alpha \bot \beta$.
  \item if $\{1, 2\} \subset J$, then $I = J$.
  \item if $(1,2)\in (I\times I')\cup(I'\times I)$ then $(1,2)\in(J\times J')\cup(J'\times J)$
    by the same argument as above.
  \end{enumerate}
  As the functions $\alpha$ and $\beta$ are completely anti-symmetric
  under permutations of variables, the terms of the sums over $I$ and
  $J$ corresponding to different cases described above are all the
  same. If we denote $\whx^{k} = x^3, \hdots, x^{k}$ and $\rmd{\whx^k}
  = \rmd{x^3} \hdots \rmd{x^k}$ for $k \in \{n, m, n + n\}$), this
  finally yields
  \begin{equation*}
    \gamma_{\alpha \wedge \beta}^{(2)}(x^1, x^2, y^1, y^2)
    = A+B+C
  \end{equation*}
  where
  \begin{equation*}
    \begin{split}
      A&:=\frac{n (n - 1)}{2} \frac{1}{\binom{n + m}{n}} \binom{n + m
        - 2}{n - 2} \int_{[0, L]^{n - 2}}\alpha(x^1, x^2, \whx^n)
      \alpha^\ast(y^1, y^2, \whx^n) \rmd{\whx^n}\\&
      =\gamma_{\alpha}^{(2)}(x^1, x^2, y^1, y^2),
    \end{split}
  \end{equation*}
  \begin{equation*}
    \begin{split}
      B&:=\frac{n (n - 1)}{2} \frac{1}{\binom{n + m}{n}} \binom{n + m
        - 2}{m - 2} \int_{[0, L]^{m - 2}}\beta(x^1, x^2, \whx^m)
      \beta^\ast(y^1, y^2, \whx^m) \rmd{\whx^m}\\&
      =\gamma_{\beta}^{(2)}(x^1, x^2, y^1, y^2)
    \end{split}
  \end{equation*}
  and
  \begin{equation*}
    \begin{split}
      C&:= \frac{n (n - 1)}{2} \frac{1}{\binom{n + m}{n}} \binom{n + m
        - 2}{m - 1} \int_{[0, L]^{n + m - 2}}\rmd{\whx^{n + m}}\\
      &\hskip2cm\left(\alpha(x^1, \hdots) \beta(x^2,\hdots)
        \alpha^\ast(y^1, \hdots)
        \beta^\ast(y^2,\hdots)\right.\\
      &\hskip3cm- \alpha(x^1, \hdots) \beta(x^2,\hdots)
      \alpha^\ast(y^2, \hdots)
      \beta^\ast(y^1,\hdots)\\
      &\hskip4cm- \alpha(x^2, \hdots) \beta(x^1,\hdots)
      \alpha^\ast(y^1, \hdots)
      \beta^\ast(y^2,\hdots)\\
      &\left.\hskip5cm+ \alpha(x^2, \hdots) \beta(x^1,\hdots)
        \alpha^\ast(y^2, \hdots) \beta^\ast(y^1,\hdots)\right) \\& =
      \frac{1}{2} \left(\gamma_{\alpha}(x^1, y^1) \gamma_{\beta}(x^2,
        y^2)
        - \gamma_{\alpha}(x^1, y^2) \gamma_{\beta}(x^2, y^1)\right. \\
      &\left.\hskip3cm- \gamma_{\alpha}(x^2, y^1) \gamma_{\beta}(x^1,
        y^2) + \gamma_{\alpha}(x^2, y^2) \gamma_{\beta}(x^1,
        y^1)\right).
    \end{split}
  \end{equation*}
  This completes the proof
  of~\eqref{eq:TwoParticleDensityMatrixOfTensorProduct}.  The proof
  for the one-particle density
  matrix~\eqref{eq:OneParticleDensityMatrixOfTensorProduct} is done
  similarly and is even simpler. This completes the proof of
  Lemma~\ref{lem:DensityMatrixFactorizationReduction}.
\end{proof}
\begin{proof}[Proof of
  Proposition~\textup{\ref{prop:DensityMatrixStructure}}]
  The identity~\eqref{eq:OneParticleDensityMatrixReduction} for
  one-particle density matrix is a direct consequence
  of~\eqref{eq:OneParticleDensityMatrixOfTensorProduct}.  We
  prove~\eqref{eq:TwoParticleDensityMatrixReduction} by induction
  on $k$.\\
  For $k = 2$, \eqref{eq:TwoParticleDensityMatrixReduction} is
  equivalent to \eqref{eq:TwoParticleDensityMatrixOfTensorProduct}
  after noting that
  \begin{equation*}
    A \otimes^s B = \frac{1}{2} \left((A + B) \otimes (A + B) - A \otimes A
      - B \otimes B\right) \text{.}
  \end{equation*}
  This remark also proves that
  \begin{equation}
    \label{eq:ModifiedTwoParticleDensityMatrixReduction}
    \gamma_\Psi^{(2)}
    = \sum_{j = 1}^k \gamma_{\zeta_j}^{(2)}
    + (\Id - \Ex) \sum_{i < j} \gamma_{\zeta_i} \otimes^s \gamma_{\zeta_j}
  \end{equation}
  which is equality~\eqref{eq:TwoParticleDensityMatrixReduction}.\\
  Let us prove~\eqref{eq:ModifiedTwoParticleDensityMatrixReduction}
  inductively. Suppose now
  that~\eqref{eq:ModifiedTwoParticleDensityMatrixReduction} holds true
  and consider
  \begin{equation*}
    \Psi_{k + 1} = \bigwedge_{j = 1}^{k + 1} \zeta_j 
    = \left(\bigwedge_{j = 1}^{k} \zeta_j\right) \wedge \zeta_{k + 1}
    = \Psi_k \wedge \zeta_{k + 1}.
  \end{equation*}
  By~\eqref{eq:TwoParticleDensityMatrixOfTensorProduct}, we get
  \begin{equation*}
    \begin{split}
      \gamma_{\Psi_{k + 1}}^{(2)} &= \gamma_{\Psi_{k}}^{(2)} +
      \gamma_{\zeta_{k + 1}}^{(2)}
      + (\Id - \Ex) \gamma_{\Psi_{k}} \otimes^s \gamma_{\zeta_{k + 1}} \\
      & = \sum_{j = 1}^k \gamma_{\zeta_j}^{(2)} + (\Id - \Ex)
      \left(\sum_{\substack{i < j\\i, j = 1, \hdots, k}}
        \gamma_{\zeta_i} \otimes^s \gamma_{\zeta_j}\right)
      + \gamma_{\zeta_{k + 1}}^{(2)} \\
      &\hskip5cm+ (\Id - \Ex) \left(\sum_{j = 1}^k
        \gamma_{\zeta_{j}}\right) \otimes^s \gamma_{\zeta_{k + 1}}
      \\
      &= \sum_{j = 1}^{k + 1} \gamma_{\zeta_j}^{(2)} + (\Id - \Ex)
      \sum_{\substack{i < j\\i, j = 1, \hdots, k + 1}}
      \gamma_{\zeta_i} \otimes^s \gamma_{\zeta_j}.
    \end{split}
  \end{equation*}
  This completes the proof of
  Proposition~\ref{prop:DensityMatrixStructure}.
\end{proof}
\subsection{The proof of Theorem~\ref{thr:2}}
\label{sec:proof-theorem2}
The proof of Theorem~\ref{thr:2} will rely on Theorem~\ref{thr:4} and
the analysis of $\Psi^U_\omega(L,n)$ performed in
section~\ref{sec:main-results-proofs-1}. The two sums
in~\eqref{eq:160} will be analyzed separately and will be split into
various components according to the lengths of the pieces coming into
play in each component. \\
As in the beginning of section~\ref{sec:ground-state-psiu}
(see~\eqref{eq:146}), write $\D
\Psi^U_\omega(L,n)=\sum_{\substack{Q\text{
      occ. }\\\overline{n}\in\N^m}}
a^Q_{\overline{n}}\Phi_{Q,\overline{n}}$.  We will first transform the
results on the ground state obtained in
section~\ref{sec:main-results-proofs-1} into a statement on the
coefficients $((a^Q_{\overline{n}}))_{Q,\overline{n}}$, namely,
\begin{Pro}
  \label{pro:4}
  There exists $\rho_0>0$ such that, for $\rho\in(0,\rho_0)$ and
  $\varepsilon\in(0,1/10)$, $\omega$ almost surely, in the
  thermodynamic limit, with probability $1-O(L^{-\infty})$, one has
  \begin{enumerate}
  \item for an occupation $Q\not\in\mathcal{Q}_\rho$
    (see~\eqref{eq:142}) and any $\overline{n}\in\N^m$, one has
    $a^Q_{\overline{n}}=0$;
  \item let $\mathcal{P}_-$ be the (indices $j$ of the) pieces
    $(\Delta_j(\omega))_j$ of lengths less than
    $3\ell_\rho(1-\varepsilon)$, and, for $Q$ an occupation, let
    $\mathcal{P}^Q_-$ be the (indices $j$ of the) pieces
    in $\mathcal{P}_-$ such that $Q_j\leq3$.\\
    Then, for $Q$, an occupation number of a ground state
    $\Psi^U_\omega(L,n)$, letting
    $(a_{\overline{n}}^Q)_{Q,\overline{n}}$ be its coefficients in the
    decomposition~\eqref{eq:146}, one has
    \begin{equation}
      \label{eq:173}
      \sum_{\substack{Q\text{ occ.}\\\overline{n}\in\N^m}}
      \#\{j\in\mathcal{P}^Q_-; n_j\geq2\}\left|a_{\overline{n}}^Q\right|^2
      \leq o\left(\frac{n\cdot\rho}{|\log\rho|}\right).
    \end{equation}
  \end{enumerate}
\end{Pro}
\noindent The second part of Proposition~\ref{pro:4} controls the
excited particles in the ground state $\Psi^U_\omega(L,n)$. Actually,
as the proof shows, we shall prove~\eqref{eq:173} not only for a
ground state of $H^U_\omega(L,n)$, but, also for any state $\Psi$
satisfying
\begin{equation}
  \label{eq:270}
  \frac{1}{n} \langle H^{U^p}_\omega(L,n)\Psi,\Psi\rangle \leq
  \densEn^0(\rho) + \pi^2 \gamma_* 
  \frac{\rho}{|\log{\rho}|^{-3}} +o\left(
    \frac{\rho}{|\log{\rho}|^{-3}}\right).
\end{equation}
\begin{proof}[Proof of Proposition~\ref{pro:4}]
  Point (a) is a rephrasing of Corollary~\ref{cor:3}.\\
  Let us prove point (b). Pick an $n$-state $\Psi$ and decompose it as
  $\D\Psi^U_\omega(L,n)=\sum_{Q\in\mathcal{Q}_\rho}\Psi_Q$. Then, if
  $E^{j,U^p}_{Q_j,n_j}$ denotes the $n_j$-th eigenvalue of $\D
  -\sum_{l=1}^{Q_j} \frac{d^2}{dx_l^2}+\sum_{1\leq k\leq l\leq
    Q_j}U^p(x_k-x_l)$ acting on $\D\bigwedge_{l=1}^{Q_j}
  L^2(\Delta_j(\omega))$ with Dirichlet boundary conditions (if
  $Q_j=0$, we set $E^{j,U^p}_{Q_j,n_j}=0$ for all $n_j$), as $H^U\geq
  H^{U^p}$ (see~\eqref{eq:245}), by~\eqref{eq:222}, one has
  \begin{equation}
    \label{eq:224}
    \begin{split}
      n\left(\densEn^0(\rho)+\pi^2\gamma_\star\rho|\log{\rho}|^{-3}
        \left(1+O\left(f_Z(|\log{\rho}|)\right)\right)\right)
      &\geq \langle H^{U^p}\Psi^{U^p},\Psi^{U^p}\rangle\\
      &\geq\sum_{\substack{Q\text{ occ.}\\\overline{n}\in\N^m}}
      \left(\sum_{\substack{j\in\mathcal{P}^Q_-\\Q_j\geq1}}
        E^{j,U^p}_{Q_j,n_j}\right)\left|a_{\overline{n}}^Q\right|^2.
    \end{split}
  \end{equation}
  One proves
  \begin{Le}
    \label{le:29}
    There exists $C>0$ such that, for $j\in\mathcal{P}^Q_-$, $Q_j\geq1$ and
    $n_j\geq2$, one has
    \begin{equation}
      \label{eq:223}
      E^{j,U^p}_{Q_j,n_j}\geq E^{j,U^p}_{Q_j,1}+\frac{1}{C\ell_\rho^2}.
    \end{equation}
  \end{Le}
  \noindent Plugging~\eqref{eq:223} into~\eqref{eq:224} yields
  \begin{multline}
    \label{eq:269}
    \sum_{\substack{Q\text{ occ.}\\\overline{n}\in\N^m}}
    \left(\sum_{\substack{j\in\mathcal{P}^Q_-\\Q_j\geq1}}
      E^{j,U^p}_{Q_j,1}\right)\left|a_{\overline{n}}^Q\right|^2 +
    \sum_{\substack{Q\text{ occ.}\\\overline{n}\in\N^m}}
    \frac{\#\{j\in\mathcal{P}^Q_-;
      n_j\geq2\}}{C\ell^2_\rho}\left|a_{\overline{n}}^Q\right|^2
    \\\leq
    n\left(\densEn^0(\rho)+\pi^2\gamma_\star\rho|\log{\rho}|^{-3}
      \left(1+O\left(f_Z(|\log{\rho}|)\right)\right)\right)
  \end{multline}
  We prove
  \begin{Le}
    \label{le:21}
    There exists $\rho_0>0$ such that, for $\rho\in(0,\rho_0)$,
    $\varepsilon\in(0,1)$ and $\omega$ almost surely, for $L$
    sufficiently large and $|n/L-\rho|$ sufficiently small, if $Q$ is
    an occupation such that
    \begin{equation}
      \label{eq:247}
      \sum_{j\in\mathcal{P}_-}
      E^{j,U^p}_{Q_j,1}\leq
      n\left(\densEn^0(\rho)+\rho|\log{\rho}|^{-3}
        \left(\pi^2\gamma_\star+\varepsilon\right)\right)
    \end{equation}
    then,
    \begin{equation}
      \label{eq:244}
      \sum_{\substack{j\in\mathcal{P}^Q_-\\Q_j\geq1}}
      E^{j,U^p}_{Q_j,1}\geq n\left(\densEn^0(\rho)+
        \rho|\log{\rho}|^{-3}\left(\pi^2\gamma_\star 
          -\frac1{\rho_0}\left(\varepsilon+
            f_Z(|\log{\rho}|)\right)\right)\right).
    \end{equation}
  \end{Le}
  \noindent Lemma~\ref{le:21} shows that, for low energy states, most
  of the energy is carried by pieces carrying three particles and
  less (compare the set $\mathcal{P}_-$ and $\mathcal{P}_-^Q$).\\
  Let us postpone the proof of this result for a while and complete
  the proof of Proposition~\ref{pro:4}. From~\eqref{eq:244}
  and~\eqref{eq:269}, as $\D \sum_{\substack{Q\text{
        occ.}\\\overline{n}\in\N^m}}
  \left|a_{\overline{n}}^Q\right|^2=1$ and $f_Z(|\log{\rho}|=o(1)$, we
  get that
  \begin{equation*}
    \sum_{\substack{Q\text{ occ.}\\\overline{n}\in\N^m}}
    \frac{\#\{j\in\mathcal{P}^Q_-;
      n_j\geq2\}}{C\ell^2_\rho}\left|a_{\overline{n}}^Q\right|^2
    \leq o\left(n\rho|\log{\rho}|^{-3}\right).
  \end{equation*}
  As $\ell_\rho\asymp|\log\rho|$, this immediately
  yields~\eqref{eq:173} and completes the proof of
  Proposition~\ref{pro:4}.
\end{proof}
\begin{proof}[Proof of Lemma~\ref{le:21}]
  By Theorem~\ref{thr:6}, for $L$ large and $n/L$ close to $\rho$, we
  have
  \begin{equation*}
    \left\langle H_\omega^{U^p}\Psi^{\textup{opt}},
      \Psi^{\textup{opt}}\right\rangle\geq
    n\left(\densEn^0(\rho)+\pi^2\gamma_\star\rho|\log{\rho}|^{-3} 
      \left(1+O\left(f_Z(|\log{\rho}|)\right)\right)\right).
  \end{equation*}
  Recall that the occupation $Q^{\text{opt}}$ of $\Psi^{\text{opt}}$
  satisfies
  \begin{equation}
    \label{eq:249}
    Q_j^{\text{opt}}=
    \begin{cases}
      0\text{ if }|\Delta_j(\omega)|\in[0,\ell_\rho-\rho x_*),\\
      1\text{ if }|\Delta_j(\omega)|\in[\ell_\rho-\rho x_*,2\ell_\rho+A_*),\\
      2\text{ if }|\Delta_j(\omega)|\in[2\ell_\rho+A_*,
      3\ell_\rho(1-\varepsilon)).
    \end{cases}
  \end{equation}
  Theorem~\ref{thr:6} shows that
  \begin{equation}
    \label{eq:260}
    \left|\left\langle H_\omega^{U^p}\Psi^{\textup{opt}},
        \Psi^{\textup{opt}}\right\rangle-
      \sum_{\substack{j\in\mathcal{P}_-\\Q^{\text{opt}}_j=1}}E^{j,U^p}_{1,1}- 
      \sum_{\substack{j\in\mathcal{P}_-\\Q^{\text{opt}}_j=2}}E^{j,U^p}_{2,1}\right|
    \lesssim n\frac{\rho}{|\log{\rho}|^{3}}
    f_Z(|\log{\rho}|).
  \end{equation}
  Let
  \begin{equation}
    \label{eq:255}
    \Delta E:= \sum_{j\in\mathcal{P}_-}
    E^{j,U^p}_{Q_j,1}-
    \sum_{\substack{j\in\mathcal{P}_-\\Q^{\text{opt}}_j=1}}E^{j,U^p}_{1,1}- 
    \sum_{\substack{j\in\mathcal{P}_-\\Q^{\text{opt}}_j=2}}E^{j,U^p}_{2,1}.
  \end{equation}
  Then, \eqref{eq:260} and assumption~\eqref{eq:247} imply that
  \begin{equation}
    \label{eq:256}
    |\Delta E|\leq    \frac{C\,n\,\rho}{|\log{\rho}|^{3}}
    \left(f_Z(|\log{\rho}|)+\varepsilon\right).
  \end{equation}
  Moreover, one has
  \begin{equation}
    \label{eq:257}
    \begin{split}
      \Delta E&\geq
      \sum_{\substack{j\in\mathcal{P}_-\\Q^{\text{opt}}_j=0}}E^{j,U^p}_{Q_j,1}
      +\sum_{\substack{j\in\mathcal{P}_-\\Q^{\text{opt}}_j=1}}
      (E^{j,U^p}_{Q_j,1}-E^{j,U^p}_{1,1}) +
      \sum_{\substack{j\in\mathcal{P}_-\\Q^{\text{opt}}_j=2}}
      (E^{j,U^p}_{Q_j,1}-E^{j,U^p}_{2,1})
      \\&=\sum_{\substack{j\in\mathcal{P}_-\\Q^{\text{opt}}_j=0\\Q_j\geq1}}
      E^{j,U^p}_{Q_j,1}
      +\sum_{\substack{j\in\mathcal{P}_-\\Q^{\text{opt}}_j=1\\Q_j\geq2}}
      (E^{j,U^p}_{Q_j,1}-E^{j,U^p}_{1,1}) +
      \sum_{\substack{j\in\mathcal{P}_-\\Q^{\text{opt}}_j=2
          \\Q_j\geq3}}(E^{j,U^p}_{Q_j,1}-E^{j,U^p}_{2,1}) \\&\hskip5cm
      -\sum_{\substack{j\in\mathcal{P}_-\\Q^{\text{opt}}_j=1\\Q_j=0}}E^{j,U^p}_{1,1}
      -\sum_{\substack{j\in\mathcal{P}_-\\Q^{\text{opt}}_j=2\\Q_j\leq1}}
      (E^{j,U^p}_{2,1}-E^{j,U^p}_{Q_j,1}).
    \end{split}
  \end{equation}
  On the other hand, as $|Q|=n=|Q^{\text{opt}}|$, using
  Lemma~\ref{le:17} as $\Psi^{U}_\omega(L,n)$
  satisfies~\eqref{eq:270}, we know that
  \begin{equation}
    \label{eq:258}
    \sum_{\substack{j\in\mathcal{P}_-\\Q^{\text{opt}}_j=1\\Q_j=0}}1+
    \sum_{\substack{j\in\mathcal{P}_-\\Q^{\text{opt}}_j=2\\Q_j\leq1}}(2-Q_j)
    =\sum_{\substack{j\in\mathcal{P}_-\\Q^{\text{opt}}_j=0\\Q_j\geq1}}Q_j+
    \sum_{\substack{j\in\mathcal{P}_-\\Q^{\text{opt}}_j=1\\Q_j\geq2}}(Q_j-1)+
    \sum_{\substack{j\in\mathcal{P}_-\\Q^{\text{opt}}_j=2\\Q_j\geq3}}(Q_j-2)
    +O(n\rho^{1+\eta}).
  \end{equation}
  Define
  \begin{equation*}
    B:=\max\left(\max_{\substack{j;\ Q_j=0\\Q^{\text{opt}}_j=1}}
      E^{j,U^p}_{1,1},\max_{\substack{j;\ Q^{\text{opt}}_j=2\\0\leq
          Q_j\leq1}}\frac{E^{j,U^p}_{2,1}-E^{j,U^p}_{Q_j,1}}{2-Q_j}\right).
  \end{equation*}
  Then,~\eqref{eq:257} implies that
  \begin{equation*}
    \begin{split}
      \Delta E&\geq
      \sum_{\substack{j\in\mathcal{P}_-\\Q^{\text{opt}}_j=0\\Q_j\geq1}}E^{j,U^p}_{Q_j,1}
      +\sum_{\substack{j\in\mathcal{P}_-\\Q^{\text{opt}}_j=1\\Q_j\geq2}}
      (E^{j,U^p}_{Q_j,1}-E^{j,U^p}_{1,1}) +
      \sum_{\substack{j\in\mathcal{P}_-\\Q^{\text{opt}}_j=2
          \\Q_j\geq3}}(E^{j,U^p}_{Q_j,1}-E^{j,U^p}_{2,1}) \\&\hskip7cm
      -B\sum_{\substack{j\in\mathcal{P}_-\\Q^{\text{opt}}_j=1\\Q_j=0}}1
      -B
      \sum_{\substack{j\in\mathcal{P}_-\\Q^{\text{opt}}_j=2\\Q_j\leq1}}
      (2-Q_j).
    \end{split}
  \end{equation*}
  Hence,~\eqref{eq:258} implies that, for some $C>0$, for $\rho$
  sufficiently small, one has
  \begin{equation}
    \label{eq:252}
    \begin{split}
      \Delta E+C\,n\,\rho^{1+\eta}&\geq
      \sum_{\substack{j\in\mathcal{P}_-\\Q^{\text{opt}}_j=0\\Q_j\geq1}}
      (E^{j,U^p}_{Q_j,1}-B)
      +\sum_{\substack{j\in\mathcal{P}_-\\Q^{\text{opt}}_j=1\\Q_j\geq2}}
      (E^{j,U^p}_{Q_j,1}-E^{j,U^p}_{1,1}-B(Q_j-1)) \\&\hskip4cm+
      \sum_{\substack{j\in\mathcal{P}_-\\Q^{\text{opt}}_j=2
          \\Q_j\geq3}}(E^{j,U^p}_{Q_j,1}-E^{j,U^p}_{2,1}-B(Q_j-2)).
    \end{split}
  \end{equation}
  Let us upper bound $B$. Recalling that for a single particle in a
  piece there is no interaction, a direct computation
  and~\eqref{eq:249} show that
  \begin{equation}
    \label{eq:248}
    \max_{\substack{j;\ Q_j=0\\Q^{\text{opt}}_j=1}}
    E^{j,U^p}_{1,1}\leq \frac{\pi^2}{(\ell_\rho-\rho x_*)^2}.
  \end{equation}
  Proposition~\ref{prop:TwoElectronProblem} and~\eqref{eq:249} show
  that, for $\rho$ sufficiently small, one has
  \begin{gather*}
    \max_{\substack{j;\ Q_j=0\\Q^{\text{opt}}_j=2}}
    \frac{E^{j,U^p}_{2,1}-E^{j,U^p}_{Q_j,1}}{2-Q_j}\leq
    \frac{5\pi^2}{2(2\ell_\rho+A^*)^2}+\frac{2\gamma}{(2\ell_\rho+A^*)^3}
    \leq \frac{\pi^2}{(\ell_\rho-\rho x_*)^2}\\
    \max_{\substack{j;\ Q_j=1\\Q^{\text{opt}}_j=2}}
    \frac{E^{j,U^p}_{2,1}-E^{j,U^p}_{Q_j,1}}{2-Q_j}\leq
    \frac{4\pi^2}{(2\ell_\rho+A^*)^2}+\frac{2\gamma}{(2\ell_\rho+A^*)^3}\leq
    \frac{\pi^2}{(\ell_\rho-\rho x_*)^2}.
  \end{gather*}
  Thus,
  \begin{equation}
    \label{eq:250}
    B\leq \frac{\pi^2}{(\ell_\rho-\rho x_*)^2}.
  \end{equation}
  Now, notice that
  \begin{itemize}
  \item for $j$ s.t. $Q^{\text{opt}}_j=0$ (see~\eqref{eq:249}):
    \begin{itemize}
    \item if $Q_j=1$, one has
      \begin{equation*}
        E^{j,U^p}_{Q_j,1}-\frac{\pi^2}{(\ell_\rho-\rho x_*)^2}\geq
        \frac{\pi^2}{|\Delta_j(\omega)|^2}-\frac{\pi^2}{(\ell_\rho-\rho x_*)^2}\geq0;
      \end{equation*}
    \item if $Q_j\geq2$, one has
      \begin{equation*}
        E^{j,U^p}_{Q_j,1}-\frac{\pi^2}{(\ell_\rho-\rho x_*)^2}\geq
        \frac1{2}E^{j,U^p}_{Q_j,1}+\frac{5\pi^2}{2|\Delta_j(\omega)|^2}-
        \frac{\pi^2}{(\ell_\rho-\rho 
          x_*)^2}\geq \frac1{2}E^{j,U^p}_{Q_j,1};
      \end{equation*}
    \end{itemize}
  \item for $j$ s.t. $Q^{\text{opt}}_j=1$ (see~\eqref{eq:249}):
    \begin{itemize}
    \item if $Q_j=2$, one has
      \begin{equation*}
        \begin{split}
          E^{j,U^p}_{Q_j,1}-E^{j,U^p}_{1,1}-\frac{\pi^2}{(\ell_\rho-\rho
            x_*)^2}&\geq
          \frac{4\pi^2}{|\Delta_j(\omega)|^2}+\frac{\gamma}{|\Delta_j(\omega)|^3}
          +o(\ell_\rho^{-3})-\frac{\pi^2}{(\ell_\rho-\rho
            x_*)^2}\\&\geq
          \frac{4\pi^2}{|2\ell_\rho+A_*+\varepsilon_\rho|^2}+
          \frac{\gamma}{|2\ell_\rho+A_*+\varepsilon_\rho|^^3}
          +o(\ell_\rho^{-3})-\frac{\pi^2}{(\ell_\rho-\rho
            x_*)^2}\\&\geq
          \frac{\pi^2}{\ell_\rho^2}-\frac{A_*\pi^2}{2\ell^3_\rho}+
          \frac{\gamma}{4\ell^3_\rho}+\frac{\pi^2\varepsilon_\rho}{2\ell^3_\rho}
          +o(\ell_\rho^{-3})-\frac{\pi^2}{(\ell_\rho-\rho
            x_*)^2}\\&\geq \frac{\pi^2\varepsilon_\rho}{2\ell^3_\rho}
          +o(\ell_\rho^{-3})\geq0
        \end{split}
      \end{equation*}
      if $\rho$ sufficiently small (see~\eqref{eq:AstarXstarDef}) and
      $|\Delta_j(\omega)|\leq2\ell_\rho+A_*-\varepsilon_\rho$; here
      $\varepsilon_\rho\to0^+$ (but not too fast) as $\rho\to0^+$;\\
      on the other hand, the number of pieces of length in
      $2\ell_\rho+A_*+[-\varepsilon_\rho,0]$ is bounded by $C\rho
      n\varepsilon_\rho$ (see Proposition~\ref{prop:IntervStatistics})
      and for such pieces, one has
      \begin{equation}
        \label{eq:254}
        \left|E^{j,U^p}_{2,1}-E^{j,U^p}_{1,1}-\frac{\pi^2}{(\ell_\rho-\rho
            x_*)^2}\right| =o(\ell_\rho^{-3});
      \end{equation}
    \item if $Q_j\geq3$, one has
      \begin{equation*}
        \begin{split}
          E^{j,U^p}_{Q_j,1}-E^{j,U^p}_{1,1}-\frac{\pi^2}{(\ell_\rho-\rho
            x_*)^2}(Q_j-1)&\geq \frac12E^{j,U^p}_{Q_j,1}+
          \frac12E^{j,0}_{Q_j,1}-\frac{\pi^2}{(\ell_\rho-\rho
            x_*)^2}(Q_j-1)\\&\geq
          \frac12E^{j,U^p}_{Q_j,1}+\frac{\pi^2}{4\ell_\rho^2}\frac5{12}
          (Q_j-1)\geq \frac12E^{j,U^p}_{Q_j,1}
        \end{split}
      \end{equation*} 
    \end{itemize}
  \item for $j$ s.t. $Q^{\text{opt}}_j=2$ (see~\eqref{eq:249}):
    \begin{itemize}
    \item if $Q_j\geq3$, one has
    \begin{equation*}
      \begin{split}
        E^{j,U^p}_{Q_j,1}-E^{j,U^p}_{2,1}-\frac{\pi^2}{(\ell_\rho-\rho
          x_*)^2}(Q_j-2)&\geq \frac13 E^{j,U^p}_{Q_j,1}+ \frac23
        E^{j,0}_{Q_j,1}-\frac{\pi^2}{(\ell_\rho-\rho
          x_*)^2}(Q_j-2)\\&\geq \frac13
        E^{j,U^p}_{Q_j,1}+\frac{\pi^2}{9(1-\varepsilon)^2\ell_\rho^2}
        \left(\frac{102}{9}-9\right) (Q_j-2)\\&\geq \frac13
        E^{j,U^p}_{Q_j,1}.
      \end{split}
    \end{equation*} 
    \end{itemize}
  \end{itemize}
  Plugging these estimates and~\eqref{eq:250} into~\eqref{eq:252}, we
  get that, for $\rho$ sufficiently small,
  \begin{multline*}
    \Delta E+
    \sum_{|\Delta_j(\omega)|\in2\ell_\rho+A_*+[-\varepsilon_\rho,0]}
    \left|E^{j,U^p}_{2,1}-E^{j,U^p}_{1,1}-\frac{\pi^2}{(\ell_\rho-\rho
        x_*)^2}\right|+C\,n\,\rho^{1+\eta}\\\geq
    \frac12\sum_{\substack{j\in\mathcal{P}_-\\Q^{\text{opt}}_j=0\\Q_j\geq2}}
    E^{j,U^p}_{Q_j,1}
    +\frac12\sum_{\substack{j\in\mathcal{P}_-\\Q^{\text{opt}}_j=1\\Q_j\geq3}}
    E^{j,U^p}_{Q_j,1} +
    \frac13\sum_{\substack{j\in\mathcal{P}_-\\Q^{\text{opt}}_j=2
        \\Q_j\geq3}}E^{j,U^p}_{Q_j,1}.
  \end{multline*}
  Hence, in view of~\eqref{eq:254} and the estimate on the number of
  terms in the sum in the left hand side, one gets
  \begin{equation}
    \label{eq:253}
    3\left(\Delta E+ o\left(n\rho\ell_\rho^{-3}\right)\right)\geq
    \sum_{\substack{j\in\mathcal{P}_-\\Q^{\text{opt}}_j=0\\Q_j\geq2}}
    E^{j,U^p}_{Q_j,1}
    +\sum_{\substack{j\in\mathcal{P}_-\\Q^{\text{opt}}_j=1\\Q_j\geq3}}
    E^{j,U^p}_{Q_j,1} +
    \sum_{\substack{j\in\mathcal{P}_-\\Q^{\text{opt}}_j=2
        \\Q_j\geq3}}E^{j,U^p}_{Q_j,1}\geq0.
  \end{equation}
  This implies that
  \begin{equation*}
    o\left(n\rho\ell_\rho^{-3}\right)\leq\Delta E=\sum_{j\in\mathcal{P}_-}
    E^{j,U^p}_{Q_j,1}-
    \sum_{\substack{j\in\mathcal{P}_-\\Q^{\text{opt}}_j=1}}E^{j,U^p}_{1,1}- 
    \sum_{\substack{j\in\mathcal{P}_-\\Q^{\text{opt}}_j=2}}E^{j,U^p}_{2,1}
  \end{equation*}
  hence, by~\eqref{eq:260}, that, for some $C>0$ and $\rho$
  sufficiently small, one has
  \begin{equation}
    \label{eq:259}
    \sum_{j\in\mathcal{P}_-}
    E^{j,U^p}_{Q_j,1}\geq
    n\left(\densEn^0(\rho)+\pi^2\gamma_\star\rho|\log{\rho}|^{-3} 
      \left(1-C\, f_Z(|\log{\rho}|)\right)\right)
  \end{equation}
  We complete the proof of Lemma~\ref{le:21} by noting that, by the
  definition of $\mathcal{P}_-^Q$, one has
  \begin{equation*}
    \begin{split}
      \sum_{\substack{j\in\mathcal{P}^Q_-\\Q_j\geq1}}
      E^{j,U^p}_{Q_j,1}&= \sum_{j\in\mathcal{P}_-}E^{j,U^p}_{Q_j,1}-
      \left(\sum_{\substack{j\in\mathcal{P}_-\\Q^{\text{opt}}_j=0\\Q_j\geq3}}
        E^{j,U^p}_{Q_j,1}
        +\sum_{\substack{j\in\mathcal{P}_-\\Q^{\text{opt}}_j=1\\Q_j\geq3}}
        E^{j,U^p}_{Q_j,1} +
        \sum_{\substack{j\in\mathcal{P}_-\\Q^{\text{opt}}_j=2
            \\Q_j\geq3}}E^{j,U^p}_{Q_j,1}\right)\\&\geq
      n\left(\densEn^0(\rho)+\pi^2\gamma_\star\rho|\log{\rho}|^{-3}
        \left(1- C(\varepsilon+f_Z(|\log{\rho}|))\right)\right)
    \end{split}
  \end{equation*}
  where the last lower bound follows from~\eqref{eq:256} and~\eqref{eq:253}.\\
  This completes the proof of Lemma~\ref{le:21}.
\end{proof}
\noindent Let us resume the proof of Theorem~\ref{thr:2}. Recall
Theorem~\ref{thr:4}; we analyze the two components
$\gamma_{\Psi^U_\omega(L,n)}^{(1),d}$ and $\gamma_{\Psi^U_\omega(L,n)}^{(1),o}$ separately.\\
Let us start with the analysis of  $\gamma_{\Psi^U_\omega(L,n)}^{(1),o}$. We prove
\begin{Le}
  \label{le:19}
  Under the assumptions of Theorem~\ref{thr:4}, in the thermodynamic
  limit, with probability $1-O(L^{-\infty})$, one has
  \begin{equation}
    \label{eq:175}
    \left\|\gamma_{\Psi^U_\omega(L,n)}^{(1),o}\right\|_{\text{tr}}\leq 3.
  \end{equation}
\end{Le}
\begin{proof}
  We recall~\eqref{eq:174} from Theorem~\ref{thr:4} and write
  \begin{equation*}
    \gamma_{\Psi^U_\omega(L,n)}^{(1),o}=\sum_{Q\text{ occ.}}\sum_{\substack{i,j=1\\ i
        \not=j\\ Q_j\geq1}}^m C_1(Q,i,j)\sum_{\tilde
      n\in\N^{m-1}}\sum_{\substack{n_i,n_j\geq1\\ n'_i,n'_j\geq1}} a^Q_{\tilde
      n_{i,j}} \overline{a^{Q'}_{\tilde n'_{i,j}}}
    \gamma^{(1)}_{\substack{Q_i,Q_j\\n_i,n_j\\n'_i,n'_j}}
  \end{equation*}
  where, by definition, in the above sums, $Q'$ satisfies $Q'_k=Q_k$
  if $k\not\in\{i,j\}$, $Q'_i=Q_i+1$
  and $Q'_j=Q_j-1$.\\
  Note that, by point (a) of Proposition~\ref{pro:4}, here and in the sequel
  when summing over the occupations $Q$, we can always restrict
  ourselves to the occupations in $\mathcal{Q}_\rho$. \\
  Decompose
  \begin{equation}
    \label{eq:176}
    \gamma_{\Psi^U_\omega(L,n)}^{(1),o}=
    \gamma_{\Psi^U_\omega(L,n)}^{(1),o,+,+}+
    \gamma_{\Psi^U_\omega(L,n)}^{(1),o,+,-}+
    \gamma_{\Psi^U_\omega(L,n)}^{(1),o,-,+}+
    \gamma_{\Psi^U_\omega(L,n)}^{(1),o,-,-}
  \end{equation}
  where (see~\eqref{eq:16},~\eqref{eq:145},~\eqref{eq:155}
  and~\eqref{eq:156})
  \begin{gather*}
    \gamma_{\Psi^U_\omega(L,n)}^{(1),o,+,+}:=\sum_{\substack{Q\text{ occ.}\\\tilde
        n\in\N^{m-1}\\i \not=j\\ Q_j\geq2\\
        Q_i\geq1}}C_1(Q,i,j) a^Q_{\tilde n_{i,j}} \overline{a^{Q'}_{\tilde
        n'_{i,j}}}\gamma^{(1),1,+,+}_{Q,Q',i,j,\tilde n},\
    \gamma_{\Psi^U_\omega(L,n)}^{(1),o,+,-}:=\sum_{\substack{Q\text{
          occ.}\\\tilde n\in\N^{m-1}\\i\not=j\\
        Q_j\geq2\\ Q_i=0}} C_1(Q,i,j) a^Q_{\tilde n_{i,j}}
    \overline{a^{Q'}_{\tilde
        n'_{i,j}}}\gamma^{(1),1,+,-}_{Q,Q',i,j,\tilde n} \\
    \gamma_{\Psi^U_\omega(L,n)}^{(1),o,-,+}:=\sum_{\substack{Q\text{ occ.}\\\tilde
        n\in\N^{m-1}\\i\not=j\\Q_j=1\\Q_i\geq1}} C_1(Q,i,j) a^Q_{\tilde
      n_{i,j}} \overline{a^{Q'}_{\tilde
        n'_{i,j}}}\gamma^{(1),1,-,+}_{Q,Q',i,j,\tilde n}, \
    \gamma_{\Psi^U_\omega(L,n)}^{(1),o,-,-}:=\sum_{\substack{Q\text{ occ.}\\\tilde
        n\in\N^{m-1}\\i\not=j\\Q_j=1\\Q_i=0}} C_1(Q,i,j) a^Q_{\tilde n_{i,j}}
    \overline{a^{Q'}_{\tilde n'_{i,j}}}\gamma^{(1),1,-,-}_{Q,Q',i,j,\tilde n}
  \end{gather*}
  and
  \begin{gather*}
    \gamma^{(1),1,+,+}_{Q,Q',i,j,\tilde
      n}(x,y):=\int_{\Delta^{Q_i}_i\times \Delta^{Q_j-1}_j}
    \left(\sum_{\substack{n_i\geq1\\ n_j\geq1}} a^Q_{\tilde
        n_{i,j}}\varphi^i_{Q_i,n_i}(z)\varphi^j_{Q_j,n_j}(x,z')
    \right) \left(\overline{\sum_{\substack{n'_j\geq1\\
            n'_j\geq1}} a^{Q'}_{\tilde
          n'_{i,j}}\varphi^i_{Q'_i,n'_i}(y,z)\varphi^j_{Q'_j,n'_j}(z')}
    \right)dz dz',
    \\
    \gamma^{(1),1,+,-}_{Q,Q',i,j,\tilde n}(x,y):= \int_{\Delta^{Q_j-1}_j}
    \left(\sum_{\substack{n_i=1\\ n_j\geq1}} a^Q_{\tilde
        n_{i,j}}\varphi^j_{Q_j,n_j}(x,z')\right)
    \left(\overline{\sum_{\substack{n'_i\geq1\\ n'_j\geq1}} a^Q_{\tilde
          n'_{i,j}}\varphi^i_{1,n'_i}(y)
        \varphi^j_{Q_j-1,n'_j}(z')}\right)dz',\\
    \gamma^{(1),1,-,+}_{Q,Q',i,j,\tilde n}(x,y):= \int_{\Delta^{Q_i}_i}
    \left(\sum_{\substack{n_i\geq1\\ n_j\geq1}} a^Q_{\tilde
        n_{i,j}}\varphi^j_{1,n_j}(x)\varphi^i_{Q_i,n_i}(z)\right)
    \left(\overline{\sum_{\substack{n'_j=1\\ n'_i\geq1}} a^Q_{\tilde n'_{i,j}}
        \varphi^i_{Q_i+1,n'_i}(y,z)}\right) dz,\\
    \text{and}\quad
    \gamma^{(1),1,-,-}_{Q,Q',i,j,\tilde n}(x,y):=\left(\sum_{\substack{n_i=1\\
          n_j\geq1}} a^Q_{\tilde n_{i,j}}\varphi^i_{1,n_j}(x)\right)
    \left(\overline{\sum_{\substack{n'_j=1\\ n'_i\geq1}}
        a^{Q'}_{\tilde n'_{i,j}}\varphi^i_{1,n'_i}(y)}\right).
  \end{gather*}
  Let us first analyze $\gamma_{\Psi^U_\omega(L,n)}^{(1),o,+,+}$. By
  Lemma~\ref{le:15}, using the orthonormality of the families
  $(\varphi^j_{Q_j,n_j})_{n_j\in\N}$ (see the beginning of
  section~\ref{sec:ground-state-psiu}), we know that
  \begin{equation*}
    \begin{split}
      \left\|\gamma^{(1),1,+,+}_{Q,Q',i,j,\tilde
          n}\right\|_{\text{tr}} &\leq \left\|\sum_{n_i,n_j}
        a^Q_{\tilde
          n_{i,j}}\varphi^i_{Q_i,n_i}\otimes\varphi^j_{Q_j,n_j}\right\|
      \cdot\left\|\sum_{n'_i,n'_j} a^{Q'}_{\tilde n'_{i,j}}
        \varphi^i_{Q'_i,n'_i}\otimes\varphi^j_{Q'_j,n'_j}\right\|
      \\&\leq \frac12 \left(\sum_{n_i,n_j} \left|a^Q_{\tilde
            n_{i,j}}\right|^2+\sum_{n_i,n_j} \left|a^{Q'}_{\tilde
            n_{i,j}}\right|^2 \right).
    \end{split}
  \end{equation*}
  Hence, by definition (see the formula following~\eqref{eq:176}) and the
  symmetry of $C_1(Q,i,j)$ in $i$ and $j$, we have
  \begin{equation*}
    \left\|\gamma_{\Psi^U_\omega(L,n)}^{(1),o,+,+} \right\|_{\text{tr}}\leq
    \sum_{\substack{i,j=1\\i\not=j}}^m
    \sum_{\substack{Q\text{ occ.}\\Q_j\geq2\\Q_i\geq1}}
    C_1(Q,i,j)\sum_{\overline{n}\in\N^m}\left|a^Q_{\overline{n}}\right|^2.
  \end{equation*}
  Now, by definition (see Theorem~\ref{thr:4}), for $Q_j\geq2$ and
  $Q_i\geq 1$, one has
  \begin{equation*}
    C_1(Q,i,j)\leq \frac{Q_i\, Q_j}{(n-1)(n-2)(n-3)}.
  \end{equation*}
  Thus,
  \begin{equation}
    \label{eq:181}
    \begin{split}
      \left\|\gamma_{\Psi^U_\omega(L,n)}^{(1),o,+,+} \right\|_{\text{tr}}&\leq
      \frac1{(n-1)(n-2)(n-3)}\sum_{\substack{Q\text{ occ.
          }\\Q_j\geq2\\Q_i\geq1}}\left(\sum_j Q_j\right)^2
      \sum_{\overline{n}\in\N^m}\left|a^Q_{\overline{n}}\right|^2\\&\leq
      \frac{n^2}{(n-1)(n-2)(n-3)}\sum_{Q,\ \overline{n}\in\N^m}
      \left|a^Q_{\overline{n}}\right|^2 =\frac{n^2}{(n-1)(n-2)(n-3)}.
    \end{split}
  \end{equation}
  Let us now analyze $\gamma_{\Psi^U_\omega(L,n)}^{(1),o,-,-}$. By the definition
  of $C_1(Q,i,j)$, we write
  \begin{equation*}
    \gamma_{\Psi^U_\omega(L,n)}^{(1),o,-,-}(x,y)=\frac1{n-1}\sum_{\substack{\tilde
        n\in\N^{m-1}\\Q\text{ occ.}}}
    \sum_{\substack{i,j=1\\ i \not=j\\ Q_j=1\\ Q_i=0}}^m
    \left(\sum_{n_i=1,n_j} a^Q_{\tilde
        n_{i,j}}\varphi^j_{1,n_j}(x)\right)
    \left(\overline{\sum_{n'_j=1,n'_i} a^{Q'}_{\tilde
          n'_{i,j}}\varphi^i_{1,n'_i}(y)}\right).
  \end{equation*}
  Thus, by Lemma~\ref{le:15}, one has
  \begin{equation}
    \label{eq:178}
    \begin{split}
      \left\|\gamma_{\Psi^U_\omega(L,n)}^{(1),o,-,-}\right\|_{\text{tr}}&\leq
      \frac1{n-1}\sum_{\substack{\tilde n\in\N^{m-1}\\Q\text{ occ.}}}
      \left\|\sum_{j,\ Q_j=1}\sum_{n_i=1, n_j} a^Q_{\tilde
          n_{i,j}}\varphi^j_{1,n_j}\right\|\cdot \left\|\sum_{i,\
          Q_i=0}^m\sum_{n'_j=1,n'_i} a^{Q'}_{\tilde
          n'_{i,j}}\varphi^i_{1,n'_i}\right\|\\&\leq
      \frac1{2n-2}\sum_{\substack{\tilde n\in\N^{m-1}\\Q\text{
            occ.}}} \left\|\sum_{j,\ Q_j=1}\sum_{n_i=1, n_j}
        a^Q_{\tilde n_{i,j}}\varphi^j_{1,n_j}\right\|^2+
      \left\|\sum_{i,\ Q_i=0}\sum_{n'_j=1,n'_i} a^{Q'}_{\tilde
          n'_{i,j}}\varphi^i_{1,n'_i}\right\|^2\\&\leq
      \frac1{2n-2}\sum_{\substack{n\in\N^m\\Q\text{ occ.}}}
      \left|a^Q_{\tilde n}\right|^2=\frac1{2n-2}.
    \end{split}
  \end{equation}
  Let us now analyze $\gamma_{\Psi^U_\omega(L,n)}^{(1),o,+,-}$. One has
  \begin{equation*}
    \begin{split}
      \gamma_{\Psi^U_\omega(L,n)}^{(1),o,+,-}&=\sum_{\substack{\tilde
          n\in\N^{m-1}\\Q\text{ occ.}}}\sum_{\substack{i \not=j\\
          Q_j\geq2\\ Q_i=0}}\frac{(n-Q_j-1)!Q_j!}{(n-1)!}
      \int_{\Delta^{Q_j-1}_j} \left(\sum_{n_i=1,n_j} a^Q_{\tilde
          n_{i,j}}\varphi^j_{Q_j,n_j}(x,z')\right)\times\\&\hskip7cm
      \left(\overline{\sum_{n'_i,n'_j} a^Q_{\tilde
            n'_{i,j}}\varphi^i_{1,n'_i}(y)
          \varphi^j_{Q_j-1,n'_j}(z')}\right)dz'\\
      &=\sum_{\substack{\tilde n\in\N^{m-1}\\Q\text{ occ.}}}\sum_{j;\
        Q_j\geq2}\frac{(n-Q_j-1)!Q_j!}{(n-1)!}  \int_{\Delta^{Q_j-1}_j}
      \left(\sum_{n_i=1,n_j} a^Q_{\tilde
          n_{i,j}}\varphi^j_{Q_j,n_j}(x,z')\right)\times\\&\hskip7cm
      \left(\overline{\sum_{i;\ Q_i=0}\sum_{n'_i,n'_j} a^Q_{\tilde
            n'_{i,j}}\varphi^i_{1,n'_i}(y)
          \varphi^j_{Q_j-1,n'_j}(z')}\right)dz'.
    \end{split}
  \end{equation*}
  Thus, using Lemma~\ref{le:15} and the orthonormality properties of
  the families $(\varphi^j_{Q_j,n_j})_{n_j\in\N}$, as
  $(n-Q_j)!\,Q_j!\leq n!$ and $\sum_j Q_j=n$, we get
  \begin{equation}
    \label{eq:177}
    \left\|\gamma_{\Psi^U_\omega(L,n)}^{(1),o,+,-}\right\|_{\text{tr}}
    \leq\frac1{n-1} \sum_{\substack{\tilde n\in\N^{m-1}\\Q\text{
          occ.}}}\sum_{j=1}^m Q_j\sum_{n_i=1,n_j} \left| a^Q_{\tilde
        n_{i,j}} \right|^2\leq \frac{n}{n-1}\sum_{\substack{
        \overline{n}\in\N^m\\ Q\text{ occ.}}}\left|
      a^Q_{\overline{n}} \right|^2.
  \end{equation}
  The term $\gamma_{\Psi^U_\omega(L,n)}^{(1),o,-,+}$ is analyzed in the same
  way. Gathering~\eqref{eq:181},~\eqref{eq:178},~\eqref{eq:177} and
  using~\eqref{eq:176}, we obtain~\eqref{eq:175} and, thus, complete
  the proof of Lemma~\ref{le:19}.
\end{proof}
\noindent Let us now turn to the analysis of $\gamma_{\Psi^U_\omega(L,n)}^{(1),d}$.
Therefore, we write
\begin{equation}
  \label{eq:179}
  \gamma_{\Psi^U_\omega(L,n)}^{(1),d}=\gamma_{\Psi^U_\omega(L,n)}^{(1),d,-}+\gamma_{\Psi^U_\omega(L,n)}^{(1),d,+}
  \text{ where }\gamma_{\Psi^U_\omega(L,n)}^{(1),d,-}:=\sum_{\substack{Q\text{
        occ.}\\\tilde n\in\N^{m-1}}}\sum_{j\in\mathcal{P}^Q_-}
  \sum_{\substack{n_j\geq1\\ n'_j\geq1}}a^Q_{\tilde
    n_j} \overline{a^Q_{\tilde n'_j}} \gamma^{(1)}_{\substack{Q_j\\n_j,n'_j}}.
\end{equation}
We prove
\begin{Le}
  \label{le:20}
  Under the assumptions of Theorem~\ref{thr:4}, for $\eta\in(0,1)$,
  there exists $\varepsilon_0>0$ and $C>1$ such that, for
  $\varepsilon\in(0,\varepsilon_0)$, in the thermodynamic limit, with
  probability $1-O(L^{-\infty})$, one has
  \begin{equation}
    \label{eq:180}
    \left\|\gamma_{\Psi^U_\omega(L,n)}^{(1),d,+}\right\|_{\text{tr}}\leq
    Cn\frac{\rho}{\ell_\rho}.
  \end{equation}  
\end{Le}
\begin{proof}
  Define
  \begin{equation}
    \label{eq:234}
    \gamma_{\Psi^U_\omega(L,n)}^{(1),d,+}=\gamma_{\Psi^U_\omega(L,n)}^{(1),d,+,+}
    +\gamma_{\Psi^U_\omega(L,n)}^{(1),d,+,0}
  \end{equation}
  where
  \begin{gather*}
    \gamma_{\Psi^U_\omega(L,n)}^{(1),d,+,+}= \sum_{\substack{Q\text{ occ.}\\\tilde
        n\in\N^{m-1}}}\sum_{j\not\in\mathcal{P}_-}\sum_{\substack{n_j\geq1\\
        n'_j\geq1}}a^Q_{\tilde n_j} \overline{a^Q_{\tilde
        n'_j}}\gamma^{(1)}_{\substack{Q_j\\n_j,n'_j}}\quad\text{and}\quad
    \gamma_{\Psi^U_\omega(L,n)}^{(1),d,+,0}= \sum_{\substack{Q\text{ occ.}\\\tilde
        n\in\N^{m-1}}} \sum_{\substack{j\in\mathcal{P}^Q_-\\Q_j\geq4}}
    \sum_{\substack{n_j\geq1\\ n'_j\geq1}}a^Q_{\tilde n_j}
    \overline{a^Q_{\tilde n'_j}}\gamma^{(1)}_{\substack{Q_j\\n_j,n'_j}}.
  \end{gather*}
  One computes
  \begin{equation*}
    \begin{split}
      \gamma_{\Psi^U_\omega(L,n)}^{(1),d,+,+}(x,y)&= \sum_{Q\text{ occ.}}
      \sum_{j\not\in\mathcal{P}_-} \sum_{\substack{n_j\geq1\\
          n'_j\geq1}}\sum_{\tilde n\in\N^{m-1}}a^Q_{\tilde n_j}
      \overline{a^Q_{\tilde
          n'_j}}\gamma^{(1)}_{\substack{Q_j\\n_j,n'_j}}(x,y) \\
      &=\sum_{\substack{Q\text{ occ.}\\\tilde n\in\N^{m-1}}}
      \sum_{j\not\in\mathcal{P}_-} Q_j \int_{\Delta^{Q_j-1}_j}
      \left(\sum_{n_j=1}^{+\infty}a^Q_{\tilde n_j} \varphi^j_{Q_j,n_j}(x,z)
      \right) \left(\overline{\sum_{n_j=1}^{+\infty}a^Q_{\tilde n_j}
          \varphi^j_{Q_j,n_j}(y,z)} \right)dz.
    \end{split}
  \end{equation*}
  Thus, by Lemma~\ref{le:15}, we get
  \begin{equation}
    \label{eq:228}
    \begin{split}
      \left\|\gamma_{\Psi^U_\omega(L,n)}^{(1),d,+,+}\right\|_{\text{tr}}&\leq
      \sum_{\substack{Q\text{ occ.}\\Q\in\mathcal{Q}_\rho\\\tilde
          n\in\N^{m-1}}} \sum_{j\not\in\mathcal{P}_-} Q_j
      \sum_{n_j=1}^{+\infty} \left|a^Q_{\tilde n_j} \right|^2 \leq
      \sum_{Q\text{ occ. in }\mathcal{Q}_\rho}
      \left(\sum_{j\not\in\mathcal{P}_-} Q_j\right)
      \sum_{\overline{n}\in\N^m} \left|a^Q_{\overline{n}}\right|^2 \\&\leq
      \max_{Q\text{ occ. in }\mathcal{Q}_\rho}
      \left(\sum_{j\not\in\mathcal{P}_-} Q_j\right)\leq Cn\rho^{1+\eta}
    \end{split}
  \end{equation}
  by Lemma~\ref{le:17}. \\
  Finally, one has
  \begin{equation*}
    \gamma_{\Psi^U_\omega(L,n)}^{(1),d,+,0}=\sum_{Q\text{ occ.}}
    \sum_{\substack{j\in\mathcal{P}_-\\Q_j\geq4}}
    \sum_{\substack{n_j\geq1\\
        n'_j\geq1}}\sum_{\tilde n\in\N^{m-1}}a^Q_{\tilde 
      n_j} \overline{a^Q_{\tilde n'_j}}\gamma^{(1)}_{\substack{Q_j\\n_j,n'_j}}.
  \end{equation*}
  Thus, the same computation as above yields
  \begin{equation*}
    \begin{split}
      \left\|\gamma_{\Psi^U_\omega(L,n)}^{(1),d,+,0}\right\|_{\text{tr}}&\leq
      \sum_{\substack{Q\text{ occ.}\\Q\in\mathcal{Q}_\rho\\\tilde
          n\in\N^{m-1}}} \left(\sum_{\substack{j,\
            |\Delta_j(\omega)|\leq 3\ell_\rho(1-\varepsilon)\\
            Q_j\geq4}} Q_j\right) \sum_{n_j=1}^{+\infty}
      \left|a^Q_{\tilde n_j} \right|^2 \leq Cn\frac{\rho}{\ell_\rho}
    \end{split}
  \end{equation*}
  by Lemma~\ref{lem:normPiecesWithParticleExcess}.\\
  This completes the proof of Lemma~\ref{le:20}.
\end{proof}
\noindent Let us now analyze $\gamma_{\Psi^U_\omega(L,n)}^{(1),d,-}$. We recall
and compute
\begin{equation*}
  \gamma_{\Psi^U_\omega(L,n)}^{(1),d,-}:=\sum_{\substack{Q\text{
        occ.}\\\tilde n\in\N^{m-1}}}\sum_{j\in\mathcal{P}^Q_-}
  \sum_{\substack{n_j\geq1\\ n'_j\geq1}}a^Q_{\tilde
    n_j} \overline{a^Q_{\tilde n'_j}}
  \gamma^{(1)}_{\substack{Q_j\\n_j,n'_j}} 
  =\sum_{\substack{Q\text{ occ.}\\\tilde n\in\N^{m-1}}}
  \sum_{j\in\mathcal{P}^Q_-}Q_j
  \left|\varphi_j^{\tilde{n}}    \right\rangle
  \left\langle\varphi_j^{\tilde n}\right|
\end{equation*}
where $\D \varphi_j^{\tilde{n}}=\sum_{n_j\geq1} a^Q_{\tilde
  n_j}\varphi_{Q_j,n_j}^j$.\\
For $\tilde n$ and $Q$ given, define the two sets
\begin{equation}
  \label{eq:236}
  \mathcal{P}_{-,+}^{Q,\tilde{n}}:= \left\{j\in\mathcal{P}^Q_-;\ a^Q_{\tilde
      n_j}=0\text{ if  }n_j\geq2\right\}\text{ and }
  \mathcal{P}_{-,-}^{Q,\tilde{n}}:= 
  \left\{j\in\mathcal{P}^Q_-;\ \exists n_j\geq2\text{ s.t. }a^Q_{\tilde
      n_j}\not=0\right\}.
\end{equation}
Define also
\begin{equation}
  \label{eq:237}
  \tilde\varphi_j^{\tilde{n}}=
  \begin{cases}
    \varphi_j^{\tilde{n}}&\text{ if }n_j=1,\\\|\varphi_j^{\tilde{n}}\|
    \varphi_{Q_j,1}^j&\text{ if }n_j\geq2.
  \end{cases}
\end{equation}
Then, we compute
\begin{equation}
  \label{eq:230}
  \begin{split}
    \gamma_{\Psi^U_\omega(L,n)}^{(1),d,-}&=\sum_{\substack{Q\text{ occ.}\\\tilde
        n\in\N^{m-1}}}\sum_{j\in\mathcal{P}_{-,-}^{Q,\tilde{n}}}Q_j
    \left|\varphi_j^{\tilde{n}} \right\rangle
    \left\langle\varphi_j^{\tilde n}\right|+ \sum_{\substack{Q\text{
          occ.}\\\tilde
        n\in\N^{m-1}}}\sum_{j\in\mathcal{P}_{-,-}^{Q,\tilde{n}}}Q_j
    \left|\varphi_j^{\tilde{n}} \right\rangle
    \left\langle\varphi_j^{\tilde n}\right|\\&
    =\sum_{\substack{Q\text{ occ.}\\\tilde
        n\in\N^{m-1}}}\sum_{j\in\mathcal{P}^Q_-}Q_j
    \left|\tilde\varphi_j^{\tilde{n}} \right\rangle
    \left\langle\tilde\varphi_j^{\tilde n}\right|+
    \sum_{\substack{Q\text{ occ.}\\\tilde
        n\in\N^{m-1}}}\sum_{j\in\mathcal{P}_{-,-}^{Q,\tilde{n}}}Q_j
    \left(\left|\varphi_j^{\tilde{n}} \right\rangle
      \left\langle\varphi_j^{\tilde n}\right|-
      \left|\tilde\varphi_j^{\tilde{n}} \right\rangle
      \left\langle\tilde\varphi_j^{\tilde n}\right|\right).
  \end{split}
\end{equation}
The second term in the sum above we estimate by
\begin{equation}
  \label{eq:231}
  \begin{split}
    \left\|\sum_{\substack{Q\text{ occ.}\\\tilde
          n\in\N^{m-1}}}\sum_{j\in\mathcal{P}_{-,-}^{Q,\tilde{n}}}Q_j
      \left(\left|\varphi_j^{\tilde{n}} \right\rangle
        \left\langle\varphi_j^{\tilde n}\right|-
        \left|\tilde\varphi_j^{\tilde{n}} \right\rangle
        \left\langle\tilde\varphi_j^{\tilde
            n}\right|\right)\right\|_{\text{tr}} &\leq
    \sum_{\substack{Q\text{ occ.}\\\tilde
        n\in\N^{m-1}}}\sum_{j\in\mathcal{P}_{-,-}^{Q,\tilde{n}}}Q_j
    \left(\left\|\varphi_j^{\tilde{n}}\right\|^2+
      \left\|\tilde\varphi_j^{\tilde{n}}\right\|^2\right)\\&\leq
    \sum_{\substack{Q\text{ occ.}\\\overline{n}\in\N^m}}
    \sum_{j\in\mathcal{P}^Q_-}\#\{j;\ n_j\geq2\}
    \left|a_{\overline{n}}^Q\right|^2\\&      \leq
    n\frac{\rho}{\rho_0|\log\rho|}\,f_Z(|\log{\rho}|). 
  \end{split}  
\end{equation}
by Lemma~\ref{le:29}.\\
As for the first term in the second equality in~\eqref{eq:230},
letting $\mathcal{P}_{\text{opt}}$ be the pieces of length less than
$3\ell_\rho(1-\varepsilon)$ where $\Psi^{\text{opt}}$ puts at least
one particle, we write
\begin{equation}
  \label{eq:232}
  \sum_{\substack{Q\text{ occ.}\\\tilde
        n\in\N^{m-1}}}\sum_{j\in\mathcal{P}^Q_-}Q_j
    \left|\tilde\varphi_j^{\tilde{n}} \right\rangle
    \left\langle\tilde\varphi_j^{\tilde n}\right|=
    \sum_{\substack{Q\text{ occ.}\\\tilde
        n\in\N^{m-1}}}\left(\sum_{j\in\mathcal{P}_{\text{opt}}}+
      \sum_{j\in\mathcal{P}^Q_-\setminus\mathcal{P}_{\text{opt}}}-
      \sum_{j\in\mathcal{P}_{\text{opt}}\setminus\mathcal{P}^Q_-}\right)Q_j
    \left|\tilde\varphi_j^{\tilde{n}} \right\rangle
    \left\langle\tilde\varphi_j^{\tilde n}\right|
\end{equation}
One computes
\begin{equation}
  \label{eq:233}
  \begin{split}
    \sum_{\substack{Q\text{ occ.}\\\tilde
        n\in\N^{m-1}}}\sum_{j\in\mathcal{P}_{\text{opt}}}Q_j
    \left|\tilde\varphi_j^{\tilde{n}} \right\rangle
    \left\langle\tilde\varphi_j^{\tilde n}\right|&=
    \sum_{j\in\mathcal{P}_{\text{opt}}}Q_j
    \left(\sum_{\substack{Q\text{ occ.}\\\overline{n}\in\N^m}}
      \left|a_{\overline{n}}^Q\right|^2\right) \left|\varphi^j_{Q_j,1}
    \right\rangle \left\langle\varphi^j_{Q_j,1}\right|\\ &=
    \sum_{j\in\mathcal{P}_{\text{opt}}}Q_j \left|\varphi^j_{Q_j,1}
    \right\rangle \left\langle\varphi^j_{Q_j,1}\right|
    =\gamma_{\Psi^{\text{opt}}}+R
  \end{split}  
\end{equation}
where $\left\|R\right\|_{\text{tr}}\leq C n\rho^{1+\eta}$.\\
By Corollary~\ref{cor:3}, we know that
\begin{gather*}
  \begin{split}
    &\left\|\sum_{\substack{Q\text{ occ.}\\\tilde n\in\N^{m-1}}}\left(
        \sum_{\substack{j\in\mathcal{P}^Q_-\setminus\mathcal{P}_{\text{opt}}
            \\|\Delta_j(\omega)|\geq \ell_{\rho}+C}}-
        \sum_{\substack{j\in\mathcal{P}_{\text{opt}}\setminus\mathcal{P}^Q_-
            \\|\Delta_j(\omega)|\geq \ell_{\rho}+C}}
      \right)Q_j \left|\tilde\varphi_j^{\tilde{n}} \right\rangle
      \left\langle\tilde\varphi_j^{\tilde
          n}\right|\right\|_{\text{tr}}\\&\hskip1cm\leq \sum_{Q\text{
        occ.}}\left(
      \sum_{\substack{j\in\mathcal{P}^Q_-\setminus\mathcal{P}_{\text{opt}}
          \\|\Delta_j(\omega)|\geq \ell_{\rho}+C}}+
      \sum_{\substack{j\in\mathcal{P}_{\text{opt}}\setminus\mathcal{P}^Q_-
          \\|\Delta_j(\omega)|\geq \ell_{\rho}+C}}
    \right)Q_j \sum_{ \overline{n}\in\N^m}\left| a^Q_{\overline{n}}
    \right|^2\\&\hskip1cm\leq C n \rho
    \max\left(\sqrt{Z(2|\log{\rho}|)},\ell^{-1}_\rho\right)
    \sum_{\substack{Q\text{ occ.}\\\overline{n}\in\N^m}}\left|
      a^Q_{\overline{n}} \right|^2= C n \rho
    \max\left(\sqrt{Z(2|\log{\rho}|)},\ell^{-1}_\rho\right)
  \end{split}
  \\\intertext{and, in the same way,} \left\|\sum_{\substack{Q\text{
          occ.}\\\tilde n\in\N^{m-1}}}\left(
      \sum_{\substack{j\in\mathcal{P}_{\text{opt}\setminus\mathcal{P}^Q_-}
          \\|\Delta_j(\omega)|< \ell_{\rho}+C}}-
      \sum_{\substack{j\in\mathcal{P}_{\text{opt}}\setminus\mathcal{P}^Q_-
          \\|\Delta_j(\omega)|< \ell_{\rho}+C}}
    \right)Q_j \left|\tilde\varphi_j^{\tilde{n}} \right\rangle
    \left\langle\tilde\varphi_j^{\tilde
        n}\right|\right\|_{\text{tr}}\leq C n \max\left(\sqrt{\rho
      Z(2|\log{\rho}|)}, \rho |\log{\rho}|^{-1}\right).
\end{gather*}
Plugging this and~\eqref{eq:233} into~\eqref{eq:232} and then
into~\eqref{eq:230}, using~\eqref{eq:231}, we obtain
\begin{gather*}
  \left\|\gamma_{\Psi^U_\omega(L,n)}^{(1),d,-}-
    \gamma_{\Psi^{\text{opt}}}^{(1)}\right\|_{\text{tr},<\ell_\rho+C}\leq
  C n \max\left(\sqrt{\rho
      Z(2|\log{\rho}|)}, \rho |\log{\rho}|^{-1}\right)\\
  \left\|\gamma_{\Psi^U_\omega(L,n)}^{(1),d,-}-
    \gamma_{\Psi^{\text{opt}}}^{(1)}\right\|_{\text{tr},\geq\ell_\rho+C}\leq
  C n \rho \max\left(\sqrt{Z(2|\log{\rho}|)},\ell^{-1}_\rho\right).
\end{gather*}
Taking into account the decomposition~\eqref{eq:179},
Theorem~\ref{thr:4} and Lemmas~\ref{le:19} and~\ref{le:20} then
completes the proof of Theorem~\ref{thr:2}.\qed
\subsection{The proof of Theorem~\ref{thr:7}}
\label{sec:proof-theorem7}
We proceed as in the proof of Theorem~\ref{thr:2}: for $\Psi^U_\omega(L,n)$ a ground state
of the Hamiltonian $ H_\omega^U(L,n)$, we analyze each of the components of the
decomposition~\eqref{eq:186} separately.\\
We prove
\begin{Le}
  \label{le:1}
  Under the assumptions of Theorem~\ref{thr:4}, in the thermodynamic
  limit, with probability $1-O(L^{-\infty})$, one has
  \begin{equation*}
    \left\|\gamma_{\Psi^U_\omega(L,n)}^{(2),d,d}\right\|_{\text{tr}}\lesssim
    n\log n\cdot\log\log n.
  \end{equation*}
\end{Le}
\begin{proof}
  Using Lemma~\ref{le:15} and the orthonormality properties of the families
  $(\varphi^j_{Q_j,n_j})_{n_j\in\N}$, we compute
  \begin{equation*}
    \begin{split}
    \left\|\gamma_{\Psi^U_\omega(L,n)}^{(2),d,d}\right\|_{\text{tr}}&\leq \sum_{Q\text{
        occ. for}\Psi^U_\omega(L,n)}\sum_{\substack{j=1\\Q_j\geq2}}^m
    \frac{Q_j(Q_j-1)}2\sum_{\tilde
      n\in\N^{m-1}} \sum_{n_j\geq1}\left|a^Q_{\tilde n_j}\right|^2\\&\leq
    \sum_{Q\text{ occ. for }\Psi^U_\omega(L,n)}\sum_{\substack{j=1\\Q_j\geq2}}^m
    \frac{Q_j(Q_j-1)}2\sum_{\tilde
      n\in\N^m}\left|a^Q_{\overline{n}}\right|^2.      
    \end{split}
  \end{equation*}
  Applying Lemmas~\ref{le:17}
  and~\ref{lem:normPiecesWithParticleExcess} yields that, in the
  thermodynamic limit, with probability $1-O(L^{-\infty})$, one has
  \begin{equation*}
    \max_{Q\text{ occ. for }\Psi^U_\omega(L,n)}\sum_{\substack{j=1\\Q_j\geq2}}^m
    \frac{Q_j(Q_j-1)}2\lesssim
    n\log n\cdot\log\log n.
  \end{equation*}
  This completes the proof of Lemma~\ref{le:1} as $\D\sum_{Q,\ \tilde
    n\in\N^m}\left|a^Q_{\overline{n}}\right|^2=1$.
\end{proof}
\begin{Le}
  \label{le:24}
  Under the assumptions of Theorem~\ref{thr:4}, in the thermodynamic
  limit, with probability $1-O(L^{-\infty})$, one has
  \begin{equation*}
    \left\|\gamma_{\Psi^U_\omega(L,n)}^{(2),2}\right\|_{\text{tr}}\leq 2.
  \end{equation*}  
\end{Le}
\begin{proof}
  Using Lemma~\ref{le:15} and the orthonormality properties of the families
  $(\varphi^j_{Q_j,n_j})_{n_j\in\N}$, we compute
  \begin{equation*}
    \left\|\gamma_{\Psi^U_\omega(L,n)}^{(2),2}\right\|_{\text{tr}}\leq
    \sum_{i\not=j}\left(\sum_{\substack{Q\text{ occ.}\\
          Q_j\geq2}}+\sum_{\substack{Q\text{ occ.}\\
          Q_i\geq1\\Q_j\geq1}}\right) \sum_{\tilde n\in\N^{m-2}} C_2(Q,i,j)
    \sum_{n_i,n_j\geq1}\left|a^Q_{\tilde n_{i,j}}\right|^2.
  \end{equation*}
  For $Q_j\geq1$ and $Q_i\geq1$, one has
  \begin{equation*}
    \begin{split}
      C_2(Q,i,j)&=\frac{(n-Q_j-Q_i-2)!Q_i!Q_j!}{2\,(n-2)!}\\&=
      \frac{(Q_i+Q_j-2)!(n-(Q_j+Q_i-2)-4)!}{(n-4)!}
      \frac{(Q_i-1)!(Q_j-1)!}{(Q_i+Q_j-2)!}
      \frac{Q_iQ_j}{2(n-2)(n-3)}\\&\leq\frac{Q_iQ_j}{2(n-2)(n-3)}.
    \end{split}
  \end{equation*}
  For $Q_j\geq2$, one has
  \begin{equation}
    \label{eq:214}
    \begin{split}
      C_2(Q,i,j)&=\frac{Q_i!(Q_j-2)!(n-4-(Q_j+Q_i-2))!}{(n-4)!}
      \frac{Q_j(Q_j-1)}{2(n-2)(n-3)}\\&\leq\frac{Q_j(Q_j-1)}{2(n-2)(n-3)}.
    \end{split}
  \end{equation}
  Thus, as $\D\sum_j Q_j=n$, one estimates
  \begin{equation*}
    \left\|\gamma_{\Psi^U_\omega(L,n)}^{(2),2}\right\|_{\text{tr}}\leq
    \frac2{2(n-2)(n-3)}\sum_{\substack{Q\text{ occ.}\\
        \overline{n}\in\N^m}} \left(\sum_j Q_j\right)^2
    \left|a^Q_{\overline{n}}\right|^2\leq \frac{n^2}{(n-2)(n-3)}.
  \end{equation*}
  This proves Lemma~\ref{le:24}.
\end{proof}
\begin{Le}
  \label{le:25}
  Under the assumptions of Theorem~\ref{thr:4}, in the thermodynamic
  limit, with probability $1-O(L^{-\infty})$, one has
  \begin{equation*}
    \left\|\gamma_{\Psi^U_\omega(L,n)}^{(2),4,2}\right\|_{\text{tr}}\leq 1.
  \end{equation*}    
\end{Le}
\begin{proof}
  Using Lemma~\ref{le:15} and the orthonormality properties of the families
  $(\varphi^j_{Q_j,n_j})_{n_j\in\N}$, we compute
  \begin{equation*}
    \left\|\gamma_{\Psi^U_\omega(L,n)}^{(2),4,2}\right\|_{\text{tr}}\leq
    \sum_{i\not=j}\sum_{\tilde
      n\in\N^{m-2}} \sum_{\substack{Q\text{ occ.}\\
        Q_j\geq2\\ Q':\ Q'_k=Q_k\text{ if
        }k\not\in\{i,j\}\\Q'_i=Q_i+2\\Q'_j=Q_j-2}}C_2(Q,i,j)
    \sum_{n_i,n_j\geq1}\left|a^Q_{\tilde n_{i,j}}\right|^2.
  \end{equation*}
  The bound~\eqref{eq:214} then yields
\begin{equation*}
    \left\|\gamma_{\Psi^U_\omega(L,n)}^{(2),4,2}\right\|_{\text{tr}}\leq
    \frac2{2(n-2)(n-3)}\sum_{\substack{Q\text{ occ.}\\
        \overline{n}\in\N^m}} \left(\sum_j Q_j\right)^2
    \left|a^Q_{\overline{n}}\right|^2\leq \frac{n^2}{2(n-2)(n-3)}.
  \end{equation*}
  This proves Lemma~\ref{le:25}.
\end{proof}
\begin{Le}
  \label{le:26}
  Under the assumptions of Theorem~\ref{thr:4}, in the thermodynamic
  limit, with probability $1-O(L^{-\infty})$, one has
  \begin{equation*}
    \left\|\gamma_{\Psi^U_\omega(L,n)}^{(2),4,3}\right\|_{\text{tr}}+
    \left\|\gamma_{\Psi^U_\omega(L,n)}^{(2),4,3'}\right\|_{\text{tr}}\leq \frac{2n}{\rho}.
  \end{equation*}
\end{Le}
\begin{proof}
  Using Lemma~\ref{le:15} and the orthonormality properties of the families
  $(\varphi^j_{Q_j,n_j})_{n_j\in\N}$, we compute
  \begin{equation*}
    \left\|\gamma_{\Psi^U_\omega(L,n)}^{(2),4,3}\right\|_{\text{tr}}\leq
    \sum_{\substack{i,j,k\\\text{distinct}}}\sum_{\tilde n\in\N^{m-3}}
    \sum_{\substack{Q\text{
          occ.}\\ Q_j\geq2\\Q':\ Q'_l=Q_l\text{ if }l\not\in\{i,j,k\}\\Q'_i=Q_i+1\\
        Q'_j=Q_j-2\\Q'_k=Q_k+1}} C_3(Q,i,j,k) \sum_{n_i,n_j,n_k\geq 1}
    \left|a^Q_{\tilde n_{i,j,k}}\right|^2.
  \end{equation*}  
  For  $Q_j\geq2$, one has
  \begin{equation}
    \label{eq:215}
    \begin{split}
      C_3(Q,i,j,k)&=\frac{Q_k!Q_i!(Q_j-2)!(n-(Q_k+Q_i+Q_j-2)-4)!}{(n-4)!}
      \frac{Q_j(Q_j-1)}{2(n-2)(n-3)}\\&\leq\frac{Q_j(Q_j-1)}{2(n-2)(n-3)}.
    \end{split}
  \end{equation}
  Hence, by Proposition~\ref{prop:IntervStatistics}, one has
  \begin{equation*}
    \left\|\gamma_{\Psi^U_\omega(L,n)}^{(2),4,3}\right\|_{\text{tr}}\leq
    \frac{1}{2(n-2)(n-3)}\sum_{\substack{\overline{n}\in\N^m\\Q\text{ occ.}}}
    \left(\sum_{j}1\right)    \left(\sum_j Q_j\right)^2
    \left|a^Q_{\overline{n}}\right|^2\leq \frac{L n^2}{2(n-2)(n-3)} 
    \leq\frac{n}{\rho}.
  \end{equation*}  
  The computation for $\gamma_{\Psi^U_\omega(L,n)}^{(2),4,3'}$ is the same except that,
  instead of~\eqref{eq:215}, one uses, for $Q_k\geq1$ and $Q_i\geq1$, 
  \begin{equation*}
    \begin{split}
      C_3(Q,i,j,k)&=\frac{(Q_k-1)!(Q_i-1)!(Q_j)!(n-(Q_j+Q_i+Q_k-2)-4)!}{(n-4)!}
      \frac{Q_kQ_i}{2(n-2)(n-3)}\\&\leq\frac{Q_kQ_i}{2(n-2)(n-3)}.
    \end{split}
  \end{equation*}
  This proves Lemma~\ref{le:25}.
\end{proof}
\begin{Le}
  \label{le:27}
  Under the assumptions of Theorem~\ref{thr:4}, in the thermodynamic
  limit, with probability $1-O(L^{-\infty})$, one has
  \begin{equation*}
    \left\|\gamma_{\Psi^U_\omega(L,n)}^{(2),4,4}\right\|_{\text{tr}}\leq n^{-1}.
  \end{equation*}
\end{Le}
\begin{proof}
  As in the proof of Lemma~\ref{le:19}, we will have to deal with the
  degenerate cases separately (see Remarks~\ref{rem:8}
  and~\ref{rem:2}).\\
  Recall~\eqref{eq:203} and write
  \begin{equation}
    \label{eq:217}
    \gamma_{\Psi^U_\omega(L,n)}^{(2),4,4}=
    \sum_{\sigma\in\{\pm\}^4}\gamma_{\Psi^U_\omega(L,n)}^{(2),4,\sigma}
  \end{equation}
  where $\sigma=(\sigma_i,\sigma_j,\sigma_k,\sigma_l)\in\{\pm1\}^4$,
  \begin{equation}
    \label{eq:216}
    \gamma_{\Psi^U_\omega(L,n)}^{(2),4,\sigma}=
    \sum_{\substack{i,j,k,l\\\text{distinct}}}\sum_{\tilde n\in\N^{m-4}}
    \sum_{\substack{Q\text{ occ.}\\
        (Q_i,Q_j,Q_k,Q_l)\in\mathcal{Q}_\sigma\\ Q':\ Q'_o=Q_o\text{
          if }o\not\in\{i,j,k,l\}\\Q'_i=Q_i-1,\
        Q'_j=Q_j-1\\Q'_k=Q_k+1,\ Q'_l=Q_l+1}} 
    C_4(Q,i,j,k,l)\sum_{\substack{n_i,n_j,n_k,n_l\geq
        1\\n'_i,n'_j,n'_k,n'_l\geq 1}}
    a^Q_{\tilde{n}_{i,j,k,l}}\overline{a^{Q'}_{\tilde{n}_{i,j,k,l}'}}
    \gamma^{(2),4,4}_{\substack{Q_i,Q_j,Q_k,Q_l\\n_i,n_j,n_k,n_l\\n'_i,n'_j,n'_k,n'l}},
  \end{equation}
  and
  \begin{equation*}
    \begin{split}
      \mathcal{Q}_\sigma&=\left\{Q_i\geq1\text{ and
        }\sigma_i(Q_i-1)\geq\frac{\sigma_k+1}2\right\}
      \cap \left\{Q_j\geq1\text{ and
        }\sigma_j(Q_j-1)\geq\frac{\sigma_j+1}2\right\}
      \\&\hskip2.5cm\cap\left\{Q_k\geq0\text{ and }\sigma_k
        Q_k\geq\frac{\sigma_k+1}2\right\}
      \cap\left\{Q_l\geq0\text{ and
        }\sigma_l Q_l\geq\frac{\sigma_l+1}2\right\}.
    \end{split}
  \end{equation*}
  A term in the right hand side of~\eqref{eq:217} degenerates if some
  $\sigma_\bullet$ takes the value $-1$.\\
  Assume now $\sigma=(1,1,1,1)$. Then, 
  \begin{equation*}
    \gamma_{\Psi^U_\omega(L,n)}^{(2),4,(1,1,1,1)}=
    \sum_{\substack{i,j,k,l\\\text{distinct}}}\sum_{\tilde n\in\N^{m-4}}
    \sum_{\substack{Q\text{ occ.}\\
        Q_i,Q_j\geq2,\ Q_k,Q_l\geq1\\ Q':\ Q'_o=Q_o\text{
          if }o\not\in\{i,j,k,l\}\\Q'_i=Q_i-1,\
        Q'_j=Q_j-1\\Q'_k=Q_k+1,\ Q'_l=Q_l+1}} 
    C_4(Q,i,j,k,l)\sum_{\substack{n_i,n_j,n_k,n_l\geq
        1\\n'_i,n'_j,n'_k,n'_l\geq 1}}
    a^Q_{\tilde{n}_{i,j,k,l}}\overline{a^{Q'}_{\tilde{n}_{i,j,k,l}'}}
    \gamma^{(2),4,4}_{\substack{Q_i,Q_j,Q_k,Q_l\\n_i,n_j,n_k,n_l\\n'_i,n'_j,n'_k,n'_l}}.
  \end{equation*}
  Using Lemma~\ref{le:15} and the orthonormality properties of the families
  $(\varphi^j_{Q_j,n_j})_{n_j\in\N}$, we compute
  \begin{equation*}
    \left\|\gamma_{\Psi^U_\omega(L,n)}^{(2),4,(+,+,+,+)}\right\|_{\text{tr}}\leq
    4\sum_{\substack{i,j,k,l\\\text{distinct}}}\sum_{\tilde n\in\N^{m-4}}
    \sum_{\substack{Q\text{ occ.}\\
        Q_i,Q_j\geq2,\ Q_k,Q_l\geq1\\ Q':\ Q'_o=Q_o\text{
          if }o\not\in\{i,j,k,l\}\\Q'_i=Q_i-1,\
        Q'_j=Q_j-1\\Q'_k=Q_k+1,\ Q'_l=Q_l+1}} 
    C_4(Q,i,j,k,l) \sum_{n_i,n_j,n_k,n_l\geq 1}
    \left|a^Q_{\tilde n_{i,j,k,l}}\right|^2.
  \end{equation*}  
  When $Q_i\geq2$, $Q_j\geq2$, $Q_k\geq1$ and $Q_l\geq1$ one has
  \begin{equation*}
    C_4(Q,i,j,k,l)\leq\frac{Q_i(Q_i-1)Q_j(Q_j-1)Q_kQ_l}
    {2n(n-2)(n-3)(n-4)(n-5)(n-6)(n-7)}.
  \end{equation*}
  Thus, by Lemma~\ref{le:17}, we obtain
  \begin{equation}
    \label{eq:218}
    \begin{split}
      \left\|\gamma_{\Psi^U_\omega(L,n)}^{(2),4,(1,1,1,1)}\right\|_{\text{tr}}&\leq
      \frac{2}{(n-5)^6}\sum_{\overline{n}\in\N^m} \sum_{Q\text{ occ.}}
      \left(\sum_{j}Q_j\right)^2 \left(\sum_j Q_j^2\right)^2
      \left|a^Q_{\overline{n}}\right|^2 \leq \frac{n^4(\log
        n)^4}{2(n-7)^6}\\&\leq n^{-1}
    \end{split}
  \end{equation}
  for $n$ large.\\
  Assume now $\sigma=(-1,-1,-1,-1)$. Then,
  \begin{equation*}
    \gamma_{\Psi^U_\omega(L,n)}^{(2),4,(-1,-1,-1,-1)}=
    \sum_{\substack{i,j,k,l\\\text{distinct}}}\sum_{\tilde n\in\N^{m-4}}
    \sum_{\substack{Q\text{ occ.}\\
        Q_i=Q_j=1,\ Q_k=Q_l=0\\ Q':\ Q'_o=Q_o\text{
          if }o\not\in\{i,j,k,l\}\\Q'_i=Q_i-1,\
        Q'_j=Q_j-1\\Q'_k=Q_k+1,\ Q'_l=Q_l+1}} 
    C_4(Q,i,j,k,l)\sum_{\substack{n_i,n_j\geq
        1\\n_k=n_l=1\\n'_i=n'_j=1\\n'_k,n'_l\geq 1}}
    a^Q_{\tilde{n}_{i,j,k,l}}\overline{a^{Q'}_{\tilde{n}_{i,j,k,l}'}}
    \gamma^{(2),4,4}_{\substack{1,1,0,0\\n_i,n_j,1,1\\1,1,n'_k,n'_l}}
  \end{equation*}
  where
  \begin{equation*}
    \begin{split}
      \gamma^{(2),4,4}_{\substack{1,1,0,0\\n_i,n_j,1,1\\1,1,n'_k,n'_l}}(x,x',y,y')&=
      \varphi^{i}_{1,n_{i}}(x) \varphi^{j}_{1,n_{j}}(x')
      \overline{\varphi^{k}_{1,n'_{k}}(y)\varphi^{l}_{1,n'_{l}}(y')} +
      \varphi^{i}_{1,n_{i}}(x') \varphi^{j}_{1,n_{j}}(x)
      \overline{\varphi^{k}_{1,n'_{k}}(y)\varphi^{l}_{1,n'_{l}}(y')}
      \\&\hskip.5cm+ \varphi^{i}_{1,n_{i}}(x) \varphi^{j}_{1,n_{j}}(x')
      \overline{\varphi^{k}_{1,n'_{k}}(y')\varphi^{l}_{1,n'_{l}}(y)} +
      \varphi^{i}_{1,n_{i}}(x') \varphi^{j}_{1,n_{j}}(x)
      \overline{\varphi^{k}_{1,n'_{k}}(y')\varphi^{l}_{1,n'_{l}}(y)}.
    \end{split}
  \end{equation*}
  As in the derivation of~\eqref{eq:178}, using Lemma~\ref{le:15} and
  the orthonormality properties of the families
  $(\varphi^j_{Q_j,n_j})_{n_j\in\N}$, we compute
  \begin{equation*}
    \begin{split}
      \left\|\gamma_{\Psi^U_\omega(L,n)}^{(2),4,(-1,-1,-1,-1)}\right\|_{\text{tr}}&\leq
      \frac2{(n-2)(n-3)}\sum_{\substack{\tilde n\in\N^{m-4}\\Q\text{ occ.}}}
      \left\|\sum_{\substack{(i,j)\\Q_i=Q_j=1}} \sum_{\substack{n_i=1\\n_j=1\\
            n_k,n_l}} a^Q_{\tilde n_{i,j,k,l}}\varphi^k_{1,n_k}\otimes
        \varphi^l_{1,n_l}\right\|^2\\&\hskip5cm+
      \left\|\sum_{\substack{(k,l)\\Q_k=Q_l=0}} \sum_{\substack{n_k=1\\n_l=1\\
            n_i,n_j}} a^Q_{\tilde n_{i,j,k,l}}\varphi^i_{1,n_i}\otimes
        \varphi^j_{1,n_j}\right\|^2\\&\leq
      \frac4{(n-3)^2}\sum_{\substack{\overline{n}\in\N^m\\Q\text{ occ.}}}
      \left|a^Q_{\overline{n}}\right|^2=\frac4{(n-3)^2}.
    \end{split}
  \end{equation*}  
  Assume now $\sigma=(-1,1,1,1)$. Then,
  \begin{equation*}
    \gamma_{\Psi^U_\omega(L,n)}^{(2),4,(-1,1,1,1)}=
    \sum_{\substack{i,j,k,l\\\text{distinct}}}\sum_{\tilde n\in\N^{m-4}}
    \sum_{\substack{Q\text{ occ.}\\
        Q_i=1,\ Q_j\geq2\\ Q_k,Q_l\geq1\\ Q':\ Q'_o=Q_o\text{
          if }o\not\in\{i,j,k,l\}\\Q'_i=Q_i-1,\
        Q'_j=Q_j-1\\Q'_k=Q_k+1,\ Q'_l=Q_l+1}} 
    C_4(Q,i,j,k,l)\sum_{\substack{n_i,n_j,n_k,n_l\geq
        1\\n'_j,n'_k,n'_l\geq 1\\n'_i=1}}
    a^Q_{\tilde{n}_{i,j,k,l}}\overline{a^{Q'}_{\tilde{n}_{i,j,k,l}'}}
    \gamma^{(2),4,4}_{\substack{1,Q_j,Q_k,Q_l\\n_i,n_j,n_k,n_l\\1,n'_j,n'_k,n'_l}}
  \end{equation*}
  where 
  \begin{equation}
    \label{eq:208}
    \begin{split}
      C_4(Q,i,j,k,l)&=\frac{(n-Q_j-Q_k-Q_l-3)!Q_j!Q_k!Q_l!}{2\,(n-2)!} \\
      &\leq \frac{Q_j(Q_j-1)Q_kQ_l}{2(n-2)(n-3)(n-4)(n-5)(n-6)}.
    \end{split}
  \end{equation}
  The operator
  $\gamma^{(2),4,4}_{\substack{1,Q_j,Q_k,Q_l\\n_i,n_j,n_k,n_l\\1,n'_j,n'_k,n'_l}}$
  is given by~\eqref{eq:166} and
  \begin{equation*}
    \begin{split}
      \sigma(x,x',y,y')&= \varphi^{i}_{1,n_{i}}(x) \int_{\Delta^{Q_{j}-1}_{j}}
      \varphi^{j}_{Q_{j},n_{j}}(x',z)
      \overline{\varphi^{j}_{Q_{j}-1,n'_{j}}(z)}dz\\&\hskip3cm\times
      \int_{\Delta^{Q_{k}}_{k}} \varphi^{k}_{Q_{k},n_{k}}(z)
      \overline{\varphi^{k}_{Q_{k}+1,n'_{k}}(y,z)}dz \int_{\Delta^{Q_{l}}_{l}}
      \varphi^{l}_{Q_{l},n_{l}}(z)
      \overline{\varphi^{l}_{Q_{l}+1,n'_{l}}(y',z)}dz.
    \end{split}
  \end{equation*}
  Hence, as in the derivation of~\eqref{eq:177}, using
  Lemma~\ref{le:15},~\eqref{eq:208} and the orthonormality properties of the
  families $(\varphi^j_{Q_j,n_j})_{n_j\in\N}$, we compute
  \begin{equation*}
    \begin{split}
      \left\|\gamma_{\Psi^U_\omega(L,n)}^{(2),4,(-1,1,1,1)}\right\|_{\text{tr}}&
      \leq\frac2{(n-2)(n-3)(n-4)(n-5)(n-6)} \sum_{\substack{\tilde
          n\in\N^m\\Q\text{ occ.}}}\left(\sum_{j=1}^m Q^2_j\right)
      \left(\sum_{j=1}^m Q_j\right)^2 \left| a^Q_{\tilde n_{i,j}}
      \right|^2\\&\leq \frac{n^{10/3}(\log n)^{2/3}}{(n-6)^5}\leq n^{-3/2}.
    \end{split}
  \end{equation*}  
  In the same way, we obtain that, if $\sigma$ contains a least one
  $-1$ then $\D \left\|\gamma_{\Psi^U_\omega(L,n)}^{(2),4,\sigma}
  \right\|_{\text{tr}}\leq n^{-1}$.\\
  This completes the proof of Lemma~\ref{le:27}.
\end{proof}
\noindent Let us now turn to the analysis of
$\gamma_{\Psi^U_\omega(L,n)}^{(2),d,o}$, the main term of
$\gamma_{\Psi^U_\omega(L,n)}^{(2)}$. The analysis will be similar of
that o $\gamma_\Psi^{(1),d}$ in the proof of Theorem~\ref{thr:4}.\\
Recall that $\mathcal{P}^Q_-$ is defined in Proposition~\ref{pro:4} and write
\begin{equation}
  \label{eq:205}
  \gamma_{\Psi^U_\omega(L,n)}^{(2),d,o}=\gamma_{\Psi^U_\omega(L,n)}^{(2),d,o,-}+\gamma_{\Psi^U_\omega(L,n)}^{(2),d,o,+}
\end{equation}
where
\begin{equation}
  \label{eq:219}
  \gamma_{\Psi^U_\omega(L,n)}^{(2),d,o,-}
  =\sum_{\substack{Q\text{ occ.}\\
      Q_i\geq1\\Q_j\geq1}} \sum_{\tilde n\in\N^{m-2}}
  \sum_{\substack{1\leq i<j\leq m\\ (i,j)\in(\mathcal{P}^Q_-)^2}}
  \sum_{\substack{n_j,n'_j\geq 1\\n_i,n'_i\geq 1}}a^Q_{\tilde
    n_{i,j}} \overline{a^Q_{\tilde
      n'_{i,j}}}\gamma^{(2),d,o}_{\substack{Q_i,Q_j\\n_i,n_j\\n'_i,n'_j}}.
\end{equation}
We prove
\begin{Le}
  \label{le:28}
  Under the assumptions of Theorem~\ref{thr:5}, for $\eta\in(0,1)$,
  there exists $\varepsilon_0>0$ such that, for
  $\varepsilon\in(0,\varepsilon_0)$, in the thermodynamic limit, with
  probability $1-O(L^{-\infty})$, one has
  \begin{equation*}
    \left\|\gamma_{\Psi^U_\omega(L,n)}^{(2),d,o,+}\right\|_{\text{tr}}\leq
    n^2\frac{\rho}{\ell_\rho}.
  \end{equation*}  
\end{Le}
\begin{proof}
  The proof follows that of Lemma~\ref{le:20}. One estimates
  \begin{equation}
    \label{eq:221}
    \begin{split}
      \left\|\gamma_{\Psi^U_\omega(L,n)}^{(2),d,o,+}\right\|_{\text{tr}}&
      =\left\|\sum_{\substack{Q\text{ occ.}\\
            Q_i\geq1\\Q_j\geq1}} \sum_{\substack{1\leq i<j\leq m\\
            (i,j)\not\in(\mathcal{P}^Q_-)^2}} \sum_{\tilde
          n\in\N^{m-2}} \sum_{\substack{n_j,n'_j\geq 1\\n_i,n'_i\geq
            1}}a^Q_{\tilde n_{i,j}} \overline{a^Q_{\tilde
            n'_{i,j}}}\gamma^{(2),d,o}_{\substack{Q_i,Q_j
            \\n_i,n_j\\n'_i,n'_j}}\right\|_{\text{tr}}\\&\leq
      \sum_{\substack{Q\text{ occ.}\\
          Q_i\geq1\\Q_j\geq1}} \sum_{\tilde n\in\N^{m-2}}
      \left(\sum_{\substack{1\leq i<j\leq m\\
            i\not\in\mathcal{P}^Q_-}} +
        \sum_{\substack{1\leq i<j\leq m\\
            j\not\in\mathcal{P}^Q_-}}\right)
      \left\|\sum_{\substack{n_j,n'_j\geq 1\\n_i,n'_i\geq
            1}}a^Q_{\tilde n_{i,j}} \overline{a^Q_{\tilde
            n'_{i,j}}}\gamma^{(2),d,o}_{\substack{Q_i,Q_j
            \\n_i,n_j\\n'_i,n'_j}}\right\|_{\text{tr}}.
    \end{split}
  \end{equation}
  Let us analyze the first sum in the right hand side
  above. Using~\eqref{eq:185}, Lemma~\ref{le:15} and the
  orthonormality properties of the families
  $(\varphi^j_{Q_j,n_j})_{n_j\in\N}$, we compute
  \begin{equation*}
    \begin{split}
      \sum_{\substack{Q\text{ occ.}\\Q_i\geq1\\Q_j\geq1}}
      \sum_{\substack{1\leq i<j\leq m\\i\not\in\mathcal{P}^Q_-}}
      \sum_{\tilde n\in\N^{m-2}}\left\| \sum_{\substack{n_j,n'_j\geq
            1\\n_i,n'_i\geq 1}}a^Q_{\tilde
          n_{i,j}}\overline{a^Q_{\tilde
            n'_{i,j}}}\gamma^{(2),d,o}_{\substack{Q_i,Q_j\\n_i,n_j\\n'_i,n'_j}}
      \right\|_{\text{tr}}&\leq \sum_{\substack{Q\text{
            occ.}\\Q_i\geq1\\Q_j\geq1}} \sum_{\substack{1\leq i<j\leq
          m\\i\not\in\mathcal{P}^Q_-}} \sum_{\tilde
        n\in\N^{m-2}}\frac{Q_iQ_j}2 \sum_{n_i,n_j\geq 1}
      \left|a^Q_{\tilde
          n_{i,j}}\right|^2\\
      &\leq \frac12\sum_{\substack{\overline{n}\in\N^m\\Q\text{
            occ.}}}  \left(\sum_{i\not\in\mathcal{P}^Q_-}Q_i\right)
      \left(\sum_jQ_j\right) \left|a^Q_{\overline{n}}\right|^2\\&\leq
      C n^2\frac{\rho}{\ell_\rho}
    \end{split}
  \end{equation*}
  as in the proof of Lemma~\ref{le:20} by Lemma~\ref{le:17}
  and~\ref{lem:normPiecesWithParticleExcess}.\\ 
  The other sum in the right hand side of~\eqref{eq:221} is analyzed
  in the same way. This completes the proof of Lemma~\ref{le:28}.
\end{proof}
\noindent Let us now analyze $\gamma_{\Psi^U_\omega(L,n)}^{(2),d,o,-}$. We proceed
as in the analysis of $\gamma_{\Psi^U_\omega(L,n)}^{(1),d}$ (see~\eqref{eq:179}
and Lemma~\ref{le:20}). We recall and compute
\begin{equation*}
    \gamma_{\Psi^U_\omega(L,n)}^{(2),d,o,-}=\sum_{\substack{Q\text{ occ.}\\
        Q_i\geq1\\Q_j\geq1\\\tilde n\in\N^{m-2}}}
    \sum_{\substack{1\leq i<j\leq m\\ (i,j)\in(\mathcal{P}^Q_-)^2}}
    \sum_{\substack{n_j,n'_j\geq 1\\n_i,n'_i\geq 1}}a^Q_{\tilde
      n_{i,j}} \overline{a^Q_{\tilde
        n'_{i,j}}}\gamma^{(2),d,o}_{\substack{Q_i,Q_j\\n_i,n_j\\n'_i,n'_j}}
=\sum_{\substack{Q\text{ occ.}\\\tilde n\in\N^{m-2}}}
    \sum_{\substack{1\leq i<j\leq m\\ (i,j)\in(\mathcal{P}^Q_-)^2}}
    \frac{Q_iQ_j}2(\text{Id}-\text{Ex})\varphi_{i,j}^{\tilde{n}}\otimes^s
    \varphi_{i,j}^{\tilde{n}}.
\end{equation*}
where $\D\varphi_{i,j}^{\tilde{n}}:=\sum_{\substack{n_i\geq 1\\n_j\geq
    1}}a^Q_{\tilde
  n_{i,j}}\varphi^i_{Q_i,n_i}\wedge\varphi^j_{Q_j,n_j}$ and the
operators Ex and $\otimes^s$ are defined in
Proposition~\ref{prop:DensityMatrixStructure}.\\
Define also
\begin{equation}
  \label{eq:238}
  \tilde\varphi_{i,j}^{\tilde{n}}=
  \begin{cases}
    \varphi_{i,j}^{\tilde{n}}&\text{ if
    }n_i+n_j=2\\\|\varphi_{i,j}^{\tilde{n}}\| \varphi_{Q_i,1}^i
    ²\wedge\varphi_{Q_j,1}^j &\text{ if }n_i+n_j\geq3.
  \end{cases}
\end{equation}
Then, recalling~\eqref{eq:236}, we compute
\begin{equation}
  \label{eq:235}
  \begin{split}
    \gamma_{\Psi^U_\omega(L,n)}^{(2),d,o,-}&=\sum_{\substack{Q\text{ occ.}\\\tilde
        n\in\N^{m-2}}}\sum_{\substack{1\leq i<j\leq m\\
        (i,j)\in(\mathcal{P}^Q_{-,-})^2}}
    \frac{Q_iQ_j}2(\text{Id}-\text{Ex})\varphi_{i,j}^{\tilde{n}}\otimes^s
    \varphi_{i,j}^{\tilde{n}}\\&\hskip5cm+ \sum_{\substack{Q\text{
          occ.}\\\tilde n\in\N^{m-2}}}\sum_{\substack{1\leq i<j\leq m\\
        i\in\mathcal{P}^Q_{-,+}\\\text{or }j\in\mathcal{P}^Q_{-,+}}}
    \frac{Q_iQ_j}2(\text{Id}-\text{Ex})\varphi_{i,j}^{\tilde{n}}\otimes^s
    \varphi_{i,j}^{\tilde{n}}\\
    & =\sum_{\substack{Q\text{ occ.}\\\tilde
        n\in\N^{m-2}}}\sum_{\substack{1\leq i<j\leq m\\
        (i,j)\in(\mathcal{P}^Q_-)^2}}
    \frac{Q_iQ_j}2(\text{Id}-\text{Ex})\tilde\varphi_{i,j}^{\tilde{n}}\otimes^s
    \tilde\varphi_{i,j}^{\tilde{n}}\\&\hskip5cm+
    \sum_{\substack{Q\text{
          occ.}\\\tilde n\in\N^{m-2}}}\sum_{\substack{1\leq i<j\leq m\\
        i\in\mathcal{P}^Q_{-,+}\\\text{or }j\in\mathcal{P}^Q_{-,+}}}
    \frac{Q_iQ_j}2(\text{Id}-\text{Ex})\left(\varphi_{i,j}^{\tilde{n}}\otimes^s
      \varphi_{i,j}^{\tilde{n}}-\tilde\varphi_{i,j}^{\tilde{n}}\otimes^s
      \tilde\varphi_{i,j}^{\tilde{n}}\right).
  \end{split}
\end{equation}
The second term in the sum above we estimate by
\begin{equation}
  \label{eq:242}
  \begin{split}
    &\left\| \sum_{\substack{Q\text{
            occ.}\\\tilde n\in\N^{m-2}}}\sum_{\substack{1\leq i<j\leq m\\
          i\in\mathcal{P}^Q_{-,+}\\\text{or }j\in\mathcal{P}^Q_{-,+}}}
      \frac{Q_iQ_j}2(\text{Id}-\text{Ex})\left(\varphi_{i,j}^{\tilde{n}}\otimes^s
        \varphi_{i,j}^{\tilde{n}}-\tilde\varphi_{i,j}^{\tilde{n}}\otimes^s
        \tilde\varphi_{i,j}^{\tilde{n}}\right)\right\|_{\text{tr}}\\
    &\hskip6cm\lesssim \sum_{\substack{Q\text{ occ.}\\\tilde
        n\in\N^{m-2}}}\sum_{\substack{1\leq i<j\leq m\\
        i\in\mathcal{P}^Q_{-,+}\\\text{or }j\in\mathcal{P}^Q_{-,+}}}Q_iQ_j
    \left(\left\|\varphi_{i,j}^{\tilde{n}}\right\|^2+
      \left\|\tilde\varphi_{i,j}^{\tilde{n}}\right\|^2\right)\\
    &\hskip6cm\lesssim \sum_{\substack{Q\text{
          occ.}\\\overline{n}\in\N^m}}
    \left(\sum_{j\in\mathcal{P}^Q_-}\#\{j;\ n_j\geq2\}\right)
    \left(\sum_{j\in\mathcal{P}^Q_-}Q_j\right)
    \left|a_{\overline{n}}^Q\right|^2\\
    &\hskip6cm\lesssim
    n^2\frac{\rho}{\rho_0|\log\rho|}\,f_Z(2|\log{\rho}|).
  \end{split}  
\end{equation}
by Lemma~\ref{le:29}.\\
As for the first term in the second equality in~\eqref{eq:235}, letting
$\mathcal{P}_{\text{opt}}$ be the pieces of length less than
$3\ell_\rho(1-\varepsilon)$ where $\Psi^{\text{opt}}$ puts at least
one particle, we write
\begin{multline}
  \label{eq:229}
  \sum_{\substack{Q\text{ occ.}\\\tilde
      n\in\N^{m-2}}}\sum_{\substack{1\leq i<j\leq m\\
      (i,j)\in(\mathcal{P}^Q_-)^2}}
  \frac{Q_iQ_j}2(\text{Id}-\text{Ex})\tilde\varphi_{i,j}^{\tilde{n}}\otimes^s
  \tilde\varphi_{i,j}^{\tilde{n}} \\= \sum_{\substack{Q\text{
        occ.}\\\tilde
      n\in\N^{m-2}}}\left(\sum_{\substack{1\leq i<j\leq m\\
        (i,j)\in(\mathcal{P}_{\text{opt}})^2}}+\sum_{\substack{1\leq i<j\leq m\\
        i\text{ or }j\text{ in }\mathcal{P}^Q_-\setminus
        \mathcal{P}_{\text{opt}}}}-\sum_{\substack{1\leq i<j\leq m\\
        i\text{ or }j\text{ in
        }\mathcal{P}_{\text{opt}}\setminus\mathcal{P}^Q_-}} \right)
  \frac{Q_iQ_j}2(\text{Id}-\text{Ex})\tilde\varphi_{i,j}^{\tilde{n}}\otimes^s
  \tilde\varphi_{i,j}^{\tilde{n}}
\end{multline}
For the first of the three sums above, one computes
\begin{equation}
  \label{eq:239}
  \begin{split}
    \sum_{\substack{Q\text{ occ.}\\\tilde
        n\in\N^{m-2}}}\sum_{\substack{1\leq i<j\leq m\\
        (i,j)\in(\mathcal{P}_{\text{opt}})^2}}
    \frac{Q_iQ_j}2(\text{Id}-\text{Ex})\tilde\varphi_{i,j}^{\tilde{n}}\otimes^s
    \tilde\varphi_{i,j}^{\tilde{n}}&=
    \sum_{\substack{1\leq i<j\leq m\\
        (i,j)\in(\mathcal{P}_{\text{opt}})^2}}\left(
      \sum_{\substack{Q\text{ occ.}\\\overline{n}\in\N^m}}
      \left|a_{\overline{n}}^Q\right|^2\right)
    \frac{Q_iQ_j}2(\text{Id}-\text{Ex})
    \tilde\varphi_{i,j}^{\tilde{n}}\otimes^s
    \tilde\varphi_{i,j}^{\tilde{n}}\\ &=
    \sum_{\substack{1\leq i<j\leq m\\
        (i,j)\in(\mathcal{P}_{\text{opt}})^2}}
    \frac{Q_iQ_j}2(\text{Id}-\text{Ex})\gamma^{(1)}_{\varphi^i_{Q_i,1}}
    \otimes^s\gamma^{(1)}_{\varphi^j_{Q_j,1}}
    \\&=\gamma_{\Psi^{\text{opt}}}^{(2)}+R
  \end{split}  
\end{equation}
where $\left\|R\right\|_{\text{tr}}\leq C n^2\rho^{1+\eta}$.\\
In the last line of~\eqref{eq:239}, we have used
Proposition~\ref{prop:DensityMatrixStructure}, the definition of
$\Psi^{\text{opt}}$~\eqref{eq:80}
and Lemma~\ref{le:17} to obtain the bound on $R$.\\
To estimate the remaining two sums in~\eqref{eq:242}, we split them
into sums where the summation over pieces is restricted to pieces
either longer than $\ell_\rho+C$ or shorter than $\ell_\rho+C$ ($C$ is
given by Corollary~\ref{cor:3}).\\
By Corollary~\ref{cor:3}, we know that
\begin{equation*}
  \begin{split}
    &\left\| \sum_{\substack{Q\text{ occ.}\\\tilde
          n\in\N^{m-2}}}\left(\sum_{\substack{1\leq i<j\leq m\\
            i\in\mathcal{P}^Q_-\setminus\mathcal{P}_{\text{opt}}\\\text{
              and }|\Delta_i(\omega)|< \ell_{\rho}+C}}
        -\sum_{\substack{1\leq i<j\leq m\\
            i\in\mathcal{P}^Q_-\setminus\mathcal{P}_{\text{opt}}\\\text{
              and }|\Delta_i(\omega)|< \ell_{\rho}+C}}\right)
      \frac{Q_iQ_j}2(\text{Id}-\text{Ex})\tilde\varphi_{i,j}^{\tilde{n}}\otimes^s
      \tilde\varphi_{i,j}^{\tilde{n}}\right\|_{\text{tr}}\\&\hskip1cm\leq
    \sum_{\substack{Q\text{ occ.}\\\tilde
        n\in\N^{m-2}}}\left(\sum_{\substack{1\leq i<j\leq m\\
          i\in\mathcal{P}^Q_-\setminus\mathcal{P}_{\text{opt}}\\\text{and
          }|\Delta_i(\omega)|< \ell_{\rho}+C}}+\sum_{\substack{1\leq i<j\leq m\\
          i\in\mathcal{P}^Q_-\setminus\mathcal{P}_{\text{opt}}\\\text{and
          }|\Delta_i(\omega)|< \ell_{\rho}+C}}\right)
    \frac{Q_iQ_j}2\left\|(\text{Id}-\text{Ex})
      \tilde\varphi_{i,j}^{\tilde{n}}\otimes^s
      \tilde\varphi_{i,j}^{\tilde{n}}\right\|_{\text{tr}}
    \\\\&\hskip1cm\leq C n^2 \rho
    \max\left(\sqrt{Z(2|\log{\rho}|)},\ell^{-1}_\rho\right)
    \sum_{\substack{Q\text{ occ.}\\\overline{n}\in\N^m}}\left|
      a^Q_{\overline{n}} \right|^2= C n^2 \max\left(\sqrt{\rho
        Z(2|\log{\rho}|)}, \rho|\log{\rho}|^{-1}\right).
  \end{split}
\end{equation*}
In the same way, we estimate
\begin{multline*}
  \left\| \sum_{\substack{Q\text{ occ.}\\\tilde
        n\in\N^{m-2}}}\left(\sum_{\substack{1\leq i<j\leq m\\
          i\in\mathcal{P}_{\text{opt}}\setminus\mathcal{P}^Q_-\\\text{and
          }|\Delta_i(\omega)|\geq\ell_{\rho}+C}}-\sum_{\substack{1\leq
          i<j\leq
          m\\i\in\mathcal{P}_{\text{opt}}\setminus\mathcal{P}^Q_-\\\text{and
          }|\Delta_i(\omega)|\geq \ell_{\rho}+C}}\right)
    \frac{Q_iQ_j}2(\text{Id}-\text{Ex})\tilde\varphi_{i,j}^{\tilde{n}}\otimes^s
    \tilde\varphi_{i,j}^{\tilde{n}}\right\|_{\text{tr}} \\\leq C n^2
  \rho\max\left(\sqrt{Z(2|\log{\rho}|)},\ell^{-1}_\rho\right)
\end{multline*}
and one has the same estimates when $i$ is replaced by $j$.\\
Plugging these estimates,~\eqref{eq:242} and~\eqref{eq:229}
into~\eqref{eq:205}, recalling~\eqref{eq:240}, we obtain
\begin{gather*}
  \left\|\left(\gamma_{\Psi^U_\omega(L,n)}^{(2),d,o,-}
      -\gamma_{\Psi^{\text{opt}}}^{(2)}\right)\car^2_{<\ell_\rho+C}
  \right\|_{\text{tr}}\leq C n^2
  \max\left(\sqrt{\rho Z(2|\log{\rho}|)},\rho|\log{\rho}|^{-1}\right)\\
  \left\| \left(\gamma_{\Psi^U_\omega(L,n)}^{(2),d,o,-}-
      \gamma_{\Psi^{\text{opt}}}^{(2)}\right)
    \left(\car-\car^2_{<\ell_\rho+C}\right)
  \right\|_{\text{tr},\geq\ell_\rho+C}\leq C n^2 \rho
  \max\left(\sqrt{Z(2|\log{\rho}|)},\ell^{-1}_\rho\right).
\end{gather*}
Taking into account the decomposition~\eqref{eq:186} and
Lemmas~\ref{le:1},~\ref{le:24},~\ref{le:25},~\ref{le:26},~\ref{le:27}
then completes the proof of Theorem~\ref{thr:2}.\qed

%%% Local Variables: 
%%% mode: latex
%%% TeX-master: "PiecesModelGroundState"
%%% ispell-local-dictionary: "american"
%%% End: 

\section{Almost sure convergence for the ground state energy per particle}
\label{sec:almost-sure-conv}
In this section, we prove that, if interactions decay sufficiently
fast at infinity, then the convergence in the thermodynamic limit of
the ground state energy per particle $E^U_\omega(L, n) / n$ to
$\densEn^U(\rho)$ holds not only in $L^2_\omega$ (see~\cite[Theorem
3.5]{MR3022666}) but also $\omega$-almost surely.\\
From the proof of~\cite[Theorem 3.5]{MR3022666}, one clearly sees that
it suffices to improve upon the sub-additive estimate given
in~\cite[Lemma 4.1]{MR3022666}. We prove
\begin{Th}
 \label{thr:8}
 Assume that the pair potential $U$ be even and such that $U\in
 L^r(\R)$ for some $r>1$ and that for some $\alpha>2$, one has
 $\D\int_0^{+\infty} x^\alpha U(x)dx<+\infty$.\\
 In the thermodynamic limit, for disjoint intervals $\Lambda_1$ and
 $\Lambda_2$ with $n_1$ and $n_2$ electrons respectively, for
 $\min(|\Lambda_1|,|\Lambda_2|)$ sufficiently large, with probability
 $1-O(\min(|\Lambda_1|,|\Lambda_2|)^{-\infty})$, one has
 \begin{equation}
   \label{eq:62}
   E^U_\omega(\Lambda_1 \cup \Lambda_2, n_1 + n_2) \leq
   E^U_\omega(\Lambda_1, n_1) + E^U_\omega(\Lambda_2, n_2) + o(n_1 + n_2).  
 \end{equation}
 Here, $E^U_\omega(\Lambda, n)$ denotes the ground state energy of
 $H^U_\omega(\Lambda,n)$ (see
 section~\ref{sec:interacting-electrons}).
\end{Th}
\noindent To apply this result to $U$ satisfying \textbf{(HU)}, it
suffices to check
\begin{Le}
  \label{le:31}
  If $U$ satisfies \textbf{\textup{(HU)}} then for any $0<\alpha<3$,
  one has $\D\int_0^{+\infty} x^\alpha U(x)dx<+\infty$.
\end{Le}
\begin{proof}
  Clearly, for $n\geq0$, one has
  \begin{equation*}
    \int_{2^n}^{2^{n+1}}x^\alpha U(x)dx\leq
    2^{\alpha(n+1)}\int_{2^n}^{2^{n+1}} U(x)dx \leq 
    2^{(\alpha-3)n+\alpha}Z(2^n). 
  \end{equation*}
  As $Z$ is bounded, summing over $n$ yields
  \begin{equation*}
    \int_1^{+\infty}x^\alpha U(x)dx\lesssim
    \sum_{n\geq1}2^{(\alpha-3)n+\alpha}<+\infty.
  \end{equation*}
  This completes the proof of Lemma~\ref{le:31}. 
\end{proof}
\noindent Thus, the sub-additive estimate~\eqref{eq:62} holds for our
model and, following the analysis provided in~\cite{MR3022666}, we
obtain Theorem~\ref{thr:1}.
\begin{proof}[Proof of Theorem~\ref{thr:8}]
  Without loss of generality, let us assume that $\Lambda_1 = [-L_1,
  0]$ and $\Lambda_2 = [0, L_2]$.  For $i \in \{1, 2\}$, we denote by
  $\Psi^U_i$ ground states of $H^U_\omega(\Lambda_i, n_i)$. In case of
  degeneracy, we may additionally choose particular ground states
  $\Psi^U_i$, $i \in \{1, 2\}$ such that each of them belongs to a fixed
  occupation subspace. Thus, occupation is well defined for
  $\Psi^U_i$. As usual, we will implicitly suppose that $\Psi^U_i$ is
  extended by zero outside $\Lambda_i^{n_i}$. Consider now
  \begin{equation*}
    \Psi = \Psi^U_1 \wedge \Psi^U_2.
  \end{equation*}
  Then,
  \begin{equation*}
    \begin{split}
      E^U_\omega(\Lambda_1 \cup \Lambda_2, n_1 + n_2) &\leq \left\langle
        H^U_\omega(\Lambda_1 \cup \Lambda_2, n_1 + n_2)
        \Psi, \Psi \right\rangle \\
      &= E^U_\omega(\Lambda_1, n_1) + E^U_\omega(\Lambda_2, n_2) + \Tr(U
      \gamma^{(1)}_{\Psi^U_1} \otimes^s \gamma^{(2)}_{\Psi^U_2}) \\
      &= E^U_\omega(\Lambda_1, n_1) + E^U_\omega(\Lambda_2, n_2) +
      \int_{\Lambda_1 \times \Lambda_2} U(x - y) \rho_{\Psi^U_1}(x)
      \rho_{\Psi^U_2}(y) \rmd{x} \rmd{y}
    \end{split}
  \end{equation*}
  The proof will be accomplished by the following
  \begin{Le}
    \label{le:30}
    Under the assumptions of Theorem~\ref{thr:8}, one has
    \begin{equation}
      \label{eq:63}
      \int_{\Lambda_1 \times \Lambda_2} U(x - y) \rho_{\Psi^U_1}(x)
      \rho_{\Psi^U_2}(y) \rmd{x} \rmd{y} = o(n_1+n_2).
    \end{equation}
  \end{Le}
  \begin{proof}
    By Proposition~\ref{pro:3}, with probability
    $1-O(\min(|\Lambda_1|,|\Lambda_2|)^{-\infty})$, for $i \in \{1,
    2\}$, the largest piece in $\Lambda_i$ is of length bounded by
    $\log{|\Lambda_i|}\cdot\log\log{|\Lambda_i|}$. This implies that
    one can partition $\Lambda_i$ into sub-intervals each containing an
    integer number of original pieces (i.e., the extremities of these
    sub-intervals coincide with the extremities of pieces given by the
    Poisson random process) of length between $\log^2{|\Lambda_i|}$
    and $2 \log^2{|\Lambda_i|}$. Let these new sub-intervals be denoted
    by $\Lambda_i^j$, $j \in \{1, \hdots, m_i\}$; we order the
    intervals in such a way that their distance to $0$ increases with
    $j$. Thus,
    \begin{equation*}
      \Lambda_i = \bigcup_{j=1}^{m_i} \Lambda_i^j
    \end{equation*}
    and
    \begin{equation}
      \label{eq:263}
      \log^2{|\Lambda_i|} \leq |\Lambda_i^j| \leq 2 \log^2{|\Lambda_i|}.
    \end{equation}
    The last inequalities and the ordering convention imply that
    \begin{equation}
      \label{eq:149}
      \dist(\Lambda_1^{j_1}, \Lambda_2^{j_2}) \geq (j_1 - 1) \cdot
      \log^2{|\Lambda_1|} + (j_2 - 1) \cdot \log^2{|\Lambda_2|}
    \end{equation}
    and
    \begin{equation}
      \label{eq:150}
      \frac{|\Lambda_i|}{2 \log^2{|\Lambda_i|}} \leq m_i 
      \leq \frac{|\Lambda_i|}{\log^2{|\Lambda_i|}}.
    \end{equation}
    We now count the number of particles that $\Psi^U_i$ puts in an
    interval $\Lambda_i^j$.  Let $\{\Delta^i_k\}_{k = 1}^{M_i}$ be the
    pieces in $\Lambda_i$ and let $Q^i_k$ be the corresponding
    occupation numbers.  According to the choice of sub-intervals
    $\Lambda_i^j$ above, each $\Lambda_i^j$ is a union of some of the
    pieces $(\Delta^i_k)_k$.  We establish the following natural
    \begin{lemma}
      \label{le:13}
      With the above notations, one has
      \begin{equation*}
        \int_{\Delta^i_k} \rho_{\Psi^U_i}(x) \rmd{x} = Q^i_k, \quad i \in
        \{1, 2\}, \quad k \in \{1, \hdots, M_i\}.
      \end{equation*}
    \end{lemma}
    \begin{proof}
      For convenience, we drop the superscript $i$ in this proof.
      Recall the decomposition~\eqref{eq:146}
      \begin{equation*}
        \Psi = \sum_{\substack{(n_k)_{1 \leq k \leq M}\\\forall k,\ n_k
            \geq 1}} a_n \bigwedge_{k = 1}^M \varphi_{n_k}^k,
      \end{equation*}
      where $\varphi_{n_k}^k$ are functions of $Q_k$ variables in the
      piece $\Delta_k$. Keeping the notations, by Theorem~\ref{thr:4},
      one has
      \begin{equation*}
        \gamma_{\Psi}^{(1)} = \sum_{k = 1}^{M} 
        \sum_{\substack{n_k \geq 1\\n_k^\prime \geq 1}} \sum_{\wtn \in
          \bbN^{M - 1}} a_{\wtn_k} 
        \overline{a_{\wtn_k^\prime}} \gamma_{n_k, n_k^\prime}^{(1)},
      \end{equation*}
      where
      \begin{equation*}
        \gamma_{n_k, n_k^\prime}^{(1)}(x, y) 
        = Q_k \int_{(\Delta_k)^{Q_k - 1}} \varphi_{n_k}^k(x, z) 
        \overline{\varphi_{n_k^\prime}^k(y, z)} \rmd{z}.
      \end{equation*}
      The off-diagonal term $\gamma_{\Psi}^{(1), o}$ vanishes because
      the functions $\Psi_{1, 2}$ were chosen of a fixed occupation.
      This immediately yields
      \begin{equation*}
        \begin{split}
          \int_{\Delta_k} \rho_\Psi(x) \rmd{x} &= Q_k
          \sum_{\substack{n_k \geq 1\\n_k^\prime \geq 1}} \sum_{\wtn
            \in \bbN^{M - 1}} a_{\wtn_k} \overline{a_{\wtn_k^\prime}}
          \int_{(\Delta_k)^{Q_k}} \varphi_{n_k}^k(x)
          \overline{\varphi_{n_k^\prime}^k(x)} \rmd{x} \\
          &= Q_k \sum_{\wtn \in \bbN^{M - 1}} \int_{(\Delta_k)^{Q_k}}
          \sum_{n_k \geq 1} |a_{\wtn_k}|^2 |\varphi_{n_k}^k(x)|^2
          \rmd{x}= Q_k,
        \end{split}
      \end{equation*}
      where, in the second equality, we used the orthogonality of
      different $Q_k$-particles levels in the piece $\Delta_k$ and, in
      the third equality, we used the fact that $\Psi$ is normalized.\\
      This completes the proof of Lemma~\ref{le:13}.
    \end{proof}
    \noindent Lemma~\ref{le:13} immediately entails
    \begin{corollary}
      \label{cor:5}
      One computes 
      \begin{equation*}
        \int_{\Lambda_i^j} \rho_{\Psi^U_i}(x) \rmd{x} 
        = \sum_{k \mid \Delta^i_k \subset \Lambda_i^j} Q^i_k, \quad i \in
        \{1, 2\}, \quad j \in \{1, \hdots, m_i\}.
      \end{equation*}
    \end{corollary}
    \noindent Next, we derive a simple bound on the number of
    particles in $\Lambda_i^j$.  The total ground state energy is bounded by
    \begin{equation*}
      E^U_\omega(\Lambda_i, n_i) \leq C \ell_\rho^{-2} n_i.
    \end{equation*}
    From the other hand, a system of $\D q = \sum_{k \mid \Delta^i_k
      \subset \Lambda_i^j} Q^i_k$ particles in $\Lambda_i^j$ has
    non interacting energy at least
    \begin{equation*}
      \sum_{s = 1}^q \frac{\pi^2 s^2}{|\Lambda_i^j|^2} \asymp q^3 |\Lambda_i^j|^{-2}.
    \end{equation*}
    This implies that
    \begin{equation*}
      q^3 |\Lambda_i^j|^{-2} \leq C \ell_\rho^{-2} n_i
    \end{equation*}
    or, equivalently,
    \begin{equation*}
      \sum_{k \mid \Delta^i_k \subset \Lambda_i^j} Q^i_k 
      \leq C_1 \left(|\Lambda_i^j| / \ell_\rho\right)^{2 / 3} n_i^{1/3}
      \leq C_2  n_i^{1/3}\log^{4/3}L_i.
    \end{equation*}
    Let us now estimate the left hand side of~\eqref{eq:63} using
    H{\"o}lder's inequality ($1/p+1/q=1$, $p,q\geq1$) as
    \begin{equation}
      \label{eq:74}
      \begin{split}
        \int_{\Lambda_1 \times \Lambda_2} U(x - y) \rho_{\Psi^U_1}(x)
        \rho_{\Psi^U_2}(y) \rmd{x} \rmd{y} &= \sum_{j_1 = 1}^{m_1}
        \sum_{j_2 = 1}^{m_2} \int_{\Lambda_1^{j_1} \times
          \Lambda_2^{j_2}} U(x - y) \rho_{\Psi^U_1}(x)
        \rho_{\Psi^U_2}(y) \rmd{x} \rmd{y} \\
        &\leq \sum_{j_1 = 1}^{m_1} \sum_{j_2 = 1}^{m_2}
        \|U\|_{p,\Lambda_1^{j_1} \times \Lambda_2^{j_2}}
        \|\rho_{\Psi^U_1}\|_q \|\rho_{\Psi^U_2}\|_q.
      \end{split}
    \end{equation}
    where we have set
    \begin{equation}
      \label{eq:262}
      \|U\|_{p,\Lambda_1^{j_1} \times \Lambda_2^{j_2}}:=
      \left(\int_{\Lambda_1^{j_1} \times \Lambda_2^{j_2}} U^p(x -
        y)dxdy\right)^{1/p}.     
    \end{equation}
    Now, recall that by~\eqref{eq:73}, for $i\in\{1,2\}$, on
    $\Lambda_i^{j_i}$, one has
    \begin{equation*}
      \|\rho_{\Psi^U_i}\|_{\infty,\Lambda_i^{j_i}}\leq
      4\|\Psi^U_i\|_{H^1(\Lambda_i^{j_i})}\|\Psi^U_i\|_{2,\Lambda_i^{j_i}}\leq
      C\left(\langle
        H^U_\omega(\Lambda_i^{j_i},n_i)\Psi^U_i,\Psi^U_i\rangle_{\Lambda_i^{j_i}}\right)^{1/2} 
      \|\Psi^U_i\|_2.
    \end{equation*}
    Hence, by Corollary~\ref{cor:5},
    \begin{equation*}
      \|\rho_{\Psi^U_i}\|_q=
      \left(\int_{\Lambda_i^{j_i }}\rho^{q-1}_{\Psi^U_i}\rho_{\Psi^U_i}
      \right)^{1/q}\leq (Q_i^{j_i})^{1/q}\left(\langle
        H^U_\omega(\Lambda_i^{j_i},n_i)\Psi^U_i,\Psi^U_i\rangle_{\Lambda_i^{j_i}}\right)^{(q-1)/2q} 
      \|\Psi^U_i\|^{(q-1)/q}_{2,\Lambda_i^{j_i}}.
    \end{equation*}
    Recalling~\eqref{eq:74}, as $\|\Psi^U_i\|_{2,\Lambda_i^{j_i}}\leq1$ for
    $i\in\{1,2\}$, we estimate
    \begin{multline}
        \label{eq:227}
        \int_{\Lambda_1 \times \Lambda_2} U(x - y) \rho_{\Psi^U_1}(x)
        \rho_{\Psi^U_2}(y) \rmd{x} \rmd{y}\\\leq \sum_{j_1 = 1}^{m_1}
        \sum_{j_2 = 1}^{m_2}\|U\|_{p,\Lambda_1^{j_1} \times
          \Lambda_2^{j_2}} (Q_1^{j_1}Q_2^{j_2})^{1/q}\left(\langle
          H^U_\omega(\Lambda_1^{j_1},n_1)\Psi^U_1,\Psi^U_1\rangle_{\Lambda_1^{j_1}}
          \langle
          H^U_\omega(\Lambda_2^{j_2},n_2)\Psi^U_2,\Psi^U_2\rangle_{\Lambda_2^{j_2}}
        \right)^{(q-1)/2q}.
    \end{multline}
    Now, as $Q_{\Psi^U_i}\lesssim n_i^{1/3}\log^{4/3}L_i\lesssim n^{1/3}\\log^{4/3}n$
    and as
    \begin{equation*}
      \langle H^U_\omega(\Lambda_i^{j_i},n_i)\Psi^U_i,
      \Psi^U_i\rangle_{\Lambda_i^{j_i}}\leq \langle
      H^U_\omega(\Lambda_i)\Psi^U_i,\Psi^U_i\rangle\leq C n_i\leq C n,
    \end{equation*}
    the estimate~\eqref{eq:227} entails
    \begin{equation}
      \label{eq:251}
      \int_{\Lambda_1 \times \Lambda_2} U(x - y) \rho_{\Psi^U_1}(x)
      \rho_{\Psi^U_2}(y) \rmd{x} \rmd{y}\lesssim n^{(3q-1)/3q}
      \left(\log n\right)^{8/(3q)}
      \sum_{j_1 = 1}^{m_1}\sum_{j_2 = 1}^{m_2}
      \|U\|_{p,\Lambda_1^{j_1} \times\Lambda_2^{j_2}}.
    \end{equation}
    Hence, to prove~\eqref{eq:62}, it suffices to choose $q$ (recall
    $q\geq1$ and $1/p+1/q=1$) such that
    \begin{equation}
      \label{eq:261}
      \sum_{j_1 = 1}^{m_1}\sum_{j_2 = 1}^{m_2}
      \|U\|_{p,\Lambda_1^{j_1} \times
        \Lambda_2^{j_2}}=o\left(n^{1/3q}\left(\log n\right)^{-8/(3q)}\right).
    \end{equation}
    Therefore, we recall~\eqref{eq:262} and using the definition of
    the $(\Lambda_i^{j_i})_{i,j}$, in particular~\eqref{eq:149}
    and~\eqref{eq:150}, we estimate
    \begin{equation}
      \label{eq:72}
      \|U\|_{p,\Lambda_1^{j_1} \times
        \Lambda_2^{j_2}}\lesssim ((j_1+j_2)|\log L|^2)^{-k/p}
      \left(\int_{\Lambda_1^{j_1} \times \Lambda_2^{j_2}} (x-y)^k U^p(x -
        y)dxdy\right)^{1/p}.
    \end{equation}
    Now, by~\eqref{eq:263}, as $U$ is even, we have
    \begin{equation}
     \label{eq:264}
     \left(\int_{\Lambda_1^{j_1} \times \Lambda_2^{j_2}} (x-y)^k U^p(x -
        y)dxdy\right)^{1/p}\lesssim (\log n)^{2/p}
      \left(\int_{\R^+}u^k U^p(u)du\right)^{1/p}.
    \end{equation}
    On the other hand, if $k/p>1$ and $\max(m_1,m_2)\lesssim L/\log L
    \lesssim n/\log n$ (with a good probability), one estimates
    \begin{equation*}
      \sum_{\substack{1\leq j_1\leq m_1\\ 1\leq j_2 \leq m_2}}
      (j_1+j_2)^{-k/p}\leq (\log n)^{k/p-2}n^{2-k/p}. 
    \end{equation*}
    Plugging this,~\eqref{eq:264} and~\eqref{eq:72} into the sum
    in~\eqref{eq:261}, we see that~\eqref{eq:261} is a consequence of
    \begin{equation*}
      (\log n)^{2-2/p+8/(3q)}n^{2-k/p-1/(3q)}=
      (\log n)^{14/3(p-1)/p}n^{5/3-(3k-1)/(3p)}=o(1).
    \end{equation*}
    as $p^{-1}+q^{-1}=1$.\\
    Thus, it suffices to find $k>0$, $p>1$ such that $u\mapsto u^{k/p}
    U(u)$ be in $L^p(\R^+)$ and
    \begin{equation*}
      \frac53-\frac{3k-1}{3p}<0.
    \end{equation*}
    Recall that, by assumption $u\mapsto u^\alpha U(u)$ is integrable
    (for some $\alpha>2$) and $U\in L^r(\R^+)$ for some $r>1$.\\
    We pick $\eta\in(0,1)$ and pick $p$ and $k$ of the form $\D
    p=1+\eta(r-1)$ and $\D k=\frac{5p+1}3+\eta$. Thus, for $\D
    r\in\left(1, \min\left(\tilde r,2\right)\right]$, setting $\D
    \tilde p:=\frac{r-p}{r-1}\in(0,1)$, we have
    \begin{equation*}
      \frac53-\frac{3k-1}{3p}=-\frac{\eta}p<0,\quad \frac{p-\tilde p}{1-\tilde p}=r
      \quad\text{and}\quad
      \frac{k}{\tilde p}=k\frac{r-1}{r-p}
      =\left(2+\frac53\eta(r-1)\right)\frac1{1-\eta}=\alpha
    \end{equation*}
    for $\eta\in(0,1)$ well chosen.\\
    For this choice of $p$, $\tilde p$ and $k$, using H{\"o}lder's
    inequality, we then estimate
    \begin{equation*}
      \int_{\R^+}u^k U^p(u)du\leq
      \left(\int_{\R^+}u^{k/\tilde p}U(u)du\right)^{\tilde p}
      \left(\int_{\R^+}U^{(p-\tilde p)/(1-\tilde
          p)}(u)du\right)^{1-\tilde p}<+\infty
    \end{equation*}
    This completes the proof of~\eqref{eq:261} and, thus, of
    Lemma~\ref{le:30}.
  \end{proof}
  \noindent Lemma~\ref{le:13} implies that, under the assumption of
  Theorem~\ref{thr:8}, in the thermodynamic limit, with probability
  exponentially close to 1, one has
  \begin{equation*}
    \int_{\Lambda_1 \times \Lambda_2} U(x - y) \rho_{\Psi^U_1}(x)
    \rho_{\Psi^U_2}(y) \rmd{x} \rmd{y} = o(n_1 + n_2).
  \end{equation*}
  This completes the proof of Theorem~\ref{thr:8}.
\end{proof}
%

%%% Local Variables: 
%%% mode: latex
%%% TeX-master: "PiecesModelGroundState"
%%% ispell-local-dictionary: "american"
%%% End: 

\section{Multiple electrons interacting in a fixed number of pieces}
\label{sec:two-part-probl}
The main goal of this section is to study a system of two interacting
electrons in the interval $[0, \ell]$ for large $\ell$ and prove
Proposition~\ref{prop:TwoElectronProblem}; this is the purpose of
section~\ref{sec:two-int-electrons}. The two-particles Hamiltonian is
given by \eqref{eq:introTwoParticleHamiltonian}. In
section~\ref{sec:ferm-neighb-piec}, we study two electrons in two
distinct pieces.\\
We shall also state and prove one result for more than two interacting
electrons in a single piece.
\subsection{Two electrons in the same piece}
\label{sec:two-int-electrons}
We now study two electrons in a large interval interacting through a
pair potential $U$, that is, the Hamiltonian defined
in~\eqref{eq:introTwoParticleHamiltonian}. We first
Proposition~\ref{prop:TwoElectronProblem}. Next, in
section~\ref{sec:bounds-ground-state}, we compare the ground state of
the interacting system with that of the non-interacting system.\\
Throughout this section, we will assume $U$ is a repulsive, even pair
interaction potential. In the present section, our assumptions on $U$
will be weaker than \textbf{(HU)}.
\subsubsection{The proof of Proposition~\ref{prop:TwoElectronProblem}}
\label{sec:proof-prop-refpr}
Scaling variables to the unit square, the two-particles Hamiltonians
$H^U(\ell, 2)$ and $\ell^{-2} H^{U^\ell}(1, 2)$ are unitarily
equivalent. Here, we have defined
\begin{equation}
  \label{eq:26}
  U^{\ell}(\cdot) := \ell^2 U(\ell\,\cdot). 
\end{equation}
Recall that, for $i \ne j$, $i, j \in \bbN$, the normalized
eigenfunctions of $H^0(1, 2)$ (i.e., of the two-particles free
Hamiltonian in a unit square) are given by the determinant
\begin{equation}
  \label{eq:21}
  \phi_{(i,j)}(x, y) = \sqrt{2}
  \left|
    \begin{matrix}
      \sin(\pi i x) & \sin(\pi j x)\\ \sin(\pi i y) &\sin(\pi j y)
    \end{matrix}\right|\text{ for } (x,y)\in[0,1]^2.
\end{equation}
For a two-component index, we will use the shorthand notation
$\bar{\imath}=(i,j)$.  For the non interacting ground state
$\phi_{(1,2)}$ we will also use the notation $\phi_0$. The
corresponding ground state energy
is $5\pi^2$ and the first excited energy level is at $10 \pi^2$.\\
We decompose $\D L^2([0,1])\wedge L^2([0,1])=\C\phi_0\overset{\perp}{\oplus}
\phi^\perp_0$. By the Schur complement formula, $E$ is the ground
state energy of $H^{U^\ell}(1,2)$ if and only if $E<10\pi^2$ and $E$
satisfies
\begin{equation}
  \label{eq:twoParticlesEigenvalueEquation}
  5 \pi^2 + U^\ell_{00} - E = U^\ell_{0+} (H_+ + U^\ell_{++} - E)^{-1}
  U^\ell_{+0} \text{,}
\end{equation}
where  $\Pi_+$ is the  orthogonal projector on  $\phi^\perp_0$
and
\begin{equation}
  \label{eq:44}
  \begin{aligned}
  U^\ell_{00}&=\langle\phi_0,U^\ell\phi_0\rangle,\quad 
  H_+=\Pi_+H^0(1,2)\Pi_+,\\
  U^\ell_{++}&=\Pi_+U^\ell\Pi_+,\quad U^\ell_{+0}=\Pi_+U^\ell\phi_0    
  \quad U^\ell_{0+}=\left(\Pi_+U^\ell\phi_0\right)^*.    
  \end{aligned}
\end{equation}
We expand the r.h.s. of \eqref{eq:twoParticlesEigenvalueEquation} as
\begin{equation}
  \begin{split}
    U^\ell_{0+} (H_+ + U^\ell_{++} - E)^{-1} U^\ell_{+0} & =
    \bigl\langle U^\ell \phi_0, (H_+ - E)^{-1/2}
    \\&\hskip2cm\times \left(\Id + (H_+ - E)^{-1/2}
      U^\ell (H_+ - E)^{-1/2}\right)^{-1} \\
    &\hskip4cm\times (H_+ - E)^{-1/2} U^\ell \phi_0\bigr\rangle \\&=
    \frac{1}{\ell} \left\langle \wtphi_\ell, A_\ell (\Id + A_\ell^*
      A_\ell)^{-1} A_\ell^* \wtphi_\ell \right\rangle
    \\
    &= \frac{1}{\ell} \left\langle \wtphi_\ell, A_\ell A_\ell^* (\Id +
      A_\ell A_\ell^*)^{-1} \wtphi_\ell \right\rangle \text{,}
  \end{split}
\end{equation}
where
\begin{equation}\label{eq:phiEllAEllNotations}
  \wtphi_\ell = \sqrt{\ell} \sqrt{U^\ell} \phi_0 
  \quad \text{and} \quad
  A_\ell = A_\ell(E) = \sqrt{U^\ell} (H_+ - E)^{-1/2} \text{.}
\end{equation}
To simplify notations we will drop the reference to the energy $E$. As
$\ell\to+\infty$, the convergence of $\left\langle \wtphi_\ell, A_\ell
  A_\ell^* (\Id + A_\ell A_\ell^*)^{-1} \wtphi_\ell \right\rangle$ is
locally uniform in $(-\infty,10\pi^2)$. To compute this limit, we
shall transform the expression $\left\langle \wtphi_\ell, A_\ell
  A_\ell^* (\Id + A_\ell A_\ell^*)^{-1} \wtphi_\ell \right\rangle$
once more.\\
Consider the domain $R_\ell = \{(u, y) \in \bbR \times [0, 1], \text{s.t. }
y + \ell^{-1} u \in [0, 1]\}$ and the change of variables
\begin{align*}
  t_\ell : R_\ell &\to [0, 1]^2 \\
  (u, y) &\mapsto \left(y + \frac{u}{\ell}, y\right).
\end{align*}
Define the partial isometry
\begin{align*}
  T_\ell : L^2([0, 1]^2) &\to L^2(\R\times[0,1]) \\
  v &\mapsto \ell^{-1/2} \car_{R_\ell}\cdot v \circ t_\ell \text{,}
\end{align*}
that is, $\D(T_\ell v)(u, y) = \frac{1}{\sqrt{\ell}}
\car_{R_\ell}(u,y) v\left(y + \frac{u}{\ell}, y\right)$.\\
One computes its adjoint
\begin{align*}
T^*_\ell :  L^2(\R\times[0,1]) &\to L^2([0, 1]^2) \\
v &\mapsto \ell^{1/2} (\car_{R_\ell}\,v) \circ t^{-1}_\ell \text{,}
\end{align*}
that is, $\D(T^*_\ell v)(x, y) = \sqrt{\ell} (\car_{R_\ell}\cdot
v)(\ell(x-y),y)$.\\
One easily checks that
\begin{equation}
  \label{eq:10}
  T_\ell T^*_\ell=\car_{R_\ell}\quad\text{ and }\quad T^*_\ell
  T_\ell=\text{Id}_{L^2([0,1]^2)}
\end{equation}
where $\car_{R_\ell}:\; L^2(\R\times[0,1]) \to L^2(\R\times[0,1])$ is
the orthogonal projector on the functions supported in $R_\ell$.\\
One then computes
\begin{equation}
  \label{eq:rhsUnitaryTransformation}
  \left\langle \wtphi_\ell, A_\ell A_\ell^* (\Id +
    A_\ell A_\ell^*)^{-1} \wtphi_\ell \right\rangle_{L^2([0, 1]^2)}
  = \left\langle \phi_\ell, K_\ell (\Id + K_\ell)^{-1} \phi_\ell
  \right\rangle_{L^2(\R\times[0,1])}
\end{equation}
where we have defined
\begin{equation}
  \label{eq:KellOperatorDefinition}
  \phi_\ell:=T_\ell
  \wtphi_\ell \quad\text{and}\quad
  K_\ell :=K_\ell(E) := T_\ell A_\ell A_\ell^* T_\ell^* \text{.}
\end{equation}
Define
\begin{itemize}
\item the following functions
  \begin{itemize}
  \item $\phi(u):=u\sqrt{U(u)}$ for $u\in\R$,
  \item $\D\chi_0(y):=\pi \sqrt{2} \left(\sin\left(3\pi
        y\right)-3\sin\left(\pi y\right)\right)$ for $y\in[0,1]$.
  \end{itemize}
\item the non negative (see~\eqref{eq:58}) operator $K$ is on
  $L^2(\R)$ by the kernel
  \begin{equation}
    \label{eq:56}
    K(u,u')=\frac{1}{2} \sqrt{U(u)}(|u + u'| - |u - u'|)
    \sqrt{U(u')}.
  \end{equation}
\end{itemize}
Define also
\begin{equation}
  \label{eq:71}
  \tilde \phi=\phi\otimes\chi_0\quad\text{ and }\quad \tilde
  K=K\otimes\text{Id}. 
\end{equation}
We prove
\begin{Le}
  \label{le:4}
  Assume that $U$ is non negative and even such that $U\in L^p(\R)$
  for some $p>1$ and $x\mapsto x^2 U(x)$ is integrable.\\
  As $\ell\to+\infty$, one has:
  \begin{enumerate}
  \item in $L^2(\R\times[0,1])$, $\phi_\ell$ converges to
    $\tilde\phi$;
  \item for $\varphi\in\Coi(\R\times(0,1))$, as $\ell\to+\infty$, the
    sequence $(K_\ell\varphi)_\ell$ converges in $L^2$-norm to $\tilde
    K\varphi$
  \end{enumerate}
\end{Le}
\noindent Proposition~\ref{prop:TwoElectronProblem} follows from this
result as we shall see now. First, we prove
\begin{Le}
  \label{le:7}
  Under the assumptions of Lemma~\ref{le:4}, all the operators
  $(K_\ell)_\ell$ and the operator $K$ are bounded respectively on
  $L^2(\R\times[0,1])$ and $L^2(\R)$.
\end{Le}
\noindent Note however that, depending on $U$, one may have
  \begin{equation*}
    \|K_\ell\|_{L^2(\R\times[0,1])\to L^2(\R\times[0,1])}
    \vers{\ell\to+\infty}+\infty.
  \end{equation*}
\begin{proof}
  By~\eqref{eq:KellOperatorDefinition}, to show the boundedness of
  $K_\ell$, it suffices to show that $\tilde
  K_\ell:=\sqrt{U_\ell}(H_+-E)^{-1}\sqrt{U_\ell}$ is bounded. Note
  that, by our assumption on $U$, $U_\ell$ is in $L^p([0,1]^2)$. Using
  the eigenfunction expansion of $-\Delta$ on $L_-^2([0,1]^2)$, we
  write
  \begin{equation}
    \label{eq:30}
    \tilde K_\ell=\sum_{\overline{j}\not=(2,1)}\frac1
    {\pi^2 |\barj|^2 - E}\sqrt{U_\ell}\phi_{\barj}\otimes
    \phi_{\barj}\sqrt{U_\ell}
  \end{equation}
  where the sum is over $\overline{j}=(i,j)$ where $(i,j)\in\N$ such
  that $i>j$.\\
  For $u\in L_-^2([0,1]^2)$, as $u\sqrt{U_\ell}\in
  L_-^{2p/(1+p)}([0,1]^2)$ and as the functions
  $(\phi_{\barj})_{\overline{j}}$ are uniformly bounded, by the
  Hausdorff-Young inequality (see e.g.~\cite{MR924157}), one has
  \begin{equation}
    \label{eq:54}
    \sum_{\overline{j}}\left|\left\langle\sqrt{U_\ell}\phi_{\barj},
        u\right\rangle\right|^{p/(p-1)}\leq C_\ell \|u\|_2^{p/(p-1)}.
  \end{equation}
  Moreover, for some $C_\ell$, one has
  $\|\sqrt{U_\ell}\phi_{\barj}\|_2\leq C_\ell$. Thus,
  by~\eqref{eq:30}, as $p>1$, we obtain
  \begin{equation*}
    \|\tilde K_\ell u\|_2\leq C_\ell
    \left(\sum_{\overline{j}\not=(2,1)}\frac1
      {(\pi^2 |\barj|^2 - E)^p}
    \right)^{1/p} \|u\|_2\leq C_\ell \|u\|_2.
  \end{equation*}
  Using the explicit kernel for $K$ given in~\eqref{eq:56}, for $u\in
  L^2(\R)$, we compute
  \begin{equation*}
    (Ku)(x)=2\sqrt{U(x)}\int_{-x}^xx'\sqrt{U(x')}u(x')dx'+
    2\sqrt{U(x)}x\int_{-\infty}^{-x}\sqrt{U(x')}(u(x')-u(-x'))dx'
  \end{equation*}
  Thus,
  \begin{equation*}
    \|K\|_{\mathcal{L}(L^2(\R))}\leq 4\sqrt{\|U\|_1\|(\cdot)^2 U(\cdot)\|_1}.
  \end{equation*}
  This completes the proof of Lemma~\ref{le:7}.
\end{proof}
\noindent By Lemma~\ref{le:7}, $\Coi(\R\times(0,1))$ is a common core
for all $K_\ell$ and $K\otimes\Id$. Thus, by~\cite[Theorem
VIII.25]{MR85e:46002}, we know that $\D K_\ell\vers{\ell\to+\infty}
K\otimes\Id$ in the strong resolvent sense. Hence, by~\cite[Theorem
VIII.20]{MR85e:46002}, the sequence $(K_\ell (\Id +
K_\ell)^{-1})_\ell$ converges to $K(\Id + K)^{-1}\otimes\Id$
strongly. These operators are all bounded uniformly by $1$ (as
$K_\ell$ and $K$ are non-negative). Thus, by point (a) of
Lemma~\ref{le:4} and~\eqref{eq:rhsUnitaryTransformation}, we obtain
\begin{equation}
  \label{eq:20}
  \begin{split}
    \left\langle \tilde\phi_\ell, A_\ell A_\ell^\star (\Id + A_\ell
      A_\ell^\star)^{-1} \tilde\phi_\ell \right\rangle &= \langle
    \phi\otimes\chi_0, \left[K (\Id + K)^{-1} \otimes
      \Id\right] \phi\otimes\chi_0 \rangle + o(1) \\
    &= \left\langle \phi, K (\Id + K)^{-1} \phi\right\rangle \cdot \int_0^1
    \chi_0^2(y) \rmd{y} + o(1)\\
    &=\pi^2\cdot\left\langle \phi, K (\Id + K)^{-1}\phi\right\rangle +
    o(1) \text{.}
  \end{split}
\end{equation}
By point (a) of Lemma~\ref{le:4}, one also computes
\begin{equation}
  \label{eq:phiEllNorm}
  \begin{split}
    \ell\, U^\ell_{00}&= \|\phi\otimes\chi_0\|^2+o(1) = \int_{\bbR}
    u^2 U(u) \rmd{u} \int_0^1 \chi_0^2(y) \rmd{y} + o(1)\\&=
    \frac52\pi^2\int_{\bbR} u^2 U(u)\rmd{u} + o(1)
  \end{split}
\end{equation}
By~\eqref{eq:phiEllNorm}, the eigenvalue
equation~\eqref{eq:twoParticlesEigenvalueEquation} yields that, under
the assumptions of Lemma~\ref{le:7}, the ground state energy of
$H^{U^\ell}(1,2)$ satisfies
\begin{equation}
  \label{eq:61}
  E^{U^\ell}([0,1],2)=5\pi^2+\frac{\gamma(U)}{\ell}+o\left(\frac1{\ell}\right)
\end{equation}
where 
\begin{equation}
  \label{eq:31}
  \gamma(U)=10 \pi^2\left[\|\phi\|^2 - \left\langle
      \phi,K(\Id+K)^{-1}\phi\right\rangle\right]=10 \pi^2\left\langle
    \phi,(\Id+K)^{-1}\phi\right\rangle.
\end{equation}
By Lemma~\ref{le:31} and assumption \textbf{(HU)}, we know that the
assumptions of Lemma~\ref{le:7} are satisfied.  This proves the
asymptotic expansion announced in
Proposition~\ref{prop:TwoElectronProblem}. To complete the proof of
this proposition, we simply note that, as $K$ is bounded by
Lemma~\ref{le:7}, by~\eqref{eq:31}, we know that $\gamma(U)=0$ if and
only if $\phi\equiv0$, i.e., if and only if $U\equiv0$..\qed
\begin{Rem}
  \label{rem:5}
  If one assumes $x\mapsto x^4U(x)$ to be integrable and $U$ to be in
  some $L^p(\R)$ ($p>1$) (which is clearly stronger than
  \textbf{(HU)}), one obtains that, $E^U([0, \ell],2)$, the ground
  state energy of the Hamiltonian defined
  in~\eqref{eq:introTwoParticleHamiltonian} admits the following more
  precise expansion
  \begin{equation}
    E^{U^\ell}([0,1],2) = 5 \pi^2 + \frac{\gamma(U)}{\ell} +
    O\left(\ell^{-2}\right).
  \end{equation}
\end{Rem}
\subsubsection{The proof of Lemma~\ref{le:4}}
\label{sec:proof-lemma-refle:4}
We start with a lemma, the result of a computation, that will be used
in several parts of the proof.
\begin{Le}
  \label{le:5}
  For $\barj = (j_1, j_2)$, $j_1>j_2$, recall that $\phi_{\barj}$, the
  $\barj$-th normalized eigenvector of $H_0$, is given
  by~\eqref{eq:21}.\\
  One has
  \begin{equation}
    \label{eq:phiBarjDecomposition}
    \phi_{\barj}\left(y + \frac{u}{\ell}, y\right)
    =\phi^0_{\barj}\left(\frac{u}{\ell}, y\right)+
    \phi^+_{\barj}\left(\frac{u}{\ell}, y\right)+
    \phi^-_{\barj}\left(\frac{u}{\ell}, y\right)
  \end{equation}
  where
  \begin{equation}
    \label{eq:22}
    \begin{aligned}
      \phi^0_{\barj}(2x, y)&:=2\sqrt{2}\sin(\pi (j_1+j_2)x)
      \sin(\pi(j_2-j_1)x)\sin(\pi j_1 y)\sin(\pi j_2 y),\\
      \phi^+_{\barj}(2x, y)&:=\sqrt{2}
      \cos(\pi (j_2-j_1)x)\sin(\pi(j_2+j_1)x)\sin(\pi(j_1-j_2)y) \\
      \phi^-_{\barj}(2x,y)&:=\sqrt{2}
      \cos(\pi(j_2+j_1)x)\sin(\pi(j_2-j_1)x)\sin(\pi(j_1+j_2)y)
    \end{aligned}
  \end{equation}
\end{Le}
\begin{proof}
  Using standard sum and product formulas for the sine and cosine, we compute
  \begin{equation*}
    \begin{split}
     \frac{1}{\sqrt{2}} \phi_{\barj}\left(y + \frac{u}{\ell}, y\right) &=
      \left|
        \begin{matrix}
          \sin\left(\pi j_1 \left(y + \frac{u}{\ell}\right)\right) &
          \sin\left(\pi j_1 y\right) \\
          \sin\left(\pi j_2 \left(y + \frac{u}{\ell}\right)\right) &
          \sin\left(\pi j_2 y\right)
        \end{matrix}\right|\\
      &= \sin\left(\pi j_1 \frac{u}{\ell}\right) \cos\left(\pi j_1
        y\right) \sin\left(\pi j_2 y\right) - \sin\left(\pi j_2
        \frac{u}{\ell}\right) \cos\left(\pi j_2 y\right) \sin\left(\pi
        j_1 y\right) \\&\hskip2cm+ \left(\cos\left(\pi j_1
          \frac{u}{\ell}\right) - \cos\left(\pi j_2
          \frac{u}{\ell}\right)\right) \sin\left(\pi j_1 y\right) \sin\left(\pi j_2 y\right)\\
      &=\frac{1}{2} \sin\left(\pi j_1 \frac{u}{\ell}\right)
      \left(\sin\left(\pi \left(j_1 +
            j_2\right) y\right) - \sin\left(\pi \left(j_1 - j_2\right) y\right)\right) \\
      &\hskip1cm- \frac{1}{2} \sin\left(\pi j_2 \frac{u}{\ell}\right)
      \left(\sin\left(\pi \left(j_1 + j_2\right)
          y\right) + \sin\left(\pi \left(j_1 - j_2\right) y\right)\right) \\
      &\hskip2cm+ \left(\cos\left(\pi j_1
          \frac{u}{\ell}\right) - \cos\left(\pi j_2
          \frac{u}{\ell}\right)\right) 
          \sin\left(\pi j_1 y\right) \sin\left(\pi j_2 y\right) \\
      &=\frac{1}{2} \left(\sin\left(\pi j_1 \frac{u}{\ell}\right) -
        \sin\left(\pi j_2 \frac{u}{\ell}\right)\right) \sin\left(\pi
        \left(j_1 +
          j_2\right) y\right) \\
      &\hskip1cm- \frac{1}{2} \left(\sin\left(\pi j_1
          \frac{u}{\ell}\right) + \sin\left(\pi j_2
          \frac{u}{\ell}\right)\right) \sin\left(\pi \left(j_1 -
          j_2\right) y\right) \\
      &\hskip2cm+ \left(\cos\left(\pi j_1
          \frac{u}{\ell}\right) - \cos\left(\pi j_2
          \frac{u}{\ell}\right)\right) 
          \sin\left(\pi j_1 y\right) \sin\left(\pi j_2 y\right).
        \end{split}
      \end{equation*}
      Thus,
      \begin{equation*}
        \begin{split}
          \frac{1}{\sqrt{2}} \phi_{\barj}\left(y + \frac{u}{\ell},
            y\right) &= \sin\left(\pi \frac{j_1 - j_2}{2}
            \frac{u}{\ell}\right) \cos\left(\pi \frac{j_1 + j_2}{2}
            \frac{u}{\ell}\right)
          \sin\left(\pi \left(j_1 + j_2\right) y\right) \\
          &\hskip1cm- \sin\left(\pi \frac{j_1 + j_2}{2}
            \frac{u}{\ell}\right) \cos\left(\pi \frac{j_1 - j_2}{2}
            \frac{u}{\ell}\right) \sin\left(\pi \left(j_1 -
              j_2\right) y\right) \\
          &\hskip2cm- 2 \sin\left(\pi \frac{j_1 - j_2}{2}
            \frac{u}{\ell}\right) \sin\left(\pi \frac{j_1 + j_2}{2}
            \frac{u}{\ell}\right) \sin\left(\pi j_1 y\right)
          \sin\left(\pi j_2 y\right) \text{.}
        \end{split}
  \end{equation*}
   This completes the proof of Lemma~\ref{le:5}.
\end{proof}
\noindent We start with the proof of point (a) of Lemma~\ref{le:4}. As
$\phi_0=\phi_{(2,1)}$, by~\eqref{eq:phiBarjDecomposition}
and~\eqref{eq:22}, using the Taylor expansion of the sine and cosine
near $0$, we compute
\begin{equation*}
  \begin{split}(T_\ell \tilde\phi_\ell)(u, y)
    &= \ell \sqrt{U(u)}\car_{R_\ell}(u,y)
    \phi_{(2,1)}\left(y + \frac{u}{\ell}, y\right)\\ &=
    u \sqrt{U(u)} \chi_0(y)\car_{R_\ell}(u,y) + \frac{u^2}{\ell} \sqrt{U(u)}
    \chi_1\left(\frac{u}{\ell}, y\right)\car_{R_\ell}(u,y)
  \end{split}
\end{equation*}
where $\chi_0$ is defined in Lemma~\ref{le:4} and $ \chi_1$ is
continuous and bounded on $\R\times[0,1]$.\\
We estimate
\begin{equation*}
  \left\|\frac{(\cdot)^2}{\ell} \sqrt{U(\cdot)}
    \chi_1\left(\frac{\cdot}{\ell},
      \cdot\right)\car_{R_\ell}\right\|_{L^2(\R\times[0,1])}^2
  \lesssim\int_{R_\ell}\frac{u^2}{\ell^2}u^2 U(u)dudy
  \leq\int_\R\frac{u^2\car_{|u|\leq\ell}}{\ell^2}u^2 U(u)du.
\end{equation*}
The last integral tends to $0$ by the dominated convergence
theorem as $u\mapsto u^2 U(u)$ is integrable.\\
This completes the proof of point (a) of Lemma~\ref{le:4}.\\
Let us now turn to the analysis of the operator family
$(K_\ell)_\ell$. It is easily seen that its kernel (we use the same
notations for the operator and its kernel) is given by
\begin{equation*}
  K_\ell(E; u, y, u', y') = \ell\car_{R_\ell\times R_\ell}
  \sqrt{U(u) U(u')} \cdot \tilde K\left(E; y + \frac{u}{\ell}, y, y'
    + \frac{u'}{\ell},y'\right)
\end{equation*}
where $\tilde K(E; x, y, x', y')$ is the kernel of $(H_+ -
E)^{-1}$. The kernel $\tilde K(E)$ is easily expressed in terms of the
eigenfunctions of $H$. Using this and the representation yielded by
Lemma~\ref{le:5} leads to the following representation for the kernel
$K_\ell$
\begin{equation}
  \label{eq:23}
  \begin{split}
    K_\ell\left(E; u, y, u', y'\right)& = \ell \car_{R_\ell\times
      R_\ell}\sum_{{\barj} \ne (2,1)} \frac{\sqrt{U\left(u\right)
        U\left(u'\right)}}{\pi^2 |\barj|^2 - E} \phi_{\barj}\left(y +
      \frac{u}{\ell}, y\right)
    \phi_{\barj}\left(y' + \frac{u'}{\ell}, y'\right) \\
    &= K_{\ell}^-\left(E; u, y, u', y'\right) + K_{\ell}^+\left(E; u,
      y, u', y'\right)+ K_{\ell}^0\left(E; u, y, u', y'\right)
  \end{split}
\end{equation}
where, for $\bullet\in\{0,+,-\}$, we have set
\begin{equation*}
  K_\ell^\bullet\left(E; u, y, u', y'\right)
  =\ell\car_{R_\ell\times R_\ell} \sum_{{\barj} \ne (2,1)}
  \frac{\sqrt{U\left(u\right) U\left(u'\right)}}{\pi^2 |\barj|^2 -
    E}\phi_{\barj}\left(y + \frac{u}{\ell}, y\right)
  \phi^\bullet_{\barj}\left(\frac{u'}{\ell},y\right).
\end{equation*}
To prove point (b) of Lemma~\ref{le:4}, if suffices to prove that, for
$v\in\Coi(\R\times(0,1))$, one has $K_lv\to \tilde Kv$ in
$L^2(\R\times[0,1])$.
We first prove
\begin{Le}
  \label{le:6}
  For $v\in\Coi(\R\times(0,1))$, one has
  \begin{enumerate}
  \item $\|K_\ell^-v\|_2\to0$ as $\ell\to+\infty$,
  \item $\|K_\ell^0v\|_2\to0$ as $\ell\to+\infty$.
  \end{enumerate}
\end{Le}
\begin{proof}
  We first study the sequence $K_\ell^+v$. We compute
  \begin{equation}
    \label{eq:24}
    (K_\ell^- v)(u, y)= \sqrt{U(u)}\sum_{\substack{j \geq 1, k \geq
        1\\(j, k) \ne (1,1)}} \frac{C_{j,k}(v)}{\pi^2 ((j + k)^2
      +j^2) - E}\car_{R_\ell}(u,y)\phi_{(j+k,j)}\left(y +
      \frac{u}{\ell}, y\right)
  \end{equation}
  where
  \begin{equation}
    \label{eq:25}
    C_{j,k}(v):=\ell\int_{-\ell}^\ell\sqrt{U(u')}
    \sin\left(\pi\frac{(2j+k)u'}{2\ell}\right)\cos\left(\pi
      \frac{k u'}{\ell}\right)c_{2j+k}(u') du' 
  \end{equation}
  and
  \begin{equation}
    \label{eq:33}
    \begin{split}
      c_j(u')&:=\int_0^1(\car_{R_\ell}v)(u',y')\sin(\pi jy')dy'=
      \int_{\max(0,-u'/\ell)}^{\min(1,1-u'/\ell)}v(u',y')\sin(\pi jy') dy'
      \\&=\int_0^1v(u',y')\sin(\pi jy') dy'
    \end{split}
  \end{equation}
  for $\ell$ sufficiently large as $v\in\Coi(\R\times(0,1))$.\\
  Integrating the last integral in~\eqref{eq:33} by parts, we obtain
  \begin{equation}
    \label{eq:35}
    \|c_j\|_{L^2(\R)}=O\left(j^{-\infty}\right).
  \end{equation}
  By~\eqref{eq:33} and~\eqref{eq:25}, as $u\mapsto u^2 U(u)$ is
  summable, we obtain
  \begin{equation}
    \label{eq:34}
    \begin{split}
      |C_{j,k}(v)|&\leq O\left((2j+k)^{-\infty}\right)\ell
      \sqrt{\int_\R
        U(u')\sin^2\left(\pi\frac{(2j+k)u'}{2\ell}\right)du'}\\
      &\leq O\left((2j+k)^{-\infty}\right)\min(\ell,2j+k)
    \end{split}
  \end{equation}
  Estimating $\|K_\ell^-v\|$ using~\eqref{eq:24} and the triangular
  inequality, as
  \begin{equation}
    \label{eq:36}
    \begin{split}
      &\int_{\R\times[0,1]}U(u)\car_{R_\ell}(u,y)\phi^2_{(j+k,j)}\left(y
        + \frac{u}{\ell}, y\right)du dy\\&\lesssim \int_\R
      U(u)\sin^2\left(\pi k\frac{u}{\ell}\right)du+ \int_\R
      U(u)\sin^2\left(\pi(2j+k)\frac{u}{\ell}\right)du \\&\lesssim
      \frac{\min^2(2j+k,\ell)+\min^2(k,\ell)}{\ell^2},
    \end{split}
  \end{equation}
  for $p\geq4$, we get
  \begin{equation*}
    \left\|K_\ell^-v\right\|\lesssim\frac1{\ell}
    \sum_{\substack{j\geq1,k\geq1\\(j,k)\ne (1,1)}}
    \frac1{(j+k)^p}.
  \end{equation*}
  Thus, one gets that $\left\|K_\ell^-v\right\|\to0$ as
  $\ell\to+\infty$. This completes the proof of point (a) of
  Lemma~\ref{le:6}.\\
  To prove point (b), as $2\sin a\,\sin b=\cos(a-b)-\cos(a+b)$, we
  compute
  \begin{equation*}
    (K_\ell^0 v)(u, y)= \sqrt{U(u)}\sum_{\substack{j \geq 1, k \geq
        1\\(j, k) \ne (1,1)}} \frac{A^-_{j,k}(v)-A^+_{j,k}(v)}{\pi^2
      ((j + k)^2 +j^2) - E}
    \car_{R_\ell}(u,y)\phi_{(j+k,j)}\left(y + \frac{u}{\ell},
      y\right)
  \end{equation*}
  where
  \begin{equation*}
    \begin{aligned}
      A^+_{j,k}(v)&:=\ell\int_{-\ell}^\ell\sqrt{U(u')}
      \sin\left(\pi\frac{(2j+k)u'}{2\ell}\right)\sin\left(\pi
        \frac{k u'}{\ell}\right)a_{2j+k}(u') du',         \\
      A^-_{j,k}(v)&:=\ell\int_{-\ell}^\ell\sqrt{U(u')}
      \sin\left(\pi\frac{(2j+k)u'}{2\ell}\right)\sin\left(\pi
        \frac{k u'}{\ell}\right)a_k(u') du'
    \end{aligned}
  \end{equation*}
  and
  \begin{equation*}
    a_k(u'):=\int_0^1(\car_{R_\ell}v)(u',y')\cos(\pi ky')dy'.
  \end{equation*}
  As in~\eqref{eq:33}, we obtain
  \begin{equation*}
    \|a_k\|_{L^2(\R)}=O\left(k^{-\infty}\right).
  \end{equation*}
  As in~\eqref{eq:34}, we obtain
  \begin{equation*}
    |A^\pm_{j,k}(v)|\leq O\left(k^{-\infty}\right)\min(\ell,k).
  \end{equation*}
  By~\eqref{eq:36}, for $p\geq2$, we then get
  \begin{equation}
    \label{eq:39}
    \left\|K_\ell^0 v\right\|\lesssim \frac1{\ell}
    \sum_{\substack{j \geq 1, k \geq 1\\(j, k) \ne (1,1)}}
    \frac{\min(\ell,k)(\min(\ell,k)+\min(\ell,j+k))}
    {k^{-p}((j+k)^2+j^2)} \lesssim \frac1{\ell}+\sum_{j \geq 1}
    \frac{\min(1,j/\ell)}{j^2}
  \end{equation}
  The last term converges to $0$ by the dominated convergence
  theorem. This completes the proof of point (b) of
  Lemma~\ref{le:6}, thus, of Lemma~\ref{le:6}.
\end{proof}
\noindent Next, we decompose $K_\ell^+$ expanding
$\phi_{\barj}(y+u/\ell,y)$ according
to~\eqref{eq:phiBarjDecomposition}. This gives
\begin{equation*}
  K_\ell^+ = K_\ell^{+, +} + K_\ell^{+, -} + K_\ell^{+, 0} \text{,}
\end{equation*}
where, for $\bullet\in\{0,+,-\}$, we have set
\begin{equation*}
  K_\ell^{+,\bullet}\left(E; u, y, u', y'\right)
  =\ell\car_{R_\ell\times R_\ell} \sum_{{\barj} \ne (2,1)}
  \frac{\sqrt{U\left(u\right) U\left(u'\right)}}{\pi^2 |\barj|^2 -
    E} \phi^\bullet_{\barj}\left(y + \frac{u}{\ell}, y\right)\times 
  \phi^+_{\barj}\left(\frac{u'}{\ell},y\right).
\end{equation*}
We now prove
\begin{Le}
  \label{le:8}
  For $v\in\Coi(\R\times(0,1))$, one has
  \begin{enumerate}
  \item $\|K_\ell^{-,+}v\|\to0$ as $\ell\to+\infty$,
  \item $\|K_\ell^{0,+}v\|\to0$ as $\ell\to+\infty$.
  \end{enumerate}
\end{Le}
\begin{proof}
  As in the proof of Lemma~\ref{le:6}, the two points in
  Lemma~\ref{le:8} are proved in very similar ways. We will only
  detail the proof of point (a).\\
  We compute
  \begin{equation}
    \label{eq:37}
    (K_\ell^{-,+} v)(u, y)= \sqrt{U(u)}\sum_{\substack{j \geq 1, k
        \geq 1\\(j, k) \ne (1,1)}} \frac{C_{j,k}(v)}{\pi^2 ((j +
      k)^2 +j^2) - E}
    \car_{R_\ell}(u,y)\phi^-_{(j+k,j)}\left(y + \frac{u}{\ell},
      y\right)
  \end{equation}
  where
  \begin{equation}
    \label{eq:38}
    C_{j,k}(v):=\ell\int_{-\ell}^\ell\sqrt{U(u')}
    \sin\left(\pi\frac{(2j+k)u'}{2\ell}\right)
    \cos\left(\pi \frac{k u'}{2\ell}\right)c_k(u') du' 
  \end{equation}
  and
  \begin{equation*}
    \begin{split}
      c_k(u')&:=\int_0^1(\car_{R_\ell}v)(u',y')\sin(\pi
      ky')dy'=\int_0^1v(u',y')\sin(\pi ky') dy'
    \end{split}
  \end{equation*}
  for $\ell$ sufficiently large as $v\in\Coi(\R\times(0,1))$.\\
  Integrating the last integral in~\eqref{eq:33} by parts, we obtain
  \begin{equation}
    \label{eq:277}
    \|c_k\|_{L^2(\R)}=O\left(k^{-\infty}\right).
  \end{equation}
  As in~\eqref{eq:34}, we obtain
  \begin{equation}
    \label{eq:41}
    |C_{j,k}(v)|\leq O\left(k^{-\infty}\right)\min(\ell,2j+k).
  \end{equation}
  Using~\eqref{eq:22}, one estimates
  \begin{equation}
    \label{eq:40}
    \sqrt{\int_{\R\times[0,1]}U(u)\car_{R_\ell}(u,y)
      \left|\phi^-_{(j+k,j)}\left(y+\frac{u}{\ell},y\right)\right|^2du
      dy}\lesssim \frac{\min(k,\ell)}{\ell}.
  \end{equation}
  Thus, for $p\geq2$, we get
  \begin{equation}
    \label{eq:42}
    \left\|K_\ell^{-,+}v\right\|\lesssim
    \sum_{\substack{j\geq1,k\geq1\\(j,k)\ne (1,1)}}
    \frac{\min(k,\ell)}{\ell}
    \frac{\min(2j+k,\ell)}{k^p((j + k)^2 +j^2)}.
  \end{equation}
  Thus, by the dominated convergence theorem, as in~\eqref{eq:39},
  one gets that $\left\|K_\ell^{-,+}v\right\|\to0$ as
  $\ell\to+\infty$. This completes the proof of point (a) of
  Lemma~\ref{le:8}.\\
  Point (b) is proved similarly except that estimate~\eqref{eq:40}
  is replaced with
  \begin{equation*}
    \sqrt{\int_{\R\times[0,1]}U(u)\car_{R_\ell}(u,y)
      \left|\phi^0_{(j+k,j)}\left(y+\frac{u}{\ell},y\right)\right|^2du
      dy}\lesssim \frac{\min(k,\ell)\min(2j+k,\ell)}{\ell^2}.
  \end{equation*}    
  Thus, taking $p>3$, estimate~\eqref{eq:42} in this case becomes
  \begin{equation*}
      \left\|K_\ell^{0,+}v\right\|\lesssim
      \sum_{\substack{j\geq1,k\geq1\\(j,k)\ne (1,1)}}
      \frac{\min(k,\ell)}{k^p}
      \frac{\min^2(2j+k,\ell)}{\ell^2((j + k)^2 +j^2)}
      \lesssim \sum_{\substack{j\geq1,k\geq1\\(j,k)\ne (1,1)}}
      \frac1{k^{p-2}} \frac{\min^2(j,\ell)}{\ell^2j^2}\lesssim
      \sum_{j\geq1} \frac{\min^2(j,\ell)}{\ell^2}\frac1{j^2}
  \end{equation*}
  which converges to $0$ as $\ell\to+\infty$.\\
  This completes the proof of Lemma~\ref{le:8}.
\end{proof}
\noindent We are now left with computing the limit of $K_\ell^{+,
  +}$ where
\begin{multline}
  \label{eq:43}
  K_\ell^{+,+}(u,y,u',y')=\sum_{\substack{j\geq1,k\geq1\\(j,k)\ne(1,1)}}
  \frac{\ell\sqrt{U(u)U(u')}}{\pi^2((j+k)^2+j^2)-E}\\\times
  \phi_{(j+k,j)}^{+}\left(y+\frac{u}{\ell},y\right)\phi_{(j,j+
    k)}^{+}\left(y'+\frac{u'}{\ell},y'\right).
\end{multline}
We prove
\begin{Le}
  \label{lem:KPlusPlusStrongLimit}
  In the strong topology, one has
  \begin{equation}
    \label{eq:KPlusPlusStrongLimit}
    K_\ell^{+, +}\to K\otimes\Id\quad\text{ as }\quad\ell\to+\infty.
  \end{equation}
  where $K$ is defined in~\eqref{eq:56}.
\end{Le}
\begin{proof}
  To simplify the computations, we note that it suffices to show the
  convergence of $K_\ell^{+, +}v$ for $v\in\Coi(\R\times(0,1))$.For
  $\ell$ sufficiently large, compute
  \begin{equation*}
    (K_\ell^{+,+}v)(u,y)=\sum_{k\geq1}\sin(\pi ky) 
    \car_{R_\ell}(u,y)c\left(K^\ell_k,u\right)
  \end{equation*}
  where
  \begin{equation}
    \label{eq:48}
    c\left(K^\ell_k,u\right):=\frac{1}{2} \sqrt{U(u)} 
    \int_{\bbR} K_k^\ell(u,u')\sqrt{U(u')} c_k(u') du',
  \end{equation}
  $u\mapsto c_k(u)$ being defined by~\eqref{eq:33}, and
  \begin{equation}
    \label{eq:45}
    K^\ell_k(u,u'):=\ell \sum_{\substack{j \in
        \bbN\\(j,k)\ne(1,1)}}\frac{\sin\left(\pi\frac{2j+k}{2\ell}
        u\right)\sin\left(\pi\frac{2j+k}{2\ell}
        u'\right)\cos\left(\pi\frac{k}{2\ell}u\right)
      \cos\left(\pi \frac{2k}{2\ell}
        u'\right)}{\pi^2\left(j+k/2\right)^2+(\pi k/2)^2-E}
  \end{equation}
  Define
  \begin{equation}
    \label{eq:50}
    \begin{aligned}
      L^\ell_k(u,u')&:=\ell \sum_{\substack{j \in
          \bbN\\(j,k)\ne(1,1)}}\frac{\sin\left(\pi\frac{2j+k}{2\ell}
          u\right)\sin\left(\pi\frac{2j+k}{2\ell}
          u'\right)}{\pi^2\left(j+k/2\right)^2},\\
      M^\ell_k(u,u')&:=K^\ell_k(u,u')-L^\ell_k(u,u'),\\
      (L_\ell^{+,+}v)(u,y)&:=\sum_{k\geq1}\sin(\pi ky)
      \car_{R_\ell}(u,y)c\left(L^\ell_k,u\right),\\
      (M_\ell^{+,+}v)(u,y)&:=\sum_{k\geq1}\sin(\pi ky)
      \car_{R_\ell}(u,y)c\left(M^\ell_k,u\right).
    \end{aligned}
  \end{equation}
  Here and in the sequel, $c\left(L^\ell_k,u\right)$ and
  $c\left(M^\ell_k,u\right)$ are defined as
  $c\left(K^\ell_k,u\right)$ in~\eqref{eq:48} with $K^\ell_k$
  replaced respectively by $L^\ell_k$ and $M^\ell_k$ \\
  Note that
  \begin{equation}
    \label{eq:27}
    \left\|L_\ell^{+,+}v\right\|^2_{L^2(\R\times[0,1])}:=
    \frac12\sum_{k\geq1}\int_0^1\left\|\car_{R_\ell}(y,\cdot)
      c\left(L^{\ell}_k,\cdot\right)\right\|^2_{L^2(\R)}dy
    \leq \frac12\sum_{k\geq1}\left\|c\left(L^{\ell}_k,
        \cdot\right)\right\|^2_{L^2(\R)}
  \end{equation}
  We prove
  \begin{Le}
    \label{le:9}
    As $\ell\to+\infty$,
    $\D\left\|M_\ell^{+,+}v\right\|_{L^2(\R\times[0,1])}\to0$
  \end{Le}
  \begin{proof}
    The proof is similar to those of Lemmas~\ref{le:6}
    and~\ref{le:8}. We write
    \begin{equation*}
      M^\ell_k(u,u')=M^{1,\ell}_k(u,u')
      +M^{2,\ell}_k(u,u')+M^{3,\ell}_k(u,u')
    \end{equation*}
    where
    \begin{equation*}
      \begin{aligned}
        M^{1,\ell}_k(u,u')&=\ell \sum_{\substack{j \in
            \bbN\\(j,k)\ne(1,1)}}\frac{\sin\left(\pi\frac{2j+k}{2\ell}
            u\right)\sin\left(\pi\frac{2j+k}{2\ell}
            u'\right)\cos\left(\pi\frac{k}{2\ell}u\right)
          \cos\left(\pi\frac{k}{2\ell}u'\right)\left((\pi
            k/2)^2-E\right)}
        {\pi^4\left(j+k/2\right)^2\left(\left(j+k/2\right)^2+(\pi
            k/2)^2-E\right)} \\ M^{2,\ell}_k(u,u')&:= \ell
        \sum_{\substack{j \in
            \bbN\\(j,k)\ne(1,1)}}\frac{\sin\left(\pi\frac{2j+k}{2\ell}
            u\right)\sin\left(\pi\frac{2j+k}{2\ell}
            u'\right)\cos\left(\pi\frac{k}{2\ell}u\right)\left(
            \cos\left(\pi\frac{k}{2\ell}u'\right)-1\right)}
        {\pi^2\left(j+k/2\right)^2}\\M^{3,\ell}_k(u,u')&:= \ell
        \sum_{\substack{j \in
            \bbN\\(j,k)\ne(1,1)}}\frac{\sin\left(\pi\frac{2j+k}{2\ell}
            u\right)\sin\left(\pi\frac{2j+k}{2\ell}
            u'\right)\left(\cos\left(\pi\frac{k}{2\ell}u\right)-1\right)}
        {\pi^2\left(j+k/2\right)^2}.
      \end{aligned}
    \end{equation*}
    Following the definitions~\eqref{eq:45} and using~\eqref{eq:27},
    we estimate
    \begin{equation*}
      \begin{split}
        \left\|M_\ell^{1,+,+}v\right\|^2_{L^2(\R\times[0,1])} &\leq
        \frac12\sum_{k\geq1}\left\|c\left(M^{1,\ell}_k,
            \cdot\right)\right\|^2_{L^2(\R)}\\
        &\lesssim \sum_{k\geq1}k^2\|c_k\|^2_{L^2(\R)}
        \sum_{\substack{j\geq1\\(j,k)\ne(1,1)}} \frac{
          (\min(2j+k,\ell))^2}{\ell(2j+k)^4}\\& \lesssim
        \frac1\ell\sum_{k\geq1}k^2\|c_k\|^2_{L^2(\R)}
      \end{split}
    \end{equation*}
    which, by~\eqref{eq:277}, converges to $0$ as $\ell$ goes to
    $+\infty$.\\
    That the term coming from $M_\ell^{2,+,+}$
    (resp. $M_\ell^{3,+,+}$) also vanishes as $\ell\to+\infty$
    follows from computations similar to those done in
    Lemma~\ref{le:6} (resp. Lemma~\ref{le:8}). This completes the
    proof of Lemma~\ref{le:9}.
  \end{proof}
  \noindent Note that
  \begin{equation}
    \label{eq:51}
    L^\ell_k(u,u'):=\frac1{\ell} \sum_{\substack{j \in
        \bbN\\(j,k)\ne(1,1)}}\frac{\sin\left(\pi\frac{2j+k}{2\ell}
        u\right)\sin\left(\pi\frac{2j+k}{2\ell}
        u'\right)}{\pi^2\left(\frac{2j+k}{2\ell}\right)^2}
  \end{equation}
  Define
  \begin{equation}
    \label{eq:46}
    a(L^+,u):=\frac{1}{2} \sqrt{U(u)} 
    \int_{\bbR}L^+(u,u')\sqrt{U(u')} c_k(u') du'
  \end{equation}
  where
  \begin{equation}
    \label{eq:47}
    L^+(u,u')
    =\int_0^{+\infty}\frac{\sin(\pi xu)\sin(\pi xu')}{\pi^2x^2}dx.
  \end{equation}
  We prove
  \begin{Le}
    \label{le:10}
    For any $k\geq1$, one has
    \begin{equation}
      \label{eq:49}
      \sup_{(u,u')\in[-\ell,\ell]^2}\frac{\left|L^\ell_k(u,u')
          -L^+(u,u')\right|}{|u||u'|}\lesssim\frac{k}{\ell}.
    \end{equation}
  \end{Le}
  \begin{proof}
    Define
    \begin{equation*}
      l(x,u,u'):=\frac{\sin(\pi xu)\sin(\pi xu')}{\pi^2x^2}.
    \end{equation*}
    Assume first $k\not=1$. As $l$ is an even function of $x$, write
    \begin{equation}
      \label{eq:52}
      L^\ell_k(u,u')=\frac1{2\ell}\sum_{\substack{j
          \in\Z}}l\left(\frac{j+k/2}{\ell},u,u'\right)-
      \frac1{2\ell}
      \sum_{j=-k}^0l\left(\frac{j+k/2}{\ell},u,u'\right).
    \end{equation}
    Using the Poisson formula, one computes
    \begin{equation}
      \label{eq:55}
      \frac1{2\ell}\sum_{\substack{j
          \in\Z}}l\left(\frac{j+k/2}{\ell},u,u'\right)=
      \frac12\sum_{\substack{j\in\Z}}e^{i\pi kj}\cdot
      \hat l\left(j,u,u'\right)
    \end{equation}
    where $\hat l(\cdot,u,u')$ is the Fourier transform of $x\mapsto
    l(x,u,u')$.\\
    By the Paley-Wiener Theorem (or by a direct computation of the
    Fourier transform), one checks that $\hat l(\cdot,u,u')$ is
    supported in $[-\pi(|u|+|u'|),\pi(|u|+|u'|)]$. Thus, for
    $-\ell\leq u,u'\leq l$, all the terms in right hand side
    of~\eqref{eq:55} vanish except the term for $j=0$. That is, for
    $-\ell\leq u,u'\leq l$, one has
    \begin{equation*}
      \frac1{2\ell}\sum_{\substack{j
          \in\Z}}l\left(\frac{j+k/2}{\ell},u,u'\right)=
      \frac12\hat l\left(0,u,u'\right)=L^+(u,u').
    \end{equation*}
    This and~\eqref{eq:55} then yields that, for $-\ell\leq u,u'\leq
    l$,
    \begin{equation*}
      L^\ell_k(u,u') -L^+(u,u')=-\frac{u\, u'}{2\ell}
      \sum_{j=-k}^0\frac{l\left(\frac{j+k/2}{\ell},u,u'\right)}{u\,u'}.
    \end{equation*}
    Now, as
    \begin{equation*}
      \sup_{(x,u,u')\in\R^3}\left|\frac{l(x,u,u')}{u\,u'}\right|<+\infty,
    \end{equation*}
    we immediately obtain~\eqref{eq:49} and complete the proof of
    Lemma~\ref{le:10} when $k\not=1$.\\
    When $k=1$, the proof is done in the same way up to a shift in
    the index $j$. This completes the proof of Lemma~\ref{le:10}.
  \end{proof}
  \noindent As $v\in\Coi(\R\times(0,1)$, one has
  \begin{equation*}
    \forall N\geq0,\quad \exists C_N>0,\quad \forall k\in\Z,\quad
    \|c_k\|_{L^2(\R)}\leq C_N\frac1{1+|k|^N}.
  \end{equation*}
  Thus, as $x\mapsto x\sqrt{U(x)}$ is square integrable, the
  bound~\eqref{eq:49} yields that, for some $C_2>0$, one has
  \begin{equation*}
    \forall k\in\Z,\quad \left\| c\left(K^\ell_k,\cdot\right)-
      c\left(L^+,\cdot\right)\right\|_{L^2([-\ell,\ell])}
    \leq \frac1{\ell}\frac{C_2}{1+|k|^2}.
  \end{equation*}
  Thus, taking into account the following computation
  \begin{equation}
    \label{eq:58}
    \begin{split}
      L^+(u, u') &= \int_{\bbR} \frac{\sin(\pi x u) \sin(\pi x
        u')}{\pi^2 x^2} \rmd{x} \\
      &= \frac{1}{2 \pi^2} \left[\int_{\bbR} \frac{\cos(\pi x (u -
          u')) - 1}{x^2} \rmd{x} + \int_{\bbR} \frac{1 - \cos(\pi x (u
          +  u'))}{x^2} \rmd{x}\right] \\
      &= \frac{1}{2 \pi^2} \left[|u - u'| \int_{\bbR} \frac{\cos(\pi
          x) - 1}{x^2} \rmd{x} + |u + u'| \int_{\bbR} \frac{1 -
          \cos(\pi x)}{x^2} \rmd{x}\right] \\
      &= \frac{1}{2} (|u + u'| - |u - u'|),
    \end{split}
  \end{equation}
  the definition of $K$,~\eqref{eq:56} and~\eqref{eq:27}, we obtain
  that
  \begin{equation*}
    \left\|L_\ell^{+,+}v-(K\otimes\car)v\right\|_{L^2(\R\times[0,1]}
    \vers{\ell\to+\infty}0
  \end{equation*}
  Thus, Lemma~\ref{lem:KPlusPlusStrongLimit} is proved.
\end{proof}
\noindent Clearly, the proof of Lemma~\ref{le:4} generalizes to
arbitrary $\phi_{(i,j)}$, a normalized eigenfunction of $H^0(1,2)$;
one thus proves
\begin{corollary}
  \label{cor:1}
  Consider two particles on $i$-th and $j$-th energy levels in an
  interval of length $\ell$. Their interaction amplitude is given by
  \begin{equation}
    \label{eq:UphiIJquadraticForm}
    \langle U\phi_{(i,j)},\phi_{(i,j)}\rangle=2\pi^2(i^2+j^2)
    \cdot\int_{\bbR}u^2 U(u)\rmd{u}\cdot\ell^{-3}(1+O(\ell^{-1})).
  \end{equation}
\end{corollary}
\subsubsection{The ground state of two interacting electrons and its
  density matrices}
\label{sec:bounds-ground-state}
Recall that $\varphi^j_{[0, \ell]}$ denotes the $j$-th normalized
eigenvector of $-\laplace^D_{|[0, \ell]}$ and $\zeta^j_{[0, \ell]}$
the $j$-th normalized eigenvector
of~\eqref{eq:introTwoParticleHamiltonian}.  In the sequel, we drop the
subscript $[0,\ell]$ as we always work on the interval $[0,\ell]$.\\
We remark that, when the interactions are absent, one has
\begin{equation}
  \label{eq:freeTwoParticleProblemGroundState}
  \zeta^{1,0}=\varphi^1\wedge\varphi^2.
\end{equation}
The next proposition estimates the difference
$\zeta^{1,U}-\zeta^{1,0}$ induced by the presence of interactions.
\begin{proposition}
  \label{prop:twoParticleProblemComparison}
  For $\ell\geq1$, one has
  \begin{equation}\label{eq:twoParticleProblemComparison}
    \left\|\zeta^{1,U}-
      \zeta^{1,0}\right\|_{L^2([0,\ell]^2)}\lesssim \ell^{-1/2}.
  \end{equation}
\end{proposition}
\begin{proof}
  Scaling the variables to the unit square (see
  section~\ref{sec:proof-prop-refpr}), it suffices to show that the
  normalized ground state of $H^{U^\ell}(1, 2)$ (see~\eqref{eq:26}),
  say, $\phi_0^{U^\ell}$ satisfies
  \begin{equation}
    \label{eq:64}
    \left\|\phi_0^{U^\ell}- \phi_0\right\|_{L^2([0,1]^2)}\lesssim \ell^{-1/2}.
  \end{equation}
  where we recall that $\phi_0=\phi_{(1,2)}$ (see~\eqref{eq:21}).\\
  Decomposing $\D L^2([0,1])\wedge L^2([0,1])=\C\phi_0\overset{\perp}{\oplus}
  \phi^\perp_0$ and defining $E_0^{U^\ell}$ to be the ground state
  energy of $H^{U^\ell}(1, 2)$, we rewrite $\phi_0^{U^\ell}$ as
  \begin{equation*}
    \phi_0^{U^\ell}=\alpha\phi_0+\wtphi,\quad \wtphi\perp\phi_0,\quad\alpha\in\R^+  
  \end{equation*}
  and the eigenvalue equation it satisfies as
  \begin{equation}
    \label{eq:twoParticlesSpectralProblemMatrixForm}
    \begin{pmatrix}
      5 \pi^2 +U^\ell_{00}-E_0^{U^\ell} & U^\ell_{0+} \\
      U^\ell_{+0} & H_++U^\ell_{++} -E_0^{U^\ell}
    \end{pmatrix}
    \begin{pmatrix}
      \alpha \\
      \wtphi
    \end{pmatrix}
    = 0.
  \end{equation}
  where the terms in the matrix are defined in~\eqref{eq:44}.\\
  Thus, to prove~\eqref{eq:64} it suffices to prove that
  \begin{equation*}
    \|\wtphi\|_{L^2([0,1])\wedge L^2([0,1])}\leq C\ell^{-1/2}.
  \end{equation*}
  By~\eqref{eq:twoParticlesSpectralProblemMatrixForm}, as
  $\phi_0^{U^\ell}$ is normalized, as $10\pi^2\leq H_++U^\ell_{++}$
  and as $\D E_0^{U^\ell}\vers{\ell\to+\infty}5\pi^2$,
  using~\eqref{eq:44} and~\eqref{eq:rhsUnitaryTransformation}, one
  computes
  \begin{equation*}
    \begin{split}
      \|\wtphi\|^2_{L^2([0,1])\wedge L^2([0,1])}&\leq U^\ell_{0+}\left(
        H_++U^\ell_{++} -E_0^{U^\ell}\right)^{-2}U^\ell_{+0}
      \leq\frac{C}{\ell}\left\langle \phi_\ell, K_\ell (\Id +
        K_\ell)^{-1} \phi_\ell\right\rangle_{L^2(\R\times[0,1])}.
    \end{split}
  \end{equation*}
  Thus,~\eqref{eq:64} is an immediate consequence of
  Lemma~\ref{le:4}. This completes the proof of
  Proposition~\ref{prop:twoParticleProblemComparison}.
\end{proof}
\noindent We obtain the following corollary for the one particle
density matrices of $\zeta^{1,U}$.
\begin{corollary}
  \label{cor:twoParticleProblemComparison}
  Under assumptions of
  Proposition~\textup{\ref{prop:twoParticleProblemComparison}}, one
  has
  \begin{equation*}
    \left\|\gamma_{\zeta^{1,U}} - \gamma_{\varphi^1} -
      \gamma_{\varphi^2}\right\|_1= O\left(\ell^{-1}\right).
  \end{equation*}
\end{corollary}
\noindent Corollary~\ref{cor:twoParticleProblemComparison} is an
immediate consequence of~\eqref{eq:twoParticleProblemComparison} and
\begin{lemma}
  \label{le:12}
  Let $\psi, \phi \in L^2([0,\ell])\wedge L^2([0,\ell])$ be two normalized
  two-particles states.  Then
  \begin{equation*}
    \|\gamma_\psi - \gamma_\phi\|_1 \leq 4 \|\psi - \phi\| \text{.}
  \end{equation*}
\end{lemma}
\begin{proof}[Proof of Lemma~\ref{le:12}]
  For $\varphi \in L^2([0,\ell])\wedge L^2([0,\ell])$, consider the operator
  $A_\varphi$ defined as
  \begin{equation*}
    (A_\varphi f)(x) = \int_0^\ell \varphi(x, y) f(y) \rmd{y} \text{.}
  \end{equation*}
  Note that $A_\varphi$ is a Hilbert-Schmidt operator and $\D
  \|A_\varphi\|_2 = \|\varphi\|$ and the one-particle density matrix
  of $\varphi$ satisfies $\gamma_\varphi = 2 A_\varphi^\ast
  A_\varphi$. Thus, for $\psi$, $\phi$ as in Lemma~\ref{le:12}, we
  obtain
  \begin{equation*}
    \begin{split}
      \|\gamma_\psi - \gamma_\phi\|_1 &= 2 \|A_\psi^\ast A_\psi -
      A_\phi^\ast A _\phi\|_1\leq 2 \left(\|A_\psi^\ast\|_2 \|A_\psi -
        A_\phi\|_2 + \|A_\psi^\ast - A_\phi^\ast\|_2
        \|A_\phi\|_2\right) \leq 4 \|\psi - \phi\| \text{.}
    \end{split}
  \end{equation*}
  This completes the proof of Lemma~\ref{le:12}.
\end{proof}
\subsection{Electrons in distinct pieces}
\label{sec:ferm-neighb-piec}
In the present section, we assume that $U$ satisfies \textbf{(HU)}
(see section~\ref{sec:interacting-electrons}); thus, it decreases
sufficiently fast at infinity (roughly better than $x^{-4}$)
and is in $L^p$ for some $p>1$.\\
Let the first piece be $\Delta_1 = [-\ell_1, 0]$ and the second be
$\Delta_2 = [a, a + \ell_2]$; so, the pieces' lengths are $\ell_1$ and
$\ell_2$, while the distance between them is denoted by $a$.  As for
the one-particle systems living in each of these pieces, we will
primarily be interested in the following three cases:
\begin{enumerate}
\item\label{it:inter11} the interaction of two eigenstates of the
  one-particle Hamiltonian on each piece, i.e., following the
  notations of section~\ref{sec:two-int-electrons}, of
  $\varphi^i_{\Delta_1}$ and $\varphi^j_{\Delta_2}$,
\item\label{it:inter12} the interaction of a one-particle eigenstate
  with a one-particle reduced density matrix of a two-particle ground
  state, i.e., $\varphi^i_{\Delta_1}$ with
  $\gamma_{\zeta^1_{\Delta_2}}$,
\item\label{it:inter22} the interaction of two one-particle density
  matrices, i.e., $\gamma_{\zeta^1_{\Delta_1}}$ and
  $\gamma_{\zeta^1_{\Delta_2}}$.
\end{enumerate}
We observe that for a one-particle eigenstate in a piece of
length $\ell$, the following uniform pointwise bound holds true:
\begin{equation}
  \label{eq:100}
  \|\varphi^i_{[0, \ell]}\|_{L^\infty} \leq \sqrt{\frac{2}{\ell}}.
\end{equation}
For the one-particle reduced density matrix we establish the following
estimates.
\begin{lemma}
  \label{le:22}
  Let $\zeta \in L^2([0,\ell])\wedge L^2([0,\ell])$ be a two particle
  state and $\gamma_\zeta(x, y)$ the kernel of the corresponding
  one-particle density matrix. Let $p\in\bbN$.  Then, $\zeta\in
  H^p([0,\ell]^2)$ implies $\gamma_\zeta \in H^p([0, \ell]^2)$ and
  \begin{equation}
    \label{eq:101}
    \|\gamma_\zeta\|_{H^p} \leq 4 \|\zeta\|_{H^p}.
  \end{equation}
  In particular, unconditionally $\|\gamma_\zeta\|_{L^2} \leq 4$.
\end{lemma}
\begin{proof}
  First recall that
  \begin{equation*}
    \gamma_\zeta(x, y) = 2 \int_0^\ell \zeta(x, z) \zeta^*(y, z) \rmd{z}.
  \end{equation*}
  Then, one differentiates under the integration sign to get
  \begin{equation*}
    \frac{\partial^p}{\partial x^p} \gamma_\zeta(x, y) 
    = 2 \int_0^\ell \partial^p_x \zeta(x, z) \zeta^*(y, z) \rmd{z}.    
  \end{equation*}
  This in turn implies by the Cauchy-Schwarz inequality that
  \begin{equation*}
    \begin{split}
      \left\|\frac{\partial^p}{\partial x^p}
        \gamma_\zeta\right\|^2_{L^2} &= 4 \int_{[0, \ell]^2}
      \left|\int_0^\ell \partial^p_x \zeta(x,
        z) \zeta^*(y, z) \rmd{z}\right|^2 \rmd{x} \rmd{y} \\
      &\leq 4 \int_{[0, \ell]^4} \left|\partial^p_x \zeta(x,
        z)\right|^2 \cdot \left|\zeta(y, z^\prime)\right|^2 \rmd{x} \rmd{y}
      \rmd{z} \rmd{z^\prime}= 4 \left\|\partial^p_x
        \zeta\right\|^2_{L^2},
    \end{split}
  \end{equation*}
  which proves \eqref{eq:101}.
\end{proof}
\begin{lemma}
  \label{le:23}
  Let $\zeta=\zeta^{1,U}_{[0,\ell]}$ be the ground state of a system
  of two interacting electrons in $[0,\ell]$. Then, $\zeta\in
  H^1([0,\ell]^2)$ and there exists a constant $C>0$ independent of
  $\ell$ such that
  \begin{equation}
    \label{eq:102}
    \|\zeta\|_{H^1} \leq C / \sqrt{\ell}.
  \end{equation}
\end{lemma}
\begin{proof}
  We use the construction of the proof of
  Proposition~\ref{prop:twoParticleProblemComparison}. Then, employing
  the same notations, for the problem scaled to the unit square one has
  \begin{equation*}
    \phi^{U^\ell}_0 = \alpha \phi_0 + \wtphi,
  \end{equation*}
  where $\phi_0$ is the ground state for a system of two
  non-interacting electrons, $|\alpha| \leq 1$ and $\wtphi\; \bot\;
  \phi_0$. Obviously, $\phi_0 \in H^p$ for all $p \in \bbN$.
  Moreover, according to
  \eqref{eq:twoParticlesSpectralProblemMatrixForm},
  \begin{equation*}
    \begin{split}
      \|\wtphi\|_{H^1} &= \left\|(H_+ + U_{++}^\ell -
        E_0^{U^\ell})^{-1} U_{+0}^{\ell} \alpha \phi_0\right\|_{H^1}
      \leq \left\|(H_+ + U_{++}^\ell - E_0^{U^\ell})^{-1}\right\|_{L^2
        \to H^1} \cdot \left\|U_{+0}^{\ell}   \phi_0\right\|_{L^2} \\
      &\leq \left\|(H_+ - E_0^{U^\ell})^{-1}\right\|_{L^2 \to H^1}
      \cdot \left\|U_{+0}^{\ell} \phi_0\right\|_{L^2}.
    \end{split}
  \end{equation*}
  Arguing as in section~\ref{sec:two-int-electrons}, one can prove that
  \begin{equation*}
    \left\|U_{+0}^{\ell} \phi_0\right\|_{L^2} \leq \left\|U^{\ell}
      \phi_0\right\|_{L^2} \leq C \sqrt{\ell} 
  \end{equation*}
  and $(H_+ - E_0^{U^\ell})^{-1}$ is a bounded operator from $L^2([0,
  1]^2)$ to $H^1([0, 1]^2)$ because $H_+$ is just a part of
  $-\laplace_2$ acting in a subspace of functions orthogonal to
  $\phi_0$ and the bottom of its spectrum is separated from
  $E_0^{U^\ell}$. Thus, we proved that
  \begin{equation*}
    \|\wtphi\|_{H^1} \leq C \sqrt{\ell}
  \end{equation*}
  which immediately implies 
  \begin{equation*}
    \|\phi_0^{U^\ell}\|_{H^1} \leq C \sqrt{\ell}.
  \end{equation*}
  Scaling back to the original domain $[0, \ell]^2$ yields
  \eqref{eq:102} and completes the proof of Lemma~\ref{le:23}.
\end{proof}

\begin{corollary}
  \label{cor:2}
  Restricted to the diagonal, the kernel of the ground state
  one-particle density matrix $x\in[0,\ell]\mapsto\gamma_\zeta(x,x)$
  is a bounded function; more precisely, there exists a constant $C>0$
  such that
  \begin{equation}
    \label{eq:103}
    \|\gamma_\zeta\|_{L^\infty([0,\ell])} \leq C/\ell.
  \end{equation}
\end{corollary}

\begin{proof}
  Remark first that, as $\zeta$ satisfies Dirichlet boundary
  conditions, so does the kernel
  $(x,y)\mapsto\gamma_\zeta(x,y)$. Using anti-symmetry, we compute
  \begin{equation}
    \label{eq:73}
    \begin{split}
      |\gamma_\zeta(x,x)|&=2\left|\int_0^x\frac{d}{dt}\left[\gamma_\zeta(t,t)\right]
      \right|=4\left|\text{Im}\left(\int_0^x\int_0^\ell\partial_t
          \zeta(t,x)\overline{\zeta(t,x)}dxdt\right) \right|
      \\&\leq4\|\partial_t\zeta\|_{L^2}\cdot\|\zeta\|_{L^2}\leq 4
      \|\zeta\|^2_{H^1}
    \end{split}
  \end{equation}
  Combining this with~\eqref{eq:102} gives~\eqref{eq:103} and
  completes the proof of Corollary~\ref{cor:2}.
\end{proof}
\noindent Having now pointwise bounds \eqref{eq:100} and
\eqref{eq:103}, we estimate the interactions in each of the three
cases described in the beginning of the current section. We will also
obtain different bounds for close enough and distant pieces $\Delta_1
= [-\ell_1,0]$ and $\Delta_2=[a,a+\ell_2]$, i.e., we will discuss
different bounds depending on whether $a$ is large or small.\\
For the case (\ref{it:inter11}) of two interacting one-particle
eigenstates we prove the following two estimates.  For long distance
interactions, i.e., when $a$ is large, we will use
\begin{lemma}
  \label{lem:inter11_farEstimate}
  Suppose $U$ satisfies \textup{\textbf{(HU)}}. Then, for $\Delta_1 =
  [-\ell_1,0]$ and $\Delta_2=[a,a+\ell_2]$, one has
  \begin{equation}
    \label{eq:104}
    \sup_{i,j}\int_{\Delta_1 \times \Delta_2} U(x - y)
    |\varphi^i_{\Delta_1}(x)|^2 \cdot
    |\varphi^j_{\Delta_2}(y)|^2 \rmd{x} \rmd{y} \leq
    \frac{2 a^{-3} Z(a)}{\max(\ell_1, \ell_2)}
  \end{equation}
  where $Z$ is defined in~\eqref{eq:98}.
\end{lemma}
\begin{proof}
  Let us suppose without loss of generality that $\Delta_1$ is the
  larger piece, i.e., $\ell_1 \geq \ell_2$. Then, using \eqref{eq:100}
  and the fact that the functions $(\varphi^i_{\Delta_j})_{i,j}$ are
  normalized, we compute
  \begin{equation*}
    \begin{split}
      \int_0^{\ell_1} \int_0^{\ell_2} U(x + y + a)
      |\varphi^i_{\Delta_1}(x)|^2 \cdot |\varphi^j_{\Delta_2}(y)|^2
      \rmd{x} \rmd{y}
    &\leq \frac{2}{\ell_1} \int_0^{\ell_1} \int_0^{\ell_2} U(x + y +
    a)  |\varphi^j_{\Delta_2}(y)|^2 \rmd{x} \rmd{y} \\
    &\leq \frac{2}{\ell_1} \sup_{y \in [0, \ell_2]} \int_0^{\ell_1}
    U(x + y + a) \rmd{x} \\
    &\leq \frac{2}{\ell_1} \int_0^{+\infty} U(x + a) \rmd{x} \\
    &= \frac{2}{\ell_1} a^{-3} Z(a), \quad a \to +\infty.
    \end{split}
  \end{equation*}
  This completes the proof of Lemma~\ref{lem:inter11_farEstimate}.
\end{proof}
\noindent On the other hand, for close by interactions, i.e., $a$
small and low-lying one-particle energy levels the following lemma
gives a more precise estimate.
\begin{lemma}
  \label{lem:inter11_closeEstimate}
  Suppose $U$ satisfies \textup{\textbf{(HU)}}. Let $(i, j) \in \{1,
  2\}^2 $. Then, for any $\eps \in (0, 2)$ and $\Delta_1 =
  [-\ell_1,0]$ and $\Delta_2=[a,a+\ell_2]$, one has
  \begin{equation}
    \label{eq:105}
    \int_{\Delta_1 \times \Delta_2} U(x - y) |\varphi^i_{\Delta_1}(x)|^2 \cdot
    |\varphi^j_{\Delta_2}(y)|^2 \rmd{x} \rmd{y} \\
    = O\left(\frac{a^{-\eps} Z(a)}{\max(\ell_1, \ell_2)^2
        \min(\ell_1, \ell_2)^{2-\eps}}\right).
  \end{equation}
  If $Z(a) = O(a^{-0})$, $a \to +\infty$, then $\eps$ can be taken to zero.
\end{lemma}
\begin{proof}
  As in the proof of the previous lemma we suppose that $\ell_1 \geq
  \ell_2$.  If $j\in\{1,2\}$ then
  \begin{equation}
    \label{eq:108}
    |\varphi_{\Delta_1}^j(x)| =
    \left|\sqrt{\frac{2}{\ell_1}}\sin\left(\frac{\pi i
          x}{\ell_1}\right)\right| \leq
    \sqrt{\frac{2}{\ell_1}} \frac{\pi |x|}{\ell_1}
  \end{equation}
  and the same type inequality holds for $\varphi^j_{\Delta_2}(y)$.
  Then, using ~\eqref{eq:108} and~\eqref{eq:100}, we compute
  \begin{equation*}
    \begin{split}
      \int_0^{\ell_1} \int_0^{\ell_2} U(x + y + a)
      |\varphi^i_{\Delta_1}(x)|^2 \cdot |\varphi^j_{\Delta_2}(y)|^2 \rmd{x}
      \rmd{y} &\leq \frac{C_1}{\ell_1^2 \ell_2^{2 - \eps}}
      \int_0^{\ell_1}
      \int_0^{\ell_2} U(x + y + a) x y^{1 - \eps} \rmd{x} \rmd{y} \\
      &\leq \frac{C_1}{\ell_1^2 \ell_2^{2 - \eps}} \int_{\bbR_+^2} U(x
      + y + a) x y^{1 - \eps} \rmd{x} \rmd{y} \\
      &= \frac{C_2}{\ell_1^2 \ell_2^{2 - \eps}} \int_0^{+\infty}
      \int_{-s}^{s} U(s + a) (s + t) (s - t)^{1 - \eps} \rmd{t} \rmd{s} \\
      &\leq \frac{C_3}{\ell_1^2 \ell_2^{2 - \eps}} \int_a^{+\infty}
      U(s) s^{3 - \eps} \rmd{s}.
    \end{split}
  \end{equation*}
  It is now only left to prove that \textbf{(HU)} and~\eqref{eq:98}
  imply that the last integral converges and is $O(a^{-\eps} Z(a))$.
  Therefore, we note that
  \begin{equation}
    \label{eq:111}
    \begin{split}
      \int_a^{+\infty} U(s) s^{3 - \eps} \rmd{s} &= \sum_{n =
        0}^{+\infty} \int_{2^n a}^{2^{n + 1} a} U(s) s^{3 - \eps}
      \rmd{s}\leq \sum_{n = 0}^{+\infty} \left(2^{n + 1} a\right)^{3 -
        \eps}
      \int_{2^n a}^{2^{n + 1} a} U(s) \rmd{s} \\
      &\leq 2^{3 - \eps} a^{-\eps} \sum_{n = 0}^{+\infty} 2^{-\eps n}
      \left(2^{n} a\right)^3 \int_{2^n a}^{+\infty} U(s) \rmd{s}= 2^{3
        - \eps} a^{-\eps} \sum_{n = 0}^{+\infty} 2^{-\eps n}
      Z\left(2^{n} a\right)\\
      &\leq C a^{-\eps} Z(a).
    \end{split}
  \end{equation}
  If $Z(a) = O(a^{-0})$, i.e., if there exists $\delta>0$ s.t. $Z(a) =
  O(a^{-\delta})$ for $a\to+\infty$, then, the sum in the second line
  of~\eqref{eq:111} converges for $\eps=0$.\\
  This concludes the proof of \eqref{eq:105}.
\end{proof}
\noindent Let us now pass to the case (\ref{it:inter12}) of
one-particle eigenstate interacting with a one-particle density matrix
of a two-particle eigenstate. For large $a$, we prove
\begin{lemma}
  \label{lem:inter12_farEstimate}
  Suppose $U$ satisfies \textup{\textbf{(HU)}}. Then, for $a$
  sufficiently large, one has
  \begin{equation}
    \label{eq:107}
    \sup_{i,j}\int_{\Delta_1 \times \Delta_2} U(x - y)
    |\varphi^i_{\Delta_1}(x)|^2 \cdot
    \gamma_{\zeta^j_{\Delta_2}}(y, y) \rmd{x} \rmd{y} \leq
    \frac{4 a^{-3} Z(a)}{\ell_1}.
  \end{equation}
\end{lemma}
\begin{proof}
  The proof follows that of Lemma~\ref{lem:inter11_farEstimate}. The
  only change concerns the replacement of the fact that
  $\varphi_{\Delta_2}^j$ is normalized, $\int_{\Delta_2}
  |\varphi_{\Delta_2}^j(y)|^2 \rmd{y} = 1$, by the fact that the trace
  of $\gamma_{\zeta^j_{\Delta_2}}$ is equal to $2$.
\end{proof}
\noindent For $a$ small, we prove
\begin{lemma}
  \label{lem:inter12_closeEstimate}
  Suppose $U$ satisfies \textup{\textbf{(HU)}}. Let $i \in \{1,
  2\}$. Then, for any $\eps \in (0, 2)$, 
  
  \begin{equation}
    \label{eq:109}
    \int_{\Delta_1 \times \Delta_2} U(x - y) |\varphi^i_{\Delta_1}(x)|^2 \cdot
    \gamma_{\zeta^j_{\Delta_2}}(y, y) \rmd{x} \rmd{y} =
    O\left(\ell_1^{-3 + \eps} \ell_2^{-1/2} a^{-\eps} Z(a)\right).
  \end{equation}
  If $Z(a) = O(a^{-0})$ as $a \to +\infty$, one can choose $\eps = 0$.
\end{lemma}
\begin{proof}
  As in the proof of Lemma~\ref{lem:inter11_closeEstimate} mixing once
  more \eqref{eq:100}, \eqref{eq:103} and \eqref{eq:108}, we obtain
  \begin{equation*}
    \begin{split}
      \int_0^{\ell_1} \int_0^{\ell_2} U(x + y + a)
      |\varphi^i_{\Delta_1}(x)|^2 \gamma_{\zeta^j_{\Delta_2}}(y, y)
      \rmd{x} \rmd{y} &\leq \frac{C_1}{\ell_1^{3 - \eps} \ell_2^{1 /
          2}} \int_0^{\ell_1}
      \int_0^{\ell_2} U(x + y + a) x^{2 - \eps} \rmd{x} \rmd{y} \\
      &\leq \frac{C_1}{\ell_1^{3 - \eps} \ell_2^{1 / 2}}
      \int_0^{+\infty} \int_a^{+\infty}
      U(x + y) x^{2 - \eps} \rmd{x} \rmd{y} \\
      &= \frac{C_2}{\ell_1^{3 - \eps} \ell_2^{1 / 2}} \int_a^{+\infty}
      \int_{-s}^{s} U(s) (s + t)^{2 - \eps} \rmd{t} \rmd{s} \\
      &\leq \frac{C_3}{\ell_1^{3 - \eps} \ell_2^{1 / 2}}
      \int_a^{+\infty}
      U(s) s^{3 - \eps} \rmd{s} \\
      &\leq \frac{C_4 a^{-\eps} Z(a)}{\ell_1^{3 - \eps} \ell_2^{1 /
          2}}.
    \end{split}
  \end{equation*}
  This completes the proof of Lemma~\ref{lem:inter12_closeEstimate}.
\end{proof}
\noindent We are left with the case~\eqref{it:inter22} of two
interacting reduced density matrices.  We do not make the difference
between close and far away pieces in this case.
\begin{lemma}\label{lem:inter22_estimate}
  Suppose $U$ satisfies \textup{\textbf{(HU)}}. Then, there exists
  $C>0$ such that
  \begin{equation}
    \label{eq:110}
    \sup_{i,j}\int_{\Delta_1 \times \Delta_2} U(x - y)
    \gamma_{\zeta^i_{\Delta_1}}(x, x) \cdot
    \gamma_{\zeta^j_{\Delta_2}}(y, y) \rmd{x} \rmd{y} \leq
    C \ell_1^{-1/2} \ell_2^{-1/2} \min(1,a^{-2}
    Z(a))
  \end{equation}
\end{lemma}
\begin{proof}
  Using \eqref{eq:103} one obtains
  \begin{equation*}
    \begin{split}
      \int_0^{\ell_1} \int_0^{\ell_2} U(x + y + a)
      \gamma_{\zeta^i_{\Delta_1}}(x, x) \gamma_{\zeta^j_{\Delta_2}}(y,
      y) \rmd{x} \rmd{y} &\leq \frac{C_1}{\sqrt{\ell_1 \ell_2}}
      \int_{\bbR_+^2} U(x + y + a)
      \rmd{x} \rmd{y} \\
      &\leq \frac{C_2}{\sqrt{\ell_1 \ell_2}} \int_0^{+\infty} U(s + a)
      s \rmd{s} \\
      &\leq \frac{C_2}{\sqrt{\ell_1 \ell_2}} \int_a^{+\infty}
      \left(\int_t^{+\infty} U(s)ds\right)dt
    \end{split}
  \end{equation*}
  Thus,
  \begin{equation*}
    \int_0^{\ell_1} \int_0^{\ell_2} U(x + y + a)
    \gamma_{\zeta^i_{\Delta_1}}(x, x) \gamma_{\zeta^j_{\Delta_2}}(y,
    y) \rmd{x} \rmd{y} \leq \frac{C_3\min(C,a^{-2}
      Z(a))}{\sqrt{\ell_1 \ell_2}}
  \end{equation*}
  where the last equality is just~\eqref{eq:111} for $\eps = 2$ and
  $\D C:=\int_0^{+\infty}\left(\int_t^{+\infty} U(s)ds\right)dt<+\infty$.\\
  This completes the proof of Lemma~\ref{lem:inter22_estimate}.
\end{proof}
\noindent Finally, we give estimates for the case of compactly
supported interaction potential $U$. We prove
\begin{lemma}
  \label{lem:UijijVois}
  Assume that $U$ has a compact support. Then, there exists $C>0$ such
  that, for $i\geq1$ and $j\geq1$, one has
  \begin{equation*}
    \langle U \phi_{(i,j)}, \phi_{(i,j)} \rangle\leq C\cdot \frac{\left[\min(i,\ell_1)
      \min(j,\ell_2)\right]^2}{\ell_1^3 \ell_2^3}.
  \end{equation*}
\end{lemma}
\begin{proof}
  Due to the anti-symmetry of the functions $(\phi_{(i,j)})_{i,j,}$, it
  suffices to compute the scalar pro\-duct on
  $[-\ell_1,0]\times[a,a+\ell_2]$. Thus,
  \begin{equation*}
    \begin{split}
      \langle U \phi_{(i,j)}, \phi_{(i,j)} \rangle &\leq \sup_{|a|\leq
        \diam(\supp(U))}\frac1{2\ell_1 \ell_2}
      \int_{[0,\ell_1]\times[0,\ell_2]}U(x+y+a)
      \\&\hskip5cm\times\sin^2\left(\frac{i \pi x}{\ell_1}\right)
      \sin^2\left(\frac{j \pi y}{\ell_2}\right)dxdy \\
      &\leq C(U)\frac{\left[\min(i,\ell_1)
          \min(j,\ell_2)\right]^2}{\ell_1^3 \ell_2^3}
    \end{split}
  \end{equation*}
  where 
  \begin{equation*}
    C(U):=\frac12\sup_{0\leq a\leq \diam(\supp(U))}
    \int_{\R^+\times\R^+}U(x+y+a)(1+x^2)(1+y^2)dxdy.
  \end{equation*}
  This completes the proof of Lemma~\ref{lem:UijijVois}.
\end{proof}
\begin{proposition}
  \label{prop:TwoElectronProblemNeighbors}
  Consider a system of two interacting electrons, one in $[0, \ell_1]$,
  another in $[\ell_1 + r, \ell_1 + r + \ell_2]$ with $r \leq R_0$.
  Then, the ground state energy of this system has the following
  asymptotic expansion
  \begin{equation}
    \label{eq:TwoNeighborsEnergyExpansion}
    E((\ell_1, r, \ell_2), (1, 1)) = \frac{\pi^2}{\ell_1^2} +
    \frac{\pi^2}{\ell_2^2} + O(\ell_1^{-6} + \ell_2^{-6}) \text{.}
  \end{equation}
\end{proposition}
\begin{proof}
  Obviously, the energy of this system is greater than the energy of
  the system without interactions that is given by the main term of
  \eqref{eq:TwoNeighborsEnergyExpansion}. Taking the ground state of a
  non-interacting system as a test function and using
  Lemma~\ref{lem:UijijVois} to estimate the quadratic form of the
  interaction potential, gives the upper bound and, thus, completes
  the proof.
\end{proof}
\subsection{The proof of Lemma~\ref{le:29}}
\label{sec:proof-lemma-29}
Recall that $E^{U}_{q,n}$ denotes the $n$-th eigenvalue of $\D
-\sum_{l=1}^{q} \frac{d^2}{dx_l^2}+\sum_{1\leq k\leq l\leq
  q}U(x_k-x_l)$ acting on $\D\bigwedge_{l=1}^{q}
L^2([0,\ell])$. Rescaling as in section~\ref{sec:proof-prop-refpr}, we
need to study the case $\ell=1$ and prove that, in this case, there
exists $C>1$ such that, for $n\geq2$ and $U^\ell$ given
by~\eqref{eq:26}, one has
\begin{equation}
  \label{eq:225}
  E^{U^\ell}_{q,n}\geq  E^{U^\ell}_{q,1}+\frac1C.
\end{equation}
Indeed in Lemma~\ref{le:29}, the length $\ell$ is assumed to be less
than $3\ell_\rho$.\\
As $q\leq3$, the same computations as in the beginning of
section~\ref{sec:proof-prop-refpr} show that $E^{U^\ell}_{q,1}$
satisfies, for some $C>1$, for $\ell$ large,
\begin{equation}
  \label{eq:226}
  E^{U^\ell}_{q,1}\leq
  E^{0}_{q,1}+\langle\phi_0,U^\ell\phi_0\rangle\leq
  E^{0}_{q,1}+\frac{C}{\ell}.
\end{equation}
On the other hand, for some $C>1$, one has
\begin{equation*}
  E^{U^\ell}_{q,n}\geq  E^{0}_{q,n}\geq E^{0}_{q,1}+\frac2{C}.
\end{equation*}
Plugging~\eqref{eq:226} into this immediately yields~\eqref{eq:225}
and completes the proof of Lemma~\ref{le:29}.\qed

%%% Local Variables: 
%%% mode: latex
%%% TeX-master: "PiecesModelGroundState"
%%% ispell-local-dictionary: "american"
%%% End: 

% ========
% APPENDIX
% ========
\appendix

\section{The statistics of the pieces}
\label{sec:auxil-results-calc}
In this appendix, we prove most of the results on the statistics of
the pieces stated in section~\ref{sec:analys-one-part}.
\subsection{Facts on the Poisson process}
\label{sec:facts-poiss-proc}
Let $\Pi$ be the support of $d\mu(\omega)$, the Poisson
process of intensity $1$ on $\R_+$ (see
section~\ref{sec:introduction}).  Let $\Pi\cap[0,L]=\{x_i;\ 1\leq
i\leq m(\omega)-1\}$ (where $x_i<x_ {i+1}$).  Then,
\begin{equation}
  \label{eq:distrNumPoints}
  \bbP(\#\Pi\cap[0,L]=k) = e^{-L} \frac{L^k}{k!}, \quad 
  k \in \bbN \text{.}
\end{equation}
The following large deviation principle is well known (and easily
checked): for any $\beta\in(1/2,1)$, one has
\begin{equation}
  \label{eq:distrNomInterv}
  \bbP(|\#(\Pi\cap[0,L])-L|\geq L^\beta) = O(L^{-\infty}).
\end{equation}
The points $(x_i)_{1\leq i\leq m(\omega)-1}$ partition the interval
$[0, L]$ in $m(\omega)$ pieces of lengths $\Delta_i$. \\
For $L>e^{e^2}$, one has
\begin{equation*}
  \begin{split}
  \bbP(\exists i;\ |\Delta_i|\geq\log L\log\log L)&\leq
  \bbP(\exists n\in[0,L]\cap\N;\ \#[\Pi\cap(n+[0,\log L\log\log
  L/2])]=0)  \\&\leq L e^{-\log L\log\log L/2}=O(L^{-\infty}).
  \end{split}
\end{equation*}
This proves Proposition~\ref{pro:3}.
\subsection{The proof of Proposition~\ref{prop:IntervStatistics}}
\label{sec:proof-prop-refpr-1}
Consider the partition of $[0,L]$ into pieces (see
section~\ref{sec:introduction}). For $a,b$ both non negative, let now
$X_{[0,L]}$ to be the number of pieces of length in $[a,a+b]$. We
first compute the expectation of $X_{[0,L]}/L$, that is, prove
\begin{Pro}
  \label{pro:5}
  For $L\geq a+b$, one has
  \begin{equation*}
    \esp\left[\frac{X_{[0,L]}}L\right]=
    e^{-a} (1-e^{-b})+\frac{e^{-a}((a+b)e^{-b}-a)}{L}=e^{-a}
    \left(1-\frac{a}{L}\right)-e^{-a-b}
    \left(1-\frac{a+b}{L}\right).
  \end{equation*}
\end{Pro}
\begin{proof}
  Let $\Pi$ be the support of the support of $d\mu(\omega)$, the
  Poisson process of intensity $1$ on $\R_+$ (see
  section~\ref{sec:introduction}). Then, one has
  \begin{equation*}
    X_{[0,L]}=\sum_{X\in\Pi}G(\Pi\cap[0,X))
  \end{equation*}
  where the set-functions $G$ is defined as
  \begin{equation}
    \label{eq:271}
    G(\Pi\cap[0,X))=\begin{cases}1&\text{ if the distance from }X\text{
        to the right most point }\\&\text{ in }\{0\}\cup(\Pi\cap[0,X))\text{
        belongs to }[a,a+b],\\0&\text{ if not.} \end{cases}
  \end{equation}
  The Palm formula (see e.g.~\cite[Lemma 2.3 ]{MR2253162}) yields
  \begin{equation*}
      \esp(X_{[0,L]})=\int_{0\leq x\leq L}\esp\left[
      G(\Pi\cap[0,x))\right]dx.
  \end{equation*}
  Now, let $\mathcal{E}$ be an exponential random variable with
  parameter $1$. As the Poisson point process has independent
  increments, one easily checks that
  \begin{equation}
    \label{eq:59}
    \esp\left[G(\Pi\cap[0,x))\right]=\pro\left(\min(x,\mathcal{E})\in[a,a+b]\right)=
    \begin{cases}
      e^{-a}\left(1-e^{-b}\right)&\text{ if }x\geq a+b,\\
      e^{-a}&\text{ if }x\in[a,a+b],\\
      0&\text{ if }x\leq a,
    \end{cases}
  \end{equation}
  Hence,
  \begin{equation*}
    \esp(X_{[0,L]})=e^{-a}\left(1-e^{-b}\right) \int_{0\leq x\leq
      L}dx+e^{-a-b}\int_a^{a+b}dx-e^{-a}\left(1-e^{-b}\right)\int_0^adx
    =e^{-a}(1-e^{-b})L-R
  \end{equation*}
  where 
  \begin{equation}
    \label{eq:75}
    R=e^{-a}((a+b)e^{-b}-a).  
  \end{equation}
  This completes the proof of Proposition~\ref{pro:5}.
\end{proof}
\noindent \noindent Let us now prove
Proposition~\ref{prop:IntervStatistics}. Therefore, set $M:=e^{-a}
(1-e^{-b})$ and partition $[0,L]=\cup_{j=1}^J[j\ell,(j+1)\ell]$ so
that $J\asymp L^\nu$ and $\ell\asymp L^{1-\nu}$ for some $\nu\in(0,1)$
to be fixed. As $(a,b)=(a_L,b_l)\in[0,\log L\cdot\log\log L]^2$, one
then has
\begin{equation}
  \label{eq:66}
  \left|X_{[0,L]}-\sum_{j=1}^JX_{[j\ell,(j+1)\ell]}\right|\leq 2J.
\end{equation}
Moreover, the random variables
$(\ell^{-1}X_{[j\ell,(j+1)\ell]})_{1\leq j\leq J}$ are independent
sub-exponential random variables. Indeed, $X_{[0,L]}$ is clearly
bounded by $\#\Pi\cap[0,L]$, the number of points the Poisson process
puts in $[0,L]$ and $L^{-1}\#\Pi\cap[0,L]$ has a Poisson law with
parameter 1. We want to use the Bernstein inequality (see
e.g.~\cite[Proposition 5.16]{MR2963170}). To estimate
$\D\left\|\ell^{-1}X_{[j\ell,(j+1)\ell]}\right\|_{\Psi_1}$ (see
e.g.~\cite[Definition 5.13]{MR2963170}), we use this bound and the
Stirling formula to get, for $p\geq1$,
\begin{equation*}
  \begin{split}
    \esp\left(\left|X_{[j\ell,(j+1)\ell]}\right|^p\right)&\leq
    e^{-\ell}\sum_{k\geq1}\frac{k^p\,\ell^k}{k!}  \leq
    e^{-\ell}\sum_{k=1}^{2p-1}\frac{k^p\,\ell^k}{k!}+
    e^{-\ell}\sum_{k\geq 2p}\frac{k^p\,\ell^k}{k!}\\
    &\leq (2p)^p+ e^{-\ell}\sum_{k\geq
      2p}\frac{k^p\,\ell^p}{k\cdots(k-p+1)}
    \frac{\ell^{k-p}}{(k-p)!}\\&\leq (2p)^p+ \ell^p\max_{k\geq
      p}\frac{(k+p)^p\,k!}{(k+p)!}\leq (2p)^p+(e\ell)^p.
  \end{split}
\end{equation*}
Hence, for $\ell\geq1$,
\begin{equation*}
  \begin{split}
    \left\|\ell^{-1}X_{[j\ell,(j+1)\ell]}\right\|_{\Psi_1}&=\frac1{\ell}
    \left\|X_{[j\ell,(j+1)\ell]}\right\|_{\Psi_1}=\frac1{\ell}
    \sup_{p\geq1}\frac1p\sqrt[p]{\esp\left(\left|X_{[j\ell,(j+1)\ell]}
        \right|^p\right)}\\&\leq \sup_{p\geq1}
    \sqrt[p]{\frac{2^p}{\ell^p}+\frac{e^p}{p^p}}\leq
    \frac{2}{\ell}+e\leq 2e.
  \end{split}
\end{equation*}
Thus, the Bernstein inequality, estimate~\eqref{eq:66} and
Proposition~\ref{pro:1} yield that there exists $\kappa>0$
(independent of $a,b$) such that, for $\alpha=\alpha(L)\geq
2(R+2)/\ell$ (here, $R$ is given by~\eqref{eq:75}), one has
\begin{equation*}
  \begin{split}
    \pro\left(\left|\frac{X_{[0,L]}}L-M\right|\geq \alpha \right)
    &\leq \pro\left(\left|\sum_{j=1}^J\frac{X_{[j\ell,(j+1)\ell]}-
          \esp[X_{[j\ell,(j+1)\ell]}]}{\ell}\right| \geq
      J\left(\alpha-\frac{R+2}{\ell}\right)\right)\\&\leq 2e^{-\kappa
      \alpha^2 J}.
  \end{split}
\end{equation*}
To obtain Proposition~\ref{prop:IntervStatistics2}, it now suffices to
take $\alpha=L^{\beta-1}$ and $(\beta,\nu)\in(0,1)$ such that $1-\beta<1-\nu$
and $2(\beta-1)+\nu>0$; this requires $\beta>2/3$.\\
The proof of Proposition~\ref{prop:IntervStatistics} is
complete.\qed\vskip.2cm\noindent
\subsection{The proof of Propositions~\ref{prop:IntervStatistics2}
  and~\ref{prop:NeighborssStatistics}}
\label{sec:proof-prop:IntervStatistics2}
For any $a,b,c,d,f,g$ all non negative, define now $X_{[0,L]}$ to be the
number of pairs of pieces such that
\begin{itemize}
\item the length of the left most piece is contained in $[a,a+b]$,
\item the length of the right most piece is contained in $[c,c+d]$,
\item the distance between the two pieces belongs to $[g,g+f]$.
\end{itemize}
Again, we first compute the expectation of $X_{[0,L]}/L$, that is, prove
\begin{Pro}
  \label{pro:1}
  For $L\geq a+b+c+d+f+g$, one has
  \begin{equation}
    \label{eq:65}
    \esp\left[\frac{X_{[0,L]}}L\right]=
    f\,e^{-a-c} (1-e^{-b})(1-e^{-d})+\frac{R_L}{L}\quad\text{ where
    }|R_L|\leq f e^{-a-c}.
  \end{equation}
\end{Pro}
\begin{proof}
  Recall that $\Pi$ denotes the support of the support of
  $d\mu(\omega)$, the Poisson process of intensity $1$ on
  $\R_+$. Then, one can rewrite
  \begin{equation*}
    X_{[0,L]}=\sum_{\substack{(X,Y)\in\Pi^2\\X<Y}}\car_{g\leq Y-X\leq
      g+f}G(\Pi\cap[0,X))\,H(\Pi\cap(Y,L])
  \end{equation*}
  where the set-functions $G$ and $H$ have been defined respectively
  by~\eqref{eq:271} and 
  \begin{equation}
    \label{eq:57}
    H(\Pi\cap(Y,L])=\begin{cases}1&\text{ if the distance from }Y\text{
        to the left most point}\\&\text{ in }\{L\}\cup(\Pi\cap(Y,L])\text{
        belongs to }[c,c+d],\\0&\text{ if not.} \end{cases}      
  \end{equation}
  The Palm formula, thus, yields
  \begin{equation*}
    \begin{split}
    \esp(X_{[0,L]})&=\int_{\substack{0\leq x,y\leq L\\g\leq y-x\leq g+f}}\esp\left[
     G(\Pi\cap[0,x))H(\Pi\cap(y,L])\right]dxdy\\&=
   \int_{\substack{0\leq x,y\leq L\\g\leq y-x\leq g+f}}\esp\left[
     G(\Pi\cap[0,x))\right]\esp\left[H(\Pi\cap(y,L])\right]dxdy
    \end{split}
  \end{equation*}
  as the random sets $\Pi\cap[0,x))$ and  $\Pi\cap(y,L])$ are
  independent.\\
  As in~\eqref{eq:59}, one checks that
  \begin{equation*}
    \esp\left[H(\Pi\cap(y,L])\right]=\pro\left(\min(L-y,\mathcal{E})\in[c,c+d]\right)=
    \begin{cases}
      e^{-c}\left(1-e^{-d}\right)&\text{ if }y\leq L-c-d,\\
      e^{-c}&\text{ if }y\in L-[c,c+d],\\
      0&\text{ if }y\geq L-c.
    \end{cases}
  \end{equation*}
  Hence,
  \begin{equation*}
    \begin{split}
    \esp(X_{[0,L]})&=e^{-a-c}\left(1-e^{-d}\right)\left(1-e^{-b}\right)
    \int_{\substack{0\leq x,y\leq L\\g\leq y-x\leq g+f}}dxdy+R_1
    \\&=f\,e^{-a-c} (1-e^{-b})(1-e^{-d})L+R_2
    \end{split}
  \end{equation*}
  where, respectively, $R_1\leq e^{-a-c}$ and
  \begin{equation}
    \label{eq:96}
    R_2\leq R:=e^{-a-c} (1 + f^2 + f g).
  \end{equation}
  This completes the proof of Proposition~\ref{pro:1}.
\end{proof}
\noindent Let us now prove
Proposition~\ref{prop:IntervStatistics2}. We want to go along the same
lines as in the proof of
Proposition~\ref{prop:IntervStatistics}. Therefore, we set
$M:=f\,e^{-a-c} (1-e^{-b})(1-e^{-d})$ and partition
$[0,L]=\cup_{j=0}^{J}[j\ell,(j+1)\ell]$ so that $J\asymp L^\nu$ and
$\ell\asymp L^{1-\nu}$ for some $\nu\in(0,1)$ to be fixed. For the
same reasons as before, the random variables
$(\ell^{-1}X_{[j\ell,(j+1)\ell]})_{1\leq j\leq
  J}$ are independent sub-exponential random variables.\\
We now need a replacement for~\eqref{eq:66}. Therefore, we set
\begin{equation}
  \label{eq:273}
  r:=1+a+b+c+d+f+g  
\end{equation}
and, for $0\leq j\leq J$, we let
\begin{itemize}
\item $Y_j$ be the number of pieces in the interval $(j+1)\ell+[-r,0]$
  of length in $[a,a+b]$ ,
\item $Z_j$ be the number of pieces in the interval $j\ell+[0,r]$ of
  length in $[c,c+d]$.
\end{itemize}
Then, we have
\begin{equation}
  \label{eq:60}
  -K_a\sum_{j=0}^JY_j-K_c\sum_{j=0}^JZ_j
  \leq X_{[0,L]}-\sum_{j=0}^JX_{[j\ell,(j+1)\ell]} 
  \leq K_a\sum_{j=0}^JY_j+K_c\sum_{j=0}^JZ_j
\end{equation}
where we have set 
\begin{equation}
  \label{eq:95}
  K_a:=1+\frac{f+g}a\quad\text{and}\quad K_c=1+\frac{f+g}c.
\end{equation}
Indeed, if a pair of pieces counted by $X_{[0,L]}$ does not have any
of its intervals in any of the $(j\ell+[-r,r])_{1\leq j\leq J}$, then
the convex closure of the pair is inside some $j\ell+[0,\ell]$, thus,
the pair is counted by $X_{[j\ell,(j+1)\ell]}$. This yields the upper
bound in~\eqref{eq:60} as, any given interval is the left
(resp. right) most interval for at most $1+(f+g)/c$
(resp. $1+(f+g)/a$) pairs satisfying both the requirements on lengths
and distance. The lower bound is obtained in the same way.\\
For $L$ sufficiently large, the random variables $(Y_j)_{1\leq j\leq
  J}$ and $(Z_j)_{1\leq j\leq J}$ are i.i.d. sub-exponen\-tial. Thus,
applying the Bernstein inequality as in the proof of
Proposition~\ref{prop:IntervStatistics} yields that, for some constant
$\kappa>0$ (independent of $(a,b,c,d,f,g)$) and $\beta\in(2/3,1)$,
with probability $1-O(J^{-\infty})=1-O(L^{-\infty})$, one has
\begin{equation}
  \label{eq:69}
  \sum_{j=1}^JY_j\leq \kappa J
  (e^{-a}+J^{\beta-1})r\quad\text{and}\quad
  \sum_{j=0}^{J-1}Z_j\leq \kappa J (e^{-c}+J^{\beta-1})r;
\end{equation}
Now, we can estimate
$\left\|\ell^{-1}X_{[j\ell,(j+1)\ell]}\right\|_{\Psi_1}$ as in the
proof of Proposition~\ref{prop:IntervStatistics}. Thus, the Bernstein
inequality and Proposition~\ref{pro:1} yield that, for some $\kappa$
(independent of $(a,b,c,d,f,g)$), for $\nu\in(2/3,1)$ and $\ell\asymp
L^{1-\nu}$, with probability $1-O(L^{-\infty})$, one has
\begin{equation*}
  \left|\sum_{j=0}^J\frac{X_{[j\ell,(j+1)\ell]}}{\ell}-
    J\,M\right| \leq \kappa \frac{R_L\,J}{\ell}.
\end{equation*}
Taking~\eqref{eq:60} and~\eqref{eq:69} into account, we get that, for
some $\kappa>0$ (independent of $(a,b,c,d,f,g)$), with probability
$1-O(L^{-\infty})$, one has
\begin{equation*}
  \left|\frac{X_{[0,L]}}L- M\right|\leq \kappa
  \frac{R+(K_a e^{-a}+ K_c e^{-c}+(K_a+K_c)J^{\beta-1})r}{\ell}.
\end{equation*}
This proves~\eqref{eq:97} where the constants are given by
\begin{equation}
  \label{eq:272}
  R(a,b,c,d,f,g)=\kappa r\left(
    R+K_a e^{-a}+ K_c e^{-c}\right)\quad\text{and}\quad
  K(a,c,f,g)=(K_a+K_c)r
\end{equation}
(see~\eqref{eq:96},~\eqref{eq:273} and~\eqref{eq:95}.)\\
The proof of Proposition~\ref{prop:IntervStatistics2} is
complete.\qed\vskip.2cm\noindent
The proof of  Proposition~\ref{prop:NeighborssStatistics} is identical
to that of  Proposition~\ref{prop:IntervStatistics2}: it suffices to
take $b=d=+\infty$.
\subsection{The proofs of Proposition~\ref{prop:NeighborssStatistics2}}
\label{sec:proof-prop-refpr-2}
This proof is essentially identical to that of
Proposition~\ref{prop:IntervStatistics2}. Let us just say a word about
the differences.\\
For $\ell,\ell',\ell'',d>0$, let now $X_{[0,L]}$ to be the number
triplets of pieces at most at a distance $d$ from each other such that
\begin{itemize}
\item the left most piece longer than $\ell$,
\item the middle piece longer than $\ell'$,
\item the right most piece longer than $\ell''$.
\end{itemize}
Then, one has
\begin{equation*}
  X_{[0,L]}=\sum_{\substack{(X,Y,W,Z)\in\Pi^4\\X<Y<W<Z}}
  \car_{\substack{0<Y-X\leq d\\ l'\leq W-Y\\0<Z-W\leq d}}
  G(\Pi\cap[0,X))\,K(\Pi\cap(Y,W))\,H(\Pi\cap(Z,L])
\end{equation*}
where the set-functions $G$ and $H$ have been defined as
\begin{equation*}
  \begin{aligned}
    G(\Pi\cap[0,X))&=\begin{cases}1&\text{ if the distance from
      }X\text{ to the right most point }\\&\text{ in
      }\{0\}\cup(\Pi\cap[0,X))\text{
        belongs to }[l,+\infty),\\0&\text{ if not,} \end{cases}\\
    K(\Pi\cap(Y,W))&=\begin{cases}1&\text{ if }\Pi\cap(Y,W)=\emptyset
      ,\\0&\text{ if not,} \end{cases}\\
    H(\Pi\cap(Z,L])&=\begin{cases}1&\text{ if the distance from
      }Z\text{ to the left most point}\\&\text{ in
      }\{L\}\cup(\Pi\cap(Z,L])\text{ belongs to
      }[l'',+\infty),\\0&\text{ if not.} \end{cases}
  \end{aligned}
\end{equation*}
Following the proof of Proposition~\ref{pro:1}, one proves
\begin{Pro}
  \label{pro:6}
  For $L$ sufficiently large, one has
  \begin{equation*}
    \esp\left[\frac{X_{[0,L]}}L\right]=
    d^2e^{-\ell-\ell'-\ell''}+\frac{R_L}{L}\quad\text{ where
    }|R_L|\leq d^2e^{-\ell-\ell'-\ell''}.
  \end{equation*}
\end{Pro}
\noindent One then derives Proposition~\ref{prop:IntervStatistics}
from Proposition~\ref{pro:6} in the same way as
Proposition~\ref{prop:IntervStatistics2} was derived from
Proposition~\ref{pro:1}.
\subsection{Proof of
  Proposition~\textup{\ref{prop:IDSandFermiEnergy}}}
\label{sec:proof-prop-text}
First of all, let us note that a piece of length $l$ in $[k \ell_E, (k
+ 1) \ell_E)$ generates exactly $k$ energy levels that do not exceed
$E$.  To count the energies less than $E$, we are only interested in
intervals of length $l$ larger than $\ell_E$.  Other intervals do not
generate any energy levels we are interested in. Thus, by Proposition
\ref{prop:IntervStatistics}, for $\beta\in(2/3,1)$, we obtain that
with probability $1 - O(L^{-\infty})$, the number of intervals
generating $k$ energy levels below energy $E$ is
\begin{equation}
  \label{eq:Ninterv1}
  L (e^{-k\ell_E}-e^{-(k+1)\ell_E})+L^{\beta}R_L\quad\text{ where }|R_L|\leq3
\end{equation} 
where $O(\cdot)$ is uniform in $k$.\\
Let $m_L=\log L\cdot\log\log L$. By Proposition~\ref{pro:3}, with
probability $1-O(L^{-\infty})$, for $L$ large, one computes
\begin{equation*}
  \begin{split}
    N_L^D(E)&=L^{-1}\sum_{k=1}^{\left[m_L/\ell_E\right]} k\cdot
    L(e^{-k\ell_E} -e^{-(k+1)\ell_E})+ m_LL^{-1+\beta}R_L\quad\text{where
    }\quad|R_L|\leq\frac{1}{\ell_E}
    \\&=\sum_{k=1}^{\left[m_L/\ell_E\right]}e^{-k\ell_E}
    -\frac{\left[m_L/\ell_E\right]}{e^{(\left[m_L/\ell_E\right]+1)\ell_E}}
    + m_LL^{-1+\beta}R_L \\
    &=\sum_{k=1}^{+\infty}e^{-k\ell_E}+ m_LL^{-1+\beta}(R_L+1)
    =\frac{e^{-\ell_E}}{1- e^{-\ell_E}}+ m_LL^{-1+\beta}(R_L+2).
  \end{split}
\end{equation*}
Thus, decreasing $\beta$ above somewhat, with probability
$1-O(L^{-\infty})$, for $L$ sufficiently large, one has
\begin{equation}
  \label{eq:densEtatsPoissonModel2}
  \left|N_L^D(E)-\frac{e^{-\ell_E}}{1-
      e^{-\ell_E}}\right|\leq L^{-1+\beta}.
\end{equation}
This proves (\ref{eq:NLifschitzTail}). Using the fact that $E\mapsto
N_L^D(E)$ is monotonous and the Lipschitz continuity of $E\mapsto
N(E)$,~\eqref{eq:densEtatsPoissonModel2} yields that, for $E_0>0$,
with probability $1-O(L^{-\infty})$, for $L$ sufficiently large, one
has
\begin{equation}
  \label{eq:densEtatsPoissonModel1}
  \sup_{E\in[0,E_0]}\left|N_L^D(E)-\frac{e^{-\ell_E}}{1-
      e^{-\ell_E}}\right|\leq L^{-1+\beta}.
\end{equation}
The formulas
(\ref{eq:FermiEnergyExpression}) and (\ref{eq:ellRho}) for the Fermi
energy and the Fermi length follow trivially. This completes the
proof of Proposition~\textup{\ref{prop:IDSandFermiEnergy}}.
\section{A simple lemma on trace class operators}
\label{sec:simple-lemma-trace}
\noindent The purpose of the present section is to prove
\begin{Le}
  \label{le:15}
  Pick $(\mathcal{H},\langle\cdot,\cdot\rangle)$ a separable Hilbert
  space and $(Z,\mu)$ a measured space with $\mu$ a positive
  measure. Consider a weakly measurable mapping $z\in Z\to
  T(z)\in\mathfrak{S}_1(\mathcal{H})$. Here,
  $\mathfrak{S}_1(\mathcal{H})$ denotes the trace class operators in
  $\mathcal{H}$, the  trace class norm being denoted by $\|\cdot\|_{\text{tr}}$.\\
  Assume 
  \begin{equation}
    \label{eq:147}
    \int_Z\|T(z)\|_{\text{tr}}d\mu(z)<+\infty.
  \end{equation}
  Then, the integral $\D T:=\int_ZT(z)d\mu(z)$ converges weakly and
  defines a trace class operator that satisfies
  \begin{equation}
    \label{eq:148}
    \|T\|_{\text{tr}}=\left\|\int_Z T(z)d\mu(z)\right\|_{\text{tr}}\leq
    \int_Z\|T(z)\|_{\text{tr}}d\mu(z). 
  \end{equation}
\end{Le}
\begin{proof}
  By assumption, for $(\varphi,\psi)\in\mathcal{H}^2$, one has
  $z\to\langle T(z)\varphi,\psi\rangle$ is measurable and bounded by
  $z\to\|T(z)\|_{\text{tr}}\|\varphi\|\|\psi\|$ which by~\eqref{eq:147} is
  integrable. It, thus, is integrable and one has
  \begin{equation*}
    \left|\int_Z\langle T(z)\varphi,\psi\rangle d\mu(z) \right|\leq
    \int_Z\left|\langle T(z)\varphi,\psi\rangle\right| d\mu(z) \leq
    \int_Z\|T(z)\|_{\text{tr}}d\mu(z)\,\|\varphi\|\,\|\psi\|.
  \end{equation*}
  Thus, the operator $\D T:=\int_ZT(z)d\mu(z)$ is well defined by
  \begin{equation*}
    \langle T\varphi,\psi\rangle :=\int_Z\langle
    T(z)\varphi,\psi\rangle d\mu(z).
  \end{equation*}
  and bounded.\\
  Let us prove that it is trace class and
  satisfies~\eqref{eq:148}. Let $(\varphi_n)_{n\geq1}$ be an
  orthonormal basis of $\mathcal{H}$. Then,
  \begin{equation*}
    |\langle T\varphi_n,\varphi_n\rangle|\leq
    \int_Z \left|\langle T(z)\varphi_n,\varphi_n\rangle
    \right| d\mu(z).
  \end{equation*}
  Thus, 
  \begin{equation*}
    \sum_{n=1}^N  |\langle T\varphi_n,\varphi_n\rangle|\leq
    \int_Z \left(\sum_{n=1}^N\left|\langle T(z)\varphi_n,\varphi_n\rangle
      \right|\right) d\mu(z)\leq\int_Z \|T(z)\|_{\text{tr}} d\mu(z).
  \end{equation*}
  Taking $N\to+\infty$ proves that, for any orthonormal basis of
  $\mathcal{H}$, say,  $(\varphi_n)_{n\geq1}$, one has
  \begin{equation*}
    \sum_{n=1}^{+\infty}  |\langle T\varphi_n,\varphi_n\rangle|
    \leq\int_Z \|T(z)\|_{\text{tr}} d\mu(z)<+\infty.
  \end{equation*}
  Thus, $T$ is trace class (see e.g.~\cite{MR85e:46002}) and
  satisfies~\eqref{eq:148}. This completes the proof of
  Lemma~\ref{le:15}.
\end{proof}
\section{Anti-symmetric tensors: the projector on anti-symmetric
  functions}
\label{sec:proj-totally-antisym}
Pick $\Psi\in L^2(\Lambda^n)$ and let $\Pi_n^{\wedge}:\
L^2(\Lambda^n)\to \bigwedge^n L^2(\Lambda)$ be the orthogonal
projector on totally anti-symmetric function. Then,
\begin{equation*}
  (\Pi_n^{\wedge}\Psi)(x)=\frac1{n!}\sum_{\substack{\sigma\text{
        permutation}\\\text{of }\{1,\cdots,n\}}}\text{sgn}\,\sigma
  \cdot\Psi(\sigma x)
\end{equation*}
where, for $x=(x_1,\cdots,x_n)$, $\sigma
x=(x_{\sigma(1)},\cdots,x_{\sigma(n)})$ and sgn
$\sigma$ is the signature of the permutation $\sigma$.\\
Hence, if $n=Q_1+\cdots+Q_m$ and, for $1\leq j\leq m$, $\varphi_j\in
\bigwedge^{Q_j}L^2(\Delta_j)$, we define
\begin{equation}
  \label{eq:243}
  \left(\prod_{j=1}^m\left\|\varphi^j\right\|\right)^{-1}
  \bigwedge_{j=1}^m\varphi^j:=   \left\|\Pi_n^{\wedge} 
    \left(\bigotimes_{j=1}^m\varphi^j\right)  \right\|^{-1}
  \Pi_n^{\wedge}\left(\bigotimes_{j=1}^m\varphi^j\right)
\end{equation}
and compute
\begin{equation*}
  \begin{split}
    \Pi_n^{\wedge} \left(\bigotimes_{j=1}^m\varphi^j\right)&=
    \frac1{n!}\sum_{\substack{\sigma\text{ permutation}\\\text{of
        }\{1,\cdots,n\}}}\text{sgn
    }\sigma\,\left(\bigotimes_{j=1}^m\varphi^j\right)(\sigma x)
    \\&=\frac1{n!}\sum_{\substack{\sigma\text{ permutation}\\\text{of
        }\{1,\cdots,n\}}}\text{sgn
    }\sigma\,\left(\prod_{j=1}^m\varphi^j(x_{\sigma(\mathcal{Q}_j)})\right)
  \end{split}
\end{equation*}
where
\begin{gather*}
  x_{\sigma(\mathcal{Q}_j)}=(x_{\sigma(Q_1+\cdots+Q_{j-1}+1)},\cdots,
  x_{\sigma(Q_1+\cdots+Q_{j-1}+Q_j)}),\\
  \mathcal{Q}_j=\{Q_1+\cdots+Q_{j-1}+1,\cdots,Q_1+\cdots+Q_{j-1}+Q_j\}.
\end{gather*}
Thus,
\begin{equation*}
  \begin{split}
    n!\cdot\Pi_n^{\wedge} \left(\bigotimes_{j=1}^m\varphi^j\right)&
    =\sum_{\substack{|A_j|=Q_j,\
        \forall1\leq j\leq m\\A_1\cup\cdots\cup A_m=\{1,\cdots,n\}\\
        A_j\cap A_{j'}=\emptyset\text{ if
        }j\not=j'}}\sum_{\substack{\sigma\text{ permutation}\\\text{of
        }\{1,\cdots,n\}\\\text{s.t. }\forall j,\
        \sigma(\mathcal{Q}_j)=A_j}}\text{sgn
    }\sigma\,\left(\prod_{j=1}^m\varphi^j(x_{\sigma(\mathcal{Q}_j)})\right)
    \\&=\sum_{\substack{|A_j|=Q_j,\
        \forall1\leq j\leq m\\A_1\cup\cdots\cup A_m=\{1,\cdots,n\}\\
        A_j\cap A_{j'}=\emptyset\text{ if
        }j\not=j'}}\left(\sum_{\substack{\sigma\text{
            permutation}\\\text{of
          }\{1,\cdots,n\}\\\text{s.t. }\forall j,\
          \sigma(\mathcal{Q}_j)=A_j}}\left(\text{sgn
        }\sigma\,\prod_{j=1}^m \text{sgn
        }\sigma_{|\mathcal{Q}_j}\right)
      \left(\prod_{j=1}^m\varphi^j(x_{A_j})\right)\right)
      \\&=\prod_{j=1}^mQ_j!\sum_{\substack{|A_j|=Q_j,\
          \forall1\leq j\leq m\\A_1\cup\cdots\cup A_m=\{1,\cdots,n\}\\
          A_j\cap A_{j'}=\emptyset\text{ if
          }j\not=j'}}\varepsilon(A_1,\cdots,A_m)\,
      \left(\prod_{j=1}^m\varphi^j(x_{A_j})\right)
  \end{split}
\end{equation*}
where we recall that $\varepsilon(A_1,\cdots,A_m)$ is the signature of
$\sigma(A_1,\cdots,A_m)$ the unique permutation of $\{1,\cdots,n\}$
such that, if $A_j=\{a_{ij};\ 1\leq i\leq Q_j,\ a_{i_1j} < a_{i_2j}
\text{ for } i_1 < i_2\}$ for $1\leq j\leq m$
then $\sigma(a_{ij})=Q_1+\cdots+Q_{j-1}+i$.\\
As $\Delta_j\cap\Delta_k=\emptyset$ if $j\not=k$, one has
\begin{equation*}
  \begin{split}
    \left\|\sum_{\substack{|A_j|=Q_j,\
          \forall1\leq j\leq m\\A_1\cup\cdots\cup A_m=\{1,\cdots,n\}\\
          A_j\cap A_{j'}=\emptyset\text{ if
          }j\not=j'}}\varepsilon(A_1,\cdots,A_m)
      \left(\prod_{j=1}^m\varphi^j(x_{A_j})\right)\right\|^2&=
    \prod_{j=1}^m\left\|\varphi^j\right\|^2\sum_{\substack{|A_j|=Q_j,\
        \forall1\leq j\leq m\\A_1\cup\cdots\cup A_m=\{1,\cdots,n\}\\
        A_j\cap A_{j'}=\emptyset\text{ if }j\not=j'}}1\\&=
    \frac{n!}{\prod_{j=1}^mQ_j!}
    \prod_{j=1}^m\left\|\varphi^j\right\|^2.
  \end{split}
\end{equation*}
Hence, by~\eqref{eq:243}, we get
\begin{equation}
  \label{eq:267}
  \left(\bigwedge_{j=1}^m\varphi^j\right)(x)=
  \sqrt{\frac{\prod_{j=1}^mQ_j!}{n!}}\sum_{\substack{|A_j|=Q_j,\
      \forall1\leq j\leq m\\A_1\cup\cdots\cup A_m=\{1,\cdots,n\}\\
      A_j\cap A_{j'}=\emptyset\text{ if
      }j\not=j'}}\varepsilon(A_1,\cdots,A_m)
  \left(\prod_{j=1}^m\varphi^j(x_{A_j})\right).
\end{equation}
%

%%% Local Variables: 
%%% mode: latex
%%% TeX-master: "PiecesModelGroundState"
%%% ispell-local-dictionary: "american"
%%% End: 

%=====================================================================
%========================= BIBLIO ====================================
%=====================================================================
% \bibliographystyle{alpha}
% \bibliography{biblio}

\def\cprime{$'$} \def\cydot{\leavevmode\raise.4ex\hbox{.}} \def\cprime{$'$}

\end{document}